%% file: main.tex
\documentclass[format=acmsmall, screen, nonacm]{acmart}
\AtBeginDocument{%
  }

\usepackage[utf8]{inputenc}
\usepackage[T1]{fontenc}
\usepackage{microtype}

\usepackage{amsmath}
\usepackage{amsthm}
\usepackage{stmaryrd}
\usepackage{xspace}
\usepackage{xparse} 

\usepackage{stackengine}
\stackMath

\usepackage{tikz}
\usetikzlibrary{graphs, positioning, fit, shapes, calc,arrows.meta}

\usepackage{ebproof}

\usepackage{thmtools}
\usepackage{mathtools}
\usepackage{xcolor}

\usepackage{hyperref}
\usepackage[capitalise]{cleveref}

\usepackage{versions}

\excludeversion{SHORT}
\includeversion{LONG}

\hypersetup{%
    linktoc=all,     
    colorlinks=false,
    linkbordercolor=purple,
    citebordercolor=teal
}

\crefname{section}{Section}{Sections}
\crefname{figure}{Figure}{Figures}
\crefname{definition}{Definition}{Definitions}
\crefname{proposition}{Proposition}{Propositions}
\crefname{lemma}{Lemma}{Lemmas}
\crefname{theorem}{Theorem}{Theorems}
\crefname{corollary}{Corollary}{Corollarys}
\crefname{conjecture}{Conjecture}{Conjectures}
\crefname{remark}{Remark}{Remarks}
\crefname{example}{Example}{Examples}
\crefname{appendix}{Appendix}{Appendixs}

\input{Styles/Macro}

\input{Styles/Syntax}

\input{Styles/Types}
\input{Styles/Colors}

\begin{document}

\title{When Types Intersect and Effects Get Handled}

\author{Stefano Catozi}
\orcid{0009-0007-0829-5789}
\affiliation{%
  \institution{University Sorbonne Paris Nord}
  \city{Villetaneuse}
  \country{France}
}
\email{catozi@lipn.univ-paris13.fr}

\author{Ugo Dal Lago}
\orcid{0000-0001-9200-070X}
\affiliation{%
  \institution{University of Bologna}
  \city{Bologna}
  \country{Italy}
}
\email{ugo.dallago@unibo.it}

\author{Taro Sekiyama}
\orcid{0000-0001-9286-230X}
\affiliation{%
  \institution{National Institute of Informatics}
  \city{Tokyo}
  \country{Japan}
}
\email{tsekiyama@acm.org}

\begin{abstract}
  We introduce a novel intersection type system for a $\lambda$-calculus with algebraic effects and handlers. The system, inherently behavioral in nature, 
  enjoys the classical properties of intersection type systems, in particular subject reduction and \emph{expansion}. It thus characterizes the set of 
  terms whose evaluation process terminates and, at the same time, allows reducing the reachability problem to type inference. This new system, the first 
  with these features for a calculus with handlers, induces a system of simple types which, although not guaranteeing termination, is type sound and admits 
  a decidable HOMC problem, unlike similar type systems like Dal Lago and Ghyselen's \HEPCF.
\end{abstract}




\newcommand{\PCF}{\ensuremath{\mathsf{PCF}}}
\newcommand{\HEPCF}{\ensuremath{\mathsf{HEPCF}}}
\newcommand{\HE}{\ensuremath{\mathsf{HE}}}
\newcommand{\HEBI}{\ensuremath{\mathsf{HEBI}}}
\newcommand{\HEB}{\ensuremath{\mathsf{HEB}}}

\maketitle

\section{Introduction}
\label{sect:introduction}

\emph{Intersection types}, introduced by Coppo and Dezani almost fifty years
ago~\cite{coppo1978new}, were initially developed as a powerful tool for the analysis of the
semantics and meta-theory of the
$\lambda$-calculus. In particular, they made it possible to
characterize fundamental program properties such as termination (spelled out in different ways, like
strong normalization or head-normal-form normalization), and they
naturally led to the development of semantic models such as filter
models~\cite{barendregt1983filter,Bernadet/Lengrand_2011_CSL}. From this perspective, 
their role was from the outset quite
different from that of traditional typing, which provides a verification mechanism for
safety properties via type checking and type inference, the main reason being
that in their most general form, intersection type theories give rise to
\emph{undecidable} type inference and type checking problems.
This distinguished aspect of intersection types enables applications which are
unattainable via simple types.

For example, Kobayashi and Ong \cite{kobayashi2009types,Kobayashi/Ong_2009_LICS} made the fundamental observation that the
decidability of the \emph{higher-order model-checking} (HOMC)~\cite{Knapik/Nisinski/Urzyczyn_2001_TLCA,Knapik/Niwinski/Urzyczyn_2002_FoSSaCS,Ong_2006_LICS,kobayashi2013model} problem can be
understood as a form of exhaustive type inference in a system of uniform
intersection types \emph{refining} an underlying system of simple
types.
The key idea is that this refinement enjoys a \emph{finite model property}: each
simple type admits only a finite number of refinements, and this finiteness
makes exhaustive search possible, ultimately guaranteeing the decidability of
the underlying model checking problem. The aforementioned decidability results apply to calculi in the style of Plotkin's \PCF\ which, although equipped 
with higher-order functions and recursion, are simply typed and only feature finite 
base types. Relaxing either of these two conditions quickly leads to undecidability~\cite{Kobayashi/Igarashi_ESOP_2013}. 
As shown in numerous subsequent works, the decidability results mentioned above 
are also fragile with respect to extensions of the underlying language with effects,
e.g. probabilistic choice \cite{Kobayashi/Dal-Lago/Grellois_2019_LICS}. 
Particularly noteworthy are the extensions of \PCF\ with \emph{algebraic effects}~\cite{Plotkin/Power_2003_ACS} and \emph{handlers}~\cite{Plotkin/Pretnar_2009_ESOP}, which have been the subject of
extensive study over the past twenty years, and which allow the task of providing an
implementation for algebraic operations to be delegated to the programmer by
reifying \emph{delimited continuations}. In presence of effect handlers,
Dal Lago and Ghyselen~\cite{Dal-Lago/Ghyselen_2024_POPL} show that HOMC becomes undecidable. 
It should be noted that although the HOMC problem is traditionally studied when the 
underlying specification is \emph{any MSO formula}, \emph{reachability properties} are
sufficient to render the problem undecidable in presence of probabilistic choice or
effect handlers. In the rest of this paper we focus on reachability properties.

Intersection type systems have themselves been studied for extensions of the
$\lambda$-calculus equipped with specific forms of \emph{effects}, such as
nondeterministic choice~\cite{DeLiguoroPiperno}, probabilistic 
choice~\cite{BreuvartDalLago,DalLagoFaggianRonchi}, and global 
state~\cite{Kesner}. More recently,
Gavazzo et al.~\cite{GalalGavazzoTregliaVanoni} have undertaken the development of a generic framework
for defining intersection types for calculi with algebraic
effects, in which the monad interpreting the algebraic operations is internalized within the type system, thereby allowing the
behavior of terms to be reflected at the level of types. 
{But what if algebraic operations are not given meaning through monads, with their interpretation delegated to programmers?}
Studying intersection type systems for effect handlers seems to be a natural approach
to the semantic analysis of higher-order programs handling effects by themselves.
{Moreover, intersection types appear as a promising tool in the quest for 
	type systems for programs with effect handlers for which the HOMC problem remains decidable: 
	it would suffice to ensure that the finite model property remains true.}
To the authors’ knowledge, however, there is still no account of intersection types in
calculi with effect handlers.

It is therefore natural to ask whether a system of intersection types can be
defined for calculi with handlers, and the aim of this article is precisely to
provide a first piece of positive evidence in this direction.
A key point to define an intersection type system enjoying the desired properties,
such as a characterization of the set of terminating terms, is to provide a type
discipline that can faithfully express programs' behavior, \emph{e.g.},
functions are assigned an intersection type that express all the usages of the
functions' arguments in their bodies.
In the presence of algebraic effects and handlers, the behavior of a program
relies on \emph{what algebraic operation is called in what order} and \emph{how\
their continuations are used by the underlying effect handler}.
To embody this, the paper introduces a new form of behavioural types that
captures the \emph{computation tree} of algebraic operations produced by each
term, together with a specification of how the term itself expects the
environment to handle these operations.
This is fully in line with the spirit of intersection types, in that there is no
loss of information.

The resulting system enjoys the properties typically expected of an intersection type system. In particular, alongside the standard subject \emph{reduction} property, it also satisfies a property known as subject \emph{expansion}, which states that types are preserved backwards along reductions. This property, in turn, allows one to prove that the type system characterizes termination, in analogy with what happens in the $\lambda$-calculus. In addition to subject expansion, another non-trivial step in the proof of this result consists in showing that all typable terms are terminating, which requires an adaptation of Tait's reducibility technique~\cite{Tait_1967_JSL} to the introduced behavioral type system.

Remarkable as the above results are, they are not the only contributions of this work. We also present a study of intersection types viewed as refinements of simple types in the presence of effect handlers. This analysis provides a precise confirmation of the undecidability result of the HOMC problem for $\HEPCF$, a simply-typed calculus with algebraic effects and handlers~\cite{Dal-Lago/Ghyselen_2024_POPL}. We achieve this result by showing that when intersection types are regarded as refinements of simple types, the finite refinement property \emph{fails}: there are simple types which can be refined in infinitely many distinct ways. However, this is not the end of the story. The very structure of behavioural types as described above suggests the existence of a novel type system, itself behavioral in nature, and for which, instead, the finite model property \emph{does} hold, and consequently the HOMC problem is decidable.  This is not the first fragment of $\HEPCF$ where the HOMC problem is decidable~\cite{Sekiyama/Unno_2024_OOPSLA,Sekiyama/DalLago/Unno_2025_OOPLSA,Endo/Terauchi_2025_APLAS}, but the methodology employed in all the previous approaches consists in defining \emph{CPS transformations} from the calculus under consideration to a finitary \PCF\ ---a variant of \PCF\ only with finite base types, where the HOMC problem is decidable~\cite{Ong_2006_LICS}---and in proving that such transformations are adequate and type-preserving. However, it was still unclear what the deeper reason is behind the (un)decidability of the HOMC problem for $\HEPCF$, or similar calculi. Our intersection type system provides semantic evidence on that: whether a finite model exists or not. Our work is thus a further witness that intersection types are indeed a powerful and flexible tool, not only in a semantic setting, but also in verification.

This paper's contributions can be summarized as follows:
\begin{itemize}
	\item
	Introduction of $\HEBI$, an intersection type system for a calculus with algebraic effects and handlers. The system combines in a novel way the classical notion of intersection types—here taking the form of a \emph{set} of types—with a form of behavioral typing in which the algebraic operations generated by a term are recorded within the underlying type. This is in the first part of Section~\ref{sect:behaviorintersection}.
	\item
	Proof of soundness and completeness of \HEBI\ with respect to termination and reachability. The proof relies on a forward and backward type preservation result (i.e., subject reduction and expansion), as well as on a reducibility argument generalizing Tait’s classical method to the behavioral types forming the core of \HEBI. This is in the second part of Section~\ref{sect:behaviorintersection}.
	\item 
	Introduction of a simply typed variant of \HEBI, called \HEB, together with a proof that the reachability problem becomes decidable. The proof is structured around the observation that \HEBI\ types can be regarded as refinements of \HEB\ types, and that this refinement enjoys the finite refinement property. All this can be found in Section~\ref{sect:simplebehavior}.
	\item
	A negative result about the finite refinement property for $\HEPCF$, shedding some light on why the latter has an
	undecidable HOMC problem. This is in Section~\ref{sect:backtopcf}.
\end{itemize}

In the rest of the paper, these contributions are presented after giving an overview (\Cref{sect:highlevel}) and some preliminaries (\Cref{sect:handlers}) about \HEPCF\ and its untyped version.
\begin{SHORT}
    For an extended version of this paper, containing additional details, see~\cite{lago2026typesintersecteffectshandled}.
\end{SHORT}
\begin{LONG}
    Most of the proofs and further details can be found in the Appendix.
\end{LONG}

\section{A High-Level View}\label{sect:highlevel}\input{informal}

\section{Effect Handlers, (Un)typed}\label{sect:handlers}\input{typeduntyped}

\section{Behavioral Intersection Typing}\label{sect:behaviorintersection}\input{behaviorintersection}

\section{Behavioral Simple Typing and Model Checking}\label{sect:simplebehavior}\input{behaviorsimple}

\section{Back to \HEPCF}\label{sect:backtopcf}\input{backtohepcf}

\section{Related Work}\label{sect:relatedwork}\input{relatedwork}

\section{Conclusion}\label{sect:conclusion}\input{conclusion}

\begin{acks}
We would like to express our gratitude to Victor Arrial and Mano Orenga for their valuable contributions during the initial phase of this work.
This work was partly supported by JSPS KAKENHI (JP24H00699) and JST CREST (JPMJCR21M3).
\end{acks}

\bibliographystyle{ACM-Reference-Format}
\bibliography{main}

\newpage
\appendix 
\input{appendix}
\end{document}

%% file: Styles/Macro.tex
\newtheorem{definition}{Definition}[section]
\newtheorem{lemma}[definition]{Lemma}
\newtheorem{theorem}[definition]{Theorem}
\newtheorem{proposition}[definition]{Proposition}
\newtheorem{corollary}[definition]{Corollary}

\newtheorem{remark}[definition]{Remark}
\newtheorem{example}[definition]{Example}

\NewDocumentCommand{\wordDefinition}		{ m }
	{\textbf{#1}}
\NewDocumentCommand{\wordEmphasis}			{ m }
	{\textit{#1}}
\NewDocumentCommand{\completedProof}		{ m }
    {#1}

\NewDocumentCommand{\spaceI}				{ }
	{\mkern1mu}
\NewDocumentCommand{\spaceII}				{ }
	{\mkern2mu}

\NewDocumentCommand{\spaceNegII}			{ }
	{\mkern-2mu}

\NewDocumentCommand{\vsep}					{ s }
	{\;\IfBooleanT{#1}{\middle}|\;}

\NewDocumentCommand{\setNaturals}			{ }
	{\mathbb{N}}

\NewDocumentCommand{\rangeI}				{ }
	{i \in I}
\NewDocumentCommand{\rangeJ}				{ }
	{j \in J}
\NewDocumentCommand{\rangeK}				{ }
	{k \in K}

\NewDocumentCommand{\bltI}					{ }
	{\color{gray}{{\rule[0.5mm]{2mm}{0.7mm}}}}

\NewDocumentCommand{\termFreeVar}		{ s m }
	{\mathtt{fv}\IfBooleanTF{}{(#1)}}
\NewDocumentCommand{\termBoundVar}		{ s m }
	{\mathtt{bv}\IfBooleanTF{}{(#1)}}

\NewDocumentCommand{\Refn}		{ }
	{\mathtt{Refn}}

%% file: Styles/Syntax.tex
\NewDocumentCommand{\comput}           { m }
    {\computColor{#1}}
\NewDocumentCommand{\val}            { m }
    {\valueColor{#1}}
\NewDocumentCommand{\handler}          { m }
    {\handlerColor{#1}}

\NewDocumentCommand{\constructor} { m }
    {\ensuremath{\mathsf{#1}}}          
\NewDocumentCommand{\reduction}     { m }
    {\mathtt{#1}}
\NewDocumentCommand{\predicate}   { m } 
    {\textsf{#1}}
\NewDocumentCommand{\abs}				{ m m }
	{\lambda #1.#2}
\NewDocumentCommand{\app}				{ O{\spaceI} m m }
	{#2#1#3}
\NewDocumentCommand{\fix}     { s m m }
    {\constructor{fix}\IfBooleanTF{#1}{}{\,#2.#3}}
\NewDocumentCommand{\ret}       { s m }
    {\constructor{return}\IfBooleanTF{#1}{}{(#2)}}
\NewDocumentCommand{\letin}          { s m m m }
    {\constructor{let}\IfBooleanTF{#1}{}{\,#2\doteq#3\,\constructor{in}\,#4}}
\NewDocumentCommand{\handle}      { s m m }
    {\IfBooleanTF{#1}{}{\constructor{with}\,#2\,}\constructor{handle}\IfBooleanTF{#1}{}{\,#3}}
\NewDocumentCommand{\case}         { s m m }
    {\constructor{case}\IfBooleanTF{#1}{}{(#2;#3)}}
\NewDocumentCommand{\effect}       { s O{\sigma} m m m }
    {#2 \IfBooleanTF{#1}{}{(#3;#4.#5)}}

\RenewDocumentCommand{\int}         { m }
    {\underline{#1}}
\NewDocumentCommand{\cst}    { m }
    {\texttt{#1}}
    
\NewDocumentCommand{\retClause} { m m }
    {\constructor{return}(#1) \mapsto #2}
\NewDocumentCommand{\effectClause} { O{\sigma} m m m }
    {#1(#2;#3) \mapsto #4}
\NewDocumentCommand{\retPred}   { s m }
    {\predicate{Rtrn}\IfBooleanTF{#1}{}{(#2)}}
\NewDocumentCommand{\effectPred}   { s m }
    {\predicate{Efct}\IfBooleanTF{#1}{}{(#2)}}

\NewDocumentCommand{\sub}				{ m m }
	{\{#1 \backslash #2\}}

\NewDocumentCommand{\setVar}        { }
    {\mathcal{X}}

\NewDocumentCommand{\setTerm}       { s }
    {\Lambda\IfBooleanTF{#1}{^\sim}{}}
\NewDocumentCommand{\setVal}        { s }
    {\Lambda_{\valueColor{\texttt{Val}}}
        \IfBooleanTF{#1}{^\sim}{}}
\NewDocumentCommand{\setComput}     { s }
    {\Lambda_{\computColor{\texttt{Cmp}}}
        \IfBooleanTF{#1}{^\sim}{}}
\NewDocumentCommand{\setHandler}    { s }
    {\Lambda_{\handlerColor{\texttt{Hdl}}}
        \IfBooleanTF{#1}{^\sim}{}}

\NewDocumentCommand{\reductArr}{ s t{'} e{_} e{^} }
	{\IfBooleanTF{#1}
        {\twoheadrightarrow}
        {\IfBooleanTF{#2}
            {\mapsto}
            {\rightarrow}}
    \IfValueT{#3}
        {_{#3}}{}
    \IfValueT{#4}
        {^{#4}}{}}

\NewDocumentCommand{\reductBeta}         { }
    {\reduction{\beta}}
\NewDocumentCommand{\reductFix}          { }
    {\reduction{fix}}
\NewDocumentCommand{\reductCase}         { }
    {\reduction{case}}
\NewDocumentCommand{\reductLetRet}       { }
    {\reduction{let\mbox{-}ret}}
\NewDocumentCommand{\reductLetEff}       { }
    {\reduction{let\mbox{-}ef}}
\NewDocumentCommand{\reductHdlRet}       { }
    {\reduction{hdl\mbox{-}ret}}
\NewDocumentCommand{\reductHdlEff}        { }
    {\reduction{hdl\mbox{-}ef}}

\NewDocumentCommand{\Hole}				{ }
    {\diamond}

\NewDocumentCommand{\ctxtComput}       { }
    {\mathtt{\computColor{C}}\plugCurry}
\NewDocumentCommand{\ctxtValue}        { }
    {\mathtt{\valueColor{V}}\plugCurry}
\NewDocumentCommand{\ctxtHandler}      { }
    {\mathtt{\handlerColor{H}}\plugCurry}

\NewDocumentCommand{\ctxtEval}            { }
    {\mathtt{\computColor{E}}\plugCurry}

\NewDocumentCommand{\plugCurry}        { d{<}{>} }
    {\IfValueT{#1}{\!\left<#1\right>}}


%% file: Styles/Types.tex
\NewDocumentCommand{\st}                { m }
    {\mathsf{#1}}

\NewDocumentCommand{\typeInt}			{ m }
	{\mathsf{#1}}
\NewDocumentCommand{\stEffect}          { e{_} m }
    {\prescript{}{\IfValueTF{#1}{#1}{\effectCtxt E}} {#2}}
\NewDocumentCommand{\effectCtxt}        { m }
    {\mathtt{#1}}
\NewDocumentCommand{\effectSpecArrow}  { }
    {\mathrel{\rightsquigarrow}}
\NewDocumentCommand{\effectSpec}       { O{\sigma} m m }
    {#1 : #2 \effectSpecArrow #3}

\RenewDocumentCommand{\it}				 { m }
	{\mathcal{#1} }
\NewDocumentCommand{\itSet}			{ s m }
	{\IfBooleanTF{#1}%
		{\{#2\}}%
		{\left\{#2\right\}}}
\NewDocumentCommand{\itSetEmpty}		{ s }
	{\IfBooleanTF{#1}{\itSet{}}{\itSet{\,}}}
\NewDocumentCommand{\itEffect}		{ s m m m }
	{\IfBooleanTF{#1}%
		{\prescript{#2[#3]\!}{}{
			\left\{\!\!\left\{ #4 \right\}\!\!\right\}
		}}
		{\prescript{#2[#3]\!}{}
		\{\!\!\{#4\}\!\!\}
		}}
\NewDocumentCommand{\itEffectReturn}	{ m }
	{\itEffect{\ret*{}}{#1}{\,}}

\NewDocumentCommand{\vdashComput}       { }
    {\mathrel{\stackon[0.5pt]{\computColor{\vdash}}{\spaceII{\scriptscriptstyle\mathtt{\computColor{C}}}}}}
\NewDocumentCommand{\vdashValue}       { }
    {\mathrel{\stackon[0.5pt]{\valueColor{\vdash}}{\spaceII{\scriptscriptstyle\mathtt{\valueColor{V}}}}}}
\NewDocumentCommand{\vdashHandler}     { }
    {\mathrel{\stackon[0.5pt]{\handlerColor{\vdash}}{\spaceII{\scriptscriptstyle\mathtt{\handlerColor{H}}}}}}
\NewDocumentCommand{\itJug}        { o m m m }
    {\IfValueT{#1}{#1 \hepcfTr} #2 \vdash #3 : #4}
\NewDocumentCommand{\itJugComput}        { o m m m }
    {\IfValueT{#1}{#1 \hepcfTr} #2 \vdashComput #3 : #4}
\NewDocumentCommand{\itJugValue}         { o m m m }
    {\IfValueT{#1}{#1 \hepcfTr} #2 \vdashValue #3 : #4}
\NewDocumentCommand{\itJugHandler}       { o m m m }
    {\IfValueT{#1}{#1 \hepcfTr} #2 \vdashHandler #3 : #4}
\NewDocumentCommand{\stJug}      { o m m m }
    {\IfValueT{#1}{#1 \hepcfTr} #2 \vdash #3 :: #4}
\NewDocumentCommand{\stJugComput}      { o m m m }
    {\IfValueT{#1}{#1 \hepcfTr} #2 \vdashComput #3 :: #4}
\NewDocumentCommand{\stJugValue}       { o m m m }
    {\IfValueT{#1}{#1 \hepcfTr} #2 \vdashValue #3 :: #4}
\NewDocumentCommand{\stJugHandler}     { o m m m}
    {\IfValueT{#1}{#1 \hepcfTr} #2 \vdashHandler #3 :: #4}
\NewDocumentCommand{\refJug}      { o m m m m }
    {\IfValueT{#1}{#1 \hepcfTr} #2 \vdash #3 : #4 :: #5}
\NewDocumentCommand{\refJugComput}      { o m m m m }
    {\IfValueT{#1}{#1 \hepcfTr} #2 \vdashComput #3 : #4 :: #5}
\NewDocumentCommand{\refJugValue}       { o m m m m }
    {\IfValueT{#1}{#1 \hepcfTr} #2 \vdashValue #3 : #4 :: #5}
\NewDocumentCommand{\refJugHandler}     { o m m m m }
    {\IfValueT{#1}{#1 \hepcfTr} #2 \vdashHandler #3 : #4 :: #5}

\NewDocumentCommand{\ruleNameStyle}  { m }
    {$\scriptstyle\text
        {$\mathtt{(#1)}$\xspace}$}

\NewDocumentCommand{\hepcfTr}                       { }
    {\mathrel{\triangleright}}
\NewDocumentCommand{\ruleNameInt}              { }
    {\ruleNameStyle{int}}
\NewDocumentCommand{\ruleNameCst}              { }
    {\ruleNameStyle{cst}}
\NewDocumentCommand{\ruleNameVar}              { }
    {\ruleNameStyle{var}}
\NewDocumentCommand{\ruleNameAbs}              { }
    {\ruleNameStyle{abs}}
    \NewDocumentCommand{\ruleNameFix}    { }
    {\ruleNameStyle{fix}}
\NewDocumentCommand{\ruleNameFixBase}    { }
    {\ruleNameStyle{fixb}}
\NewDocumentCommand{\ruleNameFixRec}    { }
    {\ruleNameStyle{fixr}}
\NewDocumentCommand{\ruleNameLetin}              { }
    {\ruleNameStyle{let}}
\NewDocumentCommand{\ruleNameApp}              { }
    {\ruleNameStyle{app}}
\NewDocumentCommand{\ruleNameCase}             { }
    {\ruleNameStyle{case}}
\NewDocumentCommand{\ruleNameRet}           { }
    {\ruleNameStyle{ret}}
\NewDocumentCommand{\ruleNameEff}           { }
    {\ruleNameStyle{eff}}
\NewDocumentCommand{\ruleNameHandle}          { }
    {\ruleNameStyle{hdl}}
\NewDocumentCommand{\ruleNameHandler}    { }
    {\ruleNameStyle{hdlr}}

\NewDocumentCommand{\ruleNameHandlerRet}    { }
    {\ruleNameStyle{hdlr_{ret}}}
\NewDocumentCommand{\ruleNameHandlerSigma}    { }
    {\ruleNameStyle{hdlr_{\sigma}}}

\NewDocumentCommand{\typeRuleInt}              { O{0} m }
    {\infer#1[\ruleNameInt]{#2}}
\NewDocumentCommand{\typeRuleCst}            { O{1} m }
    {\infer#1[\ruleNameCst]{#2}}
\NewDocumentCommand{\typeRuleVar}              { O{0} m }
    {\infer#1[\ruleNameVar]{#2}}
\NewDocumentCommand{\typeRuleAbs}              { O{1} m }
    {\infer#1[\ruleNameAbs]{#2}}
\NewDocumentCommand{\typeRuleFix}    { O{2} m }
    {\infer#1[\ruleNameFix]{#2}}
\NewDocumentCommand{\typeRuleFixBase}    { O{1} m }
    {\infer#1[\ruleNameFixBase]{#2}}
\NewDocumentCommand{\typeRuleFixRec}    { O{2} m }
    {\infer#1[\ruleNameFixRec]{#2}}
\NewDocumentCommand{\typeRuleLetin}              { O{3} m }
    {\infer#1[\ruleNameLetin]{#2}}
\NewDocumentCommand{\typeRuleApp}              { O{2} m }
    {\infer#1[\ruleNameApp]{#2}}
\NewDocumentCommand{\typeRuleCase}             { O{2} m }
    {\infer#1[\ruleNameCase]{#2}}
\NewDocumentCommand{\typeRuleRet}           { O{1} m }
    {\infer#1[\ruleNameRet]{#2}}
\NewDocumentCommand{\typeRuleEff}           { O{2} m }
    {\infer#1[\ruleNameEff]{#2}}
\NewDocumentCommand{\typeRuleHandle}          { O{2} m }
    {\infer#1[\ruleNameHandle]{#2}}
\NewDocumentCommand{\typeRuleHandler}    { O{2} m }
    {\infer#1[\ruleNameHandler]{#2}}

\NewDocumentCommand{\typeRuleHandlerRet}    { O{2} m }
    {\infer#1[\ruleNameHandlerRet]{#2}}
\NewDocumentCommand{\typeRuleHandlerSigma}    { O{2} m }
    {\infer#1[\ruleNameHandlerSigma]{#2}}

\NewDocumentCommand{\handlerCompat}     { }
    {\mathrel{\#}}
\NewDocumentCommand{\replaceLeaf}     { m o }
    {\IfNoValueTF{#2}
    {\mathrel{-\:\!\!\!\! \left\{#1\right\} \:\!\!\!\!  \rightarrow}}
    {\mathrel{-\:\!\!\!\! \left\{#1\right\} \:\!\!\!\!  
        \overrightarrow{\, \vspace{-0.87em}^{#2} \;}}_{}}
    }

\NewDocumentCommand{\compatEmpty}              { }
    {\ruleNameStyle{\emptyset}}
\NewDocumentCommand{\compatDiff}               { }
    {\ruleNameStyle{diff}}
\NewDocumentCommand{\compatId}                 { }
    {\ruleNameStyle{id}}
\NewDocumentCommand{\compatEmptyRule}              { m }
    {\infer0[\compatEmpty]{#1}}
\NewDocumentCommand{\compatDiffRule}               { O{1} m }
    {\infer#1[\compatDiff]{#2}}
\NewDocumentCommand{\compatIdRule}                 { m }
    {\infer1[\compatId]{#1}}

\NewDocumentCommand{\redCand}{ s m }
    {\llbracket #2\rrbracket\IfBooleanTF{#1}{^\sim}{}}

\NewDocumentCommand{\refrel}{ m m }{#1 \sqsubset #2}
\newcommand{\refrelsym}{\sqsubset}
\NewDocumentCommand{\algosym}{ }{\mathcal{A}}
\NewDocumentCommand{\algo}{ m m }{\algosym(#1,#2)}
\NewDocumentCommand{\algoleaf}{ m m m m}{\mathcal{B}(#1,#2,#3,#4)}
\NewDocumentCommand{\graph}{ m m m }{\mathcal{G}(#1,#2,#3)}

%% file: Styles/Colors.tex
\definecolor{valueColor}{RGB}{247, 122, 0}
\definecolor{computColor}{RGB}{123, 198, 25}
\definecolor{handlerColor}{RGB}{178, 0, 255}

\NewDocumentCommand{\valueColor}        { m }
    {{\color{valueColor} #1}}
\NewDocumentCommand{\computColor}       { m }
    {{\color{computColor} #1}}
\NewDocumentCommand{\handlerColor}      { m }
    {{\color{handlerColor} #1}}

\definecolor{errorColor}{RGB}{255, 0, 0}

%% file: informal.tex
In this section we aim to provide an informal account of our contribution, delving slightly more into the details while avoiding overly technical discussions.

First of all, let us try to understand in what sense intersection types can be seen as a way of steadily increasing the expressive power of simple types, while preserving their good properties. In the simply typed $\lambda$-calculus, as is well known, the term $(\lambda x.\lambda y.\langle x\;3,y\;\mathtt{false}\rangle)(\lambda z.z)(\lambda w.w)$ is typable, whereas its close relative 
\begin{equation}\label{equ:untypable}
	(\lambda x.\langle x\;3,x\;\mathtt{false}\rangle)(\lambda z.z)
\end{equation}
is \emph{not}. This is due to the fact that simple types do not support any form of type \emph{polymorphism}. Simple types can be extended with polymorphic features in at least two different ways. On the one hand, one may add \emph{parametric polymorphism} in the style of that present in System \textsf{F}~\cite{Girard_1972_PhD,Reynolds_1974_PS}.  In this setting, the variable $x$ in (\ref{equ:untypable}) would be assigned the type $\forall\alpha.\alpha\rightarrow\alpha$. Alternatively, one may opt for a form of \emph{ad-hoc}, \emph{bounded} polymorphism, in which the same variable is instead assigned an \emph{intersection type} of the form $(\mathit{int}\rightarrow\mathit{int})\cap(\mathit{bool}\rightarrow\mathit{bool})$, meaning that the variable may be used as (and requires the call site to give) a function of type $\mathit{int}\rightarrow\mathit{int}$ \emph{and} $\mathit{bool}\rightarrow\mathit{bool}$. From many perspectives, the former solution is preferable, as is well known. There is, however, a distinctive feature of intersection types that is not shared by parametric polymorphism, namely their ability to \emph{characterize} terms satisfying interesting termination properties. Ultimately, this is due to the fact that a fundamental lemma---namely, the so called \emph{anti-substitution} lemma (from which subject expansion follows)---holds naturally for intersection types, whereas it cannot hold in parametric polymorphism. 

The anti-substitution lemma, in a form that we keep deliberately oversimplified, states that whenever there exists a derivation $\Pi$ of a typing judgment $\vdash t\{x/v\} : A$, it is possible to extract from it a type $B$ such that $x : B \vdash t : A$ and $\vdash v : B$. Clearly, this is not true in simple types nor with parametric polymorphism, whereas it becomes feasible in presence of the intersection construct, which allows one to collect into $B$ \emph{all the types} that $\Pi$ assigns to the (possibly many) occurrences of $v$ found within $t\{x/v\}$. In fact, parametric polymorphism cannot satisfy the anti-substitution lemma because not all pairs of distinct types can be unified into one, i.e., it is not possible to find one type $B$ accounting for all possible types $\Pi$ attributes to $v$, because polymorphism imposes a uniformity condition. If such a property held, parametrically polymorphic type systems would be complete for termination, and in fact they are not, see e.g. \cite{sorensen2006lectures} for more details and a counterexample.
As an example of how intersection types enable the anti-substitution property, consider a derivation $\Pi$ of $\vdash \langle x\;3,x\;\mathtt{false}\rangle\{x/\lambda z.z\} : \mathit{int} \times \mathit{bool}$ (where the subject term is the result of reducing (\ref{equ:untypable})), the anti-substitution lemma guarantees that $x : (\mathit{int}\rightarrow\mathit{int})\cap(\mathit{bool}\rightarrow\mathit{bool}) \vdash \langle x\;3,x\;\mathtt{false}\rangle : \mathit{int} \times \mathit{bool}$ and $\vdash \lambda z.z : (\mathit{int}\rightarrow\mathit{int})\cap(\mathit{bool}\rightarrow\mathit{bool})$ are also derivable, where the intersection type assigned to $x$ and $\lambda z.z$ collects all the types assigned to the occurrences of $\lambda z.z$ in $\Pi$. This collecting operation is made possible by the liberal nature of ad-hoc polymorphism, and by the fact that the types assigned to the substituted value $v$ can be completely uncorrelated \emph{and} the resulting polymorphic type keeps track of \emph{all} of them. Furthermore, this property leads to subject expansion, which states that, given terms $t$ and $s$ such that $t$ reduces to $s$ and $\vdash s : A$ holds, $\vdash t : A$ also holds. This is fundamentally different from what happens in parametric polymorphism, and guarantees that \emph{any} terminating term is typable. For example, given a term $t$ that terminates at a constant value, say, $0$, since $0$ is well typed, the subject expansion guarantees that the term $t$ is also well typed.\footnote{It is easy to generalize to terms that terminate at functions.} The inverse direction---i.e., well-typed terms always terminate---also holds because, intuitively, (finite, non-idempotent) intersection types ensure that each value is used only finitely many times, whereas non-terminating terms must refer to the same (recursive) value infinitely many times (see, for example, \cite{Bernadet/Lengrand_2011_CSL} for details).

Now, suppose that we want to generalize the discussion we have just carried out to a programming language for effect handlers. Assume we start from a calculus in the style of $\HEPCF$ \cite{Dal-Lago/Ghyselen_2024_POPL}. This calculus can be seen as an extension of the simply typed $\lambda$-calculus with two features. On the one hand, terms are allowed to perform algebraic operations, intended to model computational effects such as various forms of choice, global store read and write operations, exceptions, input and output. Each such operation $\sigma$ can be invoked in the form $\effect{\val v}{y}{\comput u}$, where the actual parameter $\val{v}$ is supplied together with a \emph{continuation} $\comput u$, which may depend on the output $y$ produced by executing the operation. On the other hand, the language provides the ability for programs to \emph{handle} these operations, that is, to specify precisely how they should be interpreted, in analogy with the classic mechanism of exception handling. Syntactically, handling is captured by a construct $\handle{\handler h}{\comput t}$ in which any operation executed by the computation $\comput t$ is managed according to the handler $\handler h$, the latter being a collection of \emph{operation clauses}, each corresponding to an algebraic operation.

For example, consider the term $\handle{\handler h}{\effect{\mathsf{()}}{x}{\constructor{if} \; x \; \constructor{then} \; \ret{1} \; \constructor{else} \; \ret{2}}}$.
This term first performs the algebraic operation $\sigma$.
Suppose that, to handle it, the handler $\handler{h}$ provides an operation clause
$\effectClause{y}{r}{\comput t}$ for $\sigma$.
Then, the handler transfers the control to the clause body $\comput t$.
The variable $y$ and $r$ bound in $\comput t$ denote an argument and a continuation, respectively, given by an operation call to be handled---therefore, in the example term, they denote the argument $\mathsf{()}$ and (the functional form of) the continuation $x.{\constructor{if} \; x \; \constructor{then} \; \ret{1} \; \constructor{else} \; \ret{2}}$.
If $\comput t$ only returns a value, i.e., $\comput t = \ret{\val v}$ for some value $\val v$, the whole term will also return $\val v$. In this case, $\sigma$ and $\handler h$ work as exception raising and handling, respectively.
By contrast, if $\comput t = \app{r\,}{\mathsf{true}}$, which invokes the continuation $r$ with Boolean value $\mathsf{true}$, then the computation $\constructor{if} \; \mathsf{true} \; \constructor{then} \; \ret{1} \; \constructor{else} \; \ret{2}$ obtained by substituting the argument $\mathsf{true}$ for the variable $x$ in the continuation will be executed.
The handler $\handler h$ can also call the continuation multiple times using
different arguments (say, $\mathsf{true}$ and $\mathsf{false}$).
This capability to resume continuations flexibly enables us to easily implement,
e.g., enumeration of all possible results of non-deterministic computations,
backtracking, mutable state, and transactions~\cite{Pretnar_2015_MFPS}.

In $\HEPCF$, computations receive types in the form ${\it E}, {\it F} := \stEffect{\st M}$, where $\st M$ is the type of the values returned by the computation and $\effectCtxt E$ is a signature type keeping track of the identity and behaviour of the algebraic operations the underlying term may produce. Handlers are viewed as \emph{computation transformers}, and this view is justified by the typing rule for the construct $\constructor{with}$-$\constructor{handle}$, which installs an effect handler ${\handler h}$ to handle the effect performed in a term {\comput t}:
\[
\begin{prooftree}
	\hypo{\itJugHandler{}{\handler h}{\it F \Rightarrow \it E}}
	\hypo{\itJugComput{}{\comput t}{\it F}}
	\typeRuleHandle{\itJugComput{}{\handle{\handler h}{\comput t}}{\it E}}
\end{prooftree}
\]
{The premise $\itJugHandler{}{\handler h}{\it F \Rightarrow \it E}$ tells us, in particular, that the handling clauses included in $\handler h$ are able to somehow \emph{transform} any computation of type $\it F$ into a computation of type $\it E$. This amounts to verifying that for each clause $\effectClause{x}{r}{\comput u}$ in $\handler h$ there exists a way of typing $\comput u$ in accordance with $\it E$.
For example, the handler $\handler h$ in the aforementioned example can be assigned type $\stEffect_{\{\effectSpec{\mathsf{unit}}{\mathsf{bool}}\}}{\mathsf{int}} \Rightarrow \stEffect_\emptyset{\mathsf{int}}$, which means that $\handler h$ transforms computations that may perform algebraic operation $\effectSpec{\mathsf{unit}}{\mathsf{bool}}$, to computations that perform no operation.
}

Suppose, now, that we want to prove the subject expansion property for this type system. Among the reduction rules governing calculi such as {\HEPCF}, there is one that is particularly interesting, namely the following:
\begin{equation}\label{equ:redhandle}
	\handle{\handler h}{\effect{\val v}{y}{\comput u}}
	\reductArr'_\reductHdlEff
	\comput t \sub{x}{\val v}\sub{r}{\abs{y}{\handle{\handler h}{\comput u}}}
\end{equation}
where the handler ${\handler h}$ is assumed to provide a clause $\effectClause{x}{r}{\comput t}$ to interpret the algebraic operation $\sigma$.
This rule tells us that handling an algebraic operation $\sigma$ by means of the effect handler $\handler h$ requires extracting from $\handler h$ the clause $\effectClause{x}{r}{\comput t}$ corresponding to $\sigma$ and passing to its body $\comput t$ two objects, namely the parameter $\val v$ and the \emph{delimited continuation} $\val w=\abs{y}{\handle{\handler h}{\comput u}}$.

Now, suppose that the RHS of reduction (\ref{equ:redhandle}) is typable, say with a computation type $\it E$, and that 
we want to show that the same type can somehow be transferred to the LHS.
Collecting the types of the various occurrences of the value
$\val v$ in the RHS into a single intersection type $\it M$ can 
be done easily as sketched above, of course provided that intersections are present.
Dealing with the (possibly many) occurrences of the continuation $\val w$ is more
challenging, depending on whether the continuation occurs in the RHS.
First, assume that the continuation does \emph{not} occur in the RHS.
In this case, we
cannot collect any type information about $\handler h$ from the RHS, except for the clause
$\effectClause{x}{r}{\comput t}$, which provides the term $\comput t$ in the RHS.
As a result, a typing derivation built for the LHS should only typecheck $\sigma$'s clause and
should \emph{not} typecheck the other clauses.
To address this typing requirement on the LHS, as shown shortly, we introduce a
new form of computation types that keeps track of what algebraic operation a
term calls \emph{first}, and the typing of a handler only typechecks the clauses
of the algebraic operations that will be called \emph{necessarily}.
Next, suppose that $\comput t$ refers to the continuation.
In such a case, the only information we can collect from the occurrences of the
continuation $\val w$ in the RHS is:
$$
\begin{prooftree}
	\hypo{\itJugHandler{}{\handler h}{\it F_i \Rightarrow \it E_i}}
	\hypo{\itJugComput{y:{\it N_i}}{\comput u}{\it F_i}}
	\typeRuleHandle{\itJugComput{y:{\it N_i}}{\handle{\handler h}{\comput u}}{\it E_i}}
\end{prooftree}
$$ where $i$ ranges over a finite set $I$ indexing all occurrences
of the variable $r$ in $\comput t$. Here, the computation types $\it E_i,
\it F_i$ and the value type $\it N_i$, which are the type information at some
occurrence of $\val w$, can be \emph{completely uncorrelated} with the type
information $\it E_j, \it F_j, \it N_j$ ($i \neq j$) at any other occurrence of
$\val w$.
We have to \emph{turn them} into the typings of $\handler h$ and $\effect{\val v}{x}{\comput u}$ in the LHS.
{We are therefore in a situation similar to the one described above concerning the anti-substitution lemma.}


The solution to the aforementioned problem we propose in this paper consists in considering computation types having a behavioral nature, that is, that aggregate the information contained in the types $\it M, \it N_i, \it F_i$ into a single composite type
\begin{equation}\label{equ:comptype}
\itEffect{\sigma}{\it M}{\it N_i \rightarrow \it F_i}_{\rangeI},
\end{equation}
called a \emph{computation type}. Computation types are behavioral, as they specify the algebraic operation terms give rise to first, here $\sigma$, together with the type of the parameters
to it, namely $\it M$. The types ${\it N_i \rightarrow \it F_i}$, instead, are all the possible types we assign, in the spirit of intersection types, to the continuation. In the context of the LHS of reduction rule (\ref{equ:redhandle}), the type (\ref{equ:comptype}) is the means by which we allow the handler $\handler h$ and the algebraic operations in $\effect{\val v}{x}{\comput u}$ to ``talk to each other''. Defining computation types this way allows us to attribute a type to algebraic operations as follows:
$$
 \begin{prooftree}
	\hypo{
		\itJugValue
		{}
		{\val v}
		{\it M}
	}
	\hypo{(
		\itJugComput
		{y : \it N_i}
		{\comput u}
		{\it F_i}
		)_{\rangeI}}
	\typeRuleEff{
		\itJugComput
		{}
		{\effect{\val v}{y}{\comput u}}
		{\itEffect{\sigma}{\it M}{\it N_i \rightarrow \it F_i}_{\rangeI}}
	}
\end{prooftree}
$$
In other words, the computation type we assign to the underlying term tells us which operation is being performed, what the type of its parameter is, \emph{but also} how the continuation $\comput u$ will be used when $\sigma$ is handled by $\handler h$. But how could we type the handler $\handler h$ itself? It would be ideal if it received a type of the form
$\itEffect{\sigma}{\it M}{\it N_i \rightarrow \it F_i}_{\rangeI}\Rightarrow{\it E}$,
because this way the type would be preserved. And indeed, this turns out to be possible by means of the following typing rule:
$$
 \begin{prooftree}
	\hypo{%
		\itJugComput%
		{  x : \it M;
			r : \itSet*{\it N_i \rightarrow \it E_i}_{\rangeI}}%
		{\comput t}%
		{\it E}%
	}%
	\hypo{(%
		\itJugHandler%
		{}%
		{\handler h}%
		{\it F_i \Rightarrow \it E_i}%
		)_{\rangeI}}%
	\hypo{\{\effectClause{x}{r}{\comput t}\} \handlerCompat \handler h}%
	\typeRuleHandlerSigma[3]{%
		\itJugHandler%
		{}%
		{\{\effectClause{x}{r}{\comput t}\} \cup \handler h}%
		{
			\itEffect{\sigma}{\it M}
			{\it N_i
				\rightarrow 
				\it F_i
			}_{\rangeI}
			\Rightarrow 
			\it E
		}%
	}%
\end{prooftree}
$$
namely the one we use in our intersection type system $\HEBI$.
This rule indicates that the clause for the algebraic operation performed first
has to be typechecked (the first premise) and the handler $\handler h$ has to be
typechecked against each type $\it F_i \Rightarrow \it E_i$ (the second
premise) which matches with the use of the continuation in $\comput t$.
The premise 
$\{\effectClause{x}{r}{\comput t}\} \handlerCompat \handler h$ serves to ensure that, if a clause for $\sigma$ is present in $\handler h$, it has the same shape, namely $\sigma(x;r)\mapsto {\comput t}$. Observe that, in this situation, the \emph{handlers} in the premise and in the conclusion of the rule may coincide. Nevertheless, these two occurrences can still be \emph{typed} differently: the type in the premise collects the types corresponding to (possibly) other usages of the operations in $\handler h$.
Therefore, if the continuation $\comput t$ calls algebraic operations (including
$\sigma$), the clauses corresponding to them are typechecked in the second
premise.

Although our computation types require determining the first algebraic operations called,
intersection types allow typechecking terms with an undetermined order of operation calls.
For instance, consider a term
$\lambda x. \constructor{if} \; x \; \constructor{then} \; \effect[\sigma_1]{x_1}{\val v_1}{\comput t_1} \; \constructor{else} \; \effect[\sigma_2]{x_2}{\val v_2}{\comput t_2}$,
which calls $\sigma_1$ or $\sigma_2$ depending on a Boolean value $x$.
Given singleton types $\mathit{true}$ and $\mathit{false}$ that are only assigned to Boolean values $\mathtt{true}$ and $\mathtt{false}$, respectively,
this term can be assigned an intersection type
$$
\mathit{true} \rightarrow \itEffect{\sigma_1}{M_1}{N_{1i} \rightarrow \it E_{1i}}_{\rangeI} \;\cap\; \mathit{false} \rightarrow \itEffect{\sigma_2}{M_2}{N_{2j} \rightarrow \it E_{2j}}_{\rangeJ} ~.
$$

Finally, it is noteworthy that this typing mechanism allows assigning different
types to different calls to the same algebraic operation.
For example, consider a term
$\effect[\sigma]{1}{x}{\effect[\sigma]{\mathtt{false}}{y}{\comput t}}$, to which
a computation type
$\itEffect{\sigma}{\mathit{int}}{M_i \rightarrow \itEffect{\sigma}{\mathit{bool}}{N_{ij} \rightarrow \it E_{ij}}_{\rangeJ}}_{\rangeI}$
may be assigned.  This computation type indicates that the first and second calls to the
operation $\sigma$ may be typed differently.
This typing is impossible in the calculi such as {\HEPCF} and others~\cite{Plotkin/Pretnar_2009_ESOP,Plotkin/Pretnar_2013_LMCS,Bauer/Pretnar_2014_LMCS,Hillerstrom/Lindley_2016_TyDe,Leijen_2017_POPL,Yoshioka/Sekiyama/Igarashi_ICFP_2024,Tang/Lindley_2026_POPL},
where all the calls to the same algebraic operation handled by the same handler have to be
typed in the same way.
Our computation types can distinguish different calls to the same operation
since they are aware of the \emph{order} of the calls, namely, they capture
(finite) \emph{computation trees} of algebraic operations produced by terms.

{But why can such an intersection type system be useful in the context of HOMC, in particular when we are interested in decidability problems? The idea is quite simple. Suppose we want to check whether a simply typed $\lambda$-term $\comput t = \val v \val w$ of Boolean type evaluates to one of the two Boolean values, for example to $\mathit{true}$. One can proceed by systematically checking whether, given a type $A$ such that $\val v$ has type $A \rightarrow \mathit{bool}$ and $w$ has type $A$, there exists a \emph{refinement} $\it A$ of $A$ such that $\val v$ has type $\it A \rightarrow \mathit{true}$ and $w$ has type $\it A$. This notion of refinement should, on the one hand, provide precise information about the value to which terms evaluate and, on the other hand, support a form of exhaustive search. Classical idempotent intersection types~\cite{coppo1978new} enjoy both of these properties. On the other hand, the behavioural types introduced in this work, although not satisfying these properties with respect to $\HEPCF$, naturally suggest a new type system, called $\HEB$, for which the HOMC problem is shown to be decidable, and which will be introduced in Section~\ref{sect:simplebehavior} below.}

%% file: typeduntyped.tex
In this section, we introduce some preliminaries related to the calculi \HE\ and \HEPCF, the latter being a typed version of the former. These calculi provide, in addition to the usual constructions of a $\lambda$-calculus in fine-grain call-by-value form, recursion, integers, and, of course, algebraic effects and handlers.

\paragraph{Terms}
Given a countably infinite set $\setVar$ of variables ($x, y, z, \ldots$) and a countably infinite set $\mathcal{O}$ of operations (ranged over by $\sigma$),
the set $\setTerm$ of \wordDefinition{terms} is given by the following inductive definitions:
\begin{equation*}
	\begin{array}{crcl}
		\textbf{(Computations)}&
		\comput t, \comput u, \comput s
		&\coloneqq& \app{\val v}{\val w}
		\vsep \effect{\val v}{x}{\comput t}
		\vsep \ret{\val v}
		\vsep \letin{x}{\comput t}{\comput u}
		\\
		&&&
		\vsep \handle{\handler h}{\comput t}
		\vsep \case{\val v}{\comput t_1, \ldots, \comput t_n}
		\\[0.2cm]
		\textbf{(Values)}&
		\val v, \val w
		&\coloneqq&  \int n
		\vsep x
		\vsep \abs{x}{\comput t}
		\vsep \fix{x}{\val v}
		\\[0.2cm]
		\textbf{(Handlers)}&
		\handler h
		&\coloneqq&
		\{\retClause{x}{\comput t}\}
		\cup \{\effectClause[\sigma_{\spaceNegII i}]{x}{r}{\comput t_i} \vsep 0 < i \leq n\}
		
		\\[0.2cm]
		\textbf{(Terms)}&
		p
		&\coloneqq&
		\comput t \vsep \val v \vsep \handler h
	\end{array}
\end{equation*}
The set $\setComput$ of \wordDefinition{computations} 
($\comput t, \comput u \ldots$) includes
\wordDefinition{applications} $\app{\val v}{\val w}$ composed of
a \emph{function} $\val v$ applied to an \emph{argument} $\val w$,
\wordDefinition{effect operations} $\effect{\val v}{x}{\comput t}$ composed of
an \emph{algebraic operation} $\sigma$ applied to a 
\emph{parameter} $\val v$ and followed by a \emph{continuation} $\comput t$ 
in which the variable $x$ is to be substituted with the result of the operation,
\wordDefinition{return operations} $\ret{\val v}$ composed of 
a \emph{resulting value} $\val v$ 
(note that a return operation can be seen as a particular 
algebraic operation with no continuation, furthermore, note that the purpose of the return construct is to embed values in computations), 
\wordDefinition{let bindings} $\letin{x}{\comput t}{\comput u}$ composed of
a \emph{definition} $\comput t$ for a value $x$ followed 
by a \emph{continuation} $\comput u$,
\wordDefinition{handling} $\handle{\handler h}{\comput t}$ composed of
a \emph{handler} $\handler h$ interpreting a \emph{handled computation} $\comput t$,
and \wordDefinition{case analysis} $\case{\val v}{\comput t_1, \dots, \comput t_n}$ composed of 
a \emph{matched value} $\val v$ and \emph{possible continutations} $(\comput t_k)_{k\le n}$. This construct represents a branching computation depending on a (positive) integer.
The set $\setVal$ of \wordDefinition{values} 
($\val v, \val w \ldots$) includes  
\wordDefinition{positive integers} $\int n \in \setNaturals^+$,
\wordDefinition{variables} $x \in \setVar$, 
\wordDefinition{abstractions} $\abs{x}{\comput t}$ composed of
a \emph{body} $\comput t$ and an \emph{abstracted variable} $x$,
and \wordDefinition{fixpoint} $\fix{x}{\val v}$ composed of 
a \emph{body} $\val v$ and a \emph{recursive variable} $x$.
Finally, each \wordDefinition{handler} $\handler h \in\setHandler$ contains
\wordDefinition{effect clauses} $\{\effectClause{x}{r}{\comput t}\}$
for different operations $\sigma$ with an \emph{interpretation} $\comput t$
based on the \wordEmphasis{parameter} $x$ of the operation and the 
\wordEmphasis{continuation} $r$. Each handler $\handler h$ 
necessarly contains an interpretation $\retClause{x}{\comput t}$
of the return operation which has no continuation. 

An effect operation $\effect{\val v}{y}{\comput t}$ can be 
interpreted as \emph{performing} the operation $\sigma$ 
with the parameter $\val v$ as input, then using the output $y$ of the
operation in the following computation $\comput t$. Correspondingly,
an effect clause $\effectClause{x}{r}{\comput u}$ of a handler
defines an \emph{interpretation} of the operation $\sigma$. The variable $x$ corresponds
to the parameter $\val v$ of the effect operation, and the variable $r$ 
is a function corresponding to the computation $\comput t$ performs after $\sigma$
depending on the output $y$ of $\sigma$.
Binding is defined as usual. In particular, notice that
in each clause $\effectClause{x}{r}{\comput t}$ of a handler $\handler h$,
the variables $x$ and $r$ are bound in the interpretation $\comput t$.

\paragraph*{Operational Semantics.}
To define the operational semantics of
	\HE\ terms, we first need to introduce \wordDefinition{evaluation contexts}, which are terms containing exactly one occurrence of the \wordDefinition{hole} $\Hole$ that should be filled by terms to be reduced. Formally, they are
	inductively defined as follows.
\begin{equation*}
	\begin{array}{rrcl}
		\textbf{(Evaluation Ctxt)}&
		\ctxtEval &\coloneqq& \Hole
		\vsep \letin{x}{\ctxtEval}{\comput t}
		\vsep \handle{\handler h}{\ctxtEval}\\
	\end{array}
\end{equation*}
Moreover, we write $\ctxtEval<t>$ for the term obtained by replacing
the only occurrence of $\Hole$  in $\ctxtEval$ with the term $t$.

~
We denote by $p\sub{x}{\val v}$ the term obtained by
the usual (capture-avoiding) meta-level \wordDefinition{substitution}
of the value $\val v$ for all free occurrences of the variable $x$ in
the term $p$. The usual notions of free variables, $\alpha$-conversion and substitution are extended to the set of terms $\setTerm$, and terms are
identified up to $\alpha$-conversion.

\begin{figure}[h!]
	\centering
			$\begin{array}{clll}
				\app{(\abs{x}{\comput t})}{\val v} 
				&\reductArr'_\reductBeta&
				\comput t\sub{x}{\val v}
				\\[0.1cm]
				\app{(\fix{x}{\val v})}{\val w}
				&\reductArr'_\reductFix&
				\app{(\val v\sub{x}{\fix{x}{\val v}})}{\val w}
				\\[0.1cm]
				\case{\int n}{\comput t_1, \ldots, \comput t_m}
				&\reductArr'_\reductCase&
				\comput t_n
				&\text{when } 0 < n \leq m
				\\[0.1cm]
				\letin{x}{\ret{\val v}}{\comput t}
				&\reductArr'_\reductLetRet&
				\comput t\sub{x}{\val v}
				\\[0.1cm]
				\letin{x}{\effect{\val v}{y}{\comput t}}{\comput u}
				&\reductArr'_\reductLetEff&
				\effect{\val v}{y}{\letin{x}{\comput t}{\comput u}}
				\\[0.1cm]
				\handle{\handler h}{\ret{\val v}}
				&\reductArr'_\reductHdlRet&
				\comput t\sub{x}{\val v}
				&\text{when } \retClause{x}{\comput t} \in \handler h
				\\[0.1cm]
				\handle{\handler h}{\effect{\val v}{y}{\comput u}}
				&\reductArr'_\reductHdlEff&
				\comput t \sub{x}{\val v}\sub{r}{\abs{y}{\handle{\handler h}{\comput u}}}
				&\text{when } \effectClause[\sigma]{x}{r}{\comput t} \in \handler h
			\end{array}$
	\caption{Reduction rules for \HE\ }
	\label{def:RewriteRules}
    \Description{Reduction rules for HE}
\end{figure}

With this,
the seven rewriting rules defined in \Cref{def:RewriteRules} 
establish reduction for \HE, in which we implicitly assume the terms to be closed. We write $\comput t\reductArr'\comput s$ if $\comput t$ rewrites to $\comput s$.
The first four rules ($\reductBeta, \reductFix, \reductCase, \reductLetRet$)
correspond to the usual reduction rules for a $\lambda$-calculus with
recursive functions.
The $\reductLetEff$-rule states that any let binding with an effect
as definition can be simplified by permuting the let binding with the
effect operation, therefore highlighting  that the operation $\sigma$ is meant 
to be performed first.
The $\reductHdlEff$-rule states that any handling with an effect operation
as handled computation can be simplified into the interpretation $\comput t$
for the operation clause of the handler $\handler h$, where the occurrences of
the parameter $x$ are replaced with the argument value $\value v$. Similarly
for the $\reductHdlRet$ rule, which states that, when a handled computation returns a value, the return clause provided by the effect handler is evaluated with the returned value.
Moreover, for the $\reductHdlEff$-rule, the occurrences of the continuation $r$ are also replaced with the
handling of $\comput u$ by $\handler h$, abstracted over the output
variable $y$ of the effect.

The $\reductHdlRet$- and $\reductHdlEff$-rules illustrate how interpretations
define the performance of return and effect operations. On the other hand,
with the $\reductLetEff$-rule, the first performed effect bubbles    up to the top of
the computation. This rule works in tandem with the $\reductHdlEff$-rule
to associate deeply embedded effect operations to their interpretations
using only local rules. Using the evaluation contexts, we can now define 
reduction as followed.

\begin{definition}[Reduction]
	\label{def:Reduction}
	Let $\comput t, \comput s \in \setComput$ be closed terms.
	We say that $\comput t$ $ \wordDefinition{reduces to}$ $\comput s$, written
	$\comput t \reductArr \comput s$ iff there is an evaluation context
	$\ctxtEval$ such that
	$\comput t=\ctxtEval<\comput u>$ and
	$\comput s=\ctxtEval<\comput r>$,   
	where $\comput u \reductArr' \comput r$.
	We write $\reductArr*$ for the reflexive and transitive closure of the reduction 
	relation $\reductArr$.
\end{definition}

\noindent Note that given the way evaluation contexts are defined, reduction is deterministic and call-by-value.

\begin{example}
	Let $\handler h = \{\retClause{x}{\effect{x}{y}{\ret{x}}}\}$. Then,
  $$
  \handle{\handler h}{\ret{\int 5}}
	\reductArr'_\reductHdlRet
	\effect{\int 5}{y}{\ret{\int 5}} ~.
  $$
	And for $\ctxtEval = \letin{z}{\Hole}{\ret z}$,
  $$
  \ctxtEval<\handle{\handler h}{\ret{\int 5}}>
	\reductArr \ctxtEval<\effect{\int 5}{y}{\ret{\int 5}}> ~.
  $$
\end{example}

The just defined calculus, called $\HE$, is genuinely untyped, and as such does not satisfy any of those desirable properties like termination or progress. In fact, type clashes can easily occur. The calculus \HE, however, is Turing-powerful, and its expressive power from the standpoint of computability theory would remain unchanged in the absence of effects and handlers, as well as without recursion. Consequently, most verification problems, including reachability problems, remain undecidable in this calculus. It is anyway worth clarifying what we mean by reachability here; we in fact adopt a very simple definition, which does not capture reachability in its full generality, but which is fine for our purposes. Given a closed term $\comput t \in \setComput$, the \wordDefinition{reachability problem} asks whether there exists a value $\val v$ such that $\comput t \reductArr* \ret{\val v}$, that is, whether the evaluation of $\comput t$ reaches any value. We note that, e.g., the halting problem can be easily encoded as one such reachability problem.

\paragraph{Simple Types}

In this section, we restrict our attention to a 
subset of terms, with the aim of guaranteeing a form
of type safety closely following a simply-typed discipline.
We associate with each computation an effect context specifying 
the input and output types of each algebraic operation.
The grammar of simple types then becomes the following.
\begin{equation*}
	\begin{array}{crcl}
		\textbf{(Value Types)}&
		\st L,\st{M}, \st{N}
		&\coloneqq& 
		\typeInt{n}
		\vsep \st{M} \rightarrow \st{U}
		\\
		\textbf{(Computation Types)}&
		\st{U}, \st{S}
		&\coloneqq& \stEffect{\st M}
		\\
		\textbf{(Handler Types)}&
		\st{H}
		&\coloneqq& \st{S} \Rightarrow \st{U}
		\\
		\textbf{(Effect Contexts)}&
		\effectCtxt E, \effectCtxt F
		&\coloneqq& \emptyset \vsep \{\effectSpec{\st{M}}{\st{N}}\} \cup \effectCtxt E
		\\
		\textbf{(General Types)}&
		\st{T}
		&\coloneqq& \st{M}
		\vsep \st{U}
		\vsep \st{H}\\
	\end{array}
\end{equation*}
The integer type\footnote{We could have considered a single, infinite type of all positive integers, but we prefer to work with finite base types here and throughout the remainder of this article, in view of the discussions about the higher-order model checking problem.} $\typeInt n$ is inhabited by all natural numbers
in the range $[1,n]$. The
computation type $\stEffect{\st M}$ indicates both the type $\st M$ of the
returned value, and the specification $\effectCtxt E$ of the effects possibly used, given as an \wordDefinition{effect context}, namely a signature specifying for every operation $\sigma$ both the type $\st M$ of the argument and the type $\st N$ of the result, to be passed to the continuation. 

We define three distinct typing judgments for values, computations, and handlers, respectively. The notations for these simple type judgments are as follows, and the corresponding inference rules are given
in \Cref{def:SimpleTypingRules}.
\begin{equation*}
	\begin{array}{ccccc}
		\stJugComput{\Delta}{\comput t}{\st U}
		&\qquad&
		\stJugValue{\Delta}{\val v}{\st M}
		&\qquad&
		\stJugHandler{\Delta}{\handler h}{\st H}
	\end{array}
\end{equation*}
\begin{figure}[h!]\centering
	\input{Definitions/stType.tex}
	\caption{\HEPCF\, typing rules.}
	\label{def:SimpleTypingRules}
    \Description{HEPCF, typing rules}
\end{figure}
Most of the typing rules of \HEPCF\ are standard, as they closely follow those of the simply typed $\lambda$-calculus. Some of them, however, deserve brief discussion. Terms of the form $\effect{\val v}{x}{\comput t}$ are assigned a type by rule ($\scriptstyle\text
        {$\mathtt{eff}$\xspace}$). Observe that, while the type of the parameter $\val v$ and that of the variable $x$ must be consistent with the effect context $\effectCtxt{E}$, the return type of the computation $\comput t$ is completely arbitrary, and becomes the type of the whole term.
Furthermore, when typing handlers, each algebraic operation is analyzed separately, and it is therefore possible to assign exactly one type to each of them. This is fully consistent with the monomorphic nature of the type system.

\HEPCF\ as we have just introduced it is type safe, which can be easily shown by proving Subject Reduction and Progress. Its expressiveness is limited, in particular due to the fact that the base types are all finite. One might think that this restriction, which in \PCF\ leads to the \emph{decidability} of the HOMC problem for a very broad class of properties, namely those captured by MSO formulas, would also make this problem decidable for \HEPCF. Unfortunately, Dal Lago and Ghyselen showed~\cite{Dal-Lago/Ghyselen_2024_POPL} that this is not the case, and that the HOMC is undecidable in \HEPCF\ already for mere reachability properties. (The latter, by the way, can be formulated for \HEPCF\ in a completely analogous way to what we did for \HE\ earlier in this section). This is proved by simulating the standard inductive type of integers through algebraic effects and handlers. We are interested in understanding the deep reason behind the jump in expressiveness between finitary $\mathsf{PCF}$ and \HEPCF. The kind of behavioral intersection typing we introduce in the next section will indeed provide a new point of view on this issue, which we defer to Section~\ref{sect:backtopcf}.

%% file: Definitions/stType.tex
$
        \begin{array}{c}
            \begin{array}{ccccc}
                \begin{prooftree}
                        \hypo{0<n \le m}
                    \typeRuleInt[1]{\stJugValue{\Delta}{\int n}{\typeInt m}}
                \end{prooftree}
                &&
                \begin{prooftree}
                    \typeRuleVar{\stJugValue{\Delta;x :: \st M}
                        {x}{\st M}}
                \end{prooftree}
            \\[0.5cm]
                \begin{prooftree}
                        \hypo{ \Delta; x :: \st M \vdash \comput t :: \st U}
                    \typeRuleAbs{\stJugValue
                        {\Delta}
                        {\abs{x}{\comput t}}
                        {\st M \rightarrow \st U}}
                \end{prooftree}
        &\hspace{0.6cm}&           
                \begin{prooftree}
                        \hypo{
                            \stJugValue
                                {\Delta; x :: \st M \rightarrow \st U}
                                {\val v}
                                {\st M \rightarrow \st U}}
                    \typeRuleFix[1]{
                        \stJugValue
                            {\Delta}
                            {\fix{x}{\val v}}
                            {\st M \rightarrow \st U}}
                \end{prooftree}
                \\[0.5cm]
            \end{array}
            \\
        \begin{array}{ccc}
                \begin{prooftree}
                        \hypo{\stJugComput
                            {\Delta}
                            {\comput t}
                            {\stEffect{\st M}}}
                        \hypo{\stJugComput
                            {\Delta; x :: \st M}
                            {\comput u}
                            {\stEffect{\st N}}}
                    \typeRuleLetin[2]{\stJugComput
                            {\Delta}
                            {\letin{x}{\comput t}{\comput u}}
                            {\stEffect{\st N}}}
                \end{prooftree}
            & \hspace{0.6cm}&
                \begin{prooftree}
                    \hypo{\stJugHandler{\Delta}{\handler h}
                        {\st S \Rightarrow \st U}}
                    \hypo{\stJugComput{\Delta}{\comput t}{\st S}}
                    \typeRuleHandle{\stJugComput
                        {\Delta}
                        {\handle{\handler h}{\comput t}}{\st U}}
                \end{prooftree}
            \end{array}
            \\[0.5cm]
            \begin{array}{ccc}
                \begin{prooftree}
                        \hypo{\stJugValue
                            {\Delta}
                            {\val v}
                            {\st M \rightarrow \st U}}
                        \hypo{\stJugValue
                            {\Delta}
                            {\val w}
                            {\st M}}
                    \typeRuleApp{\stJugComput
                        {\Delta}
                        {\app{\val v}{\val w}}
                        {\st U}}
                \end{prooftree}
            &\hspace{0.6cm} &
                \begin{prooftree}
                        \hypo{\stJugValue
                            {\Delta}{\val v}{\typeInt n}}
                        \hypo{(\stJugComput
                            {\Delta}{\comput t_k}{\st U}
                            )_{ 0<k\le n}}
                    \typeRuleCase{\stJugComput
                        {\Delta}
                        {\case{\val v}{\comput t_1, \ldots, \comput t_n}}
                        {\st U}}
                \end{prooftree}
            \end{array}
            \\[0.5cm]
            \begin{array}{ccc}
                \begin{prooftree}
                        \hypo{\stJugValue
                            {\Delta}
                            {\val v}
                            {\st M}}
                    \typeRuleRet{\stJugComput
                        {\Delta}
                        {\ret{\val v}}
                        {\stEffect{\st M}}}
                \end{prooftree}
            &\hspace{0.8cm}&
                \begin{prooftree}
                        \hypo{\stJugValue
                            {\Delta}
                            {\val v}
                            {\st M}}
                        \hypo{\stJugComput
                            {\Delta; x :: \st N}
                            {\comput t}
                            {\stEffect{\st L}}}
                        \hypo{\effectSpec{\st M}{\st N} 
                            \in \effectCtxt E}
                    \typeRuleEff[3]{
                        \stJugComput
                            {\Delta}
                            {\effect{\val v}{x}{\comput t}}
                            {\stEffect{\st L}}
                        }
                \end{prooftree}
            \end{array}
            \\[0.5cm]
    \begin{prooftree}
    \hypo{(
        \stJugComput
            {\Delta; x :: \st L_i;r :: \st N_i \rightarrow \st U}
            {\comput t_i}
            {\st U})_{0<i\leq n}
            }
            \hypo{
                \begin{matrix}
            \stJugComput
            {\Delta; y :: \st M}
            {\comput t}
            {\st U}\\
            \{\sigma_i:{\st L_i}\rightsquigarrow{\st N_i}\in \st E \vsep 0 < i \leq n\}
                \end{matrix}
            }
            %
            %
            \typeRuleHandler[2]{\stJugHandler
                {\Delta}
                {\{\retClause{y}{\comput t}\} \cup\{\effectClause[\sigma_{\spaceNegII i}]{x}{r}{\comput t_i} \vsep 0 < i \leq n\}}
                {\stEffect{\st M} \Rightarrow \st U}
            }
            \end{prooftree}
    \end{array}
    $

%% file: behaviorintersection.tex
It is time to introduce our new type discipline of behavioral intersection typing, which will result in the \HEBI\ type system. After formalizing it, we show its metatheoretical properties.

\subsection{The {\HEBI} Type System}
We first define the type language of the {\HEBI} type system.
As with intersection types, values are typed with \wordEmphasis{a set} of types. 
This corresponds to the term being typable with \emph{every type} in the set. 
Moreover, the type system  is \emph{behavioral} in nature, keeping track of 
the effect produced and of their order. 
The grammars of {\HEBI} typing is defined as follows.
\begin{equation*}
	\begin{array}{crcl}
		\textbf{(Value Types)}&
		\it M, \it N
		&\coloneqq& \itSetEmpty
		\vsep \itSet{\typeInt{n}}
		\vsep \itSet{\it M_i \rightarrow \it E_i}_{\rangeI}
		\\
		\textbf{(Computation Types)}&
		\it E, \it F, \it G
		&\coloneqq& \itEffectReturn{\it M} 
        \vsep\itEffect{\sigma}{\it M}{\it N_i \rightarrow \it E_i}_{\rangeI}
		
		\\
		\textbf{(Handler Types)}&
		\it H
		&\coloneqq& \it E \Rightarrow \it F
		
		\\
		\textbf{(General Types)}&
		\tau, \rho
		&\coloneqq& \it M \vsep \it E \vsep \it H
		
	\end{array}
\end{equation*}
{where $I$ is a finite set.}
A \wordDefinition{value type} $\it M$ caracterizes how a value 
can be used by a computation. It can either be 
an \wordEmphasis{interger type} $\{\typeInt{n}\}$ associated a specific value $\int n$,
a \wordEmphasis{function type} $\itSet{\it M_i \rightarrow \it E_i}_{\rangeI}$
associated to functions that compute for every $\rangeI$ a process of type $\it E_i$
when applied to a value of type $\it M_i$,
or the \wordEmphasis{empty type} $\itSetEmpty$ corresponding to any unused value\footnote{We tacitly assume that $\itSetEmpty$ and $\itSet{\it M_i\rightarrow\it E_i}_{i\in \emptyset}$ denote the same type.}.
A \wordDefinition{computation type} $\itEffect{\sigma}{\it M}{\it N_i \rightarrow \it E_i}_{\rangeI}$
describes the tree of branching effects performed by
a computation. Its \wordEmphasis{root} (or starting) effect
is the operation $\sigma$ with a parameter of type $\it M$.
The \wordEmphasis{children} types $\it E_i$ of this root
are indexed by the possible types $\it N_i$ of the operation's output 
(each of these children corresponding to a possible process of the
effect's continuation). Similarly, the computation
type $\itEffectReturn{\it M}$ corresponds to a computation returning
a value of type $\it M$. A \wordDefinition{handler type} 
$\it E \Rightarrow \it F$ describes handlers
transforming computations of type $\it E$, into
ones of type $\it F$, as in {\HEPCF}. We note, however, that whereas in 
\HEBI\ a computation type is behavioral and a handler type therefore configures
 itself as a proper tree transformer, in \HEPCF\ a computation type only captures 
 the type of the parameter to the return command, if any, the computation 
 will call, together with a typing for the algebraic operations.
{We give in \Cref{fig:itExample} the visualisation of a computational \HEBI\ type that can be associated to the computation $\effect{\int 3}{x}{\case{x}{\ret x,\ 
				\effect[\delta]{\int 2}{f}{\app{f}{\int 3}}}}$. A type $\itEffect{\sigma}{\it M}{\it N_i \rightarrow \it E_i}_{\rangeI}$ is represented as a tree whose root is $\sigma[\it M]$, whose children are $\it E_i$, and whose edges are labeled by the corresponding $\it N_i$.} 

\begin{figure}[ht!]
	\centering
	\input{Visuals/itExample.tex}
	\caption{Visualisation of the \HEBI\ type 
		$\itEffect*{\sigma}{\typeInt 3}{
			\itSet{\typeInt{1}} \rightarrow \itEffectReturn{\typeInt 1},
			\itSet{\typeInt{2}} \rightarrow \it E}$, 
		where 
        $\it E = \itEffect*{\delta}{\itSetEmpty}{
			\itSet{\itSet{\typeInt 3} \rightarrow \it F}
			\rightarrow \it F,
			\itSet{\itSet{\typeInt 3} \rightarrow \it G}
			\rightarrow \it G}$,
		$\it F = \itEffectReturn{\typeInt 5}$ and $\it G = \itEffect{\delta}{\typeInt 3}{}$.}
	\label{fig:itExample}
    \Description{Visualisation of an \HEBI\ type.}
\end{figure}

As we said, values in \HEBI\ are typed as sets modelling intersections. Therefore, given two value types $\it M,\it N$, we can define their union $\it M\cup \it N$. It is important to notice, however, that the value $\int n$ can receive the value types $\itSet{}$ and $\itSet{n}$. Thus, the union $\itSet{m}\cup \itSet{n}$ is defined if and only if $n=m$. The union operator $\cup$ is therefore partial.
We also define a subset relation on value types in the standard manner, written $\it N\subseteq \it M$.

A type context $\Gamma = x_1 : \it M_1 ; \dots x_n : \it M_n$
is a total function from the set $\setVar$ of variables to the set of
value types, where every variable $x_i$ is associated to $\it M_i$, and
any $y$ not mentioned explicitly in $\Gamma$, is associated to the 
empty type $\itSetEmpty$. We define the sum of two type contexts 
$\Gamma + \Delta$ as
$\{x : \it M\cup\it N \vsep 
\text{where $x : \it M\in\Gamma$ and $x : \it N\in\Delta$}\}$. Notice 
that, once again this is a \emph{partial} binary operation on contexts, 
since $\cup$ is partial on value types.
%

\begin{definition}
	Let $\Gamma, \Delta$ be type contexts. We say $\Delta\subseteq\Gamma$ iff for every $x:\it N\in \Delta $, one has $x:\it M \in \Gamma$ with $\it N\subseteq\it M$. We write $Sub(\Gamma)$ for the set of type contexts $\Delta$ such that $\Delta\subseteq\Gamma$. 
\end{definition}

\paragraph*{Typing jugements}

As in \HEPCF, there are three kinds of typing jugements, written as follows.
\begin{equation*}
	\begin{array}{ccccc}
		\itJugComput{\Gamma}{\comput t}{\it E}
		&\qquad&
		\itJugValue{\Gamma}{\val v}{\it M}
		&\qquad&
		\itJugHandler{\Gamma}{\handler h}{\it H}
	\end{array}
\end{equation*}

\noindent 
A typing jugement of the form $\itJugComput{\Gamma}{\comput t}{\it E}$
states that, in the context $\Gamma$, evaluating the computation
$\comput t$ produces a computation tree coherent with the one described
by $\it E$.
A typing jugement $\itJugValue{\Gamma}{\val v}{\it M}$
indicates that the value $\val v$ can be used as prescribed by \emph{every}
type in the set $\it M$ (in the context $\Gamma$).
Finally, the jugement $\itJugHandler{\Gamma}{\handler h}{\it E \Rightarrow \it F}$
indicates that the handler $\handler h$ can transform a computation
of type $\it E$ into one of type $\it F$ in the context $\Gamma$.
The typing rules are defined in \Cref{def:EffectTreeTypingRules}.
Again, we write $\itJug{\Gamma}{p}{\tau}$ for one of the three judgments, 
depending on the nature of $p$ and $\tau$.
We also write $\itJug[\Pi]{\Gamma}{p}{\tau}$ if the jugement
$\itJug{\Gamma}{p}{\tau}$ is the conclusion of the proof tree $\Pi$.
The typing rules use also the two auxilary relations, to which 
we will devote the next two paragraphs. One of these relations is defined in \Cref{def:HandlerCompatAndLeafReplacement}.

\begin{figure}\centering
	\input{Definitions/etType.tex}
	\caption{\HEBI, typing rules.}
	\label{def:EffectTreeTypingRules}
    \Description{HEBI, typing rules}
\end{figure}

\begin{figure}\centering
	\center
	\input{Definitions/LeafReplacementHandlerCompat.tex}
	\caption{\HEBI, Leaf Replacement.}
	\label{def:HandlerCompatAndLeafReplacement}
    \Description{HEBI, Leaf Replacement}
\end{figure}

\paragraph*{Handler Compatibility} Let $\handler h\in \setHandler$ be an handler. The relation $\{\effectClause{y}{r}{\comput t}\}\handlerCompat\handler h$ holds if for any $\{\effectClause{y}{r}{\comput u}\}\in \handler h$, it holds that $\comput t=\comput u$. This relation certifies that the clause $\effectClause{y}{r}{\comput t}$
is not in conflict with any clause for $\sigma$ in $\handler h$. The handler compatibility on $\constructor{return}$ clauses is defined similarly: $\{\retClause{y}{\comput t}\}\handlerCompat \handler h$ holds if for any $\{\retClause{y}{\comput u}\}\in \handler h$, $\comput t=\comput u$. In this case, the relation certifies that $\retClause{y}{\comput t}$ is the only $\constructor{return}$ clause in $\handler h$. Observe that the handler compatibility relation is not 
required in \HEPCF, where the handling of all operations is done through terms that 
are forced to adhere to the same (\emph{non}-behavioral) return type and to the same signature. 
As a consequence, there would be no meaningful reason to type them multiple times. This situation 
changes in \HEBI, where behavioral information becomes an integral part of the type and handling
terms can be thus treated polymorphically.

\paragraph*{Leaf Replacement} $\it E \replaceLeaf{\it M_i \rightarrow \it
G_i}[\rangeI] \it F$, is a relation stating that the type $\it E$ is transformed
to the type $\it F$ by replacing each type in the form $\itEffectReturn{\it M_i}$ occurring
as a leaf of $\it E$ with  the corresponding type $\it G_i$. Intuitively, this represents
the composition of computation trees produced by sequential terms.
To see it, consider a term ${\letin{x}{\comput t}{\comput u}}$ with the term $\comput t$ having 
a type $\it E$.
When $\comput t$ terminates, its continuation $\comput u$ refers to a value $\val v$ returned by $\comput t$.
The type of $\val v$ matches $\itEffectReturn{\it M_i}$ for some $i$ in $\it E$
because $\itEffectReturn{\it M_i}$ corresponds to a \emph{leaf} of the
computation tree produced by $\comput t$, namely, a branch of its evaluation
ending at the leaf.
Suppose that $\comput u$ also produces another tree ot type $\it G_i$ using $\val v$.
Then, the produced tree represents the continuation of the process that returns the value $\val v$.
Therefore, replacing the leaf corresponding to the return process by $\it G_i$
results in composing the tree of $\comput t$ with that of $\comput u$.
The leaf replacement relation implements this idea formally.
\Cref{fig:LeafReplacement} illustrates how it works on a simple example.
One may note that it is possible to have identical leaves that are each replaced
by different trees (ie. $\it M_i = \it M_j$ with $\it G_i \neq \it G_j$).
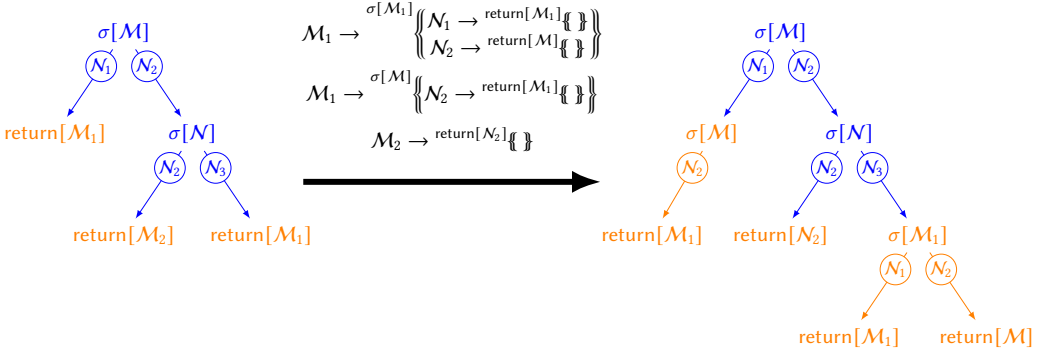
\begin{figure}[ht!]
	\centering
	\scalebox{0.79}{\input{Visuals/LeafReplacement.tex}}
	\caption{\HEBI\ Leaf Replacement: an Example}
	\label{fig:LeafReplacement}
    \Description{HEBI, Leaf Replacement: and example}
\end{figure}

Most of the typing rules presented in \Cref{def:HandlerCompatAndLeafReplacement}
just capture our intention explained thus far, but a few remarks might help
understanding some of the most unusual rules.
First, note that we adopt a multiplicative-style type system, so that whenever a typing rule has multiple premises, the type contexts of the corresponding subderivations are \emph{summed} in the conclusion. Moreover, the rules {\ruleNameVar} and {\ruleNameInt} do not explicitly admit weakening. The context is intended to record the types associated with the different occurrences of variables. If a variable is weakened, this indicates that it is not used, which is represented by assigning it the empty type $\{\}$. Indeed, as already noted, the environment $\Gamma$ is a total mapping and any variable which is not explicitly mentioned in context expressions is assumed to be mapped to the empty type. Consequently, when a variable does not occur free in a term, the environment maps it to the empty type. In premise of the rule {\ruleNameFixBase}, we explicitly write $x:\{\}$ to emphasize that $x$ is required to be mapped precisely to $\{\}$ in the premise.
Fixpoints are typed by two rules: {\ruleNameFixBase} and
{\ruleNameFixRec}.
The former corresponds to the ``base case'', hence assuming recursive variables to be unused by assigning the empty type to it.
The latter corresponds to the ``inductive case'', recursively checking a fixpiont $\fix{x}{\val v}$ against a type $\it N$ if
the body $\val v$ is typechecked under the assignment of $\it N$ to the recursive variable $x$.
Since typing rules are inductively interpreted, the typechecking process for fixpoints can
apply {\ruleNameFixRec} only finitely many times, so, in turn, recursive
variables can be used only finitely many times.
This restriction leads to the guarantee of termination for any well-typed terms.
Finally, the typing rules for handlers ({\ruleNameHandlerRet} and
{\ruleNameHandlerSigma}) only typecheck the clause corresponding to the
operation (including $\ret$) performed first in the underlying term being handled, as explained
in \Cref{sect:highlevel}; if the continuation relies on other clauses, they
are typechecked \emph{within} the typing derivations for the continuation. 
{Finally, by examining rule {\ruleNameHandlerSigma}, we also see why we defined $\{\effectClause{y}{r}{\comput t}\}\handlerCompat\handler h$ in this way, rather than simply requiring that ${\effectClause{y}{r}{\comput t}}$ does not occur in $\handler h$. As noted earlier, the handlers appearing in the premises and in the conclusion of the rule \emph{may} coincide. This makes it possible to type multiple uses of the same effect clause.}

\subsection{Metatheoretical Properties}

Now that the \HEBI\ type system has been properly defined,
we study its properties. Notably, as with intersection types~\cite{IntersectionTypes_survey},
it should verify both subject \wordEmphasis{reduction} and subject \wordEmphasis{expansion}.
In other words, the \HEBI\ typability of a computation 
should be conserved through reduction and expansion
(the opposite of reduction).
Subject reduction ensures that the \HEBI\ type 
system is type sound, while subject expansion would ensure that it is expressive enough. Furthermore, to ensure safety, the system should verify a form of progess (\Cref{lem:HEBIProgress}), so that closed computation terms are either in the form of a return statement, a $\sigma$-operation, or can be reduced, \emph{i.e.} there are no spurious closed and irreducible computations.
In order to prove these properties, let us first note a few
useful lemmas about \HEBI\ types.

\subsection{Basic Properties} \label{sect:HEBIElementaryProperties}

\Cref{lem:Splitting,lem:IntersectionLemma}
given below ensure that value types indeed act as intersections
of their elements.
In other words, terms typed with a value type $\itSet{\rho_i}_{\rangeI}$
are exactly terms typed with $\itSet{\rho_j}_{\rangeJ}$ for
every $J\subseteq I$. Furthermore, the empty type corresponds to an 
empty intersection, and therefore should include all terms.
\begin{restatable}[\HEBI, Splitting Lemma]{lemma}{hebiSplitting}
	\label{lem:Splitting}%
	Let $\val v \in \setVal$ such that
	$\itJugValue[\Xi]{\Gamma}{\val v}{\it M}$ and $I$ such that $\it M=\bigcup\limits_{\rangeI}\it M_i$.
	Then for each $\rangeI$ there exist a derivation
	$\itJugValue[\Pi_i]{\Gamma_i}{\val v}{\it M_i}$ such that $\Gamma= +_{\rangeI}\Gamma_i$.
\end{restatable}
\begin{restatable}[\HEBI, Intersection Lemma]{lemma}{hebiIntersection}
	\label{lem:IntersectionLemma}%
	Let $\val v\in\setVal$ be a value and $I$ a set, 
	such that for every $\rangeI$, there is a $\Pi_i$ such that
	$\itJugValue[\Pi_i]{\Gamma_i}{\val v}{\it M_i}$.
	If $+_{\rangeI}\Gamma_i$ and $\bigcup\limits_{\rangeI}\it M_i$ are defined, 
	then there exists $\itJugValue[\Pi]{+_{\rangeI}\Gamma_i}{\val v}{\bigcup_{\rangeI}\it M_i}$.
\end{restatable}
\noindent
Both of these lemmas can be proved by induction on the structure of the derivations 
which are assumed to exist by hypothesis.

In turn, these results constitute an essential intermediate step in the proof of a pair of key lemmas concerning the relationship between typability and the substitution operation. The first of these results is the classical Substitution Lemma, which states that the operation preserves typing in the forward direction.
\begin{restatable}[\HEBI, Substitution Lemma]{lemma}{hebiSub}
\label{lem:SubstitutionLemma}%
	Let $p \in \setTerm$ be any term and $\val v \in
	\setVal$ closed such that
	$\itJug[\Pi]{\Gamma; x : \it M}{p}{\tau}$ and
	$\itJugValue[\Pi_{\val v}]{\emptyset}{\val v}{\it M}$.
	Then there exists
	$\itJug[\Xi]{\Gamma}
	{p \sub{x}{\val v}}{\tau}$.
\end{restatable}
\begin{proof}[Proof Sketch]
As usual, the proof of the substitution lemma proceeds by induction on the structure of the term into which substitution is performed; in this case, one can equivalently proceed by induction on the structure of $\Pi$. The argument is mostly standard: it is perhaps only worth noting that, in those inductive cases concerning rules with more than one premise, Lemma~\ref{lem:Splitting} is used in an essential way.
\end{proof}
\noindent
Contrary to what happens in most traditional type systems, a result dual to Lemma~\ref{lem:SubstitutionLemma} also holds here, stating that typability is preserved in the backward direction as well.
\begin{restatable}[\HEBI, Anti-substitution Lemma]{lemma}{hebiAntiSub} 
	\label{lem:antiSubstLemma}%
	Let $p \in \setTerm$
	be a term and $\val v\in\setVal$ a closed value.
	If $\itJug[\Pi]{\Gamma}{p\sub{x}{\val v}}{\tau}$, then there exists a type $\it M$ such that 
	$\itJugValue[\Pi_\val v]{\emptyset}{\val v}{\it M}$ and
	$\itJug[\Pi_p]{\Gamma, x : \it M}{p}{\tau}$.
\end{restatable}
\begin{proof}[Proof Sketch]
In this case, it is convenient to build the proof as an induction on the structure of $p$. The base case is therefore the one in which $p$ is a variable, and in that case the thesis follows trivially from the hypothesis concerning $\val v$. The most interesting inductive cases, once again, are those concerning typing rules with at least two premises. In such cases, Lemma \ref{lem:IntersectionLemma} is crucial, and its hypotheses are satisfied precisely because the value $\val v$ is closed and therefore must necessarily be either an abstraction or a ground value. This ensures that the types we assign to $\val v$ are mutually compatible.
\end{proof}

The last auxiliary lemma we consider is about evaluation contexts, and can be proved itself by induction on the structure of derivations. It tells
us that from any type derivation for a closed term in the form $\ctxtEval<\comput t>$ one can extract one for $\comput t$.
\begin{restatable}[\HEBI, Context Typing]{lemma}{hebiContextTyping}
	\label{lem:ContextTyping}%
	Let $\comput t \in \setComput$ and 
	$\ctxtEval$ a computation context such that 
	$\itJugComput[\Pi]{\emptyset}{\ctxtEval<\comput t>}{\it E}$.
	Then there exists an \HEBI\ type $\it F$ 
	such that $\itJugComput{\emptyset}{\comput t}{\it F}$ and 
	for every $\comput u$ such that $\itJugComput{\emptyset}{\comput u}{\it F}$, then it exists
	$\itJugComput[\Xi]{\emptyset}{\ctxtEval<\comput u>}{\it E}$.
\end{restatable}

Using the previous lemmas, we may now show that \HEBI\ typing 
is preserved along reduction and expansion. These properties
are particularly important because they certify that \HEBI\ types
express properties that are stable by reduction,
but also that  these properties are precise enough to be passed on to
any expansion. Thanks to \Cref{lem:ContextTyping}, 
it suffices to consider the seven rewrite rules presented in \Cref{def:RewriteRules}.

\begin{lemma}[\HEBI, Subject Reduction and Expansion] 
	\label{lem:ElementarySubjectReduction}%
	Let $\comput t, \comput s \in \setComput$ closed terms such that $\comput t \reductArr\comput s$. There
	exists
	$\itJugComput[\Pi]{\emptyset}{\comput t}{\it E}$
	if and only if there exists
	$\itJugComput[\Xi]{\emptyset}{\comput s}{\it E}$.
\end{lemma}

{We can also prove, by induction on the derivation, the following Progress Lemma.}

\begin{restatable}[\HEBI, Progress]{lemma}{hebiProgress}
	 \label{lem:HEBIProgress}
    Let $\comput t\in \setComput$ be a closed term such that $\itJugComput[\Pi]{\emptyset}{\comput t}{\it E}$. Then either there exists $\comput s \in \setComput$ such that $\comput t \reductArr \comput s$, or $\comput t$ is a normal form computation, that is, either $\comput t = \ret{\val v}$ or $\comput t = \effect{\val v}{x}{\comput u}$ for some $\val v \in \setVal$ and $\comput u \in \setComput$.
\end{restatable}

\subsection{Reducibility, Reachability, and Termination}
\label{sect:reducibility}
Thanks to Lemma \ref{lem:ElementarySubjectReduction} and to the fact that values can all receive
the empty intersection type, we can conclude that all terminating terms are themselves typable.
The converse inclusion, however, is not a direct consequence of Lemma \ref{lem:ElementarySubjectReduction}
and, as usual, requires an argument of a higher logical complexity, which we are going to describe.

To precisely characterize the properties of \HEBI\ types,
we define a \textbf{reducibility set} for each such type.
The reducibility set of a type $\tau$ corresponds to 
the set of terms that are expected to \emph{behave according} to $\tau$, although not necessarily typable.
For example, the reducibility set of a value type $\it M$ 
consists of values that can be used as intended by all types in $\it M$, while
the reducibility set of a computation type $\it E$ includes all
of computations that behave according to $\it E$, i.e., that perform the
operations specified in $\it E$ in such a way that the continuations can receive
the types as from $\it E$. This definition is intended to ensure that the behavior of terms is 
preserved even when they are encapsulated within an evaluation context,
provided that the context is well-formed.
Reducibility sets of types, contexts and judgements are defined as follows, by induction.

\begin{definition}[Reducibility Sets]
	\label{def:ReducibilityCandidates}
	\input{Definitions/etReducibilityCandidates.tex}
\end{definition}
Some of the inductive cases in \Cref{def:ReducibilityCandidates} are standard. In particular, the treatment of 
values is exactly the same as in traditional intersection types, in which the Tait's reducibility
is well known to work well. Functional values, in particular, are handled by universally quantifying over the elements of the reducibility set of the argument type, requiring that the corresponding application is itself reducible. The treatment of computations and handlers, on the other hand, is novel and deserves some discussion. While the reducibility set of the type $\itEffectReturn{\it M}$ is defined in a very natural way, that of the type $\itEffect{\sigma}{\it M}{\it N_i \rightarrow \it E_i}_{\rangeI}$ 
may appear less conventional: for a term to belong to it, it is required not only that the term reduces to one 
in the appropriate canonical form, but also that the continuation is itself reducible seen as an open term. It is also noteworthy that the reducibility set $\redCand{\it E \Rightarrow \it F}$, whose elements are handlers, behaves structurally analogously to a genuine function type, where, however, the role of term application is played by the construct $\mathsf{with}$.

The first result on reducibility sets we present concerns type contexts, and states that the operator $+$ on type contexts behaves exactly as expected, namely that it gets interpreted as set-theoretic intersection.
\begin{remark}[Context Reducibility] 
	\label{rem:ContextReducibility}%
	If $\vec{\val v} \in \redCand{\Gamma + \Delta}$
	then $\vec{\val v} \in \redCand{\Gamma} \cap \redCand{\Delta}$.
\end{remark}

Another important step is that reducibility sets are stable under reduction and expansion, that is, that a \emph{saturation} property holds.
\begin{restatable}[Reducibility Expansion and Reduction]{lemma}{reducibilityExpansion}
	 \label{lem:ReducibilityExpansionReduction}%
    If $\comput t \reductArr \comput s$ then 
    $\comput s \in \redCand{\it E}$ iff $\comput t \in \redCand{\it E}$.
\end{restatable}
\noindent
Its proof follows from the fact that reduction is deterministic: for every $\comput t$ there is \emph{at most} one term $\comput s$ such that $\comput t \reductArr \comput s$.
The final step establishes the connection between typability and membership in reducibility sets.
\begin{restatable}[Fundamental Lemma]{lemma}{typableReducibility} 
	\label{lem:TypableReducibility}%
	If $\itJug[\Pi]{\Gamma}{p}{\tau}$ then $p \in \redCand{\Gamma \vdash \tau}$.
\end{restatable}
\begin{proof}[Proof Sketch]
	This goes by induction on the structure of $\Pi$. In some inductive cases, we make use of Lemma~\ref{lem:ReducibilityExpansionReduction}.
\end{proof}
This characterization allows us to establish a strong
connection between reachability and \HEBI\ typability.
\begin{theorem}[\HEBI\ Reachability]
	\label{th:EffectTreeReachability}%
	For every closed $\comput t \in \setComput$, $n\in\setNaturals^+$ it holds that
	$\itJugComput{\emptyset}{\comput t}{\itEffectReturn{\itSet{\typeInt{n}}}}$
	if and only if $\comput t \reductArr* \ret{\int n}$.
\end{theorem}
\begin{proof}
	Suppose that $\comput t \in \setComput$ is a closed term and assume $\itJugComput{\emptyset}{\comput t}{\itEffectReturn{\itSet{\typeInt{n}}}}$. From \Cref{lem:TypableReducibility} it follows that $t\in\redCand{\itEffectReturn{\itSet{\typeInt{n}}}}$, from which, by Definition~\ref{def:ReducibilityCandidates}, we deduce that $\comput t \reductArr* \ret{\int n}$. Conversely, if $\comput t \reductArr* \ret{\int n}$, let $m$ be the number of reduction steps. By induction on $m$ one proves that necessarily $\itJugComput{\emptyset}{\comput t}{\itEffectReturn{\itSet{\typeInt{n}}}}$: in the base case this is immediate, while in the inductive step Lemma~\ref{lem:ElementarySubjectReduction} applies.
\end{proof}
The type system \HEBI\ characterizes not only reachability in the sense of Theorem~\ref{th:EffectTreeReachability}, but also termination. In this respect, we say that a closed term $\comput t$ is terminating if the only reduction sequence $t = t_1 \rightarrow t_2 \rightarrow t_3\rightarrow \cdots$ has finite length. Not only are all terms typable in \HEBI\ terminating, as can be readily deduced from Theorem~\ref{lem:TypableReducibility}, but the converse also holds.

%% file: Visuals/itExample.tex
\begin{tabular}{c}
    $

    \hspace{2.5cm}
\begin{tikzpicture}
    [level distance= 22mm, sibling distance= 50mm, 
    edge from parent/.style={-latex}]

\tikzstyle{index}=
    [rectangle, fill=white,
    inner sep=2pt, scale=0.75]
\tikzstyle{C}=[blue]
\tikzstyle{V-}=[draw=blue, text=valueColor]
 
  \node[C] {$ \sigma[\valueColor{\typeInt 3}]$}
    child {node[C] {$\ret*{}[\valueColor{\typeInt 1}]$}
        edge from parent[V-] node[index,V-] {$\typeInt 1$}
    }
    child {node[C] {$ \delta[\valueColor{\itSetEmpty}]$}
        child {node[C] {$\ret*{}[\valueColor{\typeInt 5}]$}
            edge from parent[V-] node[index,V-] {$\itSet{\itSet{\typeInt 3} \rightarrow \itEffectReturn{\typeInt 5}}$}
        }
        child {node[C] {$\delta[\valueColor{\typeInt 3}]$}
            edge from parent[V-] node[index,V-] {$\itSet{\itSet{\typeInt 3} \rightarrow \itEffect{\delta}{\typeInt 3}{}}$}
        }
        edge from parent[V-] node[index,V-] {$\typeInt 2$}
    };
\end{tikzpicture}
\hspace{2.5cm}
    $
\end{tabular}

%% file: Definitions/etType.tex
$
        \begin{array}{c}
            \begin{array}{ccccc}
                \begin{prooftree}
                        \hypo{\it M \in \{\itSetEmpty, \itSet{\typeInt{n}}\}}
                    \typeRuleInt[1]{\itJugValue{\emptyset}{\int n}{\it M}}
                \end{prooftree}
            & \hspace{0.3cm}&
                \begin{prooftree}
                    \typeRuleVar{\itJugValue{x : \it M}{x}{\it M}}
                \end{prooftree}

               &\hspace{0.3cm} &
               \begin{prooftree}
                        \hypo{(\itJugComput{\Gamma_i; x : \it M_i}{\comput t}{\it E_i})_{\rangeI}}
                    \typeRuleAbs{\itJugValue{+_{\rangeI} \Gamma_i}{\abs{x}{\comput t}}{\itSet{\it M_i \rightarrow \it E_i}_{\rangeI}}}
                \end{prooftree}
               
            \end{array} \\[0.5cm]
        
         \begin{array}{ccc}
            \begin{prooftree}
                    \hypo{
                        \itJugValue
                            {\Gamma_1; x : \itSetEmpty}
                            {\val v}
                            {\itSet{\it M_i \rightarrow \it E_i}_{\rangeI}}}
                \typeRuleFixBase{
                    \itJugValue
                        {\Gamma_1}
                        {\fix{x}{\val v}}
                        {\itSet{\it M_i \rightarrow \it E_i}_{\rangeI}}}
            \end{prooftree}& \hspace{0.3cm}&

                \begin{prooftree}
                        \hypo{\itJugValue{\Gamma_1}{\val v}{\itSet{\it M \rightarrow \it E}}}
                        \hypo{\itJugValue{\Gamma_2}{\val w}{\it M}}
                    \typeRuleApp{\itJugComput{\Gamma_1 + \Gamma_2}{\app{\val v}{\val w}}{\it E}}
                \end{prooftree}
            
         \end{array} \\[0.5cm]
            
         \begin{prooftree}
                    \hypo{
                        \itJugValue
                            {\Gamma_1; x : \it N}
                            {\val v}
                            {\itSet{\it M_i \rightarrow \it E_i}_{\rangeI}}}
                    \hypo{
                        \itJugValue
                            {\Gamma_2}
                            {\fix{x}{\val v}}
                            {\it N}}
                \typeRuleFixRec{
                    \itJugValue
                        {\Gamma_1 + \Gamma_2}
                        {\fix{x}{\val v}}
                        {\itSet{\it M_i \rightarrow \it E_i}_{\rangeI}}}
            \end{prooftree}
           
            \\[0.5cm]
            
                \begin{prooftree}
                        \hypo{\itJugComput{\Gamma}{\comput t}{\it E}}
                        \hypo{
                            \it E
                                \replaceLeaf{\it M_i \rightarrow \it G_i}[\rangeI]
                            \it F}
                        \hypo{
                            (\itJugComput
                                {\Gamma_i; x : \it M_i}
                                {\comput u}
                                {\it G_i}
                            )_{\rangeI}}
                    \typeRuleLetin{
                        \itJugComput
                            {\Gamma + (+_{\rangeI} \Gamma_i)}
                            {\letin{x}{\comput t}{\comput u}}
                            {\it F}
                    }
                \end{prooftree}\\[0.5cm]
              
            \begin{array}{ccc}
                
                \begin{prooftree}
                        \hypo{\itJugValue{\Gamma_1}{\val v}{\itSet{\typeInt{m}}}}
                        \hypo{0 < m \leq n}
                       \hypo{\itJugComput{\Gamma_2}{\comput t_m}{\it E}}
                    \typeRuleCase[3]{\itJugComput{\Gamma_1 + \Gamma_2}{\case{\val v}{\comput t_1, \ldots, \comput t_n}}{\it E}}
                \end{prooftree}
                & \hspace{0.3cm}&
                \begin{prooftree}
                        \hypo{\itJugValue{\Gamma}{\val v}{\it M}}
                    \typeRuleRet{\itJugComput{\Gamma}{\ret{\val v}}{\itEffectReturn{\it M}}}
                \end{prooftree}
            \\[0.5cm]
            \end{array}
        \\
            \begin{array}{c}
                \begin{prooftree}
                        \hypo{
                            \itJugValue
                                {\Gamma}
                                {\val v}
                                {\it M}
                        }
                        \hypo{(
                            \itJugComput
                                {\Gamma_i; x : \it N_i}
                                {\comput t}
                                {\it E_i}
                        )_{\rangeI}}
                    \typeRuleEff{
                        \itJugComput
                            {\Gamma + (+_{\rangeI} \Gamma_i)}
                            {\effect{\val v}{x}{\comput t}}
                            {\itEffect{\sigma}{\it M}{\it N_i \rightarrow \it E_i}_{\rangeI}}
                        }
                \end{prooftree}
            \hspace{0.8cm}
                \begin{prooftree}
                    \hypo{\itJugHandler{\Gamma_1}{\handler h}{\it F \Rightarrow \it E}}
                    \hypo{\itJugComput{\Gamma_2}{\comput t}{\it F}}
                    \typeRuleHandle{\itJugComput{\Gamma_1 + \Gamma_2}{\handle{\handler h}{\comput t}}{\it E}}
                \end{prooftree}
            \\[0.5cm]
            \begin{prooftree}
                \hypo{\itJugComput
                    {\Gamma;
                        y : \it M}
                    {\comput t}
                    {\it E}}
                \hypo{\{\retClause{y}{\comput t}\} \handlerCompat \handler h}%
            \typeRuleHandlerRet[2]{\itJugHandler
                {\Gamma}
                {\{\retClause{y}{\comput t}\}\cup \handler h }%
                {\itEffectReturn{\it M} \Rightarrow \it E}
            }%
        \end{prooftree}\\[0.5cm]

        \begin{prooftree}
                        \hypo{%
                            \itJugComput%
                                {\Gamma;
                                    y : \it M;
                                    r : \itSet*{\it N_i \rightarrow \it G_i}_{\rangeI}}%
                                {\comput t}%
                                {\it E}%
                            }%
                        \hypo{(%
                            \itJugHandler%
                                {\Gamma_i}%
                                {\handler h}%
                                {\it F_i \Rightarrow \it G_i}%
                        )_{\rangeI}}%
                        \hypo{\{\effectClause{y}{r}{\comput t}\} \handlerCompat \handler h}%
                    \typeRuleHandlerSigma[3]{%
                        \itJugHandler%
                            {\Gamma + (+_{\rangeI} \Gamma_i)}%
                            {\{\effectClause{y}{r}{\comput t}\} \cup \handler h}%
                            {
                                \itEffect{\sigma}{\it M}
                                    {\it N_i
                                        \rightarrow 
                                    \it F_i
                                }_{\rangeI}
                                    \Rightarrow 
                                \it E
                            }%
                    }%
                \end{prooftree}
            \end{array}
        \end{array}
    $

%% file: Definitions/LeafReplacementHandlerCompat.tex
\begin{tabular}{c}$
    \begin{array}{c}
        \begin{prooftree}
            \infer0{
                \itEffectReturn{\it M}
                    \replaceLeaf{\it M \rightarrow \it E}
                \it E
            }
        \end{prooftree}
    \quad\;
        \begin{prooftree}
                \hypo{\left(
                    \it F_i
                        \replaceLeaf{
                            \it M_k^i \rightarrow \it G_k^i
                            }[\rangeK_i]
                    \it E_i\right)_{\rangeI}
                }
            \infer1{
                \itEffect{\sigma}{\it M}
                        {\it N_i \rightarrow \it F_i}_{\rangeI}
                \replaceLeaf{
                    \it M_k^i \rightarrow \it G_k^i
                    }[\rangeK_i, \rangeI]
                \itEffect{\sigma}{\it M}
                    {\it N_i \rightarrow \it E_i}_{\rangeI}}
        \end{prooftree}
\end{array}
$\end{tabular}

%% file: Visuals/LeafReplacement.tex
\begin{tikzpicture}
    [level distance= 17mm, sibling distance= 23mm, 
    edge from parent/.style={-latex}]

\tikzstyle{index}=
    [near start, circle, fill=white,
    inner sep=0.7pt, scale=0.9]
\tikzstyle{E}=[blue]
\tikzstyle{G}=[orange]
\tikzstyle{E-}=[draw=blue, text=blue]
\tikzstyle{G-}=[draw=orange, text=orange]

  \node[E] at (-5.5,2) {$ \sigma[\it M]$}
    child {node[G] {$\ret*{}[\it M_1]$}
        edge from parent[E-] node[index,E-] {$\it N_1$}
    }
    child {node[E] {$ \sigma[\it N]$}
        child {node[G] {$\ret*{}[\it M_2]$}
            edge from parent[E-] node[index,E-] {$\it N_2$}
        }
        child {node[G] {$\ret*{}[\it M_1]$}
            edge from parent[E-] node[index,E-] {$\it N_3$}
        }
        edge from parent[E-] node[index,E-] {$\it N_2$}
    };

    \draw[line width=1mm, -{Latex[length=5mm]}]      (-2.5,-0.5)   -- (2.5,-0.5);
    \node at (0,2) {$
        \it M_1 \rightarrow \itEffect*{\sigma}{\it M_1}
                { \begin{matrix}
                    \it N_1 \rightarrow \itEffectReturn{\it M_1}\\
                    \it N_2 \rightarrow \itEffectReturn{\it M}
                \end{matrix} } $};
    \node at (0,1) {$
        \it M_1 \rightarrow \itEffect*{\sigma}{\it M}
                {\it N_2 \rightarrow \itEffectReturn{\it M_1}} $};
    \node at (0,0.2) {$ \it M_2 \rightarrow \itEffectReturn{\it N_2}$};

    \node[E] at (5.5,2) {$ \sigma[\it M]$}
    child {node[G] {$ \sigma[\it M]$}
        child {node[G, xshift=-10mm] {$\ret*{}[\it M_1]$}
            edge from parent[G-] node[index,G-] {$\it N_2$}
        }
        edge from parent[E-] node[index,E-] {$\it N_1$}
    }
    child {node[E] {$ \sigma[\it N]$}
        child {node[G] {$\ret*{}[\it N_2]$}
            edge from parent[E-] node[index,E-] {$\it N_2$}
        }
        child {node[G] {$ \sigma[\it M_1]$}
            child {node[G] {$\ret*{}[\it M_1]$}
                edge from parent[G-] node[index,G-] {$\it N_1$}
            }
            child {node[G] {$\ret*{}[\it M]$}
                edge from parent[G-] node[index,G-] {$\it N_2$}
            }
            edge from parent[E-] node[index,E-] {$\it N_3$}
        }
        edge from parent[E-] node[index,E-] {$\it N_2$}
    };
\end{tikzpicture}

%% file: Definitions/etReducibilityCandidates.tex
First, we define for every type $\tau$, for every typing context $\Gamma$ and 
for every expression in the form $\Gamma\vdash\tau$, the corresponding 
reducibility sets  $\redCand{\tau}$, $\redCand{\Gamma}$, and $\redCand{\Gamma\vdash\tau}$ as follows.
\begin{equation*}
    \hspace{-1cm}
    \begin{array}{rcl}
    \\
        \redCand{\itSet{\typeInt{n}}}
            &\coloneqq&
                \{\int n\}
    \\
        \redCand{\itSet{\it M_i \rightarrow \it E_i}_{\rangeI}}
            &\coloneqq&
                \bigcap\limits_{\rangeI} \{\val v \in \setVal
                \vsep \forall\; \val w \in \redCand{\it M_i},
                \; \app{\val v}{\val w} \in \redCand{\it E_i}\}
    \\[0.5cm]

        \redCand{\itEffectReturn{\it M}}
            &\coloneqq&
                \{\comput t \in \setComput
                \vsep \exists\; \val v \in \redCand{\it M},\;
                \comput t \reductArr* \ret{\val v}\}
   \\
        \redCand{\itEffect{\sigma}{\it M}{\it N_i \rightarrow \it E_i}_{\rangeI}}
            &\coloneqq& 
                \{\comput t \in \setComput
                \vsep \exists\; \val v \in \redCand{\it M},\,
                \exists\; \comput u \in \bigcap\limits_{\rangeI}
                    \redCand{x : \it N_i \vdash \it E_i},
                
        \\
            && \hspace{5cm}
                \comput t \reductArr* \effect{\val v}{x}{\comput u}
            \}
    \\[0.5cm]
        \redCand{\it E \Rightarrow \it F}
            &\coloneqq& \{\handler h \in \setHandler
                \vsep \forall\; \comput t \in \redCand{\it E},
                \; \handle{\handler h}{\comput t} \in \redCand{\it F}\}
                \\[0.5cm]

        \redCand{\Gamma}
            &\coloneqq& \{(\val v_i)_{\rangeI}\in\setVal^I
                \vsep \forall \rangeI,\; \val v_i\in\redCand{\it M_i} \} 
                \\
            && \hspace{3cm} \text{ where } \Gamma = (x_i : \it M_i)_{\rangeI}
    \\

        \redCand{\Gamma \vdash \tau}
            &\coloneqq& \{p \in \setVal \cup \setComput \cup \setHandler
                \vsep \forall\; \vec {\val v} \in \redCand{\Gamma},
                \; p\sub{\Gamma}{\vec{\val v}} \in \redCand{\tau}\}
    \end{array}
\end{equation*}
Where for $\Gamma = x_1 : \it M_1\dots x_n : \it M_n$
and $\vec{\val v} \in \redCand{\Gamma}$,
$p\sub{\Gamma}{\vec{\val v}} \coloneqq p\{x_1\backslash \val v_1\dots x_n\backslash \val v_n\}$.

%% file: behaviorsimple.tex
In the last section, we introduced a notion of behavioral intersection typing for \HE\ terms and proved that it precisely characterizes termination and reachability. In doing so, we observed that intersection types naturally acquire a behavioural interpretation, whereby the operations performed become an integral part of the underlying type.

The price to pay for all that, of course, is that type checking necessarily becomes undecidable, as is to be expected in a system featuring intersection types. At this point, however, it is natural to ask whether dispensing with intersections while retaining the behavioural component yields a system with satisfactory metatheoretic properties. By this we mean that typing is preserved under forward reduction, (see \Cref{thm:HEBElementarySubjectReduction}) and that it must satisfy the same form of progess already explained in the case of \HEBI\ (see \Cref{lem:HEBProgress}). On the other hand, we do not require the system to satisfy subject expansion, as is standard for non-intersection type systems.
In this section, we show that this question admits a positive answer. Specifically, we prove that a system, denoted \HEB, obtained from \HEBI\ in the manner outlined above, not only enjoys these desirable properties, but also features a decidable reachability problem.
We establish this decidability result as follows.
First, we define a \textbf{refinement relation} between \HEB\ typing derivations and those of {\HEBI}.
This relation allows us to identify the space of \HEBI\ typing derivations accommodating a term that is well typed in {\HEB}, and ultimately suggests a way to explore it.
We then show that the space is finite and, based on it, that the underlying reachability problem is decidable.
This technique is essentially the same as that employed in parts of the literature on HOMC~\cite{Kobayashi/Ong_2009_LICS,melliesgrellois,clairambaultmurawski}, and it is revisited here by making explicit use of the refinement relation between typing derivations, thereby yielding a more streamlined presentation.

\subsection{The {\HEB} Type System}
\label{sect:hebsystem}
The grammars defining \HEB\ types are as follows.
\begin{equation*}
	\begin{array}{crcl}
		\textbf{(Value Types)}&
		\st M, \st N
		&\coloneqq& \typeInt{n}
		\vsep \st M \rightarrow \st E
		\\
		\textbf{(Computation Types)}&
		\st E, \st F, \st G
		&\coloneqq& \itEffectReturn{\st M} 
        \vsep\itEffect{\sigma}{\st M}{\st N \rightarrow \st E}
		
		\\
		\textbf{(Handler Types)}&
		\st H
		&\coloneqq& \st E \Rightarrow \st F
		
		\\
		\textbf{(General Types)}&
		\st T,\st U
		&\coloneqq& \st M \vsep \st E \vsep \st H
		
	\end{array}
\end{equation*}
As can be readily observed, the \HEB\ types are obtained in a natural way from those of \HEBI\ through a process of uniformization of intersections, while preserving their behavioural structure. In particular, $\lambda$-abstractions in \HEB\ are assigned a \emph{single} arrow type of the form $\st M \rightarrow \st E$ and, similarly, computations that produce the effect $\sigma$ are typed as $\itEffect{\sigma}{\st M}{\st N \rightarrow \st E}$. In \HEBI, computation types use intersections to describe the type-level tree of branching effects produced by a computation. In contrast, \HEB\ prohibits this mechanism, ensuring that the algebraic operations performed by a well-typed computation in \HEB\ do not branch and can be \emph{linearly} ordered. But beware: this does not mean that handlers make a linear use of the continuation, which can in fact be copied; rather, all the multiple uses of the continuation still receive \emph{the same type}.

\begin{sloppypar} The typing rules of the system are presented in Figure~\ref{def:HEBTypingRules}, while the new notion of Leaf Replacement---substantially simplified with respect to that of \HEBI---are given in Figure~\ref{def:SimpleLeafReplacement}. We write $\Delta$ for type contexts of {\HEB}.
The resulting system enjoys robust metatheoretic properties.
Type preservation can be established using entirely standard techniques and, in light of the corresponding property for \HEBI, this is not surprising. Although it does not follow as a direct consequence of the preservation result for \HEBI, it deserves nonetheless to be stated and proved explicitly\end{sloppypar}

\begin{figure}\centering
	\input{Definitions/estType.tex}
	\caption{\HEB, Typing Rules}
	\label{def:HEBTypingRules}
    \Description{HEB, typing rules}
\end{figure}
\begin{figure}
	\centering
	\input{Definitions/SimpleLeafReplacement.tex}
	\caption{\HEB, Leaf Replacement}
	\label{def:SimpleLeafReplacement}
    \Description{HEB, Leaf replacement}
\end{figure}

\begin{theorem}[\HEB\ Subject Reduction]
	\label{thm:HEBElementarySubjectReduction}
	Suppose that $\comput t,\comput s\in \setComput$, that $\stJugComput{\emptyset}{\comput t}{\st E}$ and that $\comput t\reductArr \comput s$. Then, there exists $\Xi$ such that $\stJugComput[\Xi]{\emptyset}{\comput s}{\st E}$.
\end{theorem}

{We also state a Progress Lemma, similarly to what we have already done for \HEBI.}

\begin{restatable}[\HEB\ Progress]{lemma}{hebProgress}
	 \label{lem:HEBProgress}
    Let $\comput t\in \setComput$ be a closed term such that $\stJugComput[\Pi]{\emptyset}{\comput t}{\st E}$. Then either there exists $\comput s \in \setComput$ such that $\comput t \reductArr \comput s$, or $\comput t$ is a normal form computation, that is, either $\comput t = \ret{\val v}$ or $\comput t = \effect{\val v}{x}{\comput u}$ for some $\val v \in \setVal$ and $\comput u \in \setComput$.
\end{restatable}

The \HEB\ typing can be regarded as a mechanism for ensuring the absence of type errors in presence of handlers. Unlike the situation for \HEBI, however, this verification technique may yield false negatives, namely, it is \emph{in}complete with respect to reachability. For example, each term that terminates but has an handler that uses the continuation in two or more ways which are incompatible type-theoretically, would not be typable in \HEB, while it would be in \HEBI. That said, given a term typable in \HEB\ with an integer type, would it be possible to \emph{automatically} determine the value to which it reduces, if any? The remainder of this section is devoted to answering this question in the affirmative. This might seem surprising, given the corresponding negative result for {\HEPCF}~\cite{Dal-Lago/Ghyselen_2024_POPL}, but the positive result in {\HEB} relies on the fact that the {\HEB} type system can still capture the computation trees produced by terms, unlike {\HEPCF} (the class of capturable trees is more restrictive than that in {\HEBI}, though). 
{A natural question to ask is then whether type checking is itself decidable for \HEB. Although we do not study this question in the present work, we conjecture that it is, particularly in a setting where the programmer is required to annotate variables with their types. Although computation types in \HEB\ have an inherently behavioral nature, checking that a handler is compatible with a given computation type does not seems to be too difficult: even when such compatibility may require assigning multiple types to the same handler, the behavioral type gives an explicit bound on the number of distinct typing. } 

{We conclude this section with a brief discussion on the expressive power of \HEB. In particular, this system seems well-suited for scenarios in which it is important to track the order in which operations are executed, as well as to assign different types to the same handler depending on the context. For example, consider a program that relies on a polymorphic algebraic operation for printing a value that may be either a string or a numeric type. In such a case, \HEB\ could help capture both the sequencing constraints and the variation in handler behavior through its typing discipline.}

\subsection{Behavioural Refinement}
This section introduces a refinement relation, which certifies that the {\HEBI} typing provides strictly more fine-grained information about terms that are typable in {\HEB}.
We start by defining a refinement relation $\refrel{\st M}{\it M}$ on pairs of an {\HEB} type $\st M$ and an {\HEBI} type $\it M$.
This relation identifies the set of {\HEBI} types $\it M$ at which a term of an {\HEB} type $\st M$ may be typable.
\begin{definition}[\HEB, Type Refinement]
	Let $\st T$ be an \HEB\ type and $\it \tau$ an \HEBI\ type. We define the refinement relation between them in \Cref{def:HEBRefinementTypes}. Let $\Delta$ be an \HEB\ type context  and $\Gamma$ be an \HEBI\ type context. We write $\refrel{\Delta}{\Gamma}$ if, for each $x:\it M\in \Gamma$, there exists $x::\st M\in\Delta$ such that $\refrel{\st M}{\it M}$.
\end{definition}
It is worth noting that when $\refrel{\st M}{\it M}$, the two types $\st M$ and $\it M$ have \emph{exactly} the same structure, but they differ in the fact that whenever intersections occur in $\it M$, \emph{all} such intersections must refine the corresponding type in $\st M$.
We define $\Refn(\st M)$ to be the set of \HEBI\ types $\it M$ such that $\refrel{\st M}{\it M}$.
Notice that $\Refn(\typeInt{m}) = \{ \itSetEmpty \} \cup \{ \itSet{\typeInt{n}} \mid 0<n \leq m \}$.
This indicates that, for any positive integer $m$ and value $\val v$, if $\stJugValue{}{\val v}{\typeInt{m}}$ holds in {\HEB}, $\itJugValue{}{\val v}{\itSetEmpty}$ or $\itJugValue{}{\val v}{\itSet{\typeInt{n}}}$ also holds in {\HEBI} for some $n$, and, furthermore, in the latter case, $0<n \leq m$ must hold.
The condition $0<n \leq m$ allows us to \emph{bound} the number of the {\HEBI} types that can be assigned to $\val v$, by $m$.
This enables computing an {\HEBI} typing derivation from an {\HEB} typing derivation (\Cref{prop:RefinementComputability}).
\begin{figure}
	\centering
	\input{Definitions/NewRefinementTypes.tex}
	\caption{\HEB, Refinement on Types.}
    \Description{HEB, refinements on types}
	\label{def:HEBRefinementTypes}
\end{figure}

We can now extend the refinement relation from types and typing contexts to typing \emph{derivations}, as shown in \Cref{def:HEBDerivationRefinement}: $\refrel{\Pi}{\Xi}$ states that an \HEBI\ derivation $\Xi$ refines an \HEB\ derivation $\Pi$. This definition is given inductively on the structure of type derivations: $\Xi$ refines $\Pi$ if the derivations of their respective premises are themselves refinements of each other, and this must hold at all levels. For integer values, the constraint imposed by the \HEBI\ derivation ensures that the only admissible refinements are those that correspond exactly to the value $\int n$ being typed. Moreover, the \ruleNameFix\ rule can be refined by each of the two {\HEBI} tying rules for fixpoints.

\begin{figure}
\centering
\input{Definitions/NewRefinementRelation.tex}
\caption{\HEB\ Refinement for Rules, Excerpt.}
\label{def:HEBDerivationRefinement}
\Description{HEB, Refinement for rules}
\end{figure}

\subsection{Behavioural Refinement Expansion}
\label{sect:HEBRedExp}
The refinement relation between type derivations we have just introduced was defined in order to restrict the search space of \HEBI\ derivations whenever the term under consideration is typable in \HEB. But is all this correct? In this section, we will prove that this is indeed the case.

Suppose we are dealing with a closed term $\comput t$ typable via the \HEB\ type derivation $\Pi$ with an integer type $\typeInt{n}$, and suppose that we wish to study to which specific value $m\in[1,n]$, if any, such a term reduces. Each such possible value can be seen as a judgement $J$ of \HEBI\ refining the conclusion of $\Pi$. Thus, in the spirit of Theorem \ref{th:EffectTreeReachability}, we might attempt to construct a type derivation for the underlying term having conclusion $J$. However, the search space remains too large if we somehow forget that $\comput t$ is typable in \HEB. We may instead restrict ourselves to searching for one such derivation $\Xi$ \emph{refining} $\Pi$. If such a derivation $\Xi$ existed, the property in question would certainly hold, thanks to Theorem~\ref{th:EffectTreeReachability}. Is it also the case, however, that whenever the reachability property is satisfied, the search can be restricted to \HEBI\ derivations refining $\Pi$? In this section we prove this property, by a relative form of subject expansion property.



It is convenient, on the \HEB\ side, to lift the reduction relation $\reductArr'$ from terms to typing derivations. We give here only a brief intuition of the construction.
The goal is to mimick the transformation induced by \Cref{thm:HEBElementarySubjectReduction} on a derivation $\Pi$ typing a computation term $\comput t$.
We start by defining a notion of substitution directly on type derivations. 
Let $\stJug[\Pi]{\Delta;x::\st M}{p}{\st T}$ and $\stJugValue[\Xi]{\emptyset}{\val v}{\st M}$ be two \HEB\ derivations and $\val v\in \setVal$ a closed term. We define a derivation $\stJug[\Pi\sub{x}{\Xi}]{\Delta}{p\sub{x}{\val v}}{\st T}$ by induction on $\Pi$. Intuitively, this is the derivation obtained by replacing each instance of a \ruleNameVar\ rule involving $x$ in $\Pi$ with the entire derivation $\Xi$, mirroring the term substitution $p\sub{x}{\val v}$ at the level of derivations. 

More formally, we have that: 
\[
\Bigl(\begin{prooftree}
     \typeRuleVar{\stJugValue[\Pi]{\Delta;x::\st M}{x}{\st M}}
\end{prooftree}    
\Bigl)\sub{x}{\Xi}=\Pi^s
\]
Where $\Pi^s$ is obtained from $\Xi$ thanks to the admissibility of weakening in the system. In the inductive cases, the substitution is propagated to the premises of the rule, \emph{e.g.}
	\[
	\begin{prooftree}
		\hypo{\stJugComput[\Pi_\comput t]
			{\Delta;x::\st M; y :: \st N}
			{\comput t}{\st E}}
		\typeRuleAbs{\stJugValue
			{\Delta;x::\st M}
			{\abs{y}{\comput t}}
			{\st N \rightarrow \st E}}
	\end{prooftree}
	\sub{x}{\Xi}=
	\begin{prooftree}
		\hypo{\stJugComput[\Pi_\comput t\sub{x}{\Xi}]
			{\Delta; y :: \st N}
			{\comput t\sub{x}{\val v}}{\st E}}
		\typeRuleAbs{\stJugValue
			{\Delta}
			{\abs{y}{\comput t}\sub{x}{\val v}}
			{\st N \rightarrow \st E}}
	\end{prooftree}
	\]

\begin{figure}
    \scalebox{0.87}{\input{Definitions/HEBReductions.tex}}
    \caption{Example of reduction between two \HEB\ derivations.}
    \label{def:HEBDerivationReduction}
    \Description{HEB, example of reduction between derivations.}
\end{figure}

Let $\comput t\in \setComput $ be a closed term and $\stJugComput[\Pi]{\emptyset}{\comput{t}}{\st E}$ an \HEB\ derivation. We can now define a reduction on derivations, written $\Pi\rightsquigarrow \Pi'$, by minimicking the rules in \Cref{def:RewriteRules}. Two examples are given in \Cref{def:HEBDerivationReduction}.

\begin{lemma}[Relative Subject Expansion] 
	\label{lem:NewElementarySubjectExpansion-ref}%
	Let $\comput t, \comput s \in \setComput$ be two closed term such that 
	$\comput t \reductArr \comput s$. Let
	$\itJugComput[\Xi']{\emptyset}{\comput s}{\it E}$ and $\stJugComput[\Pi']{\emptyset}{\comput s}{\st E}$ be such that $\refrel{\Pi'}{\Xi'}$ and $\stJugComput[\Pi]{\emptyset}{\comput t}{\st E}$ be such that $\Pi\rightsquigarrow \Pi'$, 
	then there exists
	$\itJugComput[\Xi]{\emptyset}{\comput t}{\it E}$ such that $\refrel{\Pi}{\Xi}$.
\end{lemma}
\noindent 
The proof of Lemma \ref{lem:NewElementarySubjectExpansion-ref} goes by cases on $\Pi\rightsquigarrow \Pi'$, and exploits some auxiliary results along the way. But why is reduction on derivations actually needed? In \Cref{lem:NewElementarySubjectExpansion-ref} we have two derivations $\Pi'$ and $\Xi'$ typing $\comput s$ such that $\refrel{\Pi'}{\Xi'}$. In general, subject expansion fails for \HEB\,, therefore we have to suppose the existence of a derivation $\Pi$ typing $\comput t$. Our goal now is to construct a refinement $\Xi$ of $\Pi$. However, $\Pi$ need not bear any relationship with $\Pi'$ and may be completely different from it. The condition $\Pi\rightsquigarrow\Pi'$ ensures that $\Pi$ is suitably close to $\Pi'$, allowing the construction of a derivation $\Xi$ from $\Xi'$ and such that $\refrel{\Pi}{\Xi}$. By the way, a similar approach has already been used in~\cite{dudenhefner2024mechanized} to prove a form of Relative Subject Expansion.

\subsection{Finite Refinements and Decidability}
\label{sect:refinementsDecidability}
Thanks to the results obtained in the previous subsection, and in particular to Relative Subject Expansion, we can say that restricting the search for an \HEBI\ derivation to those derivations that refine a given \HEB\ derivation does not result in any loss of generality. However, it is not immediate that this can be done \emph{automatically}. This is indeed the case, and the key observation from which to start to show it is that every  \HEB\ type, by construction, admits a finite number of refinements in \HEBI:
\begin{restatable}[Finite Refinement Property]{lemma}{finiteRefinement}
	\label{lem:FiniteRefinement}
    Let $\st T$ be an \HEB\ type. There exists a finite number of \HEBI\ types $\tau$ such that $\refrel{\st T}{\tau}$.
\end{restatable}
\noindent
This is an easy induction on the structure of $\st T$. Noticeably, in the inductive cases one exploits the fact that the powerset of a finite set is finite. As a consequence, the number of refinements of $\st T$ can be exponential on the size of $\st T$.

The existence of a refinement with conclusion $J$ for a given \HEB\ derivation $\Pi$ can be expressed by means of a logical formula $\algo{\Pi}{J}$, defined by induction on $\Pi$, some of whose defining clauses are given in Figure \ref{def:AlgorithmLogic}.
\begin{figure}
    \centering
	\scalebox{0.89}{\input{Visuals/AlgorithmLogic.tex}}
	\caption{Description of the function $\mathcal{A}$.}
	\label{def:AlgorithmLogic}
    \Description{Description of the function A}
\end{figure}
When the underlying term is, for example, a $\lambda$-abstraction, the formula $\algo{\Pi}{J}$ is naturally obtained as the logical conjunction of some  formulas produced by applying the same algorithm to the subderivations of $\Pi$. When, instead, the term is an application, a fixed point, or a handled term, it is necessary to existentially quantify over the refinements of the type that such typing rule ``hides'' and that is not visible in the conclusion of $\Pi$. It is worth noting that the formula $\algo{\Pi}{J}$ has been constructed so as to be interpretable solely on the basis of the objects $\Pi$ and $J$ and, therefore, is a closed formula. It thus makes sense to write $\models\algo{\Pi}{J}$, to mean that such a formula is true.

The first thing we prove about the function $\mathcal{A}$ is that the produced logical condition indeed captures the informal specification around which $\mathcal{A}$  was constructed. This is expressed by the following result.
\begin{proposition}[Correctness of $\mathcal{A}$]
    \label{prop:AlgoCorrectness}
	Let $\Pi$ be an \HEB\ type derivation with conclusion $J_s$ and $J_r$ an \HEBI\ type judgement such that $\refrel{J_s}{J_r}$. It holds that $\models\algo{\Pi}{J_r}$ if and only if there exists an \HEBI\ type derivation $\Xi$, with conclusion $J_r$, such that $\refrel{\Pi}{\Xi}$.
\end{proposition}
The function $\mathcal{A}$, in addition to returning a formula whose truth value is the expected one, is  computable, which can be established by observing that all the quantifiers appearing in the produced formula range over sets which, thanks to Lemma~\ref{lem:FiniteRefinement}, we know to be finite. Consequently:
\begin{restatable}[Computability of $\mathcal{A}$]{proposition}{refinementComputability}
 \label{prop:RefinementComputability}
	The function $\mathcal{A}$ is computable, i.e., there exists an algoritm that given in input a type derivation $\Pi$ and
	a judgement $J$ determine whether the formula $\algo{\Pi}{J}$ holds.
\end{restatable}
We finally reached our target result.
\begin{restatable}[Decidability]{theorem}{decidability}
	The reachability problem for terms typable in \HEB\ is decidable.
\end{restatable}

%% file: Definitions/estType.tex
$
        \begin{array}{c}
           \hspace{0.4cm} \begin{array}{ccccc}
                \begin{prooftree}
                    \hypo{0<n\le m}
                    \typeRuleInt[1]{\stJugValue{\Delta}{\int n}{\typeInt{m}}}
                \end{prooftree}
                &\hspace{0.87cm} &
                \begin{prooftree}
                    \typeRuleVar{\stJugValue{\Delta;x :: \st M}{x}{\st M}}
                \end{prooftree}
           
             & \hspace{0.87cm} &
                \begin{prooftree}
                        \hypo{\stJugComput{\Delta; x :: \st M}{\comput t}{\st E}}
                    \typeRuleAbs{\stJugValue{\Delta}{\abs{x}{\comput t}}{\st M \rightarrow \st E}}
                \end{prooftree}\hspace{0.4cm}
                \end{array}\\[0.5cm]
            \begin{array}{ccc}
            \begin{prooftree}
                    \hypo{
                        \stJugValue
                            {\Delta; x :: \st M \rightarrow \st E}
                            {\val v}
                            {\st M \rightarrow \st E}}
                    
                \typeRuleFix[1]{
                    \stJugValue
                        {\Delta }
                        {\fix{x}{\val v}}
                        {\st M \rightarrow \st E}}
            \end{prooftree}

            &\hspace{0.9cm} &

                \begin{prooftree}
                        \hypo{\stJugValue{\Delta}{\val v}{\st M \rightarrow \st E}}
                        \hypo{\stJugValue{\Delta}{\val w}{\st M}}
                    \typeRuleApp{\stJugComput{\Delta}{\app{\val v}{\val w}}{\st E}}
                \end{prooftree}

            \end{array}
        
            \\[0.5cm]
            
                \begin{prooftree}
                        \hypo{\stJugComput{\Delta}{\comput t}{\st E}}
                        \hypo{
                            \st E
                                \replaceLeaf{\st M \rightarrow \st G}
                            \st F}
                        \hypo{
                            \stJugComput
                                {\Delta; x :: \st M}
                                {\comput u}
                                {\st G}
                            }
                    \typeRuleLetin{
                        \stJugComput
                            {\Delta}
                            {\letin{x}{\comput t}{\comput u}}
                            {\st F}
                    }
                \end{prooftree}\\[0.5cm]
           
            \begin{array}{ccc}
                \begin{prooftree}
                        \hypo{\stJugValue{\Delta}{\val v}{\typeInt{n}}}
                        \hypo{(\stJugComput{\Delta}{\comput t_i}{\st E}
                            )_{ 0<i\le n}}
                    \typeRuleCase[2]{\stJugComput{\Delta }{\case{\val v}{\comput t_1, \ldots, \comput t_n}}{\st E}}
                \end{prooftree}
                &\hspace{0.9cm}&
                \begin{prooftree}
                        \hypo{\stJugValue{\Delta}{\val v}{\st M}}
                    \typeRuleRet{\stJugComput{\Delta}{\ret{\val v}}{\itEffectReturn{\st M}}}
                \end{prooftree}
            \\[0.5cm]
            \end{array}
        \\
            \begin{array}{c}
                \begin{prooftree}
                        \hypo{
                            \stJugValue
                                {\Delta}
                                {\val v}
                                {\st M}
                        }
                        \hypo{
                            \stJugComput
                                {\Delta; x :: \st N}
                                {\comput t}
                                {\st E}
                        }
                    \typeRuleEff{
                        \stJugComput
                            {\Delta}
                            {\effect{\val v}{x}{\comput t}}
                            {\itEffect{\sigma}{\st M}{\st N \rightarrow \st E}}
                        }
                \end{prooftree}
            \qquad\qquad\quad
                \begin{prooftree}
                    \hypo{\stJugHandler{\Delta}{\handler h}{\st F \Rightarrow \st E}}
                    \hypo{\stJugComput{\Delta}{\comput t}{\st F}}
                    \typeRuleHandle{\stJugComput{\Delta}{\handle{\handler h}{\comput t}}{\st E}}
                \end{prooftree}
            \\[0.5cm]
          \begin{prooftree}
                \hypo{\stJugComput
                    {\Delta;
                        y :: \st M}
                    {\comput t}
                    {\st E}}
                \hypo{\{\retClause{y}{\comput t}\} \handlerCompat \handler h}%
            \typeRuleHandlerRet[2]{\stJugHandler
                {\Delta}
                {\{\retClause{y}{\comput t}\}\cup \handler h}%
                {\itEffectReturn{\st M} \Rightarrow \st E}
            }%
        \end{prooftree}\\[0.5cm]

            \begin{prooftree}
                        \hypo{%
                            \stJugComput%
                                {\Delta;
                                    y :: \st M;
                                    r :: \st N \rightarrow \st G}%
                                {\comput t}%
                                {\st E}%
                            }%
                        \hypo{%
                            \stJugHandler%
                                {\Delta}%
                                {\handler h}%
                                {\st F \Rightarrow \st G}%
                        }%
                        \hypo{\{\effectClause{y}{r}{\comput t}\} \handlerCompat \handler h}%
                    \typeRuleHandlerSigma[3]{%
                        \stJugHandler%
                            {\Delta}%
                            {\{\effectClause{y}{r}{\comput t}\} \cup \handler h}%
                            {
                                \itEffect{\sigma}{\st M}
                                    {\st N
                                        \rightarrow 
                                    \st F
                                }
                                    \Rightarrow 
                                \st E
                            }%
                    }%
                \end{prooftree}

            \end{array}
        \end{array}
    $

%% file: Definitions/SimpleLeafReplacement.tex
$
    \begin{array}{ccc}
        \begin{prooftree}
            \infer0{
                \itEffectReturn{\st M}
                    \replaceLeaf{\st M \rightarrow \st E}
                \st E
            }
        \end{prooftree}
        &\hspace{0.5cm}&
        \begin{prooftree}
                \hypo{
                    \st F
                        \replaceLeaf{
                            \st M \rightarrow \st G
                            }
                    \st E
                }
            \infer1{
                \itEffect{\sigma}{\st M'}
                        {\st N \rightarrow \st F}
                \replaceLeaf{
                    \st M \rightarrow \st G
                    }
                \itEffect{\sigma}{\st M'}
                    {\st N \rightarrow \st E}}
        \end{prooftree}
    \end{array}
$

%% file: Definitions/NewRefinementTypes.tex
$
        \begin{array}{c}
            \begin{array}{ccccccc}
                \begin{prooftree}
                    \infer0{
                        \refrel{\typeInt n}{\itSetEmpty}
                    }
                \end{prooftree}
                &\hspace{0.6cm}&
                \begin{prooftree}
                    \hypo{0<n \le m}
                    \infer1{\refrel{\typeInt{m}}{\itSet{\typeInt{n}}}
                    }
                \end{prooftree}
                  &\hspace{0.6cm}&
                   \begin{prooftree}

                    \hypo{\refrel{\st{M}}{\it M} }

                    \hypo{\left(
                        \refrel{\st{N}}{\it N_i} 
                    \right)_{\rangeI}}

                    \hypo{\left(
                        \refrel{\st E}{\it E_i} 
                    \right)_{\rangeI}}
                \infer3{\refrel{\itEffect{\sigma}{\st M}{\st N \rightarrow \st E}}{\itEffect{\sigma}{\it M}{\it N_i \rightarrow \it E_i}_{\rangeI}}
                }
            \end{prooftree}
            \end{array}
        \\[0.5cm]

        \begin{array}{ccccc}

                \begin{prooftree}
                    \hypo{\left(
                        \refrel{\st{M}}{\it M_i}
                    \right)_{\rangeI}}

                    \hypo{\left(
                        \refrel{\st{U}}{\it E_i} 
                    \right)_{\rangeI}}
                \infer2{\refrel{\st{M} \rightarrow \st{U}}{\itSet{\it M_i \rightarrow \it E_i}_{\rangeI}}
                }
            \end{prooftree}
             & \hspace{0.6cm}&
                \begin{prooftree}
                        \hypo{\refrel{\st{M}}{\it M} }
                    \infer1{\refrel{\itEffectReturn{\st M}}{\itEffectReturn{\it M}}
                    }
                \end{prooftree}
        &\hspace{0.6cm}&
                \begin{prooftree}
                        \hypo{\refrel{\st S}{\it E} }
                        \hypo{\refrel{\st U}{\it F}}
                    \infer2{
                        \refrel{\st S \Rightarrow \st U}{\it E \Rightarrow \it F}
                    }
                \end{prooftree}
       
            \end{array}
        \end{array}
    $

%% file: Definitions/NewRefinementRelation.tex
$
    \begin{array}{c}
        \begin{array}{ccc}
            \begin{prooftree}
            \hypo{0<n \le m}
            \typeRuleInt[1]{\itJugValue{\Delta}{\int n}{\typeInt{m}}}
        \end{prooftree}
        \refrelsym
        \begin{prooftree}
                \hypo{\it M \in \{\itSetEmpty, \itSet{\typeInt{n}}\}}
            \typeRuleInt[1]{\itJugValue{\emptyset}{\int n}{\it M}}
        \end{prooftree}
        &\hspace{0.1cm} &
        \begin{prooftree}
            \typeRuleVar{\stJugValue{\Delta;x :: \st M}
                {x}{\st M}}
        \end{prooftree}
        \refrelsym
        \begin{prooftree}
            \typeRuleVar{\itJugValue{x : \it M }
                {x}{\it M}}
        \end{prooftree}\\
        & & 
         \text{if } \refrel{\st M}{\it M}
        \end{array}
    
    \\[0.7cm]
        \begin{prooftree}
                \hypo{
                    \stJugValue[\Pi_1]
                        {\Delta; x::\st M\rightarrow \st E}
                        {\val v}
                        {\st M \rightarrow \st E}}
            \typeRuleFix[1]{
                \stJugValue
                    {\Delta}
                    {\fix{x}{\val v}}
                    {\st M \rightarrow \st E}}
        \end{prooftree}
        \refrelsym 
         \begin{prooftree}
                \hypo{
                    \itJugValue[\Xi_1]
                        {\Gamma_1; x:\itSetEmpty}
                        {\val v}
                        {\itSet{\it M_i \rightarrow \it E_i}_{\rangeI}}}
            \typeRuleFixBase{
                \stJugValue
                    {\Gamma_1 }
                    {\fix{x}{\val v}}
                    {\itSet{\it M_i \rightarrow \it E_i}_{\rangeI}}}
        \end{prooftree}
        \\[0.5cm]
        \text{if } \refrel{\Pi_1}{\Xi_1}\\[0.3cm]

        \begin{prooftree}
                \hypo{
                    \stJugValue[\Pi_1]
                        {\Delta; x::\st M\rightarrow \st E}
                        {\val v}
                        {\st M \rightarrow \st E}}
            \typeRuleFix[1]{
                \stJugValue[\Pi]
                    {\Delta}
                    {\fix{x}{\val v}}
                    {\st M \rightarrow \st E}}
        \end{prooftree}
        \\[0.5cm]
        \refrelsym 
         \begin{prooftree}
                \hypo{
                    \itJugValue[\Xi_1]
                        {\Gamma_1; x:\it N}
                        {\val v}
                        {\itSet{\it M_i \rightarrow \it E_i}_{\rangeI}}}
                \hypo{
                    \itJugValue[\Xi_2]
                        {\Gamma_2}
                        {\fix{x}{\val v}}
                        {\it N}}
            \typeRuleFixRec{
                \itJugValue
                    {\Gamma_1 + \Gamma_2}
                    {\fix{x}{\val v}}
                    {\itSet{\it M_i \rightarrow \it E_i}_{\rangeI}}}
        \end{prooftree}
        \\[0.5cm]
        \text{if } \refrel{\Pi_1}{\Xi_1}  \text{ and } \refrel{\Pi}{\Xi_2}\\[0.4cm]

        \begin{prooftree}
                \hypo{\stJugComput[\Pi_1]{\Delta}{\comput t}{\st E}}
            \hypo{
                    \st E
                        \replaceLeaf{\st M \rightarrow \st G}
                    \st F}
                \hypo{\stJugComput[\Pi_2]
                        {\Delta; x :: \st M}
                        {\comput u}
                        {\st G}}
            \typeRuleLetin[3]{
                \stJugComput
                    {\Delta}
                    {\letin{x}{\comput t}{\comput u}}
                    {\st F}
            }
        \end{prooftree}
    \\[0.5cm]
    \refrelsym 
    \begin{prooftree}
                \hypo{\itJugComput[\Xi]{\Gamma}{\comput t}{\it E}}
                \hypo{
                    \it E
                        \replaceLeaf{\it M_i \rightarrow \it G_i}[\rangeI]
                    \it F}
                \hypo{
                    (\itJugComput[\Xi_i]
                        {\Gamma_i; x : \it M_i}
                        {\comput u}
                        {\it G_i}
                    )_{\rangeI}}
            \typeRuleLetin{
                \itJugComput
                    {\Gamma +_{\rangeI} \Gamma_i}
                    {\letin{x}{\comput t}{\comput u}}
                    {\it F}
            }
        \end{prooftree}\\[0.5cm]
        \text{if } \refrel{\Pi_1}{\Xi} \text{ and } \refrel{\Pi_2}{\Xi_i} \text{ for each } \rangeI
        \\[0.4cm]
     \begin{prooftree}
                \hypo{\stJugValue[\Pi_1]
                    {\Delta}
                    {\val v}
                    {\st M}}
                \hypo{
                    \stJugComput[\Pi_2]
                        {\Delta; x :: \st N}
                        {\comput t}
                        {\st E}}

            \typeRuleEff[2]{
                \stJugComput
                    {\Delta }
                    {\effect{\val v}{x}{\comput t}}
                    {\itEffect{\sigma}{\st M}{\st N \rightarrow \st E}}
                }
        \end{prooftree}
    \\[0.5cm]
    \refrelsym 
     \begin{prooftree}
                \hypo{\itJugValue[\Xi]
                    {\Gamma}
                    {\val v}
                    {\it M}}
                \hypo{(
                    \itJugComput[\Xi_i]
                        {\Gamma_i; x : \it N_i}
                        {\comput t}
                        {\it E_i}
                )_{\rangeI}}
            \typeRuleEff{
                \itJugComput
                    {\Gamma +_{\rangeI} \Gamma_i}
                    {\effect{\val v}{x}{\comput t}}
                    {\itEffect{\sigma}{\it M}{\it N_i \rightarrow \it E_i}_{\rangeI}}
                }
        \end{prooftree}\\[0.5cm]
        \text{if } \refrel{\Pi_1}{\Xi} \text{ and } \refrel{\Pi_2}{\Xi_i} \text{ for each }\rangeI
        \\[0.4cm]
        \begin{prooftree}
    \hypo{
        \stJugComput[\Pi_1]
            {\Delta; y :: \st M;r :: \st N \rightarrow \st G}
            {\comput t}
            {\st E}
            }
            \hypo{\begin{matrix}
                \{\effectClause{y}{r}{\comput t}\} \handlerCompat \handler h\\
               \stJugHandler[\Pi_2]
                    {\Delta}%
                    {\handler h}%
                    {\st F \Rightarrow \st G}
            \end{matrix}}
            \typeRuleHandlerSigma[2]{\stJugHandler
                {\Delta }
                {\{\effectClause{y}{r}{\comput t}\} \cup \handler h}
                {\itEffect{\sigma}{\st M}
                                    {\st N
                                        \rightarrow 
                                    \st F
                                }
                                    \Rightarrow 
                                \st E}
            }
            \end{prooftree}
        \\[0.6cm] \refrelsym
        \begin{prooftree}
                        \hypo{%
                            \itJugComput[\Xi]
                                {\Gamma;
                                    y : \it M;
                                    r : \itSet*{\it N_i \rightarrow \it G_i}_{\rangeI}}%
                                {\comput t}%
                                {\it E}%
                            }%

                    \hypo{\begin{matrix}
                \{\effectClause{y}{r}{\comput t}\} \handlerCompat \handler h\\
               (
                            \itJugHandler[\Xi_i]
                                {\Gamma_i}%
                                {\handler h}%
                                {\it F_i \Rightarrow \it G_i}%
                        )_{\rangeI}
            \end{matrix}}

                    \typeRuleHandlerSigma[2]{%
                        \itJugHandler%
                            {\Gamma +_{\rangeI} \Gamma_i}%
                            {\{\effectClause{y}{r}{\comput t}\} \cup \handler h}%
                            {
                                \itEffect{\sigma}{\it M}
                                    {\it N_i
                                        \rightarrow 
                                    \it F_i
                                }_{\rangeI}
                                    \Rightarrow 
                                \it E
                            }%
                    }%
        \end{prooftree}\\[0.7cm]
        \text{if } \refrel{\Pi_1}{\Xi} \text{ and } \refrel{\Pi_2}{\Xi_i} \text{ for each }\rangeI
    \end{array}
$

%% file: Definitions/HEBReductions.tex
$
\begin{array}{ccc}
\begin{prooftree}
			\hypo{\stJugComput[\Pi_{\comput u}]{ x :: \st M}{\comput u}{\st E}}
			\typeRuleAbs{\stJugValue{\emptyset}{\abs{x}{\comput u}}{\st M\rightarrow \st E}}
			\hypo{\stJugValue[\Pi_{\val v}]{\emptyset}{\val v}{\st M}}
			\typeRuleApp{\stJugComput{\emptyset }{\app{(\abs{x}{\comput u})}{\val v}}{\st E}}
		\end{prooftree}
		& \rightsquigarrow & \stJugComput[\Pi_{\comput u}\sub{x}{\Pi_\val v}]{\emptyset}{\comput u\sub{x}{\val v}}{\st E}\\ \\

		\begin{prooftree}
			\hypo{
				\stJugValue[\Pi_p]
				{x :: \st M \rightarrow \st E}
				{\val v}
				{\st M \rightarrow \st E}}
			\typeRuleFix[1]{
				\stJugValue[\Pi_\val v]
				{\emptyset}
				{\fix{x}{\val v}}
				{\st M \rightarrow \st E}}
			\hypo{\stJugValue{\emptyset}{\val w}{\st M}}
			\typeRuleApp{\stJugComput{\emptyset}{(\fix{x}{\val v})\val w}{\st E}}
			
		\end{prooftree} 
		& \rightsquigarrow & 
		\begin{prooftree}
            \hypo{
                \begin{matrix}
			\hypo{\stJugValue{\emptyset}{\val w}{\st M}}\\
            \hypo{\stJugValue[\Pi_{p}\sub{x}{\Pi_\val v}]{\emptyset}{\val v\sub{x}{\fix{x}{\val v}}}{\st M \rightarrow \st E}}
                \end{matrix}
            }

			\typeRuleApp[1]{\stJugComput{\emptyset}{\app{\val v\sub{x}{\fix{x}{\val v}}}{\val w}}{\st E}}
			
		\end{prooftree} 
    \end{array}
$

%% file: Visuals/AlgorithmLogic.tex
$
\begin{array}{ccl} 
    \\[-0.2cm]   
    \algosym\left(
        \begin{prooftree}
            \hypo{0<n\le m}
                \infer1{\stJugValue{\Delta}{\int n}{\typeInt{m}}}
        \end{prooftree}
    ,\,
    \itJugValue{\Gamma}{\int n}{\it M\in \itSet{\itSetEmpty, \itSet{\typeInt i}}}
    \right)
    &=&
    (\Gamma=\emptyset) \wedge ((\it M=\itSetEmpty) \lor (i=n))\\[0.5cm]
    
    \algosym\left(
        \begin{prooftree}
            \infer0{\stJugValue{\Delta;x::\st M}{x}{\st M}}
        \end{prooftree}
    ,\,
    \itJugValue{\Gamma;x:\it N}{x}{\it M}
    \right)
    &=&
   (\Gamma=\emptyset) \land (\it M=\it N) \land (\refrel{\st M}{\it M})\\[0.5cm]

    \algosym\left(
        \begin{prooftree}
                        \hypo{\stJugComput[\Pi']{\Delta; x :: \st M}{\comput t}{\st E}}
                    \infer1{\stJugValue{\Delta}{\abs{x}{\comput t}}{\st M \rightarrow \st E}}
                \end{prooftree}
    ,\,
    \itJugValue{\Gamma}{\abs{x}{\comput t}}{\itSet*{\it M_i \rightarrow \it E_i}_{\rangeI}}
    \right)
    &=&\begin{array}{c}
        \forall \rangeI.\, \exists \Gamma_i\subseteq \Gamma.\, (\Gamma=+_{\rangeI}\Gamma_i)\\
        \bigwedge_{i\in I}\algo{\Pi'}{\itJugComput{\Gamma_i;x:\it M_i} {\comput t}{\it E_i}} \\
    \end{array}\\[0.5cm]
      
    \algosym\left(
        \begin{array}{c}
                    \begin{prooftree}
                    \hypo{
                        \stJugValue[\Pi_1]
                            {\Delta; x :: \st M \rightarrow \st E}
                            {\val v}
                            {\st M \rightarrow \st E}}
                    
                \infer1{
                    \stJugValue
                        {\Delta}
                        {\fix{x}{\val v}}
                        {\st M \rightarrow \st E}}
            \end{prooftree}
    ,\\[0.3cm]
   \qquad\qquad\quad\qquad\itJugValue{\Gamma}{\fix{x}{\val v}}{\it N}
        \end{array}
    \right)
    &=&\begin{array}{c}
        \exists n\leq |\Refn(\st M\rightarrow\st E)|\cdot |Sub(\Gamma)|.\\
        \forall j\leq n.\, \exists \Gamma_j\subseteq \Gamma.\, (\Gamma=+_{j\leq n}\Gamma_i).\,\\
        \land \forall k\leq n+1.\,\exists \refrel{\st M}{\it M_k}.\\
        \it M_0=\itSetEmpty \land \it M_{n+1}= \it N\\
        \bigwedge_{j\leq n}\algo{\Pi_1}{\itJugValue{\Gamma_j;x:\it M_j}{\val v}{\it M_{j+1}}} \\
    \end{array}\\[1cm]

    \algosym\left(
        \begin{array}{c}
            \begin{prooftree}   
                        \hypo{\stJugValue[\Pi_1]{\Delta}{\val v}{\st M \rightarrow \st E}}         

            \hypo{\stJugValue[\Pi_2]{\Delta}{\val w}{\st M}}
            \infer2{\stJugComput{\Delta}{\app{\val v}{\val w}}{\st E}}
                \end{prooftree}
    ,\\[0.3cm]
   \qquad\qquad\qquad\qquad\quad\qquad\itJugValue{\Gamma}{\app{\val v}{\val w}}{ \it E}
        \end{array}\right)
    &=&\begin{array}{c}
        \exists \refrel{\st M}{\it M}.\,\exists\Gamma_1,\Gamma_2\subseteq \Gamma .\, (\Gamma=\Gamma_1+\Gamma_2)\\
        \wedge\algo{\Pi_1}{\itJugValue{\Gamma_1}{\val v}{\it M \rightarrow\it E}} \\
         \wedge\algo{\Pi_2}{\itJugValue{\Gamma_2}{\val w}{\it M}} \\
        
    \end{array}\\[0.8cm]

    \algosym\left(
        \begin{array}{c}
        \begin{prooftree}
                    \hypo{\stJugHandler[\Pi_1]{\Delta}{\handler h}{\st F \Rightarrow \st E}}
                    \hypo{\stJugComput[\Pi_2]{\Delta}{\comput t}{\st F}}
                    \infer2{\stJugComput{\Delta}{\handle{\handler h}{\comput t}}{\st E}}
                \end{prooftree}
    ,\\[0.3cm]
   \qquad\qquad\qquad\itJugComput{\Gamma}{\handle{\handler h}{\comput t}}{\it E}
        \end{array}
    \right)
    &=&\begin{array}{c}
        \exists \refrel{\st F}{\it F} .\,\exists\Gamma_1,\Gamma_2\subseteq \Gamma .\, (\Gamma=\Gamma_1+\Gamma_2)\\
        \wedge\algo{\Pi_1}{\itJugHandler{\Gamma_1}{\handler h}{\it F \Rightarrow \it E}}\\
        \wedge \algo{\Pi_2}{\itJugComput{\Gamma_2}{\comput t}{\it F}}\\
    \end{array}\\[0.8cm]

    \algosym\left(
        \begin{array}{c}
        \begin{prooftree}

                        \hypo{
                            \begin{matrix}
                                \stJugComput[\Pi_1]%
                                {\Delta;
                                    y :: \st M;
                                    r :: \st N \rightarrow \st G}%
                                {\comput t}%
                                {\st E}\\
                            \stJugHandler[\Pi_2]%
                                {\Delta}%
                                {\handler h}%
                                {\st F \Rightarrow \st G}%
                            \end{matrix}
                        }
                    \infer1{%
                        \stJugHandler%
                            {\Delta}%
                            {\{\effectClause{y}{r}{\comput t}\} \cup \handler h}%
                            {
                                \itEffect{\sigma}{\st M}
                                    {\st N
                                        \rightarrow 
                                    \st F
                                }
                                    \Rightarrow 
                                \st E
                            }%
                    }%
                \end{prooftree}
    ,\\[0.5cm]
   \itJugHandler
                {\Gamma}
                {\{\effectClause{y}{r}{\comput t}\} \cup \handler h}%
                {\itEffect{\sigma}{\it M}
                    {\it N_i\rightarrow \it F_i
                        }_{\rangeI}
                    \Rightarrow \it E}
        \end{array}
    \right)
    \hspace{-0.25cm}&=\hspace{-0.25cm}&\hspace{-0.10cm}\begin{array}{c}
        \exists \Gamma'\subseteq \Gamma .\, \forall \rangeI.\\
        \exists \refrel{\st G}{\it G_i}.\,
         \exists \Gamma_i\subseteq \Gamma.\, (\Gamma=\Gamma'+_{\rangeI}\Gamma_i)\\
        \wedge\algo{\Pi_1}{\itJugComput{\Gamma';y:\it M;r:\itSet{\it N_i \rightarrow \it G_i}_{\rangeI}}{\comput t}{\it E}}\\
        \bigwedge_{\rangeI}\algo{\Pi_2}{\itJugHandler{\Gamma_i}{\handler h}{\it F_i\Rightarrow\it G_i}}\\
    \end{array}
\end{array}
$

%% file: backtohepcf.tex
In the previous section, we have seen how \HEBI\ can be regarded not only as a type system complete for termination and reachability, but also that such completeness properties are preserved, although in a relative sense, when \HEBI\ is viewed as a refinement language for \HEB. In this section, we will show how the very same technique can also be applied to \HEPCF.

In \Cref{def:HEPCFRefinementTypes}, we define a refinement relation between \HEPCF\ and \HEBI\ types, as we already did for \HEB. 
\begin{figure}[h!]
	\centering
	\input{Definitions/RefinementTypes.tex}
	\caption{\HEPCF\, Type Refinement}
	\label{def:HEPCFRefinementTypes}
    \Description{HEPCF, Type Refinement}
\end{figure}
Following exactly the same recipe we used for \HEB, the refinement relation can be lifted from types to judgements and derivations in a very natural way. Moreover, reduction itself can be given on derivations rather than terms, following the recipe we used in \Cref{sect:HEBRedExp} for \HEB. All this leads quite quickly to the following:
\begin{lemma}[Relative Subject Expansion for \HEPCF] 
	\label{lem:ElementarySubjectExpansion-ref}%
	Let $\comput t, \comput s \in \setComput$ be two closed terms such that 
	$\comput t \reductArr \comput s$. If there exists
	$\itJugComput[\Xi']{\emptyset}{\comput s}{\it E}$ and $\stJugComput[\Pi']{\emptyset}{\comput s}{\st U}$ such that $\refrel{\Pi'}{\Xi'}$ and also there exist $\stJugComput[\Pi]{\emptyset}{\comput t}{\st U}$ such that $\Pi\rightsquigarrow \Pi'$, 
	then there exists
	$\itJugComput[\Xi]{\emptyset}{\comput t}{\it E}$ such that $\refrel{\Pi}{\Xi}$.
\end{lemma}

We can therefore conclude that, for the purposes of the reachability problem, the search space of \HEBI\ derivations for a given \HEPCF\ program typed through $\Pi$ can be restricted to those \HEBI\ derivations that refine $\Pi$. Why does this \emph{not} contradict the undecidability of the reachability problem for \HEPCF? Simply because the finite refinement property cannot hold for \HEPCF, as can be easily verified by inspecting the rules in \Cref{def:HEPCFRefinementTypes}. Indeed, let us consider the \HEPCF\ computation type $\stEffect{\st 1}$ where the effect context $\effectCtxt E$ attributes $\sigma$ the type $\effectSpec{\st{1}}{\st{1}}$. 
It is immediate to check that all the following \HEBI\ types refine $\stEffect{\st 1}$:

\[
\begin{array}{cc}
    \it E_0\coloneq\itEffectReturn{\itSet{\typeInt{1}}} & 
    \it E_{n+1}\coloneq\itEffect{\sigma}{\itSet{\typeInt{1}}}{\itSet{\typeInt{1}}\rightarrow \it E_n}
\end{array}
\]

%% file: Definitions/RefinementTypes.tex
$
        \begin{array}{c}
            \begin{array}{ccccccc}
                \begin{prooftree}
                    \infer0{
                        \refrel{\typeInt n}{\itSetEmpty}
                    }
                \end{prooftree}
                &\hspace{0.3cm}&
                \begin{prooftree}
                    \hypo{0<n \le m}
                    \infer1{\refrel{\typeInt{m}}{\itSet{\typeInt{n}}}
                    }
                \end{prooftree}
            
        &\hspace{0.3cm}&  
            \begin{prooftree}
                    \hypo{\effectSpec{\st{M}}{\st{N}} \in \effectCtxt E}

                    \hypo{\refrel{\st{M}}{\it M} }

                    \hypo{\left(
                        \refrel{\st{N}}{\it N_i} 
                    \right)_{\rangeI}}

                    \hypo{\left(
                        \refrel{\stEffect{\st L}}{\it E_i} 
                    \right)_{\rangeI}}
                \infer4{\refrel{\stEffect{\st L}}{\itEffect{\sigma}{\it M}{\it N_i \rightarrow \it E_i}_{\rangeI}}
                }
            \end{prooftree}
            \end{array}
                \\[0.5cm]

        \begin{array}{ccccc}
            \begin{prooftree}
                    \hypo{\left(
                        \refrel{\st{M}}{\it M_i}
                    \right)_{\rangeI}}

                    \hypo{\left(
                        \refrel{\st{U}}{\it E_i} 
                    \right)_{\rangeI}}
                \infer2{\refrel{\st{M} \rightarrow \st{U}}{\itSet{\it M_i \rightarrow \it E_i}_{\rangeI}}
                }
            \end{prooftree}
            & \hspace{0.5cm}&
                \begin{prooftree}
                        \hypo{\refrel{\st{M}}{\it M} }
                    \infer1{\refrel{\stEffect{\st{M}}}{\itEffectReturn{\it M}}
                    }
                \end{prooftree}
        &\hspace{0.5cm}&
                \begin{prooftree}
                        \hypo{\refrel{\st S}{\it E} }
                        \hypo{\refrel{\st U}{\it F}}
                    \infer2{
                        \refrel{\st S \Rightarrow \st U}{\it E \Rightarrow \it F}
                    }
                \end{prooftree}
            \end{array}
        \end{array}
    $

%% file: relatedwork.tex
\subsection{Intersection Type Systems}
Since their introduction~\cite{coppo1978new} fifty years ago, intersection types have not only been a type-theoretic foundation of ad-hoc polymorphism~\cite{Reynolds_1988_TR}, but also a powerful tool for analyzing the semantic properties of higher-order programs, when the latter are expressed as $\lambda$-terms or related calculi. 
In languages as diverse as Java~\cite{bracha1998making}, Scala~\cite{odersky2004overview}, and Typescript~\cite{bierman2014understanding}, variations on the theme of intersection types have been adopted in order to provide type systems with a degree of flexibility that is incompatible with classical parametric polymorphism.

The original idea of intersection types (namely, to allow a single term to be assigned a type capturing different and completely unrelated types) has been generalized in various ways and over calculi that extend the $\lambda$-calculus along heterogeneous directions. 
It is worth recalling, in this context, the so-called non-idempotent intersection types~\cite{decarvalho2007}, which make it possible to capture quantitative refinements of the usual termination property, such as the number of steps that possibly different abstract machines would take to evaluate the term being typed~\cite{decarvalho2007,AccattoliDalLagoVanoni}.
Concerning the axis of programming languages, it is also worth mentioning the extension of intersection types to languages with control operators, such as the $\lambda\mu$-calculus~\cite{vanBakel/Barbanera/deLiguoro_2018_LMCS}, and, more recently, their extension to a calculus with algebraic effects~\cite{gavazzo2024monadic}. 
The latter, however, as already observed, is fundamentally different from the system introduced in this paper, in that effects are pre-interpreted and given meaning through the notion of monad~\cite{moggi1991notions}, whereas in the present work the role traditionally played by monads is instead taken by handlers. This means that there is no need to capture, say, probabilities, non-determinism, or exceptions in a native way, and the main difficulty becomes the interaction between operations and handlers.

The use of intersection types is seen as a decision algorithm in so-called higher-order model checking (HOMC) was pioneered by Kobayashi \cite{kobayashi2009types}, and has subsequently been studied in various directions, from its connection with the relational model of linear logic \cite{melliesgrellois} to its link with linearity \cite{clairambaultmurawski}. To the authors' knowledge, the systems described in this article are the first examples of applying intersection types to HOMC within the setting of calculi with handlers.

\subsection{Reasoning about Programs Handling Effects}
Algebraic effect handlers allow the user to implement effects within programs
in a structured, modular, and composable manner.
Their expressivity is underpinned by the ability to reify delimited
continuations as first-class values.
On the other hand, from the denotational perspective, algebraic effect handlers
can also be viewed as transformers of computation trees (or, equivalently,
\emph{free models}) over algebraic operations that underlying terms
perform~\cite{Plotkin/Pretnar_2009_ESOP,Plotkin/Pretnar_2013_LMCS,Pretnar_2015_MFPS}.
This view enables embedding algebraic effects and handlers in a
language~\cite{Kiselyov/Ishii_2015_Haskell,Schrijvers/Pirog/Wu/Jaskelioff_2019_Haskell}.

A major approach to reasoning about effectful higher-order programs is
\emph{(type-and-)effect systems}, and effect systems for algebraic effects and
handlers have been studied
extensively~\cite{Plotkin/Pretnar_2009_ESOP,Plotkin/Pretnar_2013_LMCS,Bauer/Pretnar_2014_LMCS,Hillerstrom/Lindley_2016_TyDe,Leijen_2017_POPL,Yoshioka/Sekiyama/Igarashi_ICFP_2024,Tang/Lindley_2026_POPL}.
However, these systems only keep track of what algebraic operations may
perform; in other words, they only approximate the algebraic operations
appearing in the computation tree produced by a program.
By contrast, our intersection type system {\HEBI} can precisely represent the
computation tree of a terminating program (non-terminating programs are ill
typed there).
This empowers us not only to characterize the termination of programs but also
to allow fine-grained typing.
For instance, as shown in \cref{sect:highlevel}, {\HEBI}
allows different calls to the same algebraic operation to take arguments of
different types.
This typing is disallowed in the aforementioned effect systems, because those
systems have to typecheck all the calls to the same algebraic operation in a
uniform way, thereby requiring the arguments to have the same type.
Effect systems aware of the order of performed effects are called
\emph{sequential}~\cite{Skalka/Smith_2004_APLAS,Skalka/Smith/VanHorn_2008_JFP,Gordon_2020_ECOOP,Gordon_2021_TOPLAS}
or
\emph{temporal}~\cite{Koskinen/Terauchi_2014_CSL-LICS,Nanjo/Unno/Koskinen/Terauchi_2018_LICS,Sekiyama/Unno_2023_POPL,Sekiyama/Unno_2025_POPL},
depending on whether they reason about only the finite behavior or even the
infinite behavior of effectful programs.
These effect systems are interested in what \emph{unhandled}---namely,
built-in---effects are performed in what order, whereas {\HEBI} can
track the order---and even branching---of \emph{handled} effects as well as how it affects the
transformation of computations along with effect handlers.

%
\emph{Dijkstra
monads}~\cite{Swamy/Weinberger/Schlesinger/Chen/Livshits_2013_PLDI} are one of
the approaches that utilize the free model view, enabling the verification of
higher-order effectful programs by monadically transforming specifications to be
verified against the programs.
Their aim is at verification rather than typability, so the typing is not
as fine-grained as our type system and no work on Dijkstra monads considers
typing properties specific to intersection type systems, such as subject
expansion.

Higher-order model checking is also extended to support algebraic effects and
handlers by \citet{Dal-Lago/Ghyselen_2024_POPL} with the undecidability result,
followed by the attempts to identify fragments where the HOMC problem is
decidable~\cite{Sekiyama/Unno_2024_OOPSLA,Sekiyama/DalLago/Unno_2025_OOPLSA,Endo/Terauchi_2025_APLAS}.
These results on the decidable fragments rely on the CPS transformations that
translate a program with algebraic effect handlers to one without them.
%
%
Our work is also aligned with this direction, providing a more semantic viewpoint
than the justification with the CPS transformations.
This provides a benefit that {\HEB} accommodates terms that are not in the
fragments identified by the prior work, although they are incomparable:
only {\HEB} typecheckes computation $\effect[\sigma]{\val v_1}{x}{\effect[\sigma]{\val v_2}{y}{\comput t}}$ with values $\val v_1, \val v_2$ of different types, whereas
only {\HEB} cannot typecheck $\case{x}{\effect[\sigma_1]{x_1}{\val v_1}{\comput t_1}, \effect[\sigma_2]{x_2}{\val v_2}{\comput t_2}}$ for $\sigma_1 \neq \sigma_2$.

\subsection{Effect Handlers and Session Types}
{Several works in the literature have already investigated the strong connections between effect handlers and session types (e.g., \cite{orchard2016effects,tangsession}). In our setting, this relationship is reflected in the resemblance between session types and the tree of possible effects associated with computations in our systems.
However, there are important differences. In particular, session types place a central emphasis on notions of \emph{choices} and their duality, which fundamentally structure the interaction. In \HEBI, this emphasis is not present: the distinction between branches is not primarily captured through dual choice constructs, but rather arises through intersections.}

%% file: conclusion.tex
This work introduces \HEBI, the first example of an intersection type system specifically designed for a calculus with handlers and algebraic effects. Noticeably, \HEBI\ is not merely an intersection type system, but it is also \emph{behavioral}, in the sense that from the type of any computation one can deduce not only which algebraic operation will be executed next, but also what the type(s) of the underlying continuation must be. On the one hand, this makes it possible to retain all relevant information inside the types. On the other hand, this guarantees \emph{compositionality}: computations and the handlers responsible for managing them can be typed independently, aggregating the corresponding type information only when needed.

What is remarkable about \HEBI\ is not only its precise characterization of terminating computations, but also the fact that it naturally suggests a new monomorphic type system (meaning that the same handler can not use the continuation in type-theoretically incompatible ways), namely \HEB, for which the reachability problem is decidable, contrary to what happens in similar type systems.
{\HEPCF\ for example, derives its expressive power primarily from \emph{not} tracking the types of algebraic operations' continuations within the computation type, thereby allowing the latter to become arbitrarily complex. This phenomenon is ultimately responsible for the undecidability of the HOMC problem in \HEPCF, as shown in \Cref{sect:backtopcf}. The use of behavioral forms of computation types in \HEB\ prevents this phenomenon to happen. At the same time, \HEB\ supports a form of polymorphism in the parameter and return types of algebraic operations that is not available in \HEPCF .}
The decidability result on \HEB\ is established by making essential use of \HEBI\ itself, employing it as a refinement of \HEB, in the style of the results by Kobayashi and coauthors concerning the use of intersection types in the context of the HOMC problem.

There are certainly several promising directions for future research that the introduction of \HEBI\,naturally suggests.
First of all, it is worth mentioning that the previously discussed fragments of \HEPCF\ for which the reachability problem is decidable were shown to be so by defining suitable CPS translations and proving that these translations are type preserving. Studying such type systems through the refinement induced by \HEBI\ would undoubtedly lead to a deeper understanding of their expressive power, something which at present is not possible. Another important direction concerns the expressive power of the calculus \HEB, which in this article has been introduced as a kind of proof of concept, but is certainly of independent interest, especially if the behavioral types it comprises were to become recursive, similarly to what happens with session types. 
Finally, similarly to what happens with intersection types and the $\lambda$-calculus, \HEBI\ types seem to induce a denotational model for \HEB\ that is worth studying in its own right. In particular, it would be worthwhile to study its relationship with the standard denotational semantics of effect handlers, which is based on trees in the free monad. They are in fact in close relationship, since an \HEB\ computation type can be inhabited only by a denotational tree in which all branches perform the same effects in the same order and finally return values of the same type. Note that the argument values of the effects performed in and the values returned by the branches may be different from each other: the branches only share which effect is performed in which order, and that, after the aforementioned sequence of effects, some value will be returned.
\\

%% file: appendix.tex
\section{Proofs of Section \ref{sect:behaviorintersection}}
In this appendix, we provide detailed proofs of all the statements presented in the main text. The organization of this appendix closely follows that of the main article, with sections arranged accordingly.

We begin by defining the notion of \emph{clash} in the \HE\ calculus. Clashes correspond to terms that are not in normal form, yet no further computation step can be performed. Most of these situations arise from typing errors.
The clashes are defined with the following form (\cref{def:Clashes}).

\begin{figure}[h!]\centering
\begin{tabular}{c}\framebox{$
    \begin{array}{r l}
        \app{\int n}{\val w}
    \\[0.1cm]
        \case{\abs{x}{\comput u}}{\comput t_1, \ldots, \comput t_m}
    \\[0.1cm]
        \case{\fix{x}{\val v}}{\comput t_1, \ldots, \comput t_m}
    \\[0.1cm]
        \case{\int n}{\comput t_1, \ldots, \comput t_m}
            &\text{ where } n \leq 0 \text{ or } n > m
    \\[0.1cm]
        \handle{\handler h}{\effect{\val v}{x}{\comput t}}
            &\text{ where } \{\effectClause{y}{r}{\comput s}\} \not\in \handler h 
    \\[0.1cm]
        \handle
            { \{\effectClause{y}{r}{\comput s}\} \cup \handler h}
            {\effect{\val v}{x}{\comput t}}
            &\text{ where } 
            \{\effectClause{y}{r}{\comput u}\} \in 
            \handler h 
            \\ & \text{ and } \comput s \neq \comput u
    \end{array}$}
\end{tabular}
\caption{The different kinds of clashes in \HE\ terms}
\label{def:Clashes}
\end{figure}

Both the definitions of clashes for $\handle
            {\handler h}
            {\effect{\val v}{x}{\comput t}}$ extends easily in the case of $\constructor{return}$ clauses. 

\subsection{Section \ref{sect:HEBIElementaryProperties}}
\label{sect:apx4-3}
\begin{remark}[Type Context Simplification] ~\\
    \label{rem:ContextSimplification}%
    Let $p \in \setTerm$ be term and $x$ a variable that is not free in $p$.
    If $\itJug{\Gamma}{p}{\it E}$ and $x : \it M \in \Gamma$, then $\it M = \itSetEmpty$. In particular, if $p$ is closed, then $\Gamma$ is empty.
\end{remark}

\begin{lemma}[\HEBI, Typability of Values] ~\\
    \label{lem:TypabilityValues}%
    For $\val v \in \setVal$, there exists
    $\itJugValue[\Pi]{\emptyset}{\val v}{\itSetEmpty}$.
\end{lemma}
\begin{proof}
    By induction on $\val v$.
\end{proof}

\begin{corollary}
    \label{Col:TypabilityOfValuesOptmialSize}%
    For any value $\val v \in \setVal$, the type derivation 
    $\itJugValue[\Pi]{\emptyset}{\val v}{\itSetEmpty}$
    defined in the previous lemma (\cref{lem:TypabilityValues})
    is the smallest possible derivation for $\val v$.
    In other words, for any $\itJugValue[\Pi']{\Gamma}{\val v}{\it M}$,
    $\Pi'$ has at least the size of $\Pi$.
\end{corollary}

\hebiSplitting*
\begin{proof}
    \input{Proofs/SplittingLemma2.tex}
\end{proof}

\hebiIntersection*
\begin{proof}
    \input{Proofs/IntersectionLemma.tex}
\end{proof}

Since $\cup$ is a partial relation, the following lemma ensures that it is always well defined for the types arising from derivations typing closed values.
\begin{lemma}[Union of Closed Values]
    \label{lem:SumClosedValues}
    Let $\val v\in \setVal$ a closed value, $\Pi$ and $\Xi$ two derivations such that $\itJugValue[\Pi]{\emptyset}{\val v}{\it M}$ and $\itJugValue[\Xi]{\emptyset}{\val v}{\it N}$ then $\it M\cup\it N$ exists.
\end{lemma}
\begin{proof}
    By cases on $\val v$. If $\val v=\int n$ then $\it M,\it N\in \itSet{\itSetEmpty,\itSet{\typeInt{n}}}$, thus $\it M\cup \it N=\itSet{\typeInt{n}}$ is defined.
    Otherwise, being closed, we have either $\val v=\abs{x}{\comput t}$ or $\val v=\fix{x}{\val w}$. In both cases the types must have the form $\it M=\itSet{\it M_i\rightarrow\it E_i}_{\rangeI}$ and $\it N=\itSet{\it M_j\rightarrow\it E_j}_{\rangeJ}$. Therefore $\it M\cup \it N=\itSet{\it M_i\rightarrow\it E_i}_{i\in I\cup J}$. 
\end{proof}

Finally, we must prove that the $\#$ relation is well behaved with respect to substitution.
\begin{remark}[Compatibility of Substitution]
    \label{lem:HandlerComputIndipendence}
    Let $\handler h\in \setHandler$, $\val v\in \setVal$ and $x$ a variable. Moreover, let $\psi$ be either $\{\retClause{x}{\comput t}\}$ or $\{\effectClause{y}{r}{\comput t}\}$, then $\psi \handlerCompat \handler h$ if and only if $\psi\sub{x}{\val v}\handlerCompat \handler h\sub{x}{\val v}$.
\end{remark}

\hebiSub*
\begin{proof}
    \completedProof{%
    \input{Proofs/SubstitutionLemma.tex}

    }
\end{proof}

\hebiAntiSub*
\begin{proof}
    \completedProof{%
        \input{Proofs/AntiSubstitutionLemma.tex}
    }
\end{proof}

\hebiContextTyping*
\begin{proof}
    \completedProof{
        \input{Proofs/ContextTyping.tex}
    }
\end{proof}

Before proving Subject Reduction and Subject Expansion, we first establish the following two lemmas. The first states that any effect clause may be added to a handler, provided that it satisfies the $\#$ relation. The second shows that this relation characterizes precisely the absence of clashes in the presence of handlers.
\begin{lemma}[\HEBI, Handler Weakening] ~\\
    \label{lem:HandlerWeakening}%
    Let $\itJugHandler[\Pi]{\Gamma}{\handler h}{\it F
    \Rightarrow \it E}$ and $\psi$ be either $\{\retClause{x}{\comput t}\}$ or $\{\effectClause{y}{r}{\comput t}\}$. If
    $\psi \handlerCompat \handler h$,
    then there exists
    $\itJugHandler[\Pi']{\Gamma}{\psi
    \cup \handler h}{\it F
    \Rightarrow \it E}$.
\end{lemma}

\begin{lemma}[Clashes and Compatibility]
    \label{lem:CompatImpliesClashFree}
    The term $\handle{\handler h}{\effect{\val v}{x}{\comput u}}$, with $\{\effectClause{y}{r}{\comput t}\} \in \handler h$, is not a clash if and only if $\{\effectClause{y}{r}{\comput t}\} \handlerCompat \handler h$. Similarly,  $\handle{\handler h}{\ret{\val v}}$, with $\{\retClause{x}{\comput t}\} \in \handler h$, is not a clash if and only if $\{\retClause{x}{\comput t}\}\handlerCompat \handler h$.
\end{lemma}

We split the proof of \cref{lem:ElementarySubjectReduction} into separate proofs of Elementary Subject Reduction and Elementary Subject Expansion. Here, by elementary we refer to the fact that we consider only the basic reduction rules of \cref{def:RewriteRules}. We can easily extend these results to the full reduction using \cref{lem:ContextTyping}, finally obtaining \cref{lem:ElementarySubjectReduction}.

\begin{lemma}[\HEBI, Elementary Subject Expansion] ~\\
    \label{lem:ElementarySubjectExpansion}%
    Let $\comput t, \comput t' \in \setComput$ such that 
    $\comput t \reductArr' \comput t'$. If there exists
    $\itJugComput[\Pi]{\emptyset}{\comput t'}{\it E}$
    then there exists
    $\itJugComput[\Pi']{\emptyset}{\comput t}{\it E}$.
\end{lemma}
\begin{proof}
    \completedProof{
        \input{Proofs/ElementarySubjectExpansion.tex}
    }
\end{proof}

\begin{lemma}[\HEBI, Elementary Subject Reduction] ~\\
    \label{lem:ElementaryElSubjectReduction}%
    Let $\comput t, \comput t' \in \setComput$ such that $\comput t$
    is typed and $\comput t \reductArr' \comput t'$. If there
    exists
    $\itJugComput[\Pi]{\emptyset}{\comput t}{\it E}$
    then there exists
    $\itJugComput[\Pi']{\emptyset}{\comput t'}{\it E}$.
\end{lemma}
\begin{proof}
    \completedProof{
        \input{Proofs/ElementarySubjectReduction.tex}
    }
\end{proof}

To conclude this section, we prove the Progress Lemma for \HEBI.

\hebiProgress*
\begin{proof}
    \input{Proofs/HEBIProgress.tex}

\end{proof}

\subsection{Section \ref{sect:reducibility}}
\reducibilityExpansion*
\begin{proof}   
    \completedProof{%
        \input{Proofs/ReductionExpansionReduction.tex}
    }
\end{proof}

\typableReducibility*
\begin{proof}
    \completedProof{%
        \input{Proofs/TypableReducibility.tex}
    }
\end{proof}

\section{Proofs of Section \ref{sect:simplebehavior}}

\subsection{Section \ref{sect:hebsystem}}
The aim of this section is to prove Subject Reduction for \HEB. Since the arguments are very similar to those already used in \cref{sect:apx4-3}, we only present in detail the proof of \cref{thm:HEBSubjectReductionEl}.

\begin{lemma}[\HEB\ Substitution Lemma] ~\\
    \label{lem:SimpleSubstitutionLemma}%
    Let $p \in \setTerm$ and $\val v \in
    \setVal$ such that
    $\stJug[\Pi_p]{\Delta; x : \st M}{p}{\st T}$ and
    $\stJugValue[\Pi_{\val v}]{\emptyset}{\val v}{\st M}$.
    Then, there exists
    $\stJug[\Pi]{\Delta}
        {p \sub{x}{\val v}}{\st T}$.
\end{lemma}
\begin{proof}
    By induction on $\Pi_p$.
\end{proof}

\begin{lemma}[\HEB\ Handler Weakening] ~\\
    \label{lem:HEBHandlerWeakening}%
    Let $\itJugHandler[\Pi]{\Gamma}{\handler h}{\st F
    \Rightarrow \st E}$. If
    $\{\effectClause{y}{r}{\comput t}\} \handlerCompat \handler h$,
    then there exists a derivation, that we call $\Pi^{(\sigma)}$, such that 
    $\itJugHandler[\Pi^{(\sigma)}]{\Gamma}{\{\effectClause{y}{r}{\comput t}\}
    \cup \handler h}{\st F
    \Rightarrow \st E}$.
\end{lemma}
\begin{proof}
    By induction on $\Pi$.
\end{proof}

\begin{theorem}[\HEB\ Elementary Subject Reduction]
    \label{thm:HEBSubjectReductionEl}
    Let $\comput t,\comput t'\in \setComput$, $\stJugComput[\Pi]{\emptyset}{\comput t}{\st E}$ and $\comput t\reductArr'\comput t'$ then it exists $\Pi'$ such that $\stJugComput[\Pi']{\emptyset}{\comput t'}{\st E}$.
\end{theorem}
\begin{proof}
    \completedProof{%
    \input{Proofs/ElementarySimpleSubjectReduction.tex}
    }
\end{proof}

To extend this result to the full reduction, thus obtaining \cref{thm:HEBElementarySubjectReduction}, we need to prove the following lemma.

\begin{lemma}[\HEB\ Context Typing]~\\
    \label{lem:HEBSimpleContextTyping}%
    Let $\comput t \in \setComput$ and 
    $\ctxtEval$ a computation context such that 
    $\stJugComput[\Pi]{\emptyset}{\ctxtEval<\comput t>}{\st E}$.
    Then there is an \HEB\ type $\st F$ 
    such that $\stJugComput{\emptyset}{\comput t}{\st F}$ and 
    for every $\comput u$ such that $\itJugComput{\emptyset}{\comput u}{\st F}$, then there exist $\Pi'$ such that 
    $\stJugComput[\Pi']{\emptyset}{\ctxtEval<\comput u>}{\st E}$.
\end{lemma}
\begin{proof}
    By induction on the context $\ctxtEval$.
\end{proof}

To conclude, we can prove a Progress Lemma similar to the one proved for \HEBI.
\hebProgress*
\begin{proof}
    The proof is similar to the one of \cref{lem:HEBIProgress}.
\end{proof}

\subsection{Section \ref{sect:HEBRedExp}}
\label{sec:apx5-3}
In this section, we prove Relative Subject Expansion for refinements between \HEBI\ and \HEB\ type derivations. We begin by stating two lemmas showing that the notion of refinement between types is well-behaved with respect to inclusion, both for individual value types and for contexts.

\begin{lemma}[Subset Refinement]
    \label{lem:RefinementSubsetType}
    Let $\st M$ be a \HEB\ value type and $\it M$ an \HEBI\  value type such that $\refrel{\st M}{\it M}$. If $\it N\subseteq \it M$ then $\refrel{\st M}{\it N}$.
\end{lemma}

\begin{corollary}[Context Refinement]
    \label{lem:ContextRefinement}
    Let $\Delta$ be a \HEB\ types context and $\Gamma$ be an \HEBI\ types one such that $\refrel{\Delta}{\Gamma}$. 
    \begin{itemize}
        \item Let $\st M$ and $\it M$ be two types such that $\refrel{\st M}{\it M}$, then $\refrel{\Delta;x::\st M}{\Gamma;x:\it M}$.
        \item For each $\Gamma'\subseteq\Gamma$ we have $\refrel{\Delta}{\Gamma'}$.
    \end{itemize}      
\end{corollary}

In \cref{def:HEBDerivationRefinementContd}, we present the missing refinement rules between \HEBI\ and \HEB\ derivations that were not included in \cref{def:HEBDerivationRefinement}. Moreover, in \cref{def:HEBDerivationReductionContd}, we introduce the reduction rules between derivations that are missing in \cref{def:HEBDerivationReduction}. We then state several properties of this refinement system that will be useful in the proof of \cref{lem:NewElementarySubjectExpansion-refEl}.

\begin{figure}
\centering
\input{Definitions/NewRefinementRelationContd.tex}
\caption{\HEB\ refinement for rules, cont'd.}
\label{def:HEBDerivationRefinementContd}
\end{figure}

\begin{figure}
    \scalebox{0.86}{\input{Definitions/HEBReductionsContd.tex}}
    \caption{Reductions between two \HEB\ derivations, cont'd.}
    \label{def:HEBDerivationReductionContd}
\end{figure}

\begin{lemma}[Conlcusion Refinement]
    \label{rem:NewRefineConclusion}
    If $\refrel{\stJug[\Pi]{\Delta}{p}{\st U}}{\itJug[\Xi]{\Gamma}{p}{\it E}}$, then $\refrel{\Delta}{\Gamma}$ and $\refrel{\st U}{\it E}$.
\end{lemma}

\begin{proposition}[Weakening Refinement]
    \label{prop:HEBClosedRefinement}
    Let $p$ be a term, $\stJug[\Pi]{\Delta_1}{p}{\st U}$ be its derivation and $\Delta_2$ be a context such that for each $x::\st M\in \Delta_1$ and $x::\st N\in \Delta_2$ then $\st M=\st N$. Let $\itJug[\Xi]{\Gamma}{p}{\it E}$ be an \HEBI\ type derivation, then $\refrel{\Pi}{\Xi}$ if and only if $\refrel{\Pi^{(\Delta_2)}}{\Xi}$.
\end{proposition}

\begin{lemma}[\HEB\,, Refinement Handler Elimination]
    \label{lem:HEBHandlerElim}
    Let $\stJugHandler[\Pi]{\Delta}{\handler h}{\st F \Rightarrow\st E}$ be a derivation, $\psi$ be either $\{\retClause{x}{\comput t}\}$ or $\{\effectClause{y}{r}{\comput t}\}$. If $\psi\handlerCompat\handler h$ and it exists an \HEBI\ derivation $\Xi$ such that $\refrel{\stJugHandler[\Pi^{(\psi)}]{\Delta}{\handler h\cup\psi}{\st F\Rightarrow \st E}}{\itJugHandler[\Xi]{\Gamma}{\handler h\cup \psi}{\it F\Rightarrow\it E}}$, then it exists a derivation $\itJugHandler[\Xi']{\Gamma}{\handler h}{\it F\Rightarrow\it E}$ such that $\refrel{\Pi}{\Xi}$.
\end{lemma}

\begin{lemma}[\HEB\ Refinement Anti-substitution Lemma] ~\\
    \label{lem:antiSubstLemma-ref}%
    Let $p \in \setTerm$
    be any term, $x$ a variable occourring free in $p$ and $\val v\in\setVal$ a closed value.
    If $\itJug[\Xi]{\Gamma}{p\sub{x}{\val v}}{\tau}$ for some \HEBI\ type $\tau$ and there exists an \HEB\ type $\st M$ such that $\stJug[\Pi_{p}]{\Delta;x::\st M}{p}{\st T}$, $\stJug[\Pi_\val v]{\emptyset}{\val v}{\st M}$ and $\refrel{\Pi_p\sub{x}{\Pi_\val v}}{\Pi}$, then it exists an \HEBI\ type $\it M$ such that 
    $\itJugValue[\refrel{\Pi_\val v}{\Xi_\val v}]{\emptyset}{\val v}{\it M}$ and
    $\itJug[\refrel{\Pi_p}{\Xi_p}]{\Gamma; x: \it M }{p}{\tau}$.
\end{lemma}
\begin{proof}
    By induction on $\Xi$.
\end{proof}

\begin{lemma}[\HEB\ Relative Elementary Subject Expansion] ~\\
    \label{lem:NewElementarySubjectExpansion-refEl}%
    Let $\comput t, \comput t' \in \setComput$ be two closed term such that 
    $\comput t \reductArr' \comput t'$. If there exists
    $\itJugComput[\Xi']{\emptyset}{\comput t'}{\it E}$ and $\stJugComput[\Pi']{\emptyset}{\comput t}{\st E}$ such that $\refrel{\Pi'}{\Xi'}$ and also there exist $\stJugComput[\Pi]{\emptyset}{\comput t}{\st E}$ such that $\Pi\rightsquigarrow \Pi'$, 
    then there exists
    $\itJugComput[\Xi]{\emptyset}{\comput t}{\it E}$ such that $\refrel{\Pi}{\Xi}$.
\end{lemma}
\begin{proof}
    \input{Proofs/NewRefinementElementarySubjectExpansion.tex}
\end{proof}

In order to prove \cref{lem:NewElementarySubjectExpansion-ref}, it only remains to prove the following lemma.
\begin{lemma}[\HEB\ Refinement Context Typing]~\\
    \label{lem:RefinementContextTyping}%
    Let $\comput t \in \setComput$ and 
    $\ctxtEval$ a computation context such that 
    $\refrel{\stJugComput[\Pi]{\emptyset}{\ctxtEval<\comput t>}{\st E}}{\itJugComput[\Xi]{\emptyset}{\ctxtEval<\comput t>}{\it E}}$.
    Then there exists an \HEBI\ type $\it F$ and an \HEB\ type $\st F$ 
    such that $\refrel{\stJugComput{\emptyset}{\comput t}{\st F}}{\stJugComput{\emptyset}{\comput t}{\it F}}$ and 
    for every $\comput u$ such that $\refrel{\itJugComput{\emptyset}{\comput u}{\st F}}{\itJugComput{\emptyset}{\comput u}{\it F}}$, then there exist $\Pi'$ and $\Xi'$ such that 
    $\refrel{\itJugComput[\Pi']{\emptyset}{\ctxtEval<\comput u>}{\st E}}{\itJugComput[\Xi']{\emptyset}{\ctxtEval<\comput u>}{\it E}}$.
\end{lemma}
\begin{proof}
    By induction on $\ctxtEval$.
\end{proof}

\subsection{Section \ref{sect:refinementsDecidability}}

We start this section giving some preliminary definitions and properties needed in order to prove the correctness and computability of $\mathcal A$.
Let $\itJugValue{\Gamma}{\fix{x}{\val v}}{\it M}$ be typable in \HEBI. By looking at the corresponding typing rule, we can notice that this implies the existence of a certain number $n$ of types $\it N_i$ and derivations $\itJugValue[\Xi_i]{\Gamma_i;x:\it N_i}{\val v}{\it N_{i+1}}$ such that $\it N_0=\itSetEmpty$, $\it N_n=\it M$ and $\Gamma=+_{i\in n}\Gamma_i$. We can see this sequence of derivations as a path in the graph $\mathcal G(\Pi,\Gamma,\st M\rightarrow \st E)$, with the additional requirements that $\refrel{\Pi}{\Xi_i}$ and $\refrel{\st M\rightarrow \st E}{\it N_i}$ for each $i\leq n$. In \cref{lem:RefinementFixCompleteness}, we prove that this path is simple, therefore the length of such sequences of derivations is bounded by the size of the graph itself.

\begin{definition}
    \label{def:GraphFixpoint}
    Let $\stJugValue[\Pi]{\Delta;x:\st M\rightarrow\st E}{\val v}{\st M\rightarrow\st E}$ be a \HEB\ type derivation and $\Gamma$ be an \HEBI\ typing context. We define a graph $\graph{\Pi}{\Gamma}{\st M\rightarrow \st E}$ having as nodes couples $(\Gamma',\it M)$, where $\Gamma'\subseteq\Gamma$ and $\refrel{\st M\rightarrow\st E}{\it M}$. There exist an arrow between two nodes $(\Gamma_1,\it M_1)$ and $(\Gamma_2,\it M_2)$ if and only if it exists an \HEBI\ type derivation $\itJugValue[\Xi]{\Gamma_\val v;x:\it M_1}{\val v}{\it M_2}$ such that $\refrel{\Pi}{\Xi}$ with $\Gamma_\val v\subseteq \Gamma$ and $\Gamma_2=\Gamma_1+\Gamma_\val v$.
\end{definition}

\begin{remark}[Finite Fixpoint]
    \label{rem:GraphFinite}
    The size of $\graph{\Pi}{\Gamma}{\st M\rightarrow \st E}$ is bounded by $|Sub(\Gamma)|\cdot|\Refn(\st M\rightarrow\st E)|$.
\end{remark}

\begin{remark}
    \label{rem:PathToDerivations}
    Let $\stJugValue[\Pi]{\Delta_s;x:\st M\rightarrow\st E}{\val v}{\st M\rightarrow\st E}$ be a \HEB\ type derivation and $\Gamma$ be an \HEBI\ typing context. If there exists a path $\phi$ between two nodes $(\Delta,\it N)$ and $(\Delta',\it N')$ in $\graph{\Pi}{\Gamma}{\st M\rightarrow \st E}$, then there exist some $\refrel{\st M}{\it M_i}$ for $i\leq |\phi|+1$, $\Gamma_i\subseteq \Gamma$ and some derivations $\refrel{\Pi}{\itJugValue[\Xi_i]{\Gamma_i;x:\it M_i}{\val v}{\it M_{i+1}}}$ for $i\leq |\phi|$ such that $\Delta=\Gamma_0$ and $\Delta'=+_{i\leq n}\Gamma_i$, $\it N=\it M_0$ and $\it N'=\it M_{|\phi|+1}$. 
\end{remark}

\begin{lemma}[Fixpoint Path]
    \label{lem:RefinementFixCompleteness}
    Let $\stJugValue[\Pi]{\Delta}{\fix{x}{\val v}}{\st M \rightarrow\st E}$ be a \HEB\ type derivation, $\Pi'$ its premise and $\itJugValue[\Xi]{\Gamma}{\fix{x}{\val v}}{\it M}$ such that $\refrel{\Pi}{\Xi}$, it exists a $\Gamma_0\subseteq \Gamma$ and a simple path $\phi$ from $(\Gamma_0,\itSetEmpty)$ to $(\Gamma,\it M)$ in $\graph{\Pi'}{\Gamma}{\st M\rightarrow \st E}$.
\end{lemma}

This lemma can be seen as a sort of dual of \cref{lem:RefinementFixCompleteness}: if we have a sequence of $\itJugValue{\Gamma_i;x:\it M_i}{\val v}{\it M_{i+1}}$ for some suitables $\it M_i$, we can merge them in a derivation of $\fix{x}{\val v}$.
\begin{lemma}
    \label{lem:RefinementFixSoundness}
    Let $\stJugValue[\Pi']{\Delta;x::\st M\rightarrow\st E}{\val v}{\st M \rightarrow\st E}$ be a \HEB\ type derivation, $n$ an integer, $\it M_i$ with $i\leq n$ be \HEBI\ types such that $\refrel{\st M \rightarrow \st E}{\it M_i}$ and $\it M_0=\itSetEmpty$. Furthermore, let $\itJugValue[\Xi_i]{\Gamma_i;x:\it M_i}{\val v}{\it M_{i+1}}$ with $i\leq n-1$ be a sequence of \HEBI\ type derivations such that $\refrel{\Pi'}{\Xi_i}$. Then, there exist $\stJugValue[\Pi]{\Delta}{\fix{x}{\val v}}{\st M\rightarrow\st E}$ and $\itJugValue[\Omega]{+_{i\leq n-1}\Gamma_i}{\fix{x}{\val v}}{\it M_n}$ such that $\refrel{\Pi}{\Omega}$.
\end{lemma}

As we said in the discussion about \HEBI\ computational types, they can be seen as trees.
The following measure counts the number of leaves in such trees.
\begin{definition}
    We define a measure on computation \HEBI\ types.
\[
    \begin{array}{lcr}
        |\itEffectReturn{\it M}|=1 && |\itEffect{\sigma}{\it M}{\it N_i \rightarrow \it E_i}_{\rangeI}|=\sum_{\rangeI}|\it E_i|
    \end{array}
\]
\end{definition}
\begin{remark}
    \label{rem:LeafMeasureFinite}
    The measure $|\it F|$ is finite for each $\it F$, since the set $I$ is finite.    
\end{remark}

The Leaf Replacement relation $\it E\replaceLeaf{\it M_i \rightarrow \it G_i}[\rangeI]\it F$ states that the type $\it F$ is created from $\it E$ by substituting each type $\itEffectReturn{\it M_i}$ in a leaf of $\it E$, with the corresponding type $\it G_i$. This imples that the size of $I$ is smaller then the number of leaves of $\it E$ (\emph{i.e.} $|I|\leq |\it E|$). Moreover $|\it E|\leq |\it F|$, since $|\it G_i|\geq 1$.  
\begin{lemma}[Finite Leaf Replacement]
    \label{lem:LeafReplaceBound}
    If $\it E\replaceLeaf{\it M_i \rightarrow \it G_i}[\rangeI]\it F$ then $|I|\leq |\it F|$.
\end{lemma}

We now define a new formula $\mathcal B$, in the same way as we did for $\mathcal A$. This formula checks whether some given types satisfy the Leaf Replacement relation.

\begin{definition}
    Let $I$ be a set, $\it E$, $\it F$, $\it M_i$ and $\it G_i$ for each $\rangeI$ be \HEBI\ types. We define a formula $\algoleaf{\it E}{\it M_1,\dots,\it M_{|I|}}{\it G_1,\dots,\it G_{|I|}}{\it F}$ that checks if the relation $\it E\replaceLeaf{\it M_i \rightarrow \it G_i}[\rangeI]\it F$ holds. If it does, we write $\models\algoleaf{\it E}{\it M_1,\dots,\it M_{|I|}}{\it G_1,\dots,\it G_{|I|}}{\it F}$. We describe this formula in \cref{def:AlgoLeafLogic} by induction on $\it E$.
\end{definition}

\begin{figure}
    \scalebox{0.9}{\input{Definitions/AlgoLeafLogic.tex}}
    \caption{Definition of $\mathcal B$}
    \label{def:AlgoLeafLogic}
\end{figure}

\begin{proposition}[Correctness of $\mathcal B$]
    \label{lem:AlgoLeafSoundness}
    \label{lem:AlgoLeafCompleteness}
    Let $I$ be a set, $\it E$, $\it F$, $\it M_i$ and $\it G_i$ for each $\rangeI$ be \HEBI\ types. It holds that $\models\algoleaf{\it E}{\it M_1,\dots,\it M_{|I|}}{\it G_1,\dots,\it G_{|I|}}{\it F}$ if and only if $\it E\replaceLeaf{\it M_i \rightarrow \it G_i}[\rangeI]\it F$.
\end{proposition}

\begin{lemma}[Computability of $\mathcal B$]
    \label{lem:AlgoLeafComputable}
    The function $\models\algoleaf{\it E}{\it M_1,\dots,\it M_{|I|}}{\it G_1,\dots,\it G_{|I|}}{\it F}$ is computable for each $\it E$, $\it F$, $\it M_i$ and $\it G_i$ with $\rangeI$.
\end{lemma}

We can now finally prove the correctness of $\mathcal A$. In \cref{def:AlgorithmContd} are described the cases of $\mathcal A$ that were not included in \cref{def:AlgorithmLogic}. 
We procede by splitting the proof of \cref{prop:AlgoCorrectness} into two separate proofs of soundness and completeness.

\begin{figure}
    \centering
    \input{Definitions/AlgorithmLogic.tex}
    \caption{Description of the formula $\mathcal A$, cont'd.}
    \label{def:AlgorithmContd}
\end{figure}

\begin{proposition}[Soundness of $\mathcal A$]
    \label{prop:AlgoSoundness}
    Let $\stJug[\Pi]{\Delta}{p}{\st T}$ be a \HEB\ derivation, $J_s \coloneq \stJug{\Delta}{p}{\st T}$ its conclusion and $J_r\coloneq \itJug{\Gamma}{p}{\tau}$ an \HEBI\ type judgement such that $\refrel{J_s}{J_r}$. If $\models\algo{\Pi}{J_r}$ then it exists an \HEBI\ derivation $\Xi$, with conclusion $J_r$, such that $\refrel{\Pi}{\Xi}$.
\end{proposition}
\begin{proof}
    \input{Proofs/AlgoSoundness.tex}
\end{proof}

\begin{proposition}[Completeness of $\mathcal A$]
    \label{prop:AlgoCompleteness}
    Let $\stJug[\Pi]{\Delta}{p}{\st T}$ be a \HEB\ type derivation, $\itJug[\Xi]{\Gamma}{p}{\tau}$ be an \HEBI\ type derivation such that $\refrel{\Pi}{\Xi}$ and $J_r$ the conclusion of $\Xi$. Then $\models\algo{\Pi}{J_r}$.
\end{proposition}
\begin{proof}
    \input{Proofs/AlgoCompleteness.tex}

\end{proof}

In order to prove the computibility of $\mathcal A$, we must show that all quantifiers appearing in its definition range over finite sets. The following lemmas establish this property for the sets of refinements $\Refn(\st T)$ and of subcontexts $Sub(\Gamma)$, where $\st T$ is an \HEB\ type and $\Gamma$ is an \HEBI\ context.

\finiteRefinement*
\begin{proof}
    By induction on $\st T$.
    
    \paragraph*{Case Integers}
    By looking at \cref{def:HEBRefinementTypes}, in case $\st T= \typeInt{m}$ we have
    \[
    \begin{array}{cc}
         \begin{prooftree}
                    \infer0{
                        \refrel{\typeInt m}{\itSetEmpty}
                    }
                \end{prooftree}
        &
         \begin{prooftree}
                    \hypo{0<n \le m}
                    \infer1{\refrel{\typeInt{m}}{\itSet{\typeInt{n}}}
                    }
                \end{prooftree}
    \end{array}
    \]
    The number of possible refinement is then finite.

    \paragraph*{Case Arrow}
    If $\st T=\st M\rightarrow\st E$, then 
    \[
         \begin{prooftree}
                    \hypo{\left(
                        \refrel{\st{M}}{\it M_i}
                    \right)_{\rangeI}}

                    \hypo{\left(
                        \refrel{\st{E}}{\it E_i} 
                    \right)_{\rangeI}}
                \infer2{\refrel{\st{M} \rightarrow \st{E}}{\itSet{\it M_i \rightarrow \it E_i}_{\rangeI}}
                }
            \end{prooftree}
    \]
    By induction hypothesis, there exists a finite number of \HEBI\ types $\it N_i$ and $\it E_i$ such that $\refrel{\st M}{\it M_i}$ and $\refrel{\st E}{\it E_i$}. Therefore, there are only finitely many possible arrow types $\it M_i\rightarrow\it E_i$. Thus, the number of types of the form $\itSet{\it M_i\rightarrow \it E_i}_{\rangeI}$ is finite.

    \paragraph*{Case Computation Effect}
    If $\st T=\itEffect{\sigma}{\st M}{\st N\rightarrow \st E}$. We have
    \[
         \begin{prooftree}

                    \hypo{\refrel{\st{M}}{\it M} }

                    \hypo{\left(
                        \refrel{\st{N}}{\it N_i} 
                    \right)_{\rangeI}}

                    \hypo{\left(
                        \refrel{\st E}{\it E_i} 
                    \right)_{\rangeI}}
                \infer3{\refrel{\itEffect{\sigma}{\st M}{\st N \rightarrow \st E}}{\itEffect{\sigma}{\it M}{\it N_i \rightarrow \it E_i}_{\rangeI}}
                }
            \end{prooftree}
    \]
    As before by induction hypothesis, the number of \HEBI\ types such that $\refrel{\st M}{\it M}$, $\refrel{\st N}{\it N_i}$ and $\refrel{\st E}{\it E_i}$ are finite. Therefore, there are finitely many types of the form $\itEffect{\sigma}{\it M}{\it N_i \rightarrow \it E_i}_{\rangeI}$.
    
    The other two cases ($\st T=\itEffectReturn{\st M}$ and $\st T=\st E\Rightarrow \st F$) are similar.
\end{proof}

\begin{corollary}[Finite Refinement for Contexts]
    \label{coro:FiniteContextRefinement}
    Let $\Delta$ be an \HEB\ contex. There exists a finite number of \HEBI\ contexts $\Gamma$ such that $\refrel{\Delta}{\Gamma}$.
\end{corollary}
\begin{proof}
    By induction on $\Delta$.
\end{proof}

\begin{lemma}[Finite Subsets]
    \label{lem:FiniteSubsets}
    Let $\it M$ an \HEBI\  value type and $\Gamma$ be an \HEBI\ typing context:
    \begin{itemize}
        \item The types $\it N$ such that $\it N\subseteq \it M$ are finite.
        \item The contexts $\Gamma'$ such that $\Gamma'\subseteq\Gamma$ are finite.
    \end{itemize}  
\end{lemma}

\refinementComputability*
\begin{proof}
    The equality checks perfomed in base cases are clearly computable. In particular, in the \ruleNameVar\ case, we have to check whether $\refrel{\st M}{\it M}$ holds, given an \HEB\ type $\st M$ and an \HEBI\ type $\it M$. The relation $\refrelsym$ between these types is defined inductively over the structure of the types, and is clearly computable.
    In most inductive cases, such as \ruleNameApp, \ruleNameFix or \ruleNameLetin, there is an existential quantification over subcontexts of a given \HEBI\ context $\Gamma$. \cref{lem:FiniteSubsets} ensures that the set $Sub(\Gamma)$ is finite. In several similar cases, there is an existential quantification over all possible refinements of a given \HEB\ type. By \cref{lem:FiniteRefinement}, the set of such refinements is finite.
    In the case of \ruleNameFix, we search for an integer $n$, bounded by $|\Refn(\st M)|\cdot |Sub(\Gamma)|$, for some \HEB\ type $\st M$ and some \HEBI\ typing context $\Gamma$. As we already noticed, both these numbers are finite for each $\st M$ and $\Gamma$.
    Similarly, in the \ruleNameLetin\ case, the existential quantification is over a set $I$ whose size is bounded by $|\it F|$, for some \HEBI\ type $\it F$. According to \cref{rem:LeafMeasureFinite}, this measure is finite for each $\it F$. Finally, we use the algorithm $\mathcal B$, whose computability is established in \cref{lem:AlgoLeafComputable}. Therefore, all quantifiers appearing in \cref{def:AlgorithmLogic,def:AlgorithmContd} range over finite sets. We can conclude that there exists an algorithm that given an \HEB\ derivation $\Pi$ and an \HEBI\ judgement $J$ determines wheter $\algo{\Pi}{J}$ holds.
\end{proof}

\decidability*
\begin{proof}
    Let $\comput t\in \setComput$ be a closed term typable with $\itEffectReturn{\typeInt m}$ in \HEB, there exists an \HEB\ derivation $\stJugComput[\Pi]{\emptyset}{\comput t}{\itEffectReturn{\typeInt m}}$. By \cref{lem:FiniteRefinement}, there are only finitely many \HEBI\ types $\itSet{\typeInt {n}}$ such that $\refrel{\itEffectReturn{\typeInt m}}{\itEffectReturn{\itSet{\typeInt n}}}$, in particular one for each $0<n\leq m$. Let $J_s$ be the conclusion of $\Pi$, we can then construct a finite amount of \HEBI\ judgements $J_n\coloneq \itJugComput{\emptyset}{\comput t}{\itEffectReturn{\itSet{\typeInt n}}}$ such that $\refrel{J_s}{J_n}$. If $\models\algo{\Pi}{J_n$} for any of those, by \cref{prop:AlgoSoundness} $\itJugComput{\emptyset}{\comput t}{\itEffectReturn{\itSet{\typeInt n}}}$, and \cref{th:EffectTreeReachability} implies $\comput t\reductArr* \ret{\int n}$. Finally, this process is decidable because of \cref{prop:RefinementComputability}.
\end{proof}

\section{Proofs of Section \ref{sect:backtopcf}}
We omit here the full proof of Subject Reduction for \HEPCF. However, we state the transformations induced by the reduction on terms, as we did for \HEB. The proof proceeds using standard techniques, in which the derivation of $\comput t'$ is constructed according to the transformations described in \cref{def:HEPCFReduction}. Since several of the transformation rules for \HEPCF\ coincide with those for \HEB, we only present those that differ.

\begin{theorem}[\HEPCF\ Subject Reduction]
     Let $\comput t,\comput t'\in \setComput$, $\stJugComput[\Pi]{\emptyset}{\comput t}{\st U}$ and $\comput t\reductArr'\comput t'$ then it exists $\Pi'$ such that $\stJugComput[\Pi']{\emptyset}{\comput t'}{\st U}$.
\end{theorem}

To obtain the full subject reduction for \HEPCF, we need to prove the following lemma.

\begin{lemma}[\HEPCF\ Context Typing]~\\
    \label{lem:HEPCFSimpleContextTyping}%
    Let $\comput t \in \setComput$ and 
    $\ctxtEval$ a computation context such that 
    $\stJugComput[\Pi]{\emptyset}{\ctxtEval<\comput t>}{\st U}$.
    Then there is an \HEB\ type $\st S$ 
    such that $\stJugComput{\emptyset}{\comput t}{\st S}$ and 
    for every $\comput u$ such that $\itJugComput{\emptyset}{\comput u}{\st S}$, then there exist $\Pi'$ such that 
    $\stJugComput[\Pi']{\emptyset}{\ctxtEval<\comput u>}{\st U}$.
\end{lemma}

A Progress Lemma can also be proved in the exact same way as we already did in \cref{lem:HEBIProgress}.
\begin{lemma}[\HEPCF, Progress]
    \label{lem:HEPCFProgress}
    Let $\comput t\in \setComput$ be a closed term such that $\stJugComput[\Pi]{\emptyset}{\comput t}{\st U}$. Then either there exists $\comput s \in \setComput$ such that $\comput t \reductArr \comput s$, or $\comput t$ is a normal form computation, that is, either $\comput t = \ret{\val v}$ or $\comput t = \effect{\val v}{x}{\comput u}$ for some $\val v \in \setVal$ and $\comput u \in \setComput$.
\end{lemma}

The proof of Relative Subject Expansion for \HEPCF\ follows the same strategy already used in the case of \HEB, in particular relying on the definitions of refinement for types and derivations, as well as the notions of reduction and substitution between derivations. Due to these similarities, we omit the full proofs. We also observe that all lemmas concerning refinements for subsets and contexts stated in \cref{sec:apx5-3} remain valid in essentially the same form for \HEPCF, since they only depend on value types, which coincide in the two systems.

In \cref{def:RefinementDerivations,def:RefinementDerivationsHandler}, we present the refinement relations between \HEBI\ and \HEPCF\ derivations. Again, given that several rules of \HEPCF\ are identical to those of \HEB, we only present those that differ.

\begin{figure}
\centering
\input{Definitions/RefinementRelation.tex}

\caption{\HEPCF, Refinement for Computation Rule}
\label{def:RefinementDerivations}
\end{figure}

\begin{figure}
\centering
\input{Definitions/RefinementRelationHand.tex}
\caption{\HEPCF, Refinement for Handler Rules}
\label{def:RefinementDerivationsHandler}
\end{figure}

\begin{figure}
    \scalebox{0.93}{\input{Definitions/HEPCFReductions.tex}}
    \caption{\HEPCF, Reduction Rules.}
        \label{def:HEPCFReduction}
\end{figure}

\begin{proposition}[\HEPCF\ Weakening Refinement]
    \label{prop:HEPCFClosedRefinement}
    Let $p$ be a term, $\stJug[\Pi]{\Delta_1}{p}{\st U}$ be its derivation and $\Delta_2$ be a context such that for each $x::\st M\in\Delta_1$  and $x::\st N\in\Delta_2$ then $\st M=\st N$. Let $\itJug[\Xi]{\Gamma}{p}{\it E}$ be an \HEBI\ type derivation, then $\refrel{\Pi}{\Xi}$ if and only if $\refrel{\Pi^{(\Delta_2)}}{\Xi}$.
\end{proposition}

\begin{lemma}[\HEPCF\ Anti-substitution Lemma] ~\\
    \label{lem:HEPCFantiSubstLemma-ref}%
    Let $p \in \setTerm$
    be any term, $x$ a variable occourring free in $p$ and $\val v\in\setVal$ a closed value.
    If $\itJug[\Xi]{\Gamma}{p\sub{x}{\val v}}{\tau}$ for some \HEBI\ type $\tau$ and there exists an \HEPCF\ type $\st M$ such that $\stJug[\Pi_{p}]{\Delta;x::\st M}{p}{\st T}$, $\stJug[\Pi_\val v]{\emptyset}{\val v}{\st M}$ and $\refrel{\Pi_p\sub{x}{\Pi_\val v}}{\Pi}$, then it exists an \HEBI\ type $\it M$ such that 
    $\itJugValue[\refrel{\Pi_\val v}{\Xi_\val v}]{\emptyset}{\val v}{\it M}$ and
    $\itJug[\refrel{\Pi_p}{\Xi_p}]{\Gamma; x: \it M }{p}{\tau}$.
\end{lemma}
\begin{proof}
    By induction on $\Xi$.
\end{proof}

\begin{lemma}[Relative Elementary Subject Expansion for \HEPCF] ~\\
    \label{lem:ElementarySubjectExpansionrefEl}%
    Let $\comput t, \comput t' \in \setComput$ be two closed term such that 
    $\comput t \reductArr' \comput t'$. If there exists
    $\itJugComput[\Xi']{\emptyset}{\comput t'}{\it E}$ and $\stJugComput[\Pi']{\emptyset}{\comput t'}{\st U}$ such that $\refrel{\Pi'}{\Xi'}$ and also there exist $\stJugComput[\Pi]{\emptyset}{\comput t}{\st U}$ such that $\Pi\rightsquigarrow \Pi'$, 
    then there exists
    $\itJugComput[\Xi]{\emptyset}{\comput t}{\it E}$ such that $\refrel{\Pi}{\Xi}$.
\end{lemma}
\begin{proof}
    \input{Proofs/RefinementElementarySubjectExpansion.tex}
\end{proof}

As usual, to obtain \cref{lem:ElementarySubjectExpansion-ref}, we need the following lemma.

\begin{lemma}[\HEPCF\ Refinement Context Typing]~\\
    \label{lem:HEPCFRefinementContextTyping}%
    Let $\comput t \in \setComput$ and 
    $\ctxtEval$ a computation context such that 
    $\refrel{\stJugComput[\Pi]{\emptyset}{\ctxtEval<\comput t>}{\st U}}{\itJugComput[\Xi]{\emptyset}{\ctxtEval<\comput t>}{\it U}}$.
    Then there exists an \HEBI\ type $\it S$ and an \HEPCF\ type $\st S$ 
    such that $\refrel{\stJugComput{\emptyset}{\comput t}{\st S}}{\stJugComput{\emptyset}{\comput t}{\it F}}$ and 
    for every $\comput u$ such that $\refrel{\itJugComput{\emptyset}{\comput u}{\st S}}{\itJugComput{\emptyset}{\comput u}{\it S}}$, then there exist $\Pi'$ and $\Xi'$ such that 
    $\refrel{\itJugComput[\Pi']{\emptyset}{\ctxtEval<\comput u>}{\st U}}{\itJugComput[\Xi']{\emptyset}{\ctxtEval<\comput u>}{\it U}}$.
\end{lemma}
\begin{proof}
    By induction on $\ctxtEval$.
\end{proof}

%% file: Proofs/SplittingLemma2.tex
By case analysis on $\val v \in \setVal$, the
form of $\Xi$ can be deduced :

\paragraph*{Case Integers} For $\val v = \int n$, $\Xi$ has the following form :

 \[
    \begin{array}{ccc}
        \begin{prooftree}
                \hypo{\it M \in \{\itSetEmpty, \itSet{\typeInt{n}}\}}
            \typeRuleInt[1]{\itJugValue{\emptyset}{\int n}{\it M}}
        \end{prooftree}
    \end{array}
    \]

Then necessarily $\it M_i \in \{\itSetEmpty, \itSet{\typeInt{n}}\}$ for each $\rangeI$,
thus we can choose $\Pi_i$ as the following :

\begin{equation*}
    \begin{prooftree}
        \typeRuleInt[0]{\itJugValue{\emptyset}{\int n}{\it M_i}}
    \end{prooftree}
\end{equation*}
And conclude $\Gamma_i=\emptyset$ for each $\rangeI$.

\paragraph*{Case Variable} For $\val v = x$, $\Xi$ is as following :

 \[
    \begin{array}{ccc}
        \begin{prooftree}
            \typeRuleVar{\itJugValue{x:\it M}{x}{\it M}}
        \end{prooftree}
    \end{array}
    \]
By hypothesis $\it M_i\subseteq \it M$ for each $\rangeI$.
Thus we can choose each $\Pi_i$ as following :

\begin{equation*}
    \begin{prooftree}
        \typeRuleVar{\itJugValue{x : \it M_i}{x}{\it M_i}}
    \end{prooftree}
\end{equation*}

And conclude that $+_{\rangeI}\Gamma_i=+_{\rangeI}(x:\it M_i)=\Gamma=x:\it M$.

\paragraph*{Case Abstraction} For $\val v = \abs{x}{\comput t}$ :

\begin{equation*}
    \begin{array}{ccc}
    \begin{prooftree}
        \hypo{\left(
            \itJugComput[\Xi^k_\comput t]{\Gamma^k_\comput t; x : \it N_k}{\comput t}{\it E_k}
        \right)_{\rangeK}}
        \typeRuleAbs{\itJugValue[\Xi]{+_{\rangeK} \Gamma^k_\comput t}{\abs{x}{\comput t}}{\itSet{\it N_k \rightarrow \it E_k}_{\rangeK}}}
    \end{prooftree}
    \end{array}
\end{equation*}

We have $\it M = \itSet{\it M_k \rightarrow \it E_k}_{\rangeK}$ so for each $\rangeI$, $\it M_i = \itSet{\it M_j \rightarrow \it E_j}_{j\in J_i}$ with $J_i\subseteq K$ and $K=\bigcup\limits_{\rangeI}J_i$.
Let  $\Gamma_i = +_{j\in J_i} \Gamma^j_\comput t$, 
we can choose $\Pi_i$ as the following for each $\rangeI$:

\begin{equation*}
    \begin{prooftree}
        \hypo{\left(
            \itJugComput[\Xi^j_\comput t]{\Gamma^j_\comput t; x : \it M_j}{\comput t}{\it E_j}
        \right)_{j\in J_i}}
        \typeRuleAbs{\itJugValue[\Pi_i]{+_{j\in J_i} \Gamma^j_\comput t}{\abs{x}{\comput t}}{\itSet{\it M_j \rightarrow \it E_j}_{j\in J_i}}}
    \end{prooftree}
\end{equation*}
Therefore $\Gamma=+_{\rangeK}\Gamma^k_\comput t= +_{\rangeI}\Gamma_i =+_{\rangeI}+_{j\in J_i}\Gamma^j_\comput t$.

\paragraph*{Case Fixpoint} For $\val v = \fix{x}{\val w}$, $\Xi$ has two possible shapes:
\begin{itemize}
\item[\bltI] Base case :
    \begin{equation*}
    \begin{prooftree}
        \hypo{
            \itJugValue[\Xi_1]
                {\Gamma; x : \itSetEmpty}
                {\val w}
                {\itSet{\it M_k \rightarrow \it E_k}_{\rangeK}}}
        \typeRuleFixBase[1]{
            \itJugValue
                {\Gamma}
                {\fix{x}{\val w}}
                {\itSet{\it M_k \rightarrow \it E_k}_{\rangeK}}}
    \end{prooftree}
    \end{equation*}

\item[\bltI] Recursive case : For $\Gamma = \Gamma_1 + \Gamma_2$ :

    \begin{equation*}
    \begin{prooftree}
        \hypo{
            \itJugValue[\Xi_1]
                {\Gamma_1; x : \it M'}
                {\val w}
                {\itSet{\it M_k \rightarrow \it E_k}_{\rangeK}}}
        \hypo{
            \itJugValue[\Xi_2]
                {\Gamma_2}
                {\fix{x}{\val w}}
                {\it M'}}
        \typeRuleFixRec{
            \itJugValue
                {\Gamma_1 + \Gamma_2}
                {\fix{x}{\val w}}
                {\itSet{\it M_k \rightarrow \it E_k}_{\rangeK}}}
    \end{prooftree}
    \end{equation*}

\end{itemize}

In both cases we have $\it M = \itSet{\it M_k \rightarrow \it E_k}_{\rangeK}$ so for each $\rangeI$, $\it M_i = \itSet{\it M_j \rightarrow \it E_j}_{j\in J_i}$ with $J_i\subseteq K$ and $K=\bigcup\limits_{\rangeI}J_i$. By induction hypothesis on $\Xi_1$ we obtain derivations $\itJugValue[\Pi^i_\val w]{\Gamma^i_\val w;x:\it M'}{\val w}{\itSet{\it M_j\rightarrow E_j}_{j\in J_i}}$ such that $\Gamma_1=+_{\rangeI}\Gamma^i_\val w$. We can then build a derivation $\Pi_i$ for each $\rangeI$ as following, for the recursive case:

 \begin{equation*}
    \begin{prooftree}
        \hypo{
            \itJugValue[\Xi^i_\val w]
                {\Gamma^i_\val w; x : \it M'}
                {\val w}
                {\itSet{\it M_j \rightarrow \it E_j}_{j\in J_i}}}
        \hypo{ \itJugValue[\Xi_2]
                {\Gamma_2}
                {\fix{x}{\val w}}
                {\it M'}}
        \typeRuleFixRec[2]{
            \itJugValue[\Pi_i]
                {\Gamma^i_\val w+\Gamma_2}
                {\fix{x}{\val w}}
                {\itSet{\it M_j \rightarrow \it E_j}_{j\in J_i}}}
    \end{prooftree}
    \end{equation*}
With $\Gamma_i=\Gamma^i_\val w+\Gamma_2$.
In the base case the construction is similar. Therefore, $\Gamma=+_{\rangeI}\Gamma_i$.

%% file: Proofs/IntersectionLemma.tex
By induction on the sum size of the $\Pi_i$ (i.e. the sum of the size of each proof).
By case analysis on $\val v$, the form of the $\Pi_i$ can be deduced :

\paragraph*{Case Integers} 
If $\val v = \int n$, each $\Pi_i$ has the following form :
\begin{equation*}
    \begin{prooftree}
            \typeRuleInt[0]{\itJugValue{\emptyset}{\int n}{\it M_i}}
    \end{prooftree}
\end{equation*}
Where for every $i$, $\it M_i \in \{\itSetEmpty, \itSet{\typeInt{n}}\}$.
Then necessarily $\bigcup_{\rangeI} \it M_i \in \{\itSetEmpty, \itSet{\typeInt{n}}\}$,
thus one can choose $\Pi$ as the following :
\begin{equation*}
    \begin{prooftree}
        \typeRuleInt[0]{\itJugValue{\emptyset}{\int n}{\bigcup\limits_{\rangeI} \it M_i}}
    \end{prooftree}
\end{equation*}

\paragraph*{Case Variable} 
For $\val v = x$, each $\Pi_i$ has the following form :
\begin{equation*}
    \begin{prooftree}
        \typeRuleVar{\itJugValue{x : \it M_i}{x}{\it M_i}}
    \end{prooftree}
\end{equation*}
Where for every $i$, $\Gamma_i = x : \it M_i$. By hypothesis $\bigcup\limits_{\rangeI}\it M_i$ exists, thus one can choose $\Pi$ as the following :
\begin{equation*}
    \begin{prooftree}
        \typeRuleVar{\itJugValue{x : \bigcup\limits_{\rangeI} \it M_i}{x}{\bigcup\limits_{\rangeI} \it M_i}}
    \end{prooftree}
\end{equation*}

\paragraph*{Case Abstraction} 
For $\val v = \abs{x}{\comput t}$, each $\Pi_i$ has the following form :
\begin{equation*}
    \begin{prooftree}
        \hypo{\left(
            \itJugComput[\Pi_{i,j}]{\Gamma_{i, j}; x : \it N_{i,j}}{\comput t}{\it E_{i,j}}
        \right)_{\rangeJ_i}}
        \typeRuleAbs{\itJugValue{+_{\rangeJ_i} \Gamma_{i,j}}{\abs{x}{\comput t}}{\itSet{\it N_{i,j} \rightarrow \it E_{i,j}}_{\rangeJ_i}}}
    \end{prooftree}
\end{equation*}
Where $\it M_i = \itSet{\it N_{i,j} \rightarrow \it E_{i,j}}_{\rangeJ_i}$ and $\Gamma_i=+_{\rangeJ_i} \Gamma_{i,j}$.
Since $\bigcup_{\rangeI} \it M_i = \itSet{\it N_{i,j} \rightarrow \it E_{i,j}}_{\rangeJ_i, \rangeI}$ and $+_{\rangeI}\Gamma_i$ both exist by hypothesis,
let $\Pi$ be the following :
\begin{equation*}
    \begin{prooftree}
        \hypo{\left(
            \itJugComput[\Pi_{i,j}]{\Gamma_{i, j}; x : \it N_{i,j}}{\comput t}{\it E_{i, j}}
        \right)_{\rangeJ_i, \rangeI}}
        \typeRuleAbs{\itJugValue
            {+_{\rangeI,\rangeJ_i} \Gamma_{i, j}}
            {\abs{x}{\comput t}}
            {\itSet{\it N_{i,j} \rightarrow \it E_{i,j}}_{\rangeJ_i, \rangeI}}}
    \end{prooftree}
\end{equation*}

\paragraph*{Case Fixpoint} 
For $\val v = \fix{x}{\val w}$, each $\Pi_i$ has two possible forms :
\begin{itemize}
\item[\bltI] Base case :
    \begin{equation*}
    \begin{prooftree}
        \hypo{
            \itJugValue[\Pi_i^1]
                {\Gamma_i; x : \itSetEmpty}
                {\val w}
                {\itSet{\it M_j \rightarrow \it E_j}_{\rangeJ_i}}}
        \typeRuleFixBase[1]{
            \itJugValue
                {\Gamma_i}
                {\fix{x}{\val w}}
                {\itSet{\it M_j \rightarrow \it E_j}_{\rangeJ_i}}}
    \end{prooftree}
    \end{equation*}

\item[\bltI] Recursive case : 
For $\Gamma_i = \Gamma_i^1 + \Gamma_i^2$ :

    \begin{equation*}
    \begin{prooftree}
        \hypo{
            \itJugValue[\Pi_i^1]
                {\Gamma_i^1; x : \it N_i}
                {\val w}
                {\itSet{\it M_j \rightarrow \it E_j}_{\rangeJ_i}}}
        \hypo{
            \itJugValue[\Pi_i^2]
                {\Gamma_i^2}
                {\fix{x}{\val w}}
                {\it N_i}}
        \typeRuleFixRec{
            \itJugValue
                {\Gamma_i^1 + \Gamma_i^2}
                {\fix{x}{\val w}}
                {\itSet{\it M_j \rightarrow \it E_j}_{\rangeJ_i}}}
    \end{prooftree}
    \end{equation*}

\end{itemize}
We first study the case in which all the $\Pi_i$ are in base form.
By  hypothesis, $+_{\rangeI}\Gamma_i$ exists, therefore $+_{\rangeI}(\Gamma_i;x:\itSetEmpty)$ exists. Moreover, $\bigcup\limits_{\rangeI}\it M_i=\itSet{\it M_j \rightarrow \it E_j
        }_{\rangeJ_i, \rangeI}$ also exists.
Then by induction hypothesis
on the $\Pi_i^1$, there is $\itJugValue[\Pi_1]
    {+_{\rangeI}\Gamma_i; x : \itSetEmpty}
    {\val w}{\itSet{\it M_j \rightarrow \it E_j
        }_{\rangeJ_i, \rangeI}}$.
One can conclude this subcase with the following $\Pi$.
\begin{equation*}
\begin{prooftree}
    \hypo{
        \itJugValue[\Pi_1]
            {+_{\rangeI}\Gamma_i; x : \itSetEmpty}
            {\val w}
            {\itSet{\it M_j \rightarrow \it E_j
                }_{\rangeJ_i, \rangeI}}}
    \typeRuleFixBase[1]{
        \itJugValue
            {+_{\rangeI}\Gamma_i}
            {\fix{x}{\val w}}
            {\itSet{\it M_j \rightarrow \it E_j}_{\rangeJ_i, \rangeI}}}
\end{prooftree}
\end{equation*}
Otherwise, let's assume there is at least one $\Pi_i$ in recursive form. For each such derivation $\Pi_i$ it exists a premise $\Pi^2_i$ and its size is strictly less then the size of $\Pi_i$.
For every $\Pi_i$ in base form, we can take as $\Pi^2_i$ the derivation $\itJugValue[\Pi_i^2]{\emptyset}{\fix{x}{\val w}}{\itSetEmpty}$, that exists by typability of values (\Cref{lem:TypabilityValues}).
The proof size of $\Pi_i^2$ is at most the size of $\Pi_i$ (\Cref{Col:TypabilityOfValuesOptmialSize}).
Since at least one $\Pi_i$ is initially in recursive form, then the sum size of all the
$\Pi_i^2$ is strictly smaller than the sum size of all the $\Pi_i$.
Moreover, since each $\it N_i$ types the term $\fix{x}{\val w}$, either $\it N_i=\itSetEmpty$ or $\it N_i=\itSet{\it N'_j\rightarrow \it F_j}_{\rangeJ}$, therefore $\bigcup\limits_{\rangeI}\it N_i$ always exists.
Thus by induction hypothesis on the $\Pi_i^2$, there exists
$\itJugValue[\Pi_2]{+_{\rangeI}\Gamma_i^2}
    {\fix{x}{\val w}}{\bigcup_{\rangeI}\it N_i}$.
We now focus on the $\Pi_i^1$. As before, by hypothesis $+_{\rangeI}\Gamma_i^1$ exists, therefore $+_{\rangeI}(\Gamma_i^1;x:\it N_i)$ exists. Also, $\bigcup\limits_{\rangeI}\it M_i=\itSet{\it M_j \rightarrow \it E_j
        }_{\rangeJ_i, \rangeI}$ exists.
Threfore, by induction hypothesis on the $\Pi_i^1$,
there is $\itJugValue[\Pi_1]
    {+_{\rangeI}\Gamma_i^1, x : \bigcup_{\rangeI} \it N_i}
    {\val w}{\itSet{\it M_j \rightarrow \it E_j
        }_{\rangeJ_i, \rangeI}}$

    Therefore one can deduce the following $\Pi$ :
\begin{equation*}
\begin{prooftree}
        \hypo{\itJugValue[\Pi_1]
            {+_{\rangeI}\Gamma_i^1, x : \bigcup\limits_{\rangeI} \it N_i}
            {\val w}{\itSet{\it M_j \rightarrow \it E_j
                }_{\rangeJ_i, \rangeI}}}
        \hypo{\itJugValue[\Pi_2]
            {+_{\rangeI}\Gamma_i^2}
            {\fix{x}{\val w}}{\bigcup_{\rangeI}\it N_i}}
    \typeRuleFixRec{\itJugValue
        {+_{\rangeI}(\Gamma_i^1+\Gamma_i^2)}
        {\fix{x}{\val w}}
        {\itSet{\it M_j \rightarrow \it E_j
            }_{\rangeJ_i, \rangeI}}}
\end{prooftree}
\end{equation*}

%% file: Proofs/SubstitutionLemma.tex
By induction on $\Pi$. By case analysis on $p$, the
form of the $\Pi$ can be deduced :

\paragraph*{Case Integers} 
Let $p = \int n$, so $\int n\sub{x}{\val v}=\int n$ and $\Pi$ has the form:
 \[
     \begin{prooftree}
                        \hypo{\it N \in \{\itSetEmpty, \itSet{\typeInt{n}}\}}
                    \typeRuleInt[1]{\itJugValue{\emptyset}{\int n}{\it N}}
                \end{prooftree}
    \]

So necessarily, $\it M=\itSetEmpty$ and $\Gamma_p=\emptyset$. Therefore, we take $\Xi=\Pi$

\paragraph*{Case Variable}
\begin{itemize}
    \item If $p = x$, $x\sub{x}{\val v}=\val v$ and $\Pi$ has the following form:
    \[
        \begin{prooftree}
            \typeRuleVar{\itJugValue[\Pi]{x:\it M}{x}{\it M}}
        \end{prooftree}
    \]
    
    Being $\Gamma_p=\emptyset$ and $\tau = \it M$. We then take $\Xi=\Pi_\val v$.

    \item If $p = y \neq x$, $y\sub{x}{\val v}=y$ and $\Pi$ has the following form :
    \[
        \begin{prooftree}
            \typeRuleVar{\itJugValue[\Pi]{y:\it N}{y}{\it N}}
        \end{prooftree}
    \]
    Therefore, $\Gamma_p=\emptyset$ and $\it M=\itSetEmpty$. Finally, we take $\Xi=\Pi$
\end{itemize}

\paragraph*{Case Abstraction} 
For $p = \abs{y}{\comput t}$, $\Pi$ has the following form :

\begin{equation*}
     \begin{prooftree}
                \hypo{(\itJugComput[\Pi_i]
                    {\Gamma_i; x : \it M_i;y: \it N_i }
                    {\comput t}{\it E_i})_{\rangeI}}
            \typeRuleAbs{\itJugValue
                {+_{\rangeI} \Gamma_i;x:\bigcup\limits_{\rangeI}\it M_i}
                {\abs{y}{\comput t}}
                {\itSet{\it N_i \rightarrow \it E_i}_{\rangeI}}}
        \end{prooftree}
\end{equation*}

By \Cref{lem:Splitting} on $\Pi_\val v$ we can obtain some derivations $\itJugValue[\Pi^i_\val v]{\emptyset}{\val v}{\it M_i}$ for $\rangeI$. 
By induction hypothesis on each $\Pi_i$ with these $\Pi^i_\val v$, there exist derivations
$\itJugComput[\Pi'_i]{\Gamma_i;y:\it N_i}{\comput t\sub{x}{\val v}}{\it E_i}$.
Thus, one can derive the following $\Xi$.
\begin{equation*}
     \begin{prooftree}
                \hypo{(\itJugComput[\Pi'_i]{\Gamma_i;y:\it N_i}{\comput t\sub{x}{\val v}}{\it E_i})_{\rangeI}}
            \typeRuleAbs{\itJugValue
                {+_{\rangeI} \Gamma_i}
                {\abs{y}{\comput t\sub{x}{\val v}}}
                {\itSet{\it N_i \rightarrow \it E_i}_{\rangeI}}}
        \end{prooftree}
\end{equation*}

\paragraph*{Case Fixpoint} 
For $p = \fix{y}{\val w}$, $\Pi$ has two possible shapes:
\begin{itemize}
    
\item[\bltI] Base case :
    \begin{equation*}
    \begin{prooftree}
        \hypo{
            \itJugValue[\Pi_1]
                {\Gamma;x:\it M; y : \itSetEmpty }
                {\val w}
                {\itSet{\it N_i \rightarrow \it E_i}_{\rangeI}}}
        \typeRuleFixBase[1]{
            \itJugValue
                {\Gamma;\it M}
                {\fix{y}{\val w}}
                {\itSet{\it N_i \rightarrow \it E_i}_{\rangeI}}}
    \end{prooftree}
    \end{equation*}

    By induction hypothesis on $\Pi_1$ with $\Pi_\val v$, there exists a derivationr $\itJugValue[\Pi'_1]{\Gamma;y:\itSetEmpty}{\val w\sub{x}{\val v}}{\itSet{\it N\rightarrow \it E}_{\rangeI}}$.
    We can the construct the following derivation $\Xi$:
    \begin{equation*}
    \begin{prooftree}
            \hypo{\itJugValue[\Pi'_1]{\Gamma;y:\itSetEmpty}{\val w\sub{x}{\val v}}{\itSet{\it N\rightarrow \it E}_{\rangeI}}}
        \typeRuleFixBase[1]{\itJugValue{\Gamma;}
            {\fix{y}{\val w}\sub{x}{\val v}}{\itSet{\it N_i \rightarrow \it E_i}_{\rangeI}}}
    \end{prooftree}
    \end{equation*}

\item[\bltI] Recursive case :
    \begin{equation*}
    \begin{prooftree}
        \hypo{
            \itJugValue[\Pi_1]
                {\Gamma_1;x:\it M_1; y :: \it N'}
                {\val w}
                {\itSet{\it N_i \rightarrow \it E_i}_{\rangeI}}}
        \hypo{
            \itJugValue[\Pi_2]
                {\Gamma_2;x:\it M_2}
                {\fix{y}{\val w}}
                {\it N'}}
        \typeRuleFixRec{
            \itJugValue[\Pi]
                {\Gamma_1 + \Gamma_2;x:\it M_1\cup \it M_2}
                {\fix{y}{\val w}}
                {\itSet{\it N_i \rightarrow \it E_i}_{\rangeI}}}
    \end{prooftree}
    \end{equation*}

    By \Cref{lem:Splitting} we can split the derivation $\Pi_\val v$ in two derivations $\itJugValue[\Pi^1_\val v]{\emptyset}{\val v}{\it M_1}$ and $\itJugValue[\Pi^2_\val v]{\emptyset}{\val v}{\it M_2}$. 
    By induction hypothesis on $\Pi_1$ with $\Pi^1_\val v$, there exists some derivations $\itJugValue[\Pi'_1]{\Gamma_1;y:\it N'}{\val w\sub{x}{\val v}}{\itSet{\it N\rightarrow \it E}_{\rangeI}}$.
    Similarly, by induction hypothesis on $\Pi_2$ with $\Pi^2_\val v$, there exists some derivations $\itJugValue[\Pi'_2]{\Gamma_2}{\fix{y}{\val w}\sub{x}{\val v}}{\it N'}$.
    We can therefore construct the following derivation $\Xi$:
    \begin{equation*}
    \begin{prooftree}
            \hypo{\itJugValue[\Pi'_1]{\Gamma_1;y:\it N'}{\val w\sub{x}{\val v}}{\itSet{\it N\rightarrow \it E}_{\rangeI}}}
            \hypo{\itJugValue[\Pi'_2]{\Gamma_2}{\fix{y}{\val w}\sub{x}{\val v}}{\it N'}}
        \typeRuleFixRec{\itJugValue
            {\Gamma_1+\Gamma_2}
            {\fix{y}{\val w}\sub{x}{\val v}}{\itSet{\it N_i \rightarrow \it E_i}_{\rangeI}}}
    \end{prooftree}
    \end{equation*}
\end{itemize}

\paragraph*{Case Let-Binding}
For $p = \letin{y}{\comput t}{\comput u}$, $\Pi$ has the following form :
\begin{equation*}
\begin{prooftree}
                \hypo{\itJugComput[\Pi_\comput t]{\Gamma;x:\it M_\comput t }{\comput t}{\it E}}
                \hypo{
                    \it E
                        \replaceLeaf{\it N_i \rightarrow \it G_i}[\rangeI]
                    \it F}
                \hypo{
                    (\itJugComput[\Pi^i_\comput u]
                        {\Gamma_i; x : \it M_i;y:\it N_i}
                        {\comput u}
                        {\it G_i}
                    )_{\rangeI}}
            \typeRuleLetin{
                \itJugComput
                    {\Gamma +_{\rangeI} \Gamma_i;x: \it M_\comput t \bigcup\limits_{\rangeI} \it M_i}
                    {\letin{y}{\comput t}{\comput u}}
                    {\it F}
            }
        \end{prooftree}
\end{equation*}

By \Cref{lem:Splitting} on $\Pi_\val v$, there exist derivations $\itJugValue[\Pi^\val v_\comput t]{\emptyset}{\val v}{\it M_\comput t}$ and $\itJugValue[\Pi_{i}^{\val v}]{\emptyset}{\val v}{\it M_i}$.
By induction hypothesis on $\Pi_{\comput t}$ and $\Pi^\val v_\comput t$ we have a derivation $\itJugComput[\Pi'_{\comput t}]{\Gamma}{\comput t\sub{x}{\val v}}{\it E}$.
Moreover, by induction on each $\Pi^i_{\comput u}$ with $\Pi_i^{\val v}$, there exist derivations
and $\itJugComput[\Pi'^i_{\comput u}]{\Gamma_i;y:\it N_i}{\comput t\sub{x}{\val v}}{\it G_i}$. 
Thus, one can derive the following $\Xi$:
\begin{equation*}
\begin{prooftree}
        \hypo{\itJugComput[\Pi'_{\comput t}]{\Gamma_\comput t}{\comput t\sub{x}{\val v}}{\it E}}
        \hypo{
                    \it E
                        \replaceLeaf{\it N_i \rightarrow \it G_i}[\rangeI]
                    \it F}
        \hypo{(\itJugComput[\Pi'^i_{\comput u}]{\Gamma_i;y:\it N_i}{\comput t\sub{x}{\val v}}{\it G_i})_{\rangeI}}
    \typeRuleLetin{
        \itJugComput
            {\Gamma_\comput t }
            {\letin{y}{\comput t\sub{x}{\val v}}{\comput u\sub{x}{\val v}}}
            {\it F}
    }
\end{prooftree}
\end{equation*}

\paragraph*{Case Application}
For $p = \app{\val w_1}{\val w_2}$, $\Pi$ has the following form :
\begin{equation*}
\begin{prooftree}
        \hypo{\itJugValue[\Pi_1]{\Gamma_1;x:\it M_1}{\val w_1}{\itSet*{\it N \rightarrow \it E}}}
        \hypo{\itJugValue[\Pi_2]{\Gamma_2;x:\it M_2}{\val w_2}{\it N}}
    \typeRuleApp{\itJugComput{\Gamma_1 + \Gamma_2;x:\it M_1\cup \it M_2}{\app{\val w_1}{\val w_2}}{\it E}}
\end{prooftree}
\end{equation*}
By \Cref{lem:Splitting} we can split the derivation $\Pi_\val v$ in two derivations $\itJugValue[\Pi^1_\val v]{\emptyset}{\val v}{\it M_1}$ and $\itJugValue[\Pi^2_\val v]{\emptyset}{\val v}{\it M_1}$.
By induction hypothesis on $\Pi_1$ with $\Pi^1_\val v$  and on $\Pi_2$ with $\Pi^2_\val v$, we obtain the derivations
$\itJugValue[\Pi'_1]{\Gamma_1}{\val w_1\sub{x}{\val v}}{\it N \rightarrow \it E}$
and $\itJugValue[\Pi'_2]{\Gamma_2}{\val w_2\sub{x}{\val v}}{\it N}$.
We can then construct the following derivation $\Xi$:
\begin{equation*}
\begin{prooftree}
        \hypo{\itJugValue[\Pi'_1]{\Gamma_1}{\val w_1\sub{x}{\val v}}{\it N \rightarrow \it E}}
        \hypo{\itJugValue[\Pi'_2]{\Gamma_2}{\val w_2\sub{x}{\val v}}{\it N}}
    \typeRuleApp{\itJugComput{\Gamma_1 + \Gamma_2}{\app{\val w_1\sub{x}{\val v}}{\val w_2\sub{x}{\val v}}}{\it E}}
\end{prooftree}
\end{equation*}

\paragraph*{Case Case}
For $p = \case{\val w}{\comput t_1, \ldots, \comput t_n}$, $\Pi$ has the following form :
\begin{equation*}
\begin{prooftree}
        \hypo{\itJugValue[\Pi_1]{\Gamma_1;x:\it M_1}{\val w}{\itSet{\typeInt{m}}}}
        \hypo{\itJugComput[\Pi_2]{\Gamma_2;x:\it M_2}{\comput t_m}{\it E}}
    \typeRuleCase{\itJugComput{\Gamma_1 + \Gamma_2;x:\it M}
        {\case{\val w}{\comput t_1, \ldots, \comput t_n}}{\it E}}
\end{prooftree}
\end{equation*}
By \Cref{lem:Splitting} we can split the derivation $\Pi_\val v$ in two derivations $\Pi^1_\val v$ and $\Pi^2_\val v$.
By induction hypothesis on $\Pi_1$ with $\Pi^1_\val v$ and $\Pi_2$ with $\Pi^2_\val v$, there exist two derivations $\itJugValue[\Pi'_1]{\Gamma_1}{\val w\sub{x}{\val v}}{\itSet{\typeInt{m}}}$ and $\itJugComput[\Pi'_2]{\Gamma_2}{\comput t_m\sub{x}{\val v}}{\it E}$.
Therfore, one can derive the following $\Xi$.
\begin{equation*}
\begin{prooftree}
        \hypo{\itJugValue[\Pi'_1]{\Gamma_1}{\val w\sub{x}{\val v}}{\itSet{\typeInt{m}}}}
        \hypo{\itJugComput[\Pi'_2]{\Gamma_2}{\comput t_m\sub{x}{\val v}}{\it E}}
    \typeRuleCase{\itJugComput{\Gamma_1 + \Gamma_2}
        {\case{\val w\sub{x}{\val v}}{\comput t_1\sub{x}{\val v}, \ldots, \comput t_n\sub{x}{\val v}}}{\it E}}
\end{prooftree}
\end{equation*}

\paragraph*{Case Return}
For $p = \ret{\val w}$, $\Pi$ has the following form :
\begin{equation*}
\begin{prooftree}
        \hypo{\itJugValue[\Pi_1]{\Gamma;x:\it M}{\val w}{\it N}}
    \typeRuleRet{\itJugComput{\Gamma;x:\it M}
        {\ret{\val w}}{\itEffectReturn{\it N}}}
\end{prooftree}
\end{equation*}

By induction hypothesis on $\Pi_1$ with $\Pi_\val v$, there exist a derivation $\itJugValue[\Pi'_1]{\Gamma_1v}{\val w\sub{x}{\val v}}{\it N}$.
Therfore, one can derive the following $\Xi$.
\begin{equation*}
\begin{prooftree}
        \hypo{\itJugValue[\Pi'_1]{\Gamma_1}{\val w\sub{x}{\val v}}{\it N}}
    \typeRuleRet{\itJugComput{\Gamma}
        {\ret{\val w\sub{x}{\val v}}}{\itEffectReturn{\it N}}}
\end{prooftree}
\end{equation*}

\paragraph*{Case Effect}
For $p = \effect{\val w}{y}{\comput t}$, $\Pi$ has the following form :
\begin{equation*}
\begin{prooftree}
        \hypo{\itJugValue[\Pi_{\val w}]
            {\Gamma_\val w;x:\it M_\val w}
            {\val w}
            {\it N}}
        \hypo{(\itJugComput[\Pi_i]
            {\Gamma_i;x:\it M_i; y : \it N_i}
            {\comput t}
            {\it E_i}
        )_{\rangeI}}
    \typeRuleEff[2]{
        \itJugComput
            {\Gamma +_{\rangeI} \Gamma_i;x:\it M_\val w\bigcup\limits_{\rangeI}\it M_i}
            {\effect{\val w}{y}{\comput t}}
            {\itEffect{\sigma}{\it N}{\it N_i \rightarrow \it E_i}_{\rangeI}}}
\end{prooftree}
\end{equation*}

By \Cref{lem:Splitting} we can split the derivation $\Pi_\val v$ in derivations $\itJugValue[\Pi^{\val w}_\val v]{\emptyset}{\val v}{\it M_\val w}$ and $\itJugValue[\Pi^i_\val v]{\emptyset}{\val v}{\it M_i}$ for each $\rangeI$,.
By induction hypothesis on $\Pi_\val w$ with $\Pi^\val w_\val v$ it exists a derivation $\itJugValue[\Pi'_\val w]{\Gamma_\val w}{\val w\sub{x}{\val v}}{\it N}$.
Moreover, by induction hypothesis on each $\Pi_i$ with $\Pi^i_\val v$ for $\rangeI$, there exist derivations $\itJugComput[\Pi'_i]{\Gamma_i;\it N_i}{\comput t\sub{x}{\val v}}{\it E_i}$.
We can construct the following derivation $\Xi$:

\begin{equation*}
\begin{prooftree}
        \hypo{\itJugValue[\Pi'_\val w]{\Gamma_\val w}{\val w\sub{x}{\val v}}{\it N}}
        \hypo{(\itJugComput[\Pi'_i]{\Gamma_i}{\comput t\sub{x}{\val v}}{\it E_i}
        )_{\rangeI}}
    \typeRuleEff[2]{
        \itJugComput
            {\Gamma +_{\rangeI} \Gamma_i}
            {\effect{\val w\sub{x}{\val v}}{y}{\comput t\sub{x}{\val v}}}
            {\itEffect{\sigma}{\it N}{\it N_i \rightarrow \it E_i}_{\rangeI}}}
\end{prooftree}
\end{equation*}

\paragraph*{Case Handler}
For $p = \handle{\handler h}{\comput t}$, $\Pi$ has the following form :
\begin{equation*}
\begin{prooftree}
        \hypo{\itJugHandler[\Pi_1]{\Gamma_1;x:\it M_1}{\handler h}{\it E \Rightarrow \it F}}
        \hypo{\itJugComput[\Pi_2]{\Gamma_2;x:\it M_2}{\comput t}{\it E}}
    \typeRuleHandle{\itJugHandler{\Gamma_1 + \Gamma_2;x:\it M_1\cup\it M_2}{\handle{\handler h}{\comput t}}{\it F}}
\end{prooftree}
\end{equation*}
By \Cref{lem:Splitting} we can split the derivation $\Pi_\val v$ in two derivations $\itJugValue[\Pi^1_\val v]{\emptyset}{\val v}{\it M_1}$ and $\itJugValue[\Pi^2_\val v]{\emptyset}{\val v}{\it M_1}$.
By induction hypothesis on $\Pi_1$ with $\Pi^1_\val v$ and $\Pi_2$ with $\Pi^2_\val v$, there exist two derivations $\itJugHandler[\Pi'_1]{\Gamma_1}{\handler h\sub{x}{\val v}}{\it E \Rightarrow\it F}$ and $\itJugComput[\Pi'_2]{\Gamma_2}{\comput t\sub{x}{\val v}}{\it E}$.
Therfore, one can derive the following $\Xi$:
\begin{equation*}
\begin{prooftree}
        \hypo{\itJugHandler[\Pi'_1]{\Gamma_1}{\handler h\sub{x}{\val v}}{\it E \Rightarrow\it F}}
        \hypo{\itJugComput[\Pi'_2]{\Gamma_2}{\comput t\sub{x}{\val v}}{\it E}}
    \typeRuleHandle{\itJugHandler{\Gamma_1 + \Gamma_2}{\handle{\handler h\sub{x}{\val v}}{\comput t}\sub{x}{\val v}}{\it E}}
\end{prooftree}
\end{equation*}

\paragraph*{Case Handler Branching Return}
For $p = \{\retClause{y}{\comput t}\}\cup \handler h$, $\Pi$ has the following form :
\begin{equation*}
    \begin{prooftree}
                \hypo{\itJugComput[\Pi_1]
                    {\Gamma;x:\it M;
                        y : \it N}
                    {\comput t}
                    {\it E}}
                \hypo{\{\retClause{y}{\comput t}\}\handlerCompat \handler h}
            \typeRuleHandlerRet[2]{\itJugHandler
                {\Gamma;x:\it M}
                {\{\retClause{y}{\comput t}\} \cup \handler h}%
                {\itEffectReturn{\it N} \Rightarrow \it E}
            }%
        \end{prooftree}
\end{equation*}

By induction hypothesis on $\Pi_1$ with $\Pi_\val v$, there exist a derivation $\itJugComput[\Pi'_1]{\Gamma;y:\it N}{\comput t\sub{x}{\val v}}{\it E}$.
Therefore, one can derive the following $\Pi'$, using \Cref{lem:HandlerComputIndipendence}.
\begin{equation*}
    \begin{prooftree}
                \hypo{\itJugComput[\Pi'_1]{\Gamma;y : \it N}{\comput t\sub{x}{\val v}}{\it E}}
                \hypo{\{\retClause{y}{\comput t\sub{x}{\val v}}\} \handlerCompat \handler h}
            \typeRuleHandlerRet[2]{\itJugHandler
                {\Gamma}
                {\{\retClause{y}{\comput t\sub{x}{\val v}}\}  \cup \handler h}%
                {\itEffectReturn{\it N} \Rightarrow \it E}
            }%
        \end{prooftree}
\end{equation*}

\paragraph*{Case Handler Branching}
For $p = \{\effectClause[\sigma]{y}{r}{\comput s}\} \cup \handler h$, $\Pi$ has the following form :

\begin{equation*}
\begin{prooftree}
                \hypo{\itJugComput[\Pi_\comput t]
                    {\Gamma_\comput t; x:\it M_\comput t
                        y : \it N' ;
                        r : \itSet*{
                                \it N_i \rightarrow \it G_i
                                }_{\rangeI}}
                    {\comput t}
                    {\it E}}
                \hypo{
                    (\itJugHandler[\Pi_i]
                        {\Gamma_i;x:\it M_i}%
                        {\handler h}%
                        {\it F_i \Rightarrow \it G_i}
                    )_{\rangeI}
                }
                \hypo{(A)}
            \typeRuleHandler[3]{\itJugHandler
                {\Gamma +_{\rangeI} \Gamma_i;x:\it M_\comput t+_{\rangeI}\it M_i}
                {\{\effectClause{y}{r}{\comput t}\} \cup \handler h}%
                {\itEffect{\sigma}{\it N'}
                    {\it N_i\rightarrow \it F_i
                        }_{\rangeI}
                    \Rightarrow \it E}}%
        \end{prooftree}
\end{equation*}
\[
(A)= \begin{matrix}
                    \{\effectClause{y}{r}{\comput t}\} 
                    \handlerCompat \handler h
            \end{matrix}
\]
By \Cref{lem:Splitting} we can split the derivation $\Pi_\val v$ in derivations $\itJugValue[\Pi^{\comput t}_\val v]{\emptyset}{\val v}{\it M_\comput t}$ and $\itJugValue[\Pi^i_\val v]{\emptyset}{\val v}{\it M_i}$ for each $\rangeI$.
By induction hypothesis on $\Pi_\comput t$ with $\Pi^\comput t_\val v$ it exists a derivation $\itJugComput[\Pi'_\comput t]
                    {\Gamma_\comput t; 
                        y : \it N' ;
                        r : \itSet*{
                                \it N_i \rightarrow \it G_i
                                }_{\rangeI}}
                    {\comput t\sub{x}{\val v}}
                    {\it E}$.

Moreover, by induction hypothesis on each $\Pi_i$ with $\Pi^i_\val v$ for $\rangeI$, there exist derivations $\itJugHandler[\Pi'_i]{\Gamma_i}{\handler h\sub{x}{\val v}}{\it F_i \Rightarrow\it G_i}$.
We can construct the following derivation $\Xi$:
\begin{equation*}
\begin{prooftree}
        \hypo{\itJugComput[\Pi'_\comput t]
                    {\Gamma; 
                        y : \it N' ;
                        r : \itSet*{
                                \it N_i \rightarrow \it G_i
                                }_{\rangeI}}
                    {\comput t\sub{x}{\val v}}
                    {\it E}
            }%
        \hypo{(\itJugHandler[\Pi'_i]{\Gamma_i}{\handler h\sub{x}{\val v}}{\it F_i \Rightarrow\it G_i})_{\rangeI}}%
    \typeRuleHandler{%
        \itJugHandler%
            {\Gamma_\comput t
            +_{\rangeI} \Gamma_i}%
            {\{\effectClause{y}{r}{\comput t\sub{x}{\val v}}\} \cup \handler h\sub{x}{\val v}}%
            {
                \itEffect{\sigma}{\it N'}
                    {\it N_i\rightarrow \it F_i
                        }_{\rangeI}
                    \Rightarrow \it E
            }%
    }%
\end{prooftree}
\end{equation*}
With $\{\effectClause{y}{r}{\comput t\sub{x}{\val v}}\} 
                    \handlerCompat \handler h\sub{x}{\val v}$ by \Cref{lem:HandlerComputIndipendence}.

%% file: Proofs/AntiSubstitutionLemma.tex
By induction on $\Pi$. By case analysis on $p$, the
form of the $\Pi$ can be deduced :

\paragraph*{Case Integers} 
Similarly for $p = \int n$, $\Pi$ has the following form :

\begin{equation*}
    \begin{prooftree}
            \typeRuleInt[0]{\itJugValue{\emptyset}{\int n}{\tau}}
    \end{prooftree}
\end{equation*}

Where $\tau \in \{\itSetEmpty, \itSet{\typeInt{n}}\}$.
Then for $\it M = \itSetEmpty$, by typability of values (\Cref{lem:TypabilityValues})
$\itJugValue[\Pi_\val v]{\emptyset}{\val v}{\itSetEmpty}$ and
$\itJugValue[\Pi_p  = \Pi]{\emptyset}{\int n}{\tau}$

\paragraph*{Case Variable}
\begin{itemize}
    \item If $p = x$, then $x\sub{x}{\val v}=\val v$ and $\Pi$ has the following form:
    \begin{equation*}
        \begin{prooftree}
            \typeRuleVar{\itJugValue{\Gamma}{\val v}{\tau}}
        \end{prooftree}
    \end{equation*}
    Since $\val v$ is closed, by \Cref{rem:ContextSimplification}, $\Gamma=\emptyset$.
    Then for $\it M = \tau$, we take $\itJugValue[\Pi_\val v = \Pi]{\emptyset}{\val v}{\it M}$.
    Moreover, we take $\itJugValue[\Pi_p]{x : \it M}{x}{\it M}$.

    \item If $p = y \neq x$, then $\Pi$ has the following form :
    \begin{equation*}
        \begin{prooftree}
            \typeRuleVar{\itJugValue{y : \it N}{y}{\it N}}
        \end{prooftree}
    \end{equation*}
    Thus, for $\it M = \itSetEmpty$, by typability of values (\Cref{lem:TypabilityValues})
    $\itJugValue[\Pi_\val v]{\emptyset}{\val v}{\itSetEmpty}$ and
    $\itJugValue[\Pi_p  = \Pi]{y : \it N}{y}{\it N}$

\end{itemize}

\paragraph*{Case Abstraction} 
For $p = \abs{y}{\comput t}$, $\Pi$ has the following form :

\begin{equation*}
    \begin{prooftree}
        \hypo{\left(
            \itJugComput[\Pi_i]{\Gamma_i; y : \it N_i}{\comput t\sub{x}{\val v}}{\it E_i}
        \right)_i}
        \typeRuleAbs{\itJugValue{+_{\rangeI} \Gamma_i}{\abs{y}{\comput t\sub{x}{\val v}}}{\itSet{\it N_i \rightarrow \it E_i}_{\rangeI}}}
    \end{prooftree}
\end{equation*}

By induction hypothesis on each $\Pi_i$, there exist a type $\it M_i$ and derivations
$\itJugValue[\Pi_i^{\val v}]{\emptyset}{\val v}{\it M_i}$
and $\itJugComput[\Pi_i^{\comput t}]{\Gamma_i^{\comput t}; x : \it M_i}{\comput t}{\it E_i}$,
where $\Gamma_i^{\comput t} = \Gamma_i; y : \it N_i$ for every $i$.
Thus, one can derive the following $\Pi_{p}$.
\begin{equation*}
    \begin{prooftree}
        \hypo{\left(
            \itJugComput[\Pi_i^{\comput t}]
                {\Gamma_i; x : \it M_i; y : \it N_i}{\comput t}{\it E_i}
        \right)_{\rangeI}}
        \typeRuleAbs{\itJugValue
            {+_{\rangeI} \Gamma_i; x : \bigcup\limits_{\rangeI}\it M_i}
            {\abs{y}{\comput t}}
            {\itSet{\it N_i \rightarrow \it E_i}_{\rangeI}}}
    \end{prooftree}
\end{equation*}
Being $\val v$ closed, by \Cref{lem:SumClosedValues} $\it M=\bigcup\limits_{\rangeI}\it M_i$ exists.
Moreover, by intersection lemma (\Cref{lem:IntersectionLemma}) on the $\Pi_i^{\val v}$, one can derive
$\itJugValue[\Pi_\val v]{\emptyset}
    {\val v}{\bigcup\limits_{\rangeI} \it M_i}$.

\paragraph*{Case Fixpoint} 
For $p = \fix{y}{\val w}$, $\Pi$ has two possible forms :
\begin{itemize}
\item[\bltI] Base case :
    \begin{equation*}
    \begin{prooftree}
        \hypo{
            \itJugValue[\Pi']
                {\Gamma; y : \itSetEmpty}
                {\val w\sub{x}{\val v}}
                {\itSet{\it N_i \rightarrow \it E_i}_{\rangeI}}}
        \typeRuleFixBase[1]{
            \itJugValue
                {\Gamma}
                {\fix{y}{\val w\sub{x}{\val v}}}
                {\itSet{\it N_i \rightarrow \it E_i}_{\rangeI}}}
    \end{prooftree}
    \end{equation*}

    By induction hypothesis on $\Pi'$, there exist a type $\it M$
    and derivations $\itJugValue[\Pi'_{\val v}]{\emptyset}{\val v}{\it M}$
    and $\itJugValue[\Pi'_{\val w}]{\Gamma; x : \it M}{\val w}
        {\itSet{\it N_i \rightarrow \it E_i}_{\rangeI}}$.
    One can thus choose $\itJugValue[\Pi_{\val v} = \Pi'_{\val v}]
        {\emptyset}{\val v}{\it M}$. 
    Moreover, the following derivation $\Pi_{p}$ holds.

    \begin{equation*}
    \begin{prooftree}
            \hypo{\itJugValue[\Pi'_{\val w}]
                {\Gamma; x : \it M; y : \itSetEmpty}
                {\val w}{\itSet{\it N_i \rightarrow \it E_i}_{\rangeI}}}
        \typeRuleFixBase[1]{\itJugValue{\Gamma; x : \it M}
            {\fix{y}{\val w}}{\itSet{\it N_i \rightarrow \it E_i}_{\rangeI}}}
    \end{prooftree}
    \end{equation*}

    Therefore, this conculdes the current case.

\item[\bltI] Recursive case :
    \begin{equation*}
    \begin{prooftree}
        \hypo{
            \itJugValue[\Pi^1]
                {\Gamma^1; y : \it N'}
                {\val w\sub{x}{\val v}}
                {\itSet{\it N_i \rightarrow \it E_i}_{\rangeI}}}
        \hypo{
            \itJugValue[\Pi^2]
                {\Gamma^2}
                {\fix{y}{\val w\sub{x}{\val v}}}
                {\it N'}}
        \typeRuleFixRec{
            \itJugValue
                {\Gamma^1 + \Gamma^2}
                {\fix{y}{\val w\sub{x}{\val v}}}
                {\itSet{\it N_i \rightarrow \it E_i}_{\rangeI}}}
    \end{prooftree}
    \end{equation*}

    By induction hypothesis on $\Pi^1$, there exist a type $\it M_1$
    and derivations $\itJugValue[\Pi^1_{\val v}]{\emptyset}{\val v}{\it M_1}$
    and $\itJugValue[\Pi^1_{\val w}]{\Gamma^1; y : \it N'; x : \it M_1}{\val w}
        {\itSet{\it N_i \rightarrow \it E_i}_{\rangeI}}$.
    Similarly, by induction hypothesis on $\Pi^2$, there exist a type $\it M_2$
    and derivations $\itJugValue[\emptyset]{\Gamma^2_{\val v}}{\val v}{\it M_2}$
    and $\itJugValue[\Pi^2_{\fix{y}{\val w}}]
        {\Gamma^2; x : \it M_2}
        {\fix{y}{\val w}}{\it N'}$.

    Being $\val v$ closed, by \Cref{lem:SumClosedValues} $\it M=\it M_1\cup\it M_2$ exists.
    By intersection lemma (\Cref{lem:IntersectionLemma}) on $\Pi^1_{\val v}$ and $\Pi^2_{\val v}$,
    one can deduce $\itJugValue[\Pi_{\val v}]
        {\emptyset}
        {\val v}{\it M_1 \cup \it M_2}$.
    Moreover, the following derivation $\Pi_{p}$ holds.

    \begin{equation*}\hspace{-1cm}
    \begin{prooftree}
            \hypo{\itJugValue[\Pi^1_{\val w}]
                {\Gamma^1; x : \it M_1; y : \it N'}
                {\val w}{\itSet{\it N_i \rightarrow \it E_i}_{\rangeI}}}
            \hypo{\itJugValue[\Pi^2_{\fix{y}{\val w}}]
                {\Gamma^2; x : \it M_2}
                {\fix{y}{\val w}}{\it N'}}
        \typeRuleFixRec{\itJugValue
            {\Gamma^1 + \Gamma^2; x : \it M_1 \cup \it M_2}
            {\fix{y}{\val w}}{\itSet{\it N_i \rightarrow \it E_i}_{\rangeI}}}
    \end{prooftree}
    \end{equation*}

\end{itemize}

\paragraph*{Case Let-Binding}
For $p = \letin{y}{\comput t}{\comput u}$, $\Pi$ has the following form :
\begin{equation*}
\begin{prooftree}
        \hypo{\itJugComput[\Pi_{\comput t}]
            {\Delta}{\comput t\sub{x}{\val v}}{\it E}}
        \hypo{\varphi}
            \ellipsis{}{
        \hypo{\it E
                \replaceLeaf{\it N_i \rightarrow \it G_i}[\rangeI]
            \it F}}
        \hypo{
            (\itJugComput[\Pi_i]
                {\Gamma_i; y : \it N_i}
                {\comput u\sub{x}{\val v}}
                {\it G_i}
            )_{\rangeI}}
    \typeRuleLetin{
        \itJugComput
            {\Delta +_{\rangeI} \Gamma_i}
            {\letin{y}{\comput t\sub{x}{\val v}}{\comput u\sub{x}{\val v}}}
            {\it F}
    }
\end{prooftree}
\end{equation*}
By induction hypothesis on $\Pi_{\comput t}$, there exist types $\it M_{\comput t}$ and derivations
$\itJugValue[\Pi_{\comput t}^{\val v}]{\emptyset}{\val v}{\it M_{\comput t}}$
and $\itJugComput[\Pi_{\comput t}']{\Delta; x : \it M_{\comput t}}{\comput t}{\it E}$.
Similarly, by induction hypothesis on each $\Pi_i$, there exist types
$\it M_i$ and derivations
$\itJugValue[\Pi_i^{\val v}]{\emptyset}{\val v}{\it M_i}$
and $\itJugComput[\Pi_i^{\comput u}]
    {\Gamma_i;y:\it N_i; x : \it M_i}
    {\comput u}{\it G_i}$.
Thus, one can derive the following $\Pi_{p}$.
\begin{equation*}
\begin{prooftree}
        \hypo{\itJugComput[\Pi_{\comput t}']
            {\Delta; x : \it M_{\comput t}}{\comput t}{\it E}}
        \hypo{\varphi}
            \ellipsis{}{
        \hypo{\it E
                \replaceLeaf{\it N_i \rightarrow \it G_i}[\rangeI]
            \it F}}
        \hypo{
            (\itJugComput[\Pi_i^{\comput u}]
                {\Gamma_i; x : \it M_i; y : \it N_i}
                {\comput u}
                {\it G_i}
            )_{\rangeI}}
    \typeRuleLetin{
        \itJugComput
            {\Delta_{\comput t} +_{\rangeI} \Gamma_i;
            x : \it M_{\comput t} \bigcup\limits_{\rangeI} \it M_i}
            {\letin{y}{\comput t}{\comput u}}
            {\it F}
    }%
\end{prooftree}
\end{equation*}
Being $\val v$ closed, by \Cref{lem:SumClosedValues} $\it M=\it M_{\comput t}\bigcup\limits_{\rangeI}\it M_i$ exists.
Moreover, by intersection lemma (\Cref{lem:IntersectionLemma}) on the $\Pi_i^{\val v}$ and 
$\Pi_{\comput t}^{\val v}$, one can derive
$\itJugValue[\Pi_\val v]{\emptyset}
    {\val v}{\it M_{\comput t} \bigcup_{\rangeI} \it M_i}$.

\paragraph*{Case Application}
For $p = \app{\val w_1}{\val w_2}$, $\Pi$ has the following form :
\begin{equation*}
\begin{prooftree}
        \hypo{\itJugValue[\Pi_1]{\Gamma_1}{\val w_1\sub{x}{\val v}}{\itSet{\it N \rightarrow \it E}}}
        \hypo{\itJugValue[\Pi_2]{\Gamma_2}{\val w_2\sub{x}{\val v}}{\it N}}
    \typeRuleApp{\itJugComput{\Gamma_1 + \Gamma_2}{\app{\val w_1\sub{x}{\val v}}{\val w_2\sub{x}{\val v}}}{\it E}}
\end{prooftree}
\end{equation*}
By induction hypothesis on $\Pi_1$, there exist types $\it M_1$ and derivations
$\itJugValue[\Pi_1^{\val v}]{\emptyset}{\val v}{\it M_1}$
and $\itJugValue[\Pi_1^{\val w_1}]{\Gamma_1; x : \it M_1}{\val w_1}{\itSet{\it N \rightarrow \it E}}$.
Similarly, by induction hypothesis on $\Pi_2$, there exist types $\it M_2$ and derivations
$\itJugValue[\Pi_2^{\val v}]{\emptyset}{\val v}{\it M_2}$
and $\itJugValue[\Pi_2^{\val w_2}]{\Gamma_2; x : \it M_2}{\val w_2}{\it N}$.
Therfore, one can derive the following $\Pi_{p}$.
\begin{equation*}
\begin{prooftree}
        \hypo{\itJugValue[\Pi_1^{\val w_1}]{\Gamma_1; x : \it M_1}{\val w_1}{\itSet{\it N \rightarrow \it E}}}
        \hypo{\itJugValue[\Pi_2^{\val w_2}]{\Gamma_2; x : \it M_2}{\val w_2}{\it N}}
    \typeRuleApp{\itJugComput{\Gamma_1 + \Gamma_2; x : \it M_1 \cup \it M_2}{\app{\val w_1}{\val w_2}}{\it E}}
\end{prooftree}
\end{equation*}
Being $\val v$ closed, by \Cref{lem:SumClosedValues} $\it M=\it M_1\cup\it M_2$ exists.
Moreover, by intersection lemma (\Cref{lem:IntersectionLemma}) on $\Pi_1^{\val v}$ and 
$\Pi_2^{\val v}$, one can derive
$\itJugValue[\Pi_\val v]{\emptyset}
    {\val v}{\it M_1 \cup \it M_2}$.

\paragraph*{Case Case}
For $p = \case{\val w}{\comput t_1, \ldots, \comput t_n}$, $\Pi$ has the following form :
\begin{equation*}
\begin{prooftree}
        \hypo{\itJugValue[\Pi_1]{\Gamma_1}{\val w\sub{x}{\val v}}{\itSet{\typeInt{m}}}}
        \hypo{\itJugComput[\Pi_2]{\Gamma_2}{\comput t_m\sub{x}{\val v}}{\it E}}
    \typeRuleCase{\itJugComput{\Gamma_1 + \Gamma_2}
        {\case{\val w\sub{x}{\val v}}{\comput t_1\sub{x}{\val v}, \ldots, \comput t_n\sub{x}{\val v}}}{\it E}}
\end{prooftree}
\end{equation*}
By induction hypothesis on $\Pi_1$, there exist types $\it M_1$ and derivations
$\itJugValue[\Pi_1^{\val v}]{\emptyset}{\val v}{\it M_1}$
and $\itJugValue[\Pi_1^{\val w}]{\Gamma_1; x : \it M_1}{\val w}{\itSet{\typeInt{m}}}$.
Similarly, by induction hypothesis on $\Pi_2$, there exist types $\it M_2$ and derivations
$\itJugValue[\Pi_2^{\val v}]{\emptyset}{\val v}{\it M_2}$
and $\itJugComput[\Pi_2^{\comput t_m}]{\Gamma_2; x : \it M_2}{\comput t_m}{\it E}$.
Therfore, one can derive the following $\Pi_{p}$.

\begin{equation*}
\begin{prooftree}
        \hypo{\itJugValue[\Pi_1^{\val w}]{\Gamma_1; x : \it M_1}{\val w}{\itSet{\typeInt{m}}}}
        \hypo{\itJugComput[\Pi_2^{\comput t_m}]{\Gamma_2; x : \it M_2}{\comput t_m}{\it E}}
    \typeRuleCase{\itJugComput{\Gamma_1 + \Gamma_2; x : \it M_1 \cup \it M_2}{\case{\val v}{\comput t_1, \ldots, \comput t_n}}{\it E}}
\end{prooftree}
\end{equation*}
Being $\val v$ closed, by \Cref{lem:SumClosedValues} $\it M=\it M_1\cup\it M_2$ exists.
Moreover, by intersection lemma (\Cref{lem:IntersectionLemma}) on $\Pi_1^{\val v}$ and 
$\Pi_2^{\val v}$, one can derive
$\itJugValue[\Pi_\val v]{\emptyset}
    {\val v}{\it M_1 \cup \it M_2}$.

\paragraph*{Case Return}
For $p = \ret{\val w}$, $\Pi$ has the following form :
\begin{equation*}
\begin{prooftree}
        \hypo{\itJugValue[\Pi']{\Gamma}{\val w\sub{x}{\val v}}{\it N}}
    \typeRuleRet{\itJugComput{\Gamma}
        {\ret{\val w\sub{x}{\val v}}}{\itEffectReturn{\it N}}}
\end{prooftree}
\end{equation*}
By induction hypothesis on $\Pi'$, there exist types $\it M$ and derivations
$\itJugValue[\Pi'_{\val v}]{\emptyset}{\val v}{\it M}$
and $\itJugValue[\Pi'_{\val w}]{\Gamma; x : \it M}{\val w}{\it N}$.
Therfore, one can derive the following $\Pi_{p}$.
\begin{equation*}
\begin{prooftree}
        \hypo{\itJugValue[\Pi'_{\val w}]{\Gamma; x : \it M}{\val w}{\it N}}
    \typeRuleRet{\itJugComput{\Gamma; x : \it M}
        {\ret{\val w}}{\itEffectReturn{\it N}}}
\end{prooftree}
\end{equation*}
Thus, one can conclude this case with
$\Pi_\val v = \Pi'_{\val v}$.

\paragraph*{Case Effect}
For $p = \effect{\val w}{y}{\comput t}$, $\Pi$ has the following form :
\begin{equation*}
\begin{prooftree}
        \hypo{\itJugValue[\Pi_{\val w}]
            {\Delta}
            {\val w\sub{x}{\val v}}
            {\it N}
        }
        \hypo{(\itJugComput[\Pi_i]
            {\Gamma_i; y : \it N_i}
            {\comput t\sub{x}{\val v}}
            {\it E_i}
        )_{\rangeI}}
    \typeRuleEff{
        \itJugComput
            {\Delta +_{\rangeI} \Gamma_i}
            {\effect{\val w\sub{x}{\val v}}{y}{\comput t\sub{x}{\val v}}}
            {\itEffect{\sigma}{\it N}{\it N_i \rightarrow \it E_i}_{\rangeI}}
        }
\end{prooftree}
\end{equation*}
By induction hypothesis on $\Pi_{\val w}$, there exist types $\it M_{\val w}$ and derivations
$\itJugValue[\Pi_{\val w}^{\val v}]{\emptyset}{\val v}{\it M_{\val w}}$
and $\itJugValue[\Pi_{\val w}']{\Delta; x : \it M_{\val w}}{\val w}{\it N}$.
Similarly, by induction hypothesis on each $\Pi_i$, there exist types
$\it M_i$ and derivations
$\itJugValue[\Pi_i^{\val v}]{\emptyset}{\val v}{\it M_i}$
and $\itJugComput[\Pi_i^{\comput t}]
    {\Gamma_i; x : \it M_i}
    {\comput t}{\it E_i}$.
Thus, one can derive the following $\Pi_{p}$.
\begin{equation*}
\begin{prooftree}
        \hypo{\itJugValue[\Pi_{\val w}']
            {\Delta; x : \it M_{\val w}}
            {\val w}
            {\it N}
        }
        \hypo{(\itJugComput[\Pi_i^{\comput t}]
            {\Gamma_i; x : \it M_i; y : \it N_i}
            {\comput t}
            {\it E_i}
        )_{\rangeI}}
    \typeRuleEff{
        \itJugComput
            {\Delta +_{\rangeI} \Gamma_i; 
                x : \it M_{\val w} \bigcup\limits_{\rangeI}\it M_i}
            {\effect{\val w}{y}{\comput t}}
            {\itEffect{\sigma}{\it N}{\it N_i \rightarrow \it E_i}_{\rangeI}}
        }
\end{prooftree}
\end{equation*}
Being $\val v$ closed, by \Cref{lem:SumClosedValues} $\it M=\it M_\val w\bigcup\limits_{\rangeI}\it M_i$ exists.
Moreover, by intersection lemma (\Cref{lem:IntersectionLemma}) on $\Pi_{\val w}^{\val v}$ and 
the $\Pi_i^{\val v}$, one can derive
$\itJugValue[\Pi_\val v]{\emptyset}
    {\val v}{\it M_{\val w} \bigcup_{\rangeI}\it M_i}$.

\paragraph*{Case Handler}
For $p = \handle{\handler h}{\comput t}$, $\Pi$ has the following form :
\begin{equation*}
\begin{prooftree}
        \hypo{\itJugHandler[\Pi_1]{\Gamma_1}{\handler h\sub{x}{\val v}}{\it E \Rightarrow \it F}}
        \hypo{\itJugValue[\Pi_2]{\Gamma_2}{\comput t\sub{x}{\val v}}{\it E}}
    \typeRuleHandle{\itJugHandler{\Gamma_1 + \Gamma_2}{\handle{\handler h}{\comput t}\sub{x}{\val v}}{\it E}}
\end{prooftree}
\end{equation*}
By induction hypothesis on $\Pi_1$, there exist types $\it M_1$ and derivations
$\itJugValue[\Pi_1^{\val v}]{\emptyset}{\val v}{\it M_1}$
and $\itJugHandler[\Pi_1^{\handler h}]{\Gamma_1; x : \it M_1}{\handler h}{\it E \Rightarrow \it F}$.
Similarly, by induction hypothesis on $\Pi_2$, there exist types $\it M_2$ and derivations
$\itJugValue[\Pi_2^{\val v}]{\emptyset}{\val v}{\it M_2}$
and $\itJugComput[\Pi_2^{\comput t}]{\Gamma_2; x : \it M_2}{\comput t}{\it E}$.
Therfore, one can derive the following $\Pi_{p}$.
\begin{equation*}
\begin{prooftree}
        \hypo{\itJugHandler[\Pi_1^{\handler h}]{\Gamma_1; x : \it M_1}{\handler h}{\it E \Rightarrow \it F}}
        \hypo{\itJugValue[\Pi_2^{\comput t}]{\Gamma_2; x : \it M_2}{\comput t}{\it E}}
    \typeRuleHandle{\itJugHandler{\Gamma_1 + \Gamma_2}{\handle{\handler h}{\comput t}\sub{x}{\val v}}{\it E}}
\end{prooftree}
\end{equation*}
Being $\val v$ closed, by \Cref{lem:SumClosedValues} $\it M=\it M_1\cup\it M_2$ exists.
Moreover, by intersection lemma (\Cref{lem:IntersectionLemma}) on $\Pi_1^{\val v}$ and 
$\Pi_2^{\val v}$, one can derive
$\itJugValue[\Pi_\val v]{\emptyset}
    {\val v}{\it M_1 \cup \it M_2}$.
 
\paragraph{Case Handler Branching Return}
For $p = \{\retClause{y}{\comput t}\}\cup \handler h$, $\Pi$ has the following form :
\begin{equation*}
    \begin{prooftree}
                \hypo{\itJugComput[\Pi_\comput t]
                    {\Gamma;
                        y : \it N}
                    {\comput t}
                    {\it E}}
               \hypo{\{\retClause{y}{\comput t\sub{x}{\val v}}\} 
                    \handlerCompat \handler h\sub{x}{\val v}} 
            \typeRuleHandlerRet[2]{\itJugHandler
                {\Gamma}
                {\{\retClause{y}{\comput t\sub{x}{\val v}}\} \cup \handler h\sub{x}{\val v}}%
                {\itEffectReturn{\it N} \Rightarrow \it E}
            }%
        \end{prooftree}
\end{equation*}
By induction hypothesis on $\Pi_\comput t$, there exist type $\it M$ and derivations $\itJugValue[\Pi_\val v]{\emptyset}{\val v}{\it M}$ and $\itJugComput[\Xi'_\comput t]{\Gamma;y:\it N;x:\it M}{\comput t}{\it E}$.
Therefore, one can derive the following $\Pi_{p}$
\begin{equation*}
    \begin{prooftree}
                \hypo{\itJugComput[\Xi'_\comput t]
                    {\Gamma;x:\it M;
                        y : \it N}
                    {\comput t}
                    {\it E}}
                
            \typeRuleHandlerRet[1]{\itJugHandler
                {\Gamma;x:\it M}
                {\{\retClause{y}{\comput t}\} \cup \handler h}%
                {\itEffectReturn{\it N} \Rightarrow \it E}
            }%
        \end{prooftree}
\end{equation*}
With $\{\retClause{y}{\comput t}\} 
                    \handlerCompat \handler h$ by \Cref{lem:HandlerComputIndipendence}.

\paragraph*{Case Handler Branching}
For $p = \{\effectClause[\sigma]{y}{r}{\comput s}\} \cup \handler h$, $\Pi$ has the following form :
\begin{equation*}
\begin{prooftree}
        \hypo{%
            \itJugComput[\Pi_{\sigma}]
                {\Delta;
                    y : \it N;
                    r : \itSet*{\it N_i \rightarrow \it G_i}_{\rangeI}}%
                {\comput t\sub{x}{\val v}}%
                {\it E}%
            }%
        \hypo{(%
            \itJugHandler[\Pi_i]
                {\Gamma_i}%
                {\handler h\sub{x}{\val v}}%
                {\it F_i \Rightarrow \it G_i}%
        )_{\rangeI}}%
    \typeRuleHandler{%
        \itJugHandler%
            {\Delta +_{\rangeI} \Gamma_i}%
            {\{\effectClause{y}{r}{\comput t\sub{x}{\val v}}\} \cup \handler h\sub{x}{\val v}}%
            {
                \itEffect
                {\sigma}{\it N}
                    {\it N_i
                        \rightarrow 
                    \it F_i
                }_{\rangeI}
                    \Rightarrow 
                \it E
            }%
    }%
\end{prooftree}
\end{equation*}
With $\{\effectClause{y}{r}{\comput t\sub{x}{\val v}}\} 
                    \handlerCompat \handler h\sub{x}{\val v}$
By induction hypothesis on each $\Pi_i$, there exist types $\it M_i$ and derivations
$\itJugValue[\Pi_i^{\val v}]{\emptyset}{\val v}{\it M_i}$
and $\itJugHandler[\Pi_i^{\handler h}]{\Gamma_i, x : \it M_i}
    {\handler h}{\it F_i \Rightarrow \it G_i}$.
Similarly, by induction hypothesis on $\Pi_{\sigma}$, there exists a type 
$\it M_{\sigma}$ and derivations
$\itJugValue[\Pi_{\sigma}^{\val v}]{\emptyset}{\val v}{\it M_\sigma}$
and $\itJugComput[\Pi_{\sigma}^{\comput t}]
    {\Delta; x : \it M_{\sigma}}
    {\comput t}{\it E}$.
Thus, one can derive the following $\Pi_{p}$.
\begin{equation*}
\begin{prooftree}
        \hypo{%
            \itJugComput[\Pi_{\sigma}^{\comput t}]
                {\Delta;
                    x : \it M_{\sigma};
                    y : \it N;
                    r : \itSet*{\it N_i \rightarrow \it G_i}_{\rangeI}}%
                {\comput t}%
                {\it E}%
            }%
        \hypo{(%
            \itJugHandler[\Pi_i^{\handler h}]
                {\Gamma_i, x : \it M_i}%
                {\handler h}%
                {\it F_i \Rightarrow \it G_i}%
        )_{\rangeI}}%
    \typeRuleHandler{%
        \itJugHandler%
            {\Delta
            +_{\rangeI} \Gamma_i;
            x : \it M_{\sigma} \bigcup\limits_{\rangeI} \it M_i}%
            {\{\effectClause{y}{r}{\comput t}\} \cup \handler h}%
            {
                \itEffect
                {\sigma}{\it N}
                    {\it N_i
                        \rightarrow 
                    \it F_i
                }_{\rangeI}
                    \Rightarrow 
                \it E
            }%
    }%
\end{prooftree}
\end{equation*}
With $\{\effectClause{y}{r}{\comput t}\} 
                    \handlerCompat \handler h$ by \Cref{lem:HandlerComputIndipendence}. 
Being $\val v$ closed, by \Cref{lem:SumClosedValues} $\it M=\it M_\sigma \bigcup\limits_{\rangeI}\it M_i$ exists.
Moreover, by intersection lemma (\Cref{lem:IntersectionLemma}) on the $\Pi_i^{\val v}$ and 
$\Pi_{\sigma}^{\val v}$, one can derive
$\itJugValue[\Pi_\val v]{\emptyset}
    {\val v}{\it M_{\sigma} \bigcup_{\rangeI} \it M_i}$.

%% file: Proofs/ContextTyping.tex
By induction on the evaluation context.

\begin{itemize}
    \item If the context is of the form $\Hole$, then $\ctxtEval<\comput t> = \comput t$ and
    $\it F = \it E$.
    Thus, for every $\comput u$ such that $\itJug{emptyset}{\comput u}{\it F}$,
    $\itJug{\emptyset}{\ctxtEval<\comput u>}{\it E}$.

    \item If the context is of the form $\letin{x}{\ctxtEval}{\comput s}$, 
    then the typing derivation is of the following form.
    \begin{equation*}
    \begin{prooftree}
            \hypo{\itJugComput[\Pi_{\comput t}]{\emptyset}{\ctxtEval<\comput t>}{\it E'}}
            \hypo{
                \it E'
                    \replaceLeaf{\it M_i \rightarrow \it G_i}[\rangeI]
                \it E}
            \hypo{
                (\itJugComput[\Pi_i]
                    { x : \it M_i}
                    {\comput s}
                    {\it G_i}
                )_{\rangeI}}
        \typeRuleLetin{
            \itJugComput
                {\emptyset}
                {\letin{x}{\ctxtEval<\comput t>}{\comput s}}
                {\it E}
        }
    \end{prooftree}
    \end{equation*}
    By induction hypothesis on $\ctxtEval$, there is a type $\it F$
    such that $\itJugComput{\emptyset}{\comput t}{\it F}$
    and for every $\comput u$ such that $\itJugComput{\emptyset}{\comput u}{\it F}$ there is a derivation
    $\itJugComput[\Xi_{\comput u}]{\emptyset}{\ctxtEval<\comput u>}{\it E'}$.
    We deduce the following derivation $\Xi$.
    \begin{equation*}
    \begin{prooftree}
            \hypo{\itJugComput[\Pi_{\comput u}]{\emptyset}{\ctxtEval<\comput u>}{\it E'}}
            \hypo{
                \it E'
                    \replaceLeaf{\it M_i \rightarrow \it G_i}[\rangeI]
                \it E}
            \hypo{
                (\itJugComput[\Pi_i]
                    {x : \it M_i}
                    {\comput s}
                    {\it G_i}
                )_{\rangeI}}
        \typeRuleLetin{
            \itJugComput
                {\emptyset}
                {\letin{x}{\ctxtEval<\comput u>}{\comput s}}
                {\it E}
        }
    \end{prooftree}
    \end{equation*}

    \item If the context is of the form $\handle{\handler h}{\ctxtEval}$, 
    then the typing derivation is of the following form.
    \begin{equation*}
    \begin{prooftree}
        \hypo{\itJugHandler[\Pi']{\emptyset}{\handler h}{\it E' \Rightarrow \it E}}
        \hypo{\itJugComput[\Pi_{\comput t}]{\emptyset}{\ctxtEval<\comput t>}{\it E'}}
        \typeRuleHandle{\itJugComput{\emptyset}{\handle{\handler h}{\ctxtEval<\comput t>}}{\it E}}
    \end{prooftree}
    \end{equation*}
    By induction hypothesis on $\ctxtEval$, there is a type $\it F$
    such that $\itJugComput{\emptyset}{\comput t}{\it F}$
    and for every $\comput u$ such that $\itJugComput{\emptyset}{\comput u}{\it F}$ there is a derivation
    $\itJugComput[\Xi_{\comput u}]{\emptyset}{\ctxtEval<\comput u>}{\it E'}$.
    We deduce the following derivation $\Xi$.
    \begin{equation*}
    \begin{prooftree}
        \hypo{\itJugHandler[\Pi']{\emptyset}{\handler h}{\it E' \Rightarrow \it E}}
        \hypo{\itJugComput[\Pi_{\comput u}]{\emptyset}{\ctxtEval<\comput u>}{\it E'}}
        \typeRuleHandle{\itJugComput{\emptyset}{\handle{\handler h}{\ctxtEval<\comput u>}}{\it E}}
    \end{prooftree}
    \end{equation*}
\end{itemize}

%% file: Proofs/ElementarySubjectExpansion.tex
By case analysis on the different rewrite rules:

\paragraph*{Case $\reductBeta$.}
Let $\comput t, \comput t' \in \setComput$ such that $\comput t'$ is
typed and $\comput t \reductArr_\reductBeta \comput t'$. By definition,
$\comput t \coloneqq \app{(\abs{x}{\comput u})}{\val v}$ and $\comput t'
\coloneqq \comput u\sub{x}{\val v}$. Since $\comput t'$ is typable,
there exists a derivation
$\itJugComput[\Pi]{\emptyset}{\comput u\sub{x}{\val v}}{\it E}$
for some type $\it E$. By anti-substitution lemma
(\Cref{lem:antiSubstLemma}), since the free variables of $\val v$
are not bounded in $\comput u$, there exist a type $\it M$ and type derivations
$\itJugValue[\Pi_\val v]{\emptyset}{\val v}{\it M}$ and 
$\itJugComput[\Pi_\comput u]{ x : \it M}{\comput u}{\it E}$.
Therefore, the following $\Pi'$ concludes this case :

\begin{equation*}
    \begin{prooftree}
        \hypo{\itJugComput[\Pi_{\comput u}]{ x : \it M}{\comput u}{\it E}}
        \typeRuleAbs{\itJugValue{\emptyset}{\abs{x}{\comput t}}{\itSet{\it M \rightarrow \it E}}}
        \hypo{\itJugValue[\Pi_{\val v}]{\emptyset}{\val v}{\it M}}
        \typeRuleApp{\itJugComput{\emptyset}{\app{(\abs{x}{\comput u})}{\val v}}{\it E}}
    \end{prooftree}
\end{equation*}

\paragraph*{Case $\reductFix$.}
Let $\comput t, \comput t' \in \setComput$ such that $\comput t'$ is
typed and $\comput t \reductArr_\reductFix \comput t'$. By definition,
$\comput t \coloneqq \app{(\fix{x}{\val v})}{\val w}$ and
$\comput t'\coloneqq
\app{\val v\sub{x}{\fix{x}{\val v}}}{\val w}$. Since
$\comput t'$ is typable, there exists a derivation
$\itJugComput[\Pi]{\emptyset}
    {\app{\val v\sub{x}{\fix{x}{\val v}}}{\val w}}
    {\it E}$
for some type $\it E$. By
case analysis, we deduce that $\Pi$ is necessarily  the
following form:

\begin{equation*}
\begin{prooftree}
        \hypo{
            \itJugValue[\Pi_{\val v}]
                {\emptyset}
                {\val v\sub{x}{\fix{x}{\val v}}}
                {\itSet*{\it M \rightarrow \it E}}}
        \hypo{
            \itJugValue[\Pi_{\val w}]
                {\emptyset}
                {\val w}
                {\it M}}
    \typeRuleApp{
        \itJugComput
            {\emptyset}
            {\app{\val v\sub{x}{\fix{x}{\val v}}}{\val w}}
                {\it E}}
\end{prooftree}
\end{equation*}

By the anti-substitution lemma
(\Cref{lem:antiSubstLemma}) with $\Pi_{\val v}$, one deduces that
there exist a type $\it N$ and derivations
$\itJugValue[\Pi_{\val v}']
    { x : \it N}
    {\val v}{\itSet*{\it M\rightarrow \it E}}$
and $\itJugValue[\Pi_{\fix*{}{}}]
    {\emptyset}
    {\fix{x}{\val v}}{\it N}$.
Let $\Pi'$ be the following typing derivation, if $\it N\neq \itSetEmpty$:

\begin{equation*}
\begin{prooftree}
            \hypo{
                \itJugValue[\Pi_{\val v}']
                    {x : \it N}
                    {\val v}{\itSet*{\it M\rightarrow \it E}}}

            \hypo{
                \itJugValue[\Pi_{\fix*{}{}}]
                    {\emptyset}
                    {\fix{x}{\val v}}{\it N}}
        \typeRuleFixRec{
            \itJugValue
                {\emptyset}
                {\fix{x}{\val v}}
                {\itSet{\it M \rightarrow \it E}}}
        \hypo{
            \itJugValue[\Pi_{\val w}]
                {\emptyset}
                {\val w}
                {\it M}}
    \typeRuleApp{
        \itJugComput
            {\emptyset}
            {\app{(\fix{x}{\val v})}{\val w}}
            {\it E}}
\end{prooftree}
\end{equation*}
Otherwise, if $\it N=\itSetEmpty$, then we can use the rule \ruleNameFixBase.

\paragraph*{Case $\reductCase$.}
Let $\comput t, \comput t' \in \setComput$ such that $\comput t'$ is
typed and $\comput t \reductArr_\reductCase \comput t'$. By definition,
$\comput t \coloneqq \case{\int n}{\comput t_1, \ldots,
\comput t_m}$ and $\comput t' \coloneqq \comput t_n$ with $0 < n \leq m$.
Since $\comput t'$ is typable, there exists a derivation
$\itJugComput[\Pi]{\emptyset}{\comput t_n}{\it E}$
for some type $\it E$.
Let $\Pi'$ be the following typing derivation:

\begin{equation*}
    \begin{prooftree}
        \typeRuleInt{\itJugValue{\emptyset}{\int n}{\itSet{\typeInt{n}}}}
        \hypo{\itJugComput[\Pi]{\emptyset}{\comput t_n}{\it E}}
        \typeRuleCase{\itJugComput{\emptyset}{\case{\int n}{\comput t_1, \ldots, \comput t_m}}{\it E}}
    \end{prooftree}
\end{equation*}

\paragraph*{Case $\reductLetRet$.}
Let $\comput t, \comput t' \in \setComput$ such that $\comput t'$ is
typed and $\comput t \reductArr_\reductLetRet \comput t'$. By
definition, $\comput t \coloneqq
\letin{x}{\ret{\val v}}{\comput u}$ and $\comput t'
\coloneqq \comput u\sub{x}{\val v}$. Since $\comput t'$ is typable,
there exists a derivation
$\itJugComput[\Pi]{\emptyset}{\comput u\sub{x}{\val v}}{\it E}$
for some type $\it E$.
By anti-substitution lemma (\Cref{lem:antiSubstLemma}), there exist a type $\it M$
and derivations $\itJugValue[\Pi_\val v]{\emptyset}{\val v}{\it M}$
and $\itJugComput[\Pi_{\comput u}]{\emptyset}{\comput u}{\it E}$.
Thus, let $\Pi'$ be the following derivation :
\begin{equation*}
\begin{prooftree}
        \hypo{\itJugValue[\Pi_\val v]{\emptyset}{\val v}{\it M}}
        \typeRuleRet{\itJugComput{\emptyset}{\ret{\val v}}{\itEffectReturn{\it M}}}
        \infer0{\itEffectReturn{\it M} \replaceLeaf{\it M \rightarrow \it E} \it E}
        \hypo{\itJugComput[\Pi_{\comput u}]{\emptyset; x : \it M}{\comput u}{\it E}}
    \typeRuleLetin{\itJugComput{\emptyset}{\letin{x}{\ret{\val v}}{\comput u}}{\it E}}
\end{prooftree}
\end{equation*}

\paragraph*{Case $\reductLetEff$.}
Let $\comput t, \comput t' \in \setComput$ such that $\comput t'$ is
typed and $\comput t \reductArr_\reductLetEff \comput t'$. By
definition, $\comput t \coloneqq
\letin{x}{\effect{\val v}{y}{\comput s}}{\comput u}$ and
$\comput t' \coloneqq
\effect{\val v}{y}{\letin{x}{\comput s}{\comput t}}$. Since
$\comput t'$ is typable, there exists a derivation
$\itJugComput[\Pi]{\emptyset}{\effect{\val v}{y}{\letin{x}{\comput s}{\comput t}}}{\it E}$
for some type $\it E$. By
case analysis, we deduce that $\Pi$ is necessarily of the following
form:
\begin{equation*}
\begin{prooftree}
    \hypo{\itJugValue[\Pi_{\val v}]{\emptyset}{\val v}{\it M}}
    \hypo{A_i}
    \delims{\left(}{\right)_{\rangeI}}
\typeRuleEff{
    \itJugComput
        {\emptyset}
        {\effect{\val v}{y}{\letin{x}{\comput s}{\comput u}}}
        {\itEffect
                {\sigma}{\it M}
                {\it N_i \rightarrow \it E_i}_{\rangeI}}}
\end{prooftree}
\end{equation*}
\begin{equation*}
    (A_i) \coloneqq
    \begin{prooftree}
        \hypo{\itJugComput[\Pi_{\comput s}]{ y : \it N_i}{\comput s}{\it G_i}}
        \hypo{\varphi_i}
        \ellipsis{}{
            \it G_i
            \replaceLeaf{\it M_k^i \rightarrow \it F_k^i}[\rangeK_i]
            \it E_i}
        \hypo{(\itJugComput[\Pi_k^i]{ y : \it N_i; x : \it M_k^i}{\comput u}{\it F_k^i})_{\rangeK_i}}
    \typeRuleLetin{
        \itJugComput
            { y : \it N_i}
            {\letin{x}{\comput s}{\comput u}}
            {\it E_i}}
    \end{prooftree}
\end{equation*}

Since $y$ is not a free variable of $\comput u$, by context simplification (\Cref{rem:ContextSimplification})
on $\Pi_k^i$ one can derive 
$\itJugComput[\Pi_k^i]
    { x : \it M_k^i}
    {\comput u}{\it F_k^i}$.
Let $\Pi'$ be the following typing derivation:
\begin{equation*}
    \begin{prooftree}
                \hypo{
                    \itJugValue[\Pi_{\val v}]
                        {\emptyset}
                        {\val v}
                        {\it M}
                }
                \hypo{(
                    \itJugComput[\Pi_{\comput s}]
                        { y : \it N_i}
                        {\comput s}
                        {\it G_i}
                )_{\rangeI}}
            \typeRuleEff{
                \itJugComput
                    {\emptyset}
                    {\effect{\val v}{y}{\comput s}}
                    {\itEffect{\sigma}{\it M}{\it N_i \rightarrow \it G_i}_{\rangeI}}
            }
            \hypo{(B)}
            \hypo{(\itJugComput[\Pi_k^i]
                { x : \it M_k^i}
                {\comput u}{\it F_k^i})_{\rangeK_i, \rangeI}}
        \typeRuleLetin{
            \itJugComput
                {\emptyset}
                {\letin{x}{\effect{\val v}{y}{\comput s}}{\comput u}}
                {\itEffect
                    {\sigma}{\it M}
                    {\it N_i \rightarrow \it E_i}_{\rangeI}}
        }
    \end{prooftree}
\end{equation*}
\begin{equation*}
    (B) \coloneqq
    \begin{prooftree}
            \hypo{\varphi_i}
            \ellipsis{}{
                \it G_i
                \replaceLeaf{\it N_i \rightarrow \it F_k^i}[\rangeK_i]
                \it E_i}
            \delims{\left(}{\right)_{\rangeI}}
        \infer1{
            \itEffect
                    {\sigma}{\it M}
                    {\it N_i \rightarrow \it G_i}_{\rangeI}
        \replaceLeaf{
            \it M_k^i \rightarrow 
            \it F_k^i}[\rangeK_i, \rangeI]
            \itEffect
                {\sigma}{\it M}
                {\it N_i \rightarrow \it E_i}_{\rangeI}}
    \end{prooftree}
\end{equation*}
One concludes this case by observing that $\Pi'$ is typing
$\comput t$.

\paragraph*{Case $\reductHdlRet$.}
Let $\comput t, \comput t' \in \setComput$ such that $\comput t'$ is
typed and $\comput t \reductArr_\reductHdlRet \comput t'$. By
definition, $\comput t \coloneqq
\handle{\handler h}{\ret{\val v}}$ and $\comput t'
\coloneqq \comput u\sub{x}{\val v}$ with
$\handler h = \{\retClause{x}{\comput u}\} \cup \handler h'$. Since
$\comput t'$ is typable, there exists a derivation
$\itJugComput[\Pi]{\emptyset}{\comput u\sub{x}{\val v}}{\it E}$
for some type $\it E$.
By anti-substitution lemma (\Cref{lem:antiSubstLemma}), there exist
a type $\it M$ and derivations 
$\itJugValue[\Pi_\val v]{\emptyset}{\val v}{\it M}$
and $\itJugComput[\Pi_{\comput u}]{ x : \it M}{\comput u}{\it E}$.

Since $\comput t$ is not a clash, by \Cref{lem:CompatImpliesClashFree}
$\{\retClause{x}{\comput u}\} \handlerCompat \handler h'$.
Therefore, let $\Pi'$ be the following derivation :
\begin{equation*}
    \begin{prooftree}
            \hypo{\itJugComput[\Pi_{\comput u}]{x : \it M}{\comput u}{\it E}}
            \hypo{\{\retClause{x}{\comput u}\} \handlerCompat \handler h'}
            \typeRuleHandlerRet{\itJugHandler{\emptyset}{\{\retClause{x}{\comput u}\} \cup \handler h'}{\itEffectReturn{\it M} \Rightarrow \it E}}
            \hypo{\itJugValue[\Pi_\val v]{\emptyset}{\val v}{\it M}}
            \typeRuleRet{\itJugComput{\emptyset}{\ret{\val v}}{\itEffectReturn{\it M}}}
        \typeRuleHandle{\itJugComput{\emptyset}{\handle{\{\retClause{x}{\comput u}\} \cup \handler h'}{\ret{\val v}}}{\it E}}
    \end{prooftree}
\end{equation*}

Ones can conclude this case by observing that $\Pi'$ is typing $\comput t$.

\paragraph*{Case $\reductHdlEff$.}
Let $\comput t, \comput t' \in \setComput$ such that $\comput t'$ is
typed and $\comput t \reductArr_\reductHdlEff \comput t'$. By definition,
$\comput t \coloneqq
\handle{\handler h}{\effect{\val v}{x}{\comput u}}$ and
$\comput t' \coloneqq
\comput s\sub{x}{\val v}\sub{r}{\abs{y}{\handle{\handler h}{\comput u}}}$
with $\effectClause[\sigma]{x}{r}{\comput s} \in
\handler h$. Since $\comput t'$ is typable, there exists a derivation
$\itJugComput[\Pi]{\emptyset}{\comput s\sub{x}{\val v}\sub{r}{\abs{y}{\handle{\handler h}{\comput u}}}}{\it E}$
for some type $\it E$. 

By anti-substitution lemma (\Cref{lem:antiSubstLemma}) on $\comput s\sub{x}{\val v}$ and 
$\abs{y}{\handle{\handler h}{\comput u}}$, there exist a type $\it L$ and derivations 
$\itJugValue[\Pi_r]{\emptyset}{\abs{y}{\handle{\handler h}{\comput u}}}{\it L}$ and 
$\itJugComput[\Pi_{\comput s}]{\Gamma; r : \it L}{\comput s\sub{x}{\val v}}{\it E}$.
By case analysis, we deduce that $\Pi_r$ is necessarily of the following
form:

\begin{equation*}
    \begin{prooftree}
                \hypo{
                    \itJugHandler[\Pi_{\handler h}^i]
                        { y : \it N_i}
                        {\handler h}
                        {\it F_i
                            \Rightarrow
                        \it G_i}}
                \hypo{
                    \itJugComput[\Pi_{\comput u}^i]
                        { y : \it N_i}
                        {\comput u}
                        {\it F_i}}
            \typeRuleHandle{
                \itJugComput
                    { y : \it N_i}
                    {\handle{\handler h}{\comput u}}
                    {\it G_i}
            }
            \delims{\left(}{\right)_{\rangeI}}
        \typeRuleAbs{
            \itJugValue
                {\emptyset}
                {\abs{y}{\handle{\handler h}{\comput u}}}
                {\itSet*{\it N_i \rightarrow \it G_i}_{\rangeI}}
        }
    \end{prooftree}
\end{equation*}

where $\it L = \itSet{\it N_i \rightarrow \it G_i}_{\rangeI}$.
By anti-substitution lemma (\Cref{lem:antiSubstLemma}) on $\comput s$,
$\val v$ and $\Pi_{\comput s}$, there exist a type $\it M$ and derivations 
$\itJugValue[\Pi_\val v]{\emptyset}
    {\val v}{\it M}$
and $\itJugComput[\Pi_{\comput s}']
    {\Gamma, r : \itSet*{\it N_i \rightarrow \it G_i}_{\rangeI}; x : \it M}
    {\comput s}{\it E}$.
Since $y$ is not a free variable of $\handler h$, by context simplification
on $\Pi_{\handler h}^i$ one can derive 
$\itJugHandler[\Pi_{\handler h}^i]
    {\emptyset}
    {\handler h}
    {\it F_i \Rightarrow \it G_i}$.

Since $\comput t$ is not a clash, by \Cref{lem:CompatImpliesClashFree}
 $\{\effectClause{x}{r}{\comput s}\} \handlerCompat \handler h$.
Let $\Pi'$ be the following derivation :

\begin{equation*}
    \begin{prooftree}
            \hypo{(B)}
                \hypo{\itJugValue[\Pi_{\val v}]
                    {\emptyset}
                    {\val v}{\it M}}
                \hypo{\left(
                    \itJugComput[\Pi_{\comput u}^i]
                        { y : \it N_i}
                        {\comput u}
                        {\it F_i}
                \right)_{\rangeI}}
            \typeRuleEff{
                \itJugComput
                    {\emptyset}
                    {\effect{\val v}{y}{\comput u}}
                    {\itEffect
                        {\sigma}{\it M}
                        {\it N_i \rightarrow \it F_i}_{\rangeI}}}
        \typeRuleHandle{
            \itJugComput
                {\emptyset}
                {\handle
                    {\{\effectClause[\sigma]{x}{r}{\comput s}\} \cup \handler h}
                    {\effect{\val v}{y}{\comput u}}}
                {\it E}}
    \end{prooftree}
\end{equation*}

\begin{equation*}
    (B) \;\coloneqq\;
    \begin{prooftree}
            \hypo{
                \itJugComput[\Pi_{\comput s}']
                    {x : \it M; 
                        r : \itSet{\it N_i \rightarrow \it G_i}_{\rangeI}}
                    {\comput s}
                    {\it E}}

            \hypo{\begin{matrix}
                \{\effectClause{x}{r}{\comput s}\} \handlerCompat \handler h\\
                (\itJugHandler[\Pi_{\handler h}^i]
                    {\emptyset}
                    {\handler h}
                    {\it F_i \Rightarrow \it G_i}
                )_{\rangeI}
            \end{matrix}}
        \typeRuleHandler[2]{
            \itJugHandler
                {\emptyset}
                {\{\effectClause{x}{r}{\comput s}\} \cup \handler h}
                {\itEffect{\sigma}{\it M}
                        {\it N_i \rightarrow \it F_i}_{\rangeI}
                \Rightarrow \it E}}
    \end{prooftree}
\end{equation*}

One can conclude this case by noticing that 
$\itJugComput[\Pi']{\Gamma}{t}{\it E}$.

%% file: Proofs/ElementarySubjectReduction.tex
By case analysis on the different rewrite rules:
\paragraph*{Case $\reductBeta$.}
Let $\comput t, \comput t' \in \setComput$ such that $\comput t$ is
typed and $\comput t \reductArr'_\reductBeta \comput t'$. By definition,
$\comput t \coloneqq \app{(\abs{x}{\comput u})}{\val v}$ and $\comput t'
\coloneqq \comput u\sub{x}{\val v}$. Since $\comput t$ is typable,
there exists a derivation
$\itJugComput[\Pi]{\emptyset}{\app{(\abs{x}{\comput u})}{\val v}}{\it E}$
for some type $\it E$. By
case analysis, we deduce that $\Pi$ is necessarily of the following
form:
\begin{equation*}
    \begin{prooftree}
        \hypo{\itJugComput[\Pi_{\comput u}]{ x : \it M}{\comput u}{\it E}}
        \typeRuleAbs{\itJugValue{\emptyset}{\abs{x}{\comput u}}{\itSet{\it M \rightarrow \it E}}}
        \hypo{\itJugValue[\Pi_{\val v}]{\emptyset}{\val v}{\it M}}
        \typeRuleApp{\itJugComput{\emptyset}{\app{(\abs{x}{\comput u})}{\val v}}{\it E}}
    \end{prooftree}
\end{equation*}
By substitution lemma (\Cref{lem:SubstitutionLemma})
with $\Pi_{\comput u}$ and $\Pi_{\val v}$, there exists
$\itJugComput[\Pi']{\emptyset}{\comput u\sub{x}{\val v}}{\it E}$,
therefore concluding this case.

\paragraph*{Case $\reductFix$.}
Let $\comput t, \comput t' \in \setComput$ such that $\comput t$ is
typed and $\comput t \reductArr'_\reductFix \comput t'$. By definition,
$\comput t \coloneqq \app{(\fix{x}{\val v})}{\val w}$ and
$\comput t'
\coloneqq
\app{\val v\sub{x}{\fix{x}{\val v}}}{\val w}$. Since
$\comput t$ is typable, there exists a derivation
$\itJugComput[\Pi]{\emptyset}{\app{(\fix{x}{\val v})}{\val w}}{\it E}$
for some type $\it E$. By
case analysis, we deduce that $\Pi$ is necessarily in one of the two
following forms:
\begin{itemize}
\item[\bltI] Base case:
    \begin{equation*}
        \hspace{-1cm}
        \begin{prooftree}
                    \hypo{
                        \itJugValue[\Pi_{\val v}]
                            { x : \itSetEmpty}
                            {\val v}
                            {\itSet*{\it M \rightarrow \it E}}}
                \typeRuleFixBase[1]{
                    \itJugValue
                        {\emptyset}
                        {\fix{x}{\val v}}
                        {\itSet{\it M \rightarrow \it E}}}
                \hypo{
                    \itJugValue[\Pi_{\val w}]
                        {\emptyset}
                        {\val w}
                        {\it M}}
            \typeRuleApp{
                \itJugComput
                    {\emptyset}
                    {\app{(\fix{x}{\val v})}{\val w}}
                    {\it E}}
        \end{prooftree}
    \end{equation*}
    By typability of values (\Cref{lem:TypabilityValues}), one
    has that
    $\itJugValue[\Pi_{\itSetEmpty}]{\emptyset}{\fix{x}{\val v}}{\itSetEmpty}$,
    hence by the substitution lemma
    (\Cref{lem:SubstitutionLemma}) with $\Pi_{\val v}$
    and $\Pi_{\itSetEmpty}$, one deduces that there exists
    $\itJugValue[\Pi'_{\val v}]{\emptyset}{\val v\sub{x}{\fix{x}{\val v}}}{\itSet*{\it M
    \rightarrow \it E}}$. Let $\Pi'$ be the
    following typing derivation:
    \begin{equation*}
        \begin{prooftree}
                \hypo{
                    \itJugValue[\Pi'_{\val v}]
                        {\emptyset}
                        {\val v\sub{x}{\fix{x}{\val v}}}
                        {\itSet*{\it M \rightarrow \it E}}}
                \hypo{
                    \itJugValue[\Pi_{\val w}]
                        {\emptyset}
                        {\val w}
                        {\it M}}
            \typeRuleApp{
                \itJugComput
                    {\emptyset}
                    {\app{\val v\sub{x}{\fix{x}{\val v}}}{\val w}}
                    {\it E}}
        \end{prooftree}
    \end{equation*}
    One concludes this case by observing that $\Pi'$ is typing
    $\comput t'$.

\item[\bltI] Recursive case:
    \begin{equation*}
        \begin{prooftree}
                    \hypo{
                        \itJugValue[\Pi_{\val v}^1]
                            {x : \it N}
                            {\val v}
                            {\itSet*{\it M \rightarrow \it E}}}
                    \hypo{
                        \itJugValue[\Pi_{\val v}^2]
                            {\emptyset}
                            {\fix{x}{\val v}}
                            {\it N}}
                \typeRuleFixRec{
                    \itJugValue
                        {\emptyset}
                        {\fix{x}{\val v}}
                        {\itSet{\it M \rightarrow \it E}}}
                \hypo{
                    \itJugValue[\Pi_{\val w}]
                        {\emptyset}
                        {\val w}
                        {\it M}}
            \typeRuleApp{
                \itJugComput
                    {\emptyset}
                    {\app{(\fix{x}{\val v})}{\val w}}
                    {\it E}}
        \end{prooftree}
    \end{equation*}
\end{itemize}
By the substitution lemma (\Cref{lem:SubstitutionLemma})
with $\Pi_{\val v}^1$ and $\Pi_{\val v}^2$, there exists
$\itJugValue[\Pi'_{\val v}]{\emptyset}{\val v\sub{x}{\fix{x}{\val v}}}{\itSet*{\it M
\rightarrow \it E}}$. Let $\Pi'$ be the following
typing derivation:
\begin{equation*}
    \begin{prooftree}
            \hypo{
                \itJugValue[\Pi'_{\val v}]
                    {\emptyset}
                    {\val v\sub{x}{\fix{x}{\val v}}}
                    {\itSet*{\it M \rightarrow \it E}}}
            \hypo{
                \itJugValue[\Pi_{\val w}]
                    {\emptyset}
                    {\val w}
                    {\it M}}
        \typeRuleApp{
            \itJugComput
                {\emptyset}
                {\app{\val v\sub{x}{\fix{x}{\val v}}}{\val w}}
                {\it E}}
    \end{prooftree}
\end{equation*}
One concludes this case by observing that $\Pi'$ is typing
$\comput t'$.

\paragraph*{Case $\reductCase$.}
Let $\comput t, \comput t' \in \setComput$ such that $\comput t$ is
typed and $\comput t \reductArr'_\reductBeta \comput t'$. By definition,
$\comput t \coloneqq \case{\int n}{\comput t_1, \ldots,
\comput t_m}$ and $\comput t' \coloneqq \comput t_n$ with $0 < n \leq m$.
Since $\comput t$ is typable, there exists a derivation
$\itJugComput[\Pi]{\emptyset}{\case{\int n}{\comput t_1,
\ldots, \comput t_m}}{\it E}$ for some type $\it E$. By case analysis, we deduce that $\Pi$ is
necessarily of the following form:
\begin{equation*}
    \begin{prooftree}
        \typeRuleInt{\itJugValue{\emptyset}{\int n}{\itSet{\typeInt{n}}}}
        \hypo{\itJugComput[\Pi']{\emptyset}{\comput t_n}{\it E}}
        \typeRuleCase{\itJugComput{\emptyset}{\case{\int n}{\comput t_1, \ldots, \comput t_m}}{\it E}}
    \end{prooftree}
\end{equation*}
One concludes this case by observing that $\Pi'$ is typing
$\comput t'$.

\paragraph*{Case $\reductLetRet$.}
Let $\comput t, \comput t' \in \setComput$ such that $\comput t$ is
typed and $\comput t \reductArr'_\reductLetRet \comput t'$. By
definition, $\comput t \coloneqq
\letin{x}{\ret{\val v}}{\comput u}$ and $\comput t'
\coloneqq \comput u\sub{x}{\val v}$. Since $\comput t$ is typable,
there exists a derivation
$\itJugComput[\Pi]{\emptyset}{\letin{x}{\ret{\val v}}{\comput u}}{\it E}$
for some type $\it E$. By
case analysis, we deduce that $\Pi$ is necessarily of the following
form:
\begin{equation*}
    \begin{prooftree}
            \hypo{\itJugValue[\Pi_\val v]{\Gamma_1}{\val v}{\it M}}
            \typeRuleRet{\itJugComput{\emptyset}{\ret{\val v}}{\itEffectReturn{\it M}}}
            \infer0{\itEffectReturn{\it M} \replaceLeaf{\it M \rightarrow \it E} \it E}

            \hypo{\itJugComput[\Pi_\comput u]{ x : \it M}{\comput u}{\it E}}
        \typeRuleLetin{\itJugComput{\emptyset}{\letin{x}{\ret{\val v}}{\comput u}}{\it E}}
    \end{prooftree}
\end{equation*}
By substitution lemma (\Cref{lem:SubstitutionLemma})
with $\Pi_{\comput u}$ and $\Pi_{\val v}$, there exists
$\itJugComput[\Pi']{\emptyset}{\comput u\sub{x}{\val v}}{\it E}$,
therefore concluding this case.

\paragraph*{Case $\reductLetEff$.}
Let $\comput t, \comput t' \in \setComput$ such that $\comput t$ is
typed and $\comput t \reductArr'_\reductLetEff \comput t'$. By
definition, $\comput t \coloneqq
\letin{x}{\effect{\val v}{y}{\comput s}}{\comput u}$ and
$\comput t' \coloneqq
\effect{\val v}{y}{\letin{x}{\comput s}{\comput t}}$. Since
$\comput t$ is typable, there exists a derivation
$\itJugComput[\Pi]{\emptyset}{\letin{x}{\effect{\val v}{y}{\comput s}}{\comput u}}{\it E}$
for some type $\it E$. By
case analysis, we deduce that $\Pi$ is necessarily of the following
form:
\begin{equation*}
    \begin{prooftree}
                \hypo{
                    \itJugValue[\Pi_{\val v}]
                        {\emptyset}
                        {\val v}
                        {\it M}
                }
                \hypo{(
                    \itJugComput[\Pi_{\comput s}]
                        {y : \it N_i}
                        {\comput s}
                        {\it G_i}
                )_{\rangeI}}
            \typeRuleEff{
                \itJugComput
                    {\emptyset}
                    {\effect{\val v}{y}{\comput s}}
                    {\itEffect{\sigma}{\it M}{\it N_i \rightarrow \it G_i}_{\rangeI}}
            }
            \hypo{(A)}

            \hypo{(
                    \itJugComput[\Pi_k^i]
                    { x : \it M_k^i}
                    {\comput u}
                    {\it F_k^i}
                )_{\rangeK_i, \rangeI}}
        \typeRuleLetin{
            \itJugComput
                {\emptyset}
                {\letin{x}{\effect{\val v}{y}{\comput s}}{\comput u}}
                {\itEffect
                    {\sigma}{\it M}
                    {\it N_i \rightarrow \it E_i}_{\rangeI}}
        }
    \end{prooftree}
\end{equation*}
\begin{equation*}
    (A) \coloneqq
    \begin{prooftree}
            \hypo{\varphi_i}
            \ellipsis{}{
                \it G_i
                \replaceLeaf{
                    \it M_k^i 
                    \rightarrow 
                    \it F_k^i
                    }[\rangeK_i]
                \it E_i}
            \delims{\left(}{\right)_{\rangeI}}
        \infer1{
            \itEffect
                {\sigma}{\it M}
                {\it N_i \rightarrow \it G_i}_{\rangeI}
            \replaceLeaf{
                    \it M_k^i 
                    \rightarrow 
                    \it F_k^i
                }[\rangeK_i, \rangeI]
            \itEffect
                {\sigma}{\it M}
                {\it N_i \rightarrow \it E_i}_{\rangeI}}
    \end{prooftree}
\end{equation*}
Let $\Pi'$ be the following typing derivation:
\begin{equation*}
    \begin{prooftree}
        \hypo{\itJugValue[\Pi_{\val v}]{\emptyset}{\val v}{\it M}}
        \hypo{B_i}
        \delims{\left(}{\right)_{\rangeI}}
    \typeRuleEff{
        \itJugComput
            {\emptyset}
            {\effect{\val v}{y}{\letin{x}{\comput s}{\comput u}}}
            {\itEffect
                    {\sigma}{\it M}
                    {\it N_i \rightarrow \it E_i}_{\rangeI}}}
    \end{prooftree}
\end{equation*}
\begin{equation*}
    (B_i) \coloneqq
    \begin{prooftree}
            \hypo{\itJugComput[\Pi_{\comput s}]{ y : \it N_i}{\comput s}{\it G_i}}
            \hypo{\varphi_i}
            \ellipsis{}{
                \it G_i
                \replaceLeaf{
                    \it M_k^i 
                    \rightarrow 
                    \it F_k^i
                    }[\rangeK_i]
                \it E_i}
            \hypo{(\itJugComput[\Pi_k^i]
                {x : \it M_k^i}
                {\comput u}
                {\it F_k^i})_{\rangeK_i}}
        \typeRuleLetin{
            \itJugComput
                { y : \it N_i}
                {\letin{x}{\comput s}{\comput u}}
                {\it E_i}}
    \end{prooftree}
\end{equation*}
One concludes this case by observing that $\Pi'$ is typing
$\comput t'$.

\paragraph*{Case $\reductHdlRet$.}
Let $\comput t, \comput t' \in \setComput$ such that $\comput t$ is
typed and $\comput t \reductArr'_\reductHdlRet \comput t'$. By
definition, $\comput t \coloneqq
\handle{\handler h}{\ret{\val v}}$ and $\comput t'
\coloneqq \comput u\sub{x}{\val v}$ with
$\retClause{\val v}{\comput u} \in \handler h$. Since
$\comput t$ is typable, there exists a derivation
$\itJugComput[\Pi]{\emptyset}{\handle{\handler h}{\ret{\val v}}}{\it E}$
for some type $\it E$. By
case analysis, we deduce that $\Pi$ is necessarily of the following
form:
\begin{equation*}
    \begin{prooftree}
            \hypo{\itJugComput[\Pi_{\comput u}]{ x : \it M}{\comput u}{\it E}}
            \hypo{\{\retClause{x}{\comput u}\} \handlerCompat \handler h'}
            \typeRuleHandlerRet{\itJugHandler{\emptyset}{\{\retClause{x}{\comput u}\} \cup \handler h'}{\itEffectReturn{\it M} \Rightarrow \it E}}
            \hypo{\itJugValue[\Pi_\val v]{\emptyset}{\val v}{\it M}}
            \typeRuleRet{\itJugComput{\Gamma_2}{\ret{\val v}}{\itEffectReturn{\it M}}}
        \typeRuleHandle{\itJugComput{\emptyset}{\handle{\{\retClause{x}{\comput u}\} \cup \handler h'}{\ret{\val v}}}{\it E}}
    \end{prooftree}
\end{equation*}
By substitution lemma (\Cref{lem:SubstitutionLemma})
with $\Pi_{\comput u}$ and $\Pi_{\val v}$, there exists
$\itJugComput[\Pi']{\emptyset}{\comput u\sub{x}{\val v}}{\it E}$,
which concluding this case since $\comput t' =
\comput u\sub{x}{\val v}$.

\paragraph*{Case $\reductHdlEff$.}
Let $\comput t, \comput t' \in \setComput$ such that $\comput t$ is
typed and $\comput t \reductArr'_\reductHdlEff \comput t'$.

By definition,
$\comput t \coloneqq
\handle{\handler h}{\effect{\val v}{x}{\comput u}}$ and
$\comput t' \coloneqq
\comput s\sub{x}{\val v}\sub{r}{\abs{x}{\handle{\handler h}{\comput u}}}$
with $\effectClause[\sigma]{x}{r}{\comput s} \in
\handler h$. Since $\comput t$ is typable, there exists a derivation
$\itJugComput[\Pi]{\emptyset}{\handle{\handler h}{\effect{\val v}{x}{\comput u}}}{\it E}$
for some type $\it E$. By
case analysis, we deduce that $\Pi$ is necessarily of the following
form:
\begin{equation*}
    \begin{prooftree}
            \hypo{(B)}
                \hypo{\itJugValue[\Pi_{\val v}]{\emptyset}{\val v}{\it M}}
                \hypo{(
                    \itJugComput[\Pi_{\comput u}^i]
                        { x : \it N_i}
                        {\comput u}
                        {\it F_i}
                )_{\rangeI}}
            \typeRuleEff{
                \itJugComput
                    {\emptyset}
                    {\effect{\val v}{x}{\comput u}}
                    {\itEffect
                        {\sigma}{\it M}
                        {\it N_i \rightarrow \it F_i}_{\rangeI}}}
        \typeRuleHandle{
            \itJugComput
                {\emptyset}
                {\handle
                    {\{\effectClause[\sigma]{x}{r}{\comput s}\} \cup \handler h'}
                    {\effect{\val v}{x}{\comput u}}}
                {\it E}}
    \end{prooftree}
\end{equation*}

\begin{equation*}
    (B) \;\coloneqq\;
    \begin{prooftree}
            \hypo{
                \itJugComput[\Pi_{\comput s}]
                    { y : \it M; r : \itSet*{\it N_i
                        \rightarrow \it G_i}_{\rangeI}}
                    {\comput s}
                    {\it E}}
        %
            \hypo{(
                \itJugHandler[\varphi_i]
                    {\emptyset}
                    {\handler h'}
                    {\it F_i \Rightarrow \it G_i}
            )_{\rangeI}}
        %
        %
        \typeRuleHandler{
            \itJugHandler
                {\emptyset}
                {\{\effectClause{x}{r}{\comput s}\} \cup \handler h'}
                {\itEffect
                        {\sigma}{\it M}{\it N_i
                    \rightarrow \it F_i}_{\rangeI}
                    \Rightarrow
                \it E}}
    \end{prooftree}
\end{equation*}
where $\{\effectClause{y}{r}{\comput s}\} \handlerCompat
\handler h'$.
By substitution lemma (\Cref{lem:SubstitutionLemma})
with $\Pi_{\comput s}$ and $\Pi_{\val v}$, there exists
$\itJugComput[\Pi_{\comput s\sub{y}{\val v}}]{ r :
\itSet*{\it N_i
\rightarrow \it G_i}_{\rangeI}}
{\comput s\sub{y}{\val v}}{\it E}$.
Since $\{\effectClause{y}{r}{\comput s}\} \handlerCompat
\handler h'$, by handler weakening (\Cref{lem:HandlerWeakening})
for every $\varphi_i$, one has that
$\itJugHandler[\varphi'_i]{\emptyset}{\handler h}
{\it F_i \Rightarrow \it G_i}$.
Let $\Pi_{\handle*{}{}}$ be the following derivation:
\begin{equation*}
    \begin{prooftree}
                \hypo{
                    \itJugHandler[\varphi_i]
                        {\emptyset}
                        {\handler h}
                        {\it F_i
                            \Rightarrow
                        \it G_i}}
                \hypo{
                    \itJugComput[\Pi_{\comput u}^i]
                        { x : \it N_i}
                        {\comput u}
                        {\it F_i}}
            \typeRuleHandle{
                \itJugComput
                    {x : \it N_i}
                    {\handle{\handler h}{\comput u}}
                    {\it G_i}
            }
            \delims{\left(}{\right)_{\rangeI}}
        \typeRuleAbs{
            \itJugValue
                {\emptyset}
                {\abs{x}{\handle{\handler h}{\comput u}}}
                {\itSet*{\it N_i \rightarrow \it G_i}_{\rangeI}}
        }
    \end{prooftree}
\end{equation*}
We conclude using the substitution lemma
(\Cref{lem:SubstitutionLemma}) with
$\Pi_{\comput s\sub{y}{\val v}}$ and $\Pi_{\handle*{}{}}$
yielding $\itJugComput[\Pi']{\emptyset}{\comput s\sub{y}{\val v}\sub{r}{\abs{x}{\handle{\handler h}{\comput u}}}}{\it E}$.

%% file: Proofs/HEBIProgress.tex
By induction on $\Pi$. By case analysis on $\comput t$, the
form of the $\Pi$ can be deduced :

\paragraph*{Case Return}
In this case, necessarily $\comput t=\ret{\val v}$ and $\Pi$ has the following form.
\begin{equation*}
\begin{prooftree}
        \hypo{\itJugValue{\emptyset}{\val v}{\it N}}
    \typeRuleRet{\itJugComput{\emptyset}
        {\ret{\val v}}{\itEffectReturn{\it N}}}
\end{prooftree}
\end{equation*}

\paragraph*{Case Effect}
In this case, necessarily $\comput t=\effect{\val v}{y}{\comput u}$ and $\Pi$ has the following form.
\begin{equation*}
\begin{prooftree}
        \hypo{\itJugValue
            {\emptyset}
            {\val v}
            {\it N}
        }
        \hypo{(\itJugComput
            { y : \it N_i}
            {\comput u}
            {\it E_i}
        )_{\rangeI}}
    \typeRuleEff{
        \itJugComput
            {\emptyset}
            {\effect{\val v}{y}{\comput u}}
            {\itEffect{\sigma}{\it N}{\it N_i \rightarrow \it E_i}_{\rangeI}}
        }
\end{prooftree}
\end{equation*}

\paragraph*{Case Application}
For $\comput t=\val v\val w$, $\Pi$ has the following form.
\begin{equation*}
\begin{prooftree}
        \hypo{\itJugValue[\Pi_1]{\emptyset}{\val v}{\itSet{\it N \rightarrow \it E}}}
        \hypo{\itJugValue[\Pi_2]{\emptyset}{\val w}{\it N}}
    \typeRuleApp{\itJugComput{\emptyset}{\app{\val v}{\val w}}{\it E}}
\end{prooftree}
\end{equation*}
Being $\itJugValue[\Pi_1]{\emptyset}{\val v}{\itSet{\it N\rightarrow\it E}}$, by case analysis on $\Pi_1$ it is easy to that either $\val v=\lambda x.\comput u$ or $\val v=\fix{x}{\val v_1}$. In both cases there exists a $\comput s$ such that $\comput t\reductArr \comput s$.

\paragraph*{Case Case}
For $\comput t = \case{\val v}{\comput t_1, \ldots, \comput t_n}$, $\Pi$ has the following form.
\begin{equation*}
\begin{prooftree}
        \hypo{\itJugValue[\Pi_1]{\emptyset}{\val w}{\itSet{\typeInt{m}}}}
    \hypo{0<m\leq n}
        \hypo{\itJugComput[\Pi_2]{\emptyset}{\comput t_m}{\it E}}
    \typeRuleCase[3]{\itJugComput{\emptyset}
        {\case{\val w}{\comput t_1, \ldots, \comput t_n}}{\it E}}
\end{prooftree}
\end{equation*}
By case analysis on $\Pi_1$, we have that $\val w=\int m$. Therefore it exists a $\comput s$ such that $\comput t\reductArr\comput s$.

\paragraph*{Case Let-Binding}
For $\comput t = \letin{y}{\comput u_1}{\comput u_2}$, $\Pi$ has the following form :
\begin{equation*}
    \hspace{-0.5cm}
\begin{prooftree}
        \hypo{\itJugComput[\Pi_{\comput t}]
            {\emptyset}{\comput u_1}{\it E}}
        \hypo{\varphi}
            \ellipsis{}{
        \hypo{\it E
                \replaceLeaf{\it N_i \rightarrow \it G_i}[\rangeI]
            \it F}}
        \hypo{
            (\itJugComput[\Pi_i]
                { y : \it N_i}
                {\comput u_2}
                {\it G_i}
            )_{\rangeI}}
    \typeRuleLetin{
        \itJugComput
            {\emptyset}
            {\letin{y}{\comput u_1}{\comput u_2}}
            {\it F}
    }
\end{prooftree}
\end{equation*}
By induction hypothesis on $\Pi_\comput t$, we have that either $\comput u_1\reductArr \comput s_1$ or $\comput u_1$ is a computation normal form. In the first case, we have a context $\ctxtEval=\letin{y}{\Hole}{\comput u_2}$ such that $\comput t=\ctxtEval<\comput u_1>\reductArr\ctxtEval<\comput s_1>=\comput s$. Otherwise, either $\comput u_1 = \ret{\val v}$ or $\comput u_1=\effect{\val v}{x}{\comput u}$. In both cases, it exists a $\comput s$ such that $\comput t\reductArr\comput s$.

\paragraph*{Case Handler}
For $\comput t = \handle{\handler h}{\comput t}$, $\Pi$ has the following form :
\begin{equation*}
\begin{prooftree}
        \hypo{\itJugHandler[\Pi_1]{\emptyset}{\handler h}{\it E \Rightarrow \it F}}
        \hypo{\itJugComput[\Pi_2]{\emptyset}{\comput u}{\it E}}
    \typeRuleHandle{\itJugHandler{\emptyset}{\handle{\handler h}{\comput u}}{\it F}}
\end{prooftree}
\end{equation*}
By induction hypothesis on $\Pi_2$, we have that either $\comput u\reductArr \comput s_1$ or $\comput u$ is a computation normal form. In the first case, we have a context $\ctxtEval=\handle{\handler h}{\Hole}$ such that $\comput t=\ctxtEval<\comput u>\reductArr\ctxtEval<\comput s_1>=\comput s$. Otherwise, either $\comput u = \ret{\val v}$ or $\comput u=\effect{\val v}{x}{\comput u_1}$. We focus on the second case. The derivation $\Pi$ has then the following form.

\begin{equation*}
    \begin{prooftree}
            \hypo{(B)}

            \hypo{
                \itJugComput
                    {\emptyset}
                    {\effect{\val v}{x}{\comput u_1}}
                    {\itEffect
                        {\sigma}{\it M}
                        {\it N_i \rightarrow \it F_i}_{\rangeI}}}
        \typeRuleHandle{
            \itJugComput
                {\emptyset}
                {\handle
                    {\{\effectClause[\sigma]{x}{r}{\comput t_1}\} \cup \handler h'}
                    {\effect{\val v}{x}{\comput u_1}}}
                {\it E}}
    \end{prooftree}
\end{equation*}

\begin{equation*}
    (B) \;\coloneqq\;
    \begin{prooftree}
            \hypo{
                \itJugComput[\Pi_{\comput t_1}]
                    { y : \it M; r : \itSet*{\it N_i
                        \rightarrow \it G_i}_{\rangeI}}
                    {\comput t_1}
                    {\it E}}
        %
            \hypo{(
                \itJugHandler[\varphi_i]
                    {\emptyset}
                    {\handler h'}
                    {\it F_i \Rightarrow \it G_i}
            )_{\rangeI}}
        %
        %
        \typeRuleHandler{
            \itJugHandler
                {\emptyset}
                {\{\effectClause{x}{r}{\comput t_1}\} \cup \handler h'}
                {\itEffect
                        {\sigma}{\it M}{\it N_i
                    \rightarrow \it F_i}_{\rangeI}
                    \Rightarrow
                \it E}}
    \end{prooftree}
\end{equation*}
where $\{\effectClause{y}{r}{\comput t_1}\} \handlerCompat
\handler h'$. Therefore, $\{\effectClause{x}{r}{\comput t_1}\}\in \handler h$ and for each $\{\effectClause{x}{r}{\comput t_2}\}\in \handler h$, $\comput t_1=\comput t_2$, so there is no clash. Finally, it exists $\comput s$ such that $\comput t\reductArr\comput s$. The case $\comput t=\ret{\val v}$ is similar.

%% file: Proofs/ReductionExpansionReduction.tex
Let $\comput t,\, \comput t'$ such that $\comput t \reductArr \comput t'$.
    If $\comput t'\in\redCand{\itEffect{\sigma}{\it M}
        {\it N_i \rightarrow \it E_i}}_{\rangeI}$,
    then $\comput t'\reductArr* \effect{\val v}{x}{\comput u}$
    where $\val v\in\redCand{\it M}$ and 
    $\forall\;\rangeI,\;\forall \handler h \in 
        \redCand{\it E_i \Rightarrow \it F_i},\;
    \abs{x}{\handle{\handler h}{\comput u}} \in 
        \redCand{\itSet{\it N_i \rightarrow \it F_i}}$.
    Therefore, $\comput t\reductArr* \effect{\val v}{x}{\comput u}$ 
    and $\comput t\in\redCand{\itEffect{\sigma}{\it M}
        {\it N_i \rightarrow \it E_i}}_{\rangeI}$.

    Conversly, if If $\comput t\in\redCand{\itEffect{\sigma}{\it M}
        {\it N_i \rightarrow \it E_i}}_{\rangeI}$,
    then $\comput t\reductArr* \effect{\val v}{x}{\comput u}$
    where $\val v\in\redCand{\it M}$ and 
    $\forall\;\rangeI,\;\forall \handler h \in 
        \redCand{\it E_i \Rightarrow \it F_i},\;
    \abs{x}{\handle{\handler h}{\comput u}} \in 
        \redCand{\itSet{\it N_i \rightarrow \it F_i}}$.
    By unicity of the reduction,
    $\comput t\reductArr\comput t'
        \reductArr* \effect{\val v}{x}{\comput u}$
    and $\comput t'\in\redCand{\itEffect{\sigma}{\it M}
        {\it N_i \rightarrow \it E_i}}_{\rangeI}$.
    The case of $\redCand{\itEffectReturn{\it M}}$ is treated 
    similarly.

%% file: Proofs/TypableReducibility.tex
By induction on $\Pi$. By case analysis on $p$, the
form of the $\Pi$ can be deduced :

\paragraph*{Case Integers} 
For $p = \int n$, $\Pi$ has the following form :
\begin{equation*}
    \begin{prooftree}
            \typeRuleInt[0]{\itJugValue{\emptyset}{\int n}{\tau}}
    \end{prooftree}
\end{equation*}
Where $\tau \in \{\itSetEmpty, \itSet{\typeInt{n}}\}$.
In any case by definition, $\int n \in \redCand{\tau}$.

\paragraph*{Case Variable}
If $p = x$, then $\Pi$ has the following form :
\begin{equation*}
    \begin{prooftree}
        \typeRuleVar{\itJugValue{x : \it M}{x}{\it M}}
    \end{prooftree}
\end{equation*}
For $\val v \in \redCand{x : \it M} = \redCand{\it M}$,
$x\sub{x}{\val v} \in \redCand{\it M}$, therfore 
$x\in\redCand{x : \it M \vdash \it M}$.

\paragraph*{Case Abstraction} 
For $p = \abs{x}{\comput t}$, $\Pi$ has the following form :
\begin{equation*}
    \begin{prooftree}
        \hypo{\left(
            \itJugComput[\Pi_i]{\Gamma_i; x : \it M_i}{\comput t}{\it E_i}
        \right)_i}
        \typeRuleAbs{\itJugValue{+_{\rangeI} \Gamma_i}{\abs{x}{\comput t}}{\itSet{\it M_i \rightarrow \it E_i}_{\rangeI}}}
    \end{prooftree}
\end{equation*}
Let $\vec{\val w} \in \redCand{+_{\rangeI} \Gamma_i}$, 
by context reducibility (\Cref{rem:ContextReducibility}) 
$\vec{\val w} \in \redCand{\Gamma_i}$.
By induction hypothesis on the $\Pi_i$,
$\comput t' \coloneqq \comput t \sub{+_{\rangeI} \Gamma_i}{\vec{\val w}}
    \in \bigcap_{\rangeI} \redCand{x : \it M_i \vdash \it E_i}$
Let $\val v \in \redCand{\it M_i}$, for some fixed $\rangeI$.
By definition,
$(\abs{x}{\comput t'}) \val v
    \reductArr \comput t' \sub{x}{\val v}
    \in \redCand{\it E_i}$.
Thus by reducibility expansion (\Cref{lem:ReducibilityExpansionReduction}),
$(\abs{x}{\comput t'})\val v \in \redCand{\it E_i}$.
This being true for all $\val v\in\redCand{\it N_i}$, 
$\abs{x}{\comput t'} \in\redCand
    {\itSet{\it M_i \rightarrow \it E_i}_{\rangeI}}$ 
and therefore $\abs{x}{\comput t}\in 
    \redCand{+_{\rangeI}\Gamma_i \vdash \itSet{\it M_i \rightarrow \it E_i}_{\rangeI}
        }$.

\paragraph*{Case Fixpoint}
For $p = \fix{x}{\val v}$, $\Pi$ has two possible forms :
\begin{itemize}
\item[\bltI] Base case :
    \begin{equation*}
    \begin{prooftree}
        \hypo{
            \itJugValue[\Pi_1]
                {\Gamma; x : \itSetEmpty}
                {\val v}
                {\itSet{\it M_i \rightarrow \it E_i}_{\rangeI}}}
        \typeRuleFixBase[1]{
            \itJugValue
                {\Gamma}
                {\fix{x}{\val v}}
                {\itSet{\it M_i \rightarrow \it E_i}_{\rangeI}}}
    \end{prooftree}
    \end{equation*}

\item[\bltI] Recursive case :

    \begin{equation*}
    \begin{prooftree}
        \hypo{
            \itJugValue[\Pi_1]
                {\Gamma_1; x : \it M'}
                {\val v}
                {\itSet{\it M_i \rightarrow \it E_i}_{\rangeI}}}
        \hypo{
            \itJugValue[\Pi_2]
                {\Gamma_2}
                {\fix{x}{\val v}}
                {\it M'}}
        \typeRuleFixRec{
            \itJugValue
                {\Gamma_1 + \Gamma_2}
                {\fix{x}{\val v}}
                {\itSet{\it M_i \rightarrow \it E_i}_{\rangeI}}}
    \end{prooftree}
    \end{equation*}

\end{itemize}
Let $\vec{\val w} \in \redCand{\Gamma_1 + \Gamma_2}$, 
by context reducibility (\Cref{rem:ContextReducibility}) 
$\vec{\val w} \in \redCand{\Gamma_1+\Gamma_2}\cap\redCand{\Gamma_2}$.
By induction hypothesis on $\Pi_1$, 
$\val v' \coloneqq \val v\sub{\Gamma_1}{\vec{\val w}} 
    \in \redCand{x : \it M' \vdash \itSet{\it M_i \rightarrow \it E_i}_{\rangeI}
        }$.
Since $\fix{x}{\val v'} \in \redCand{\itSetEmpty}$
and $\itJugValue{}{\fix{x}{\val v'}}{\itSetEmpty}$
(by typability of values \Cref{lem:TypabilityValues}),
we can assume to be in the recursive case with the induction hypothesis on $\Pi_2$
asserting $\fix{x}{\val v'} \in \redCand{\it M'}$
(where $\it M'$ is possibly empty).

Let $\val w' \in \redCand{\it M_i}$, for some fixed $\rangeI$.
By Definition
$(\fix{x}{\val v'}) \val w'
    \reductArr (\val v' 
        \sub{x}{(\fix{x}{\val v})})
        \val w'
    \in \redCand{\it E_i}$.
Thus by reducibility expansion (\Cref{lem:ReducibilityExpansionReduction}),
$(\fix{x}{\val v'})
    \val w' \in \redCand{\it E_i}$.
This being true for all $\val w'\in\redCand{\it M_i}$, 
$\fix{x}{\val v'} 
    \in\redCand{\itSet{\it M_i \rightarrow \it E_i}_{\rangeI}}$ 
and therefore $\fix{x}{\val v}\in 
    \redCand{\Gamma_1 +\Gamma_2 \vdash \itSet{\it M_i \rightarrow \it E_i}_{\rangeI}
        }$.

\paragraph*{Case Let-Binding}
In order to prove this case, let us prove a small lemma before.
\begin{lemma}\label{lem:AUX_TypableReducibility}
    The following proprety $(P)$ is true for all $\it E, \it F, (\it M_i, \it G_i)_{\rangeI}$
    such that $\it E\replaceLeaf{\it M_i\rightarrow\it G_i}[\rangeI]\it F$.

    \[ \forall \comput t\in\redCand{\it E}, 
        \forall \comput u\in\bigcap\limits_{\rangeI} \redCand{x : \it M_i \vdash \it G_i
            }, \letin{x}{\comput t}{\comput u}\in\redCand{\it F}.\]
\end{lemma}
\begin{proof}
    By induction on $\varphi \hepcfTr \it E\replaceLeaf{\it M_i\rightarrow\it G_i}[\rangeI]\it F$.
    \begin{itemize}
        \item \emph{Base case.} Let us assume $\varphi$ has the following form.
            \begin{equation*}
            \begin{prooftree}
                \infer0{
                    \itEffectReturn{\it M}
                        \replaceLeaf{\it M \rightarrow \it F}
                    \it F
                }
            \end{prooftree}
            \end{equation*}
            Let $\comput t\in\redCand{\itEffectReturn{\it M}}$ and 
            $\comput u\in\redCand{x : \it M \vdash \it F}$.
            There exists a value $\val v \in \redCand{\it M}$
            $\letin{x}{\comput t}{\comput u}\reductArr* 
                \letin{x}{\ret{\val v}}{\comput u} \reductArr
                \comput u\sub{x}{\val v} \in \redCand{\it F}$. 
            By reducibility expansion (\Cref{lem:ReducibilityExpansionReduction}),
            one can conclude $\letin{x}{\comput t}{\comput u} \in \redCand{\it F}$.\\

        \item \emph{Recursive case.} Let us assume $\varphi$ has the following form.
            \begin{equation*}
            \begin{prooftree}
                    \hypo{\varphi_i}
                        \ellipsis{}{
                            \it E_i
                            \replaceLeaf{\it M_k^i \rightarrow \it G_k^i}[\rangeK_i]
                            \it F_i}
                    \delims{\left(}{\right)_{\rangeI}}
                \infer1{
                    \itEffect{\sigma}{\it M}
                            {\it N_i \rightarrow \it E_i}_{\rangeI}
                    \replaceLeaf{
                        \it M_k^i \rightarrow \it G_k^i
                        }[\rangeK_i, \rangeI]
                    \itEffect{\sigma}{\it M}
                        {\it N_i \rightarrow \it F_i}_{\rangeI}}
            \end{prooftree}
            \end{equation*}
            Let $\comput t\in\redCand{\itEffect{\sigma}{\it M}
                    {\it N_i \rightarrow \it E_i}_{\rangeI}}$ and 
            $\comput u\in\bigcap_{\rangeK_i, \rangeI}
                \redCand{x : \it M_k^i \vdash \it G_k^i}$.
            By definition, there exists a value $\val v \in \redCand{\it M}$
            and $\comput s\in\bigcap_{\rangeI} \redCand{y : \it N_i \vdash \it E_i}$
            such that $\comput t \reductArr* \effect{\val v}{y}{\comput s}$.
            Therefore $\letin{x}{\comput t}{\comput u}\reductArr* 
                \letin{x}{\effect{\val v}{y}{\comput s}}{\comput u} \reductArr
                \effect{\val v}{y}{\letin{x}{\comput s}{\comput u}}$.
            
            Let $\rangeI$ and $\val w\in\redCand{\it N_i}$.
            Since $y$ is a free variable of $\comput u$, 
            $(\letin{x}{\comput s}{\comput u})\sub{y}{\val w} =
                \letin{x}{\comput s\sub{y}{\val w}}{\comput u}$.
            Therfore by induction hypothesis with $\comput s\sub{y}{\val w}\in\redCand{\it E_i}$ and 
            $\comput u\in\bigcap_{\rangeK_i}
                \redCand{x : \it M_k^i \vdash \it G_k^i}$,
            $\letin{x}{\comput s\sub{y}{\val w}}{\comput u} \in \redCand{\it F_i}$.
            This being true for all $\rangeI$ and $\val w\in\redCand{\it N_i}$,
            $\letin{x}{\comput s}{\comput u} \in 
                \bigcap_{\rangeI} \redCand{y : \it N_i \vdash \it F_i}$.
            And since $\val v\in\redCand{\it M}$, we can deduce
            $\effect{\val v}{y}{\letin{x}{\comput s}{\comput u}} 
                \in \redCand{\itEffect{\sigma}{\it M}
                    {\it N_i \rightarrow \it F_i}_{\rangeI}}$.

            Therfore by reducibility expansion (\Cref{lem:ReducibilityExpansionReduction}),
            one can conclude $\letin{x}{\comput t}{\comput u} 
                \in \redCand{\itEffect{\sigma}{\it M}
                    {\it N_i \rightarrow \it F_i}_{\rangeI}}$.
    \end{itemize}

\end{proof}
Let us now use this lemma to prove the following case. 
For $p = \letin{x}{\comput t}{\comput u}$, $\Pi$ has the following form :
\begin{equation*}
    \hspace{-0.5cm}
\begin{prooftree}
        \hypo{\itJugComput[\Pi_{\comput t}]
            {\Delta}{\comput t}{\it E}}
        \hypo{\varphi}
            \ellipsis{}{
        \hypo{\it E
                \replaceLeaf{\it M_i \rightarrow \it G_i}[\rangeI]
            \it F}}
        \hypo{
            (\itJugComput[\Pi_i]
                {\Gamma_i; x : \it M_i}
                {\comput u}{\it G_i}
            )_{\rangeI}}
    \typeRuleLetin{
        \itJugComput
            {\Delta +_{\rangeI} \Gamma_i}
            {\letin{x}{\comput t}{\comput u}}{\it F}
    }
\end{prooftree}
\end{equation*}
By induction hypothesis on the $\Pi_i$ (resp. on $\Pi_{\comput t}$ ), 
$\comput u \in \bigcap_{\rangeI} \redCand{\Gamma_i, x : \it M_i \vdash \it G_i}$
(resp. $\comput t\in\redCand{\Delta\vdash \it E}$). Let 
$\vec{\val w} \in \redCand{\Delta+_{\rangeI}\Gamma_i}$,
by the previous lemma (\Cref{lem:AUX_TypableReducibility}) on
$\it E \replaceLeaf{\it M_i \rightarrow \it G_i}[\rangeI] \it F$,
$\letin{x}{\comput t\sub{\Delta}{\vec{\val w}}}
    {\comput u\sub{+_{\rangeI}\Gamma_i}{\vec{\val w}}}\in \redCand{\it F}$.
Therefore, $\letin{x}{\comput t}{\comput u}\in
    \redCand{\Delta +_{\rangeI} \Gamma_i \vdash \it F}$.

\paragraph*{Case Application}
For $p = \app{\val v_1}{\val v_2}$, $\Pi$ has the following form :
\begin{equation*}
\begin{prooftree}
        \hypo{\itJugValue[\Pi_1]{\Gamma_1}{\val v_1}{\itSet{\it M \rightarrow \it E}}}
        \hypo{\itJugValue[\Pi_2]{\Gamma_2}{\val v_2}{\it M}}
    \typeRuleApp{\itJugComput{\Gamma_1 + \Gamma_2}{\app{\val v_1}{\val v_2}}{\it E}}
\end{prooftree}
\end{equation*}
Let $\vec{\val w} \in \redCand{\Gamma_1 + \Gamma_2}$.
By induction hypothesis on $\Pi_1$ (resp. on $\Pi_2$), 
$\val v_1'\coloneqq \val v_1\sub{\Gamma_1}{\vec{\val w}}
    \in\redCand{\itSet{\it M \rightarrow \it E}}$
(resp. $\val v_2'\coloneqq\val v_2\sub{\Gamma_2}{\vec{\val w}}
    \in\redCand{\it M}$).
By definition $\app{\val v_1'}{\val v_2'} \in \redCand{\it E}$
and so $\app{\val v_1}{\val v_2} \in \redCand{\Gamma_1+\Gamma_2 \vdash \it E}$.

\paragraph*{Case Case}
For $p = \case{\val v}{\comput t_1, \ldots, \comput t_n}$, $\Pi$ has the following form :
\begin{equation*}
\begin{prooftree}
        \hypo{\itJugValue[\Pi_1]{\Gamma_1}{\val v}{\itSet{\typeInt{m}}}}
        \hypo{\itJugComput[\Pi_2]{\Gamma_2}{\comput t_m}{\it E}}
    \typeRuleCase{\itJugComput{\Gamma_1 + \Gamma_2}
        {\case{\val v}{\comput t_1, \ldots, \comput t_n}}{\it E}}
\end{prooftree}
\end{equation*}
Let $\vec{\val w} \in \redCand{\Gamma_1 + \Gamma_2}$,
by context reducibility (\Cref{rem:ContextReducibility}) 
$\vec{\val w} \in \redCand{\Gamma_1}\cap\redCand{\Gamma_2}$.
By induction hypothesis on $\Pi_1$,
$\val v'\coloneqq\val v\sub{\Gamma_1}{\vec{\val w}}
    \in \redCand{\typeInt{m}}=\{\int m\}$,
and similarly by induction hypothesis on $\Pi_2$, 
$\comput t_m'\coloneqq\comput t_m\sub{\Gamma_2}{\vec{\val w}} 
    \in \redCand{\it E}$.
We also define $\comput t_k'\coloneqq\comput t_k\sub{\Gamma_1+\Gamma_2}{\vec{\val w}}$
for all $k \le n$.
Therfore, 
$\case{\val v'}{\comput t_1', \ldots, \comput t_n'}=
    \case{\int m}{\comput t_1', \ldots, \comput t_n'}
    \reductArr \comput t_m' \in \redCand{\it E}$.
Thus by reducibility expansion (\Cref{lem:ReducibilityExpansionReduction}),
$\case{\val v'}{\comput t_1', \ldots, \comput t_n'} \in \redCand{\it E}$,
and so $\case{\val v}{\comput t_1, \ldots, \comput t_n}\in
    \redCand{\Gamma_1+\Gamma_2 \vdash \it E}$.

\paragraph*{Case Return}
For $p = \ret{\val v}$, $\Pi$ has the following form :
\begin{equation*}
\begin{prooftree}
        \hypo{\itJugValue[\Pi']{\Gamma}{\val v}{\it M}}
    \typeRuleRet{\itJugComput{\Gamma}
        {\ret{\val v}}{\itEffectReturn{\it M}}}
\end{prooftree}
\end{equation*}
Let $\vec{\val w} \in \redCand{\Gamma}$,
by induction hypothesis on $\Pi'$,
$\val v\sub{\Gamma}{\vec{\val w}}\in\redCand{\it M}$.
Therfore, $\ret{\val v}\sub{\Gamma}{\vec{\val w}}
    \in\redCand{\itEffectReturn{\it M}}$
and so $\ret{\val v}
    \in\redCand{\Gamma \vdash \itEffectReturn{\it M}}$.

\paragraph*{Case Effect}
For $p = \effect{\val v}{y}{\comput t}$, $\Pi$ has the following form :
\begin{equation*}
\begin{prooftree}
        \hypo{\itJugValue[\Pi_{\val v}]
            {\Delta}
            {\val v}
            {\it M}
        }
        \hypo{(\itJugComput[\Pi_i]
            {\Gamma_i; y : \it N_i}
            {\comput t}
            {\it E_i}
        )_{\rangeI}}
    \typeRuleEff{
        \itJugComput
            {\Delta +_{\rangeI} \Gamma_i}
            {\effect{\val v}{y}{\comput t}}
            {\itEffect{\sigma}{\it M}{\it N_i \rightarrow \it E_i}_{\rangeI}}
        }
\end{prooftree}
\end{equation*}
Let $\vec{\val w} \in \redCand{\Delta +_{\rangeI} \Gamma_i}$
and let $\rangeI$.
By context reducibility (\Cref{rem:ContextReducibility}) 
$\vec{\val w} \in \redCand{\Gamma_i}$
and $\vec{\val w} \in \redCand{\Delta}$.
By induction hypothesis on $\Pi_i$, 
$\comput t'\coloneqq \comput t\sub{\Gamma_i}{\vec{\val w}}
    \in\redCand{\it N_i \vdash \it E_i}$,
similarly by induction hypothesis on $\Pi_{\val v}$,
$\val v'\coloneqq\val v\sub{\Delta}{\vec{\val w}}
    \in\redCand{\it M}$.
This being true for all $\rangeI$,
$\comput t'\in\bigcap_{\rangeI} \redCand{y : \it N_i \vdash \it E_i}$
and therefore
$\effect{\val v'}{y}{\comput t'}
    \in\redCand{\itEffect{\sigma}{\it M}
        {\it N_i \rightarrow \it E_i}_{\rangeI}}$.
Thus, $\effect{\val v}{y}{\comput t}
    \in\redCand{\Delta+_{\rangeI}\Gamma_i \vdash
    \itEffect{\sigma}{\it M}
        {\it N_i \rightarrow \it E_i}_{\rangeI}}$.

\paragraph*{Case Handle}
For $p = \handle{\handler h}{\comput t}$, $\Pi$ has the following form :
\begin{equation*}
\begin{prooftree}
        \hypo{\itJugHandler[\Pi_1]{\Gamma_1}{\handler h}{\it E \Rightarrow \it F}}
        \hypo{\itJugComput[\Pi_2]{\Gamma_2}{\comput t}{\it E}}
    \typeRuleHandle{\itJugHandler{\Gamma_1 + \Gamma_2}
        {\handle{\handler h}{\comput t}}{\it F}}
\end{prooftree}
\end{equation*}
Let $\vec{\val w} \in \redCand{\Gamma_1 + \Gamma_2}$,
by context reducibility (\Cref{rem:ContextReducibility}) 
$\vec{\val w} \in \redCand{\Gamma_1}\cap\redCand{\Gamma_2}$.
By induction hypothesis on $\Pi_1$ (resp. on $\Pi_2$), 
$\handler h'\coloneqq\handler h\sub{\Gamma_1}{\vec{\val w}}
    \in\redCand{\it E \Rightarrow \it F}$
(resp. $\comput t'\coloneqq\comput t\sub{\Gamma_2}{\vec{\val w}}
    \in\redCand{\it E}$).
Therfore, $\handle{\handler h'}{\comput t'}
    \in \redCand{\it F}$
and so $\handle{\handler h}{\comput t} \in \redCand{\Gamma_1+\Gamma_2 \vdash \it F}$.

\paragraph{Case Handler Return}

For $p=\{\retClause{y}{\comput t}\}\cup \handler h$, $\Pi$ has the following form:
   \[\begin{prooftree}
                \hypo{\itJugComput[\Pi']
                    {\Gamma;
                        y : \it M}
                    {\comput t}
                    {\it E}}
                \hypo{\{\retClause{y}{\comput t}\}\handlerCompat \handler h}
            \typeRuleHandlerRet[2]{\itJugHandler
                {\Gamma}
                {\{\retClause{y}{\comput t}\}\cup \handler h }%
                {\itEffectReturn{\it M} \Rightarrow \it E}
            }%
        \end{prooftree}\]
Let $\vec{\val w}\in \redCand{\Gamma}$. By induction hypothesis on $\Pi'$, $\comput t'\coloneq \comput t\sub{\Gamma}{\vec{\val w}}\in \redCand{\it M\vdash \it E}$. Let $\comput s\in \redCand{\itEffectReturn{\it M}}$, there exists $\val v\in \redCand{\it M}$ such that $\comput s\reductArr* \ret{\val v}$. Therefore, $\handle{\{\retClause{y}{\comput t}\}\cup \handler h}{\comput s}\reductArr* \handle{\{\retClause{y}{\comput t}\}\cup \handler h}{\ret{\val v}}\reductArr \comput t'\sub{y}{\val v}\in \redCand{\it E}$. By \Cref{lem:ReducibilityExpansionReduction}, $ \handle{\{\retClause{y}{\comput t'}\}\cup \handler h}{\ret{\val v}}\in \redCand{\it E}$. This being true for all $\comput s \in \redCand{\itEffectReturn{\it M}}$, $\{\retClause{y}{\comput t'}\cup \handler h\}\in \redCand{\itEffectReturn{\it M}\Rightarrow \it E}$ and therefore $\{\retClause{y}{\comput t}\cup \handler h\}\in \redCand{\Gamma\vdash\itEffectReturn{\it M}\Rightarrow \it E}$.

\paragraph*{Case Handler}
For $p = \{\effectClause[\sigma]{y}{r}{\comput s}\} \cup \handler h$, $\Pi$ has the following form :
\begin{equation*}
\begin{prooftree}
        \hypo{%
            \itJugComput[\Pi_{\sigma}]
                {\Delta;
                    y : \it M;
                    r : \itSet*{\it N_i \rightarrow \it G_i}_{\rangeI}}%
                {\comput t}%
                {\it E}%
            }%
        \hypo{
        \begin{matrix}
        \{\effectClause{y}{r}{\comput t}\} \cup \handler h\\
        (%
            \itJugHandler[\Pi_i]
                {\Gamma_i}%
                {\handler h}%
                {\it F_i \Rightarrow \it G_i}%
        )_{\rangeI}
        \end{matrix}}%

    \typeRuleHandler{%
        \itJugHandler%
            {\Delta +_{\rangeI} \Gamma_i}%
            {\{\effectClause{y}{r}{\comput t}\} \cup \handler h}%
            {\itEffect{\sigma}{\it M}
                {\it N_i\rightarrow \it F_i
            }_{\rangeI} \Rightarrow \it E
            }%
    }%
\end{prooftree}
\end{equation*}
Let $\vec{\val w} \in \redCand{\Delta+_{\rangeI} \Gamma_i}=
    \redCand{\Delta}\bigcap_{\rangeI}\redCand{\Gamma_i}$.
By induction hypothesis on the $\Pi_i$, 
$\handler h' := \handler h\sub{+_{\rangeI}\Gamma_i}{\vec{\val w}}
    \in\bigcap_{\rangeI} \redCand{\it F_i \Rightarrow \it G_i}$.
Similarly, by induction hypothesis on $\Pi_{\sigma}$,
$\comput t'\coloneqq \comput t\sub{\Delta}{\vec{\val w}}\in
    \redCand{
        y : \it M;
        r : \itSet*{\it N_i \rightarrow \it G_i}_{\rangeI}
    \vdash \it E}$.
Let $\comput s \in 
    \redCand{\itEffect{\sigma}{\it M}
        {\it N_i\rightarrow \it F_i}_{\rangeI}}$.
There exists $\val v\in\redCand{\it M}$ and
$\comput u\in\bigcap_{\rangeI} \redCand{x : \it N_i \vdash \it F_i}$
such that $\comput s\reductArr* \effect{\val v}{x}{\comput u}$.
Therefore $\handle{(\{\effectClause{y}{r}{\comput t'}\} \cup \handler h')}
        \comput s 
    \reductArr* \handle{(\{\effectClause{y}{r}{\comput t'}\} \cup \handler h')}        \effect{\val v}{x}{\comput u}
    \reductArr
        \comput t'
        \sub{y}{\val v}
        \sub{r}{\abs{x}{\handle{
            (\{\effectClause{y}{r}{\comput t'}\} \cup \handler h')}
            {\comput u}}}$.

Let us show that $\abs{x}{\handle{
    (\{\effectClause{y}{r}{\comput t'}\} \cup \handler h')}{\comput u}} \in 
    \redCand{\itSet*{\it N_i \rightarrow \it G_i}_{\rangeI}}$.
For $\rangeI$, by handler weakening (\Cref{lem:HandlerWeakening})
and substitution lemma (\Cref{lem:SubstitutionLemma}) on $\Pi_i$, 
there exists $\itJugHandler[\Pi_i']{\Gamma_i}
    {\{\effectClause{y}{r}{\comput t'}\} \cup \handler h'}
    {\it F_i \Rightarrow \it G_i}$.
Since each derivation $\Pi_i'$ is smaller than the 
derivation $\Pi_i$, we can apply the induction hypothesis on the
$\Pi_i'$ and have $\{\effectClause{y}{r}{\comput t'}\} \cup \handler h'
    \in\bigcap_{\rangeI} \redCand{\it F_i \Rightarrow \it G_i}$
Therefore, by definition of $\comput u$, $\abs{x}{\handle{
    (\{\effectClause{y}{r}{\comput t'}\} \cup \handler h')}{\comput u}} \in 
    \redCand{\itSet*{\it N_i \rightarrow \it G_i}_{\rangeI}}$.

Thus $\comput t'\sub{y}{\val v}
    \sub{r}{\abs{x}{\handle{
    (\{\effectClause{y}{r}{\comput t'}\} \cup \handler h')}{\comput u}}}
    \in \redCand{\it E}$
and by reducibility expansion (\Cref{lem:ReducibilityExpansionReduction}),
$\handle{(\{\effectClause{y}{r}{\comput t'}\} \cup \handler h')}
    \comput s \in \redCand{\it E}$.
This being true for all $\comput s\in
    \redCand{\itEffect{\sigma}{\it M}
    {\it N_i\rightarrow \it F_i}_{\rangeI}}$, 
$\{\effectClause{y}{r}{\comput t'}\} \cup \handler h'
    \in\redCand{\itEffect{\sigma}{\it M}
        {\it N_i\rightarrow \it F_i}_{\rangeI} \Rightarrow \it E}$ 
and therefore $\{\effectClause{y}{r}{\comput t}\} \cup \handler h
    \in\redCand{\Delta +_{\rangeI} \Gamma_i \vdash
        \itEffect{\sigma}{\it M}
            {\it N_i\rightarrow \it F_i}_{\rangeI} \Rightarrow \it E}$.

%% file: Proofs/ElementarySimpleSubjectReduction.tex
By case analysis on the different rewrite rules:
\paragraph*{Case $\reductBeta$.}
Let $\comput t, \comput t' \in \setComput$ such that $\comput t$ is
typed and $\comput t \reductArr'_\reductBeta \comput t'$. By definition,
$\comput t \coloneqq \app{(\abs{x}{\comput u})}{\val v}$ and $\comput t'
\coloneqq \comput u\sub{x}{\val v}$. Since $\comput t$ is typable,
there exists a derivation
$\stJugComput[\Pi]{\emptyset}{\app{(\abs{x}{\comput u})}{\val v}}{\st E}$
for some type $\st E$. By
case analysis, we deduce that $\Pi$ is necessarily of the following
form:
\begin{equation*}
    \begin{prooftree}
        \hypo{\stJugComput[\Pi_{\comput u}]{ x :: \st M}{\comput u}{\st E}}
        \typeRuleAbs{\stJugValue{\emptyset}{\abs{x}{\comput u}}{\itSet{\st M \rightarrow \st E}}}
        \hypo{\stJugValue[\Pi_{\val v}]{\emptyset}{\val v}{\st M}}
        \typeRuleApp{\stJugComput{\emptyset}{\app{(\abs{x}{\comput u})}{\val v}}{\st E}}
    \end{prooftree}
\end{equation*}
By substitution lemma (\Cref{lem:SimpleSubstitutionLemma})
with $\Pi_{\comput u}$ and $\Pi_{\val v}$, there exists
$\stJugComput[\Pi']{\emptyset}{\comput u\sub{x}{\val v}}{\st E}$,
therefore concluding this case.

\paragraph*{Case $\reductFix$.}
Let $\comput t, \comput t' \in \setComput$ such that $\comput t$ is
typed and $\comput t \reductArr'_\reductFix \comput t'$. By definition,
$\comput t \coloneqq \app{(\fix{x}{\val v})}{\val w}$ and
$\comput t'
\coloneqq
\app{\val v\sub{x}{\fix{x}{\val v}}}{\val w}$. Since
$\comput t$ is typable, there exists a derivation
$\stJugComput[\Pi]{\emptyset}{\app{(\fix{x}{\val v})}{\val w}}{\st E}$
for some type $\st E$. By
case analysis, we deduce that $\Pi$ is necessarily in the
following form:
    \begin{equation*}
         \begin{prooftree}
                    \hypo{
                        \stJugValue[\Pi_1]
                            {x :: \st M \rightarrow \st E}
                            {\val v}
                            {\st M \rightarrow \st E}}
                    
                \typeRuleFix[1]{
                    \stJugValue[\Pi_{\fix*{}{}}]
                        {\emptyset}
                        {\fix{x}{\val v}}
                        {\st M \rightarrow \st E}}
                \hypo{\stJugValue[\Pi_2]{\emptyset}{\val w}{\st M}}
                \typeRuleApp{\stJugComput{\emptyset}{\app{(\fix{x}{\val v})}{\val w}}{\st E}}
            \end{prooftree}
    \end{equation*}

By the substitution lemma (\Cref{lem:SimpleSubstitutionLemma})
with $\Pi_1$ and $\Pi_{\fix*{}{}}$, there exist a derivation
$\stJugValue[\Pi'_1]{\emptyset}{\val v\sub{x}{\fix{x}{\val v}}}{\st M
\rightarrow \st E}$. Let $\Pi'$ be the following
typing derivation:
\begin{equation*}
    \begin{prooftree}
            \hypo{
                \stJugValue[\Pi'_1]
                    {\emptyset}
                    {\val v\sub{x}{\fix{x}{\val v}}}
                    {\st M \rightarrow \st E}}
            \hypo{
                \stJugValue[\Pi_{\val w}]
                    {\emptyset}
                    {\val w}
                    {\st M}}
        \typeRuleApp{
            \stJugComput
                {\emptyset}
                {\app{\val v\sub{x}{\fix{x}{\val v}}}{\val w}}
                {\st E}}
    \end{prooftree}
\end{equation*}

\paragraph*{Case $\reductCase$.}
Let $\comput t, \comput t' \in \setComput$ such that $\comput t$ is
typed and $\comput t \reductArr'_\reductBeta \comput t'$. By definition,
$\comput t \coloneqq \case{\int n}{\comput t_1, \ldots,
\comput t_m}$ and $\comput t' \coloneqq \comput t_n$ with $0 < n \leq m$.
Since $\comput t$ is typable, there exists a derivation
$\stJugComput[\Pi]{\emptyset}{\case{\int n}{\comput t_1,
\ldots, \comput t_m}}{\st E}$ for some type $\st E$. By case analysis, we deduce that $\Pi$ is
necessarily of the following form:
\begin{equation*}
     \begin{prooftree}
                        \hypo{\stJugValue{\emptyset}{\val v}{\typeInt{m}}}
                        \hypo{(\stJugComput[\Pi_i]{\emptyset}{\comput t_i}{\st E}
                            )_{ 0<i\le m}}
                    \typeRuleCase[2]{\stJugComput{\emptyset }{\case{\val v}{\comput t_1, \ldots, \comput t_m}}{\st E}}
                \end{prooftree}
\end{equation*}
One concludes this case by observing that $\Pi_n$ is typing
$\comput t'$.

\paragraph*{Case $\reductLetRet$.}
Let $\comput t, \comput t' \in \setComput$ such that $\comput t$ is
typed and $\comput t \reductArr'_\reductLetRet \comput t'$. By
definition, $\comput t \coloneqq
\letin{x}{\ret{\val v}}{\comput u}$ and $\comput t'
\coloneqq \comput u\sub{x}{\val v}$. Since $\comput t$ is typable,
there exists a derivation
$\stJugComput[\Pi]{\emptyset}{\letin{x}{\ret{\val v}}{\comput u}}{\st E}$
for some type $\st E$. By
case analysis, we deduce that $\Pi$ is necessarily of the following
form:
\begin{equation*}
    \begin{prooftree}
            \hypo{\stJugValue[\Pi_\val v]{\emptyset}{\val v}{\st M}}
            \typeRuleRet{\stJugComput{\emptyset}{\ret{\val v}}{\itEffectReturn{\st M}}}
            \infer0{\itEffectReturn{\st M} \replaceLeaf{\st M \rightarrow \st E} \st E}

            \hypo{\stJugComput[\Pi_\comput u]{ x :: \st M}{\comput u}{\st E}}
        \typeRuleLetin{\stJugComput{\emptyset}{\letin{x}{\ret{\val v}}{\comput u}}{\st E}}
    \end{prooftree}
\end{equation*}
By substitution lemma (\Cref{lem:SimpleSubstitutionLemma})
with $\Pi_{\comput u}$ and $\Pi_{\val v}$, there exists
$\stJugComput[\Pi']{\emptyset}{\comput u\sub{x}{\val v}}{\st E}$,
therefore concluding this case.

\paragraph*{Case $\reductLetEff$.}
Let $\comput t, \comput t' \in \setComput$ such that $\comput t$ is
typed and $\comput t \reductArr'_\reductLetEff \comput t'$. By
definition, $\comput t \coloneqq
\letin{x}{\effect{\val v}{y}{\comput s}}{\comput u}$ and
$\comput t' \coloneqq
\effect{\val v}{y}{\letin{x}{\comput s}{\comput t}}$. Since
$\comput t$ is typable, there exists a derivation
$\stJugComput[\Pi]{\emptyset}{\letin{x}{\effect{\val v}{y}{\comput s}}{\comput u}}{\st E}$
for some type $\st E$. By
case analysis, we deduce that $\Pi$ is necessarily of the following
form:
\begin{equation*}
    \begin{prooftree}
                \hypo{
                    \stJugValue[\Pi_{\val v}]
                        {\emptyset}
                        {\val v}
                        {\st M}
                }
                \hypo{
                    \stJugComput[\Pi_{\comput s}]
                        { y :: \st N}
                        {\comput s}
                        {\st G}}
            \typeRuleEff{
                \stJugComput
                    {\emptyset}
                    {\effect{\val v}{y}{\comput s}}
                    {\itEffect{\sigma}{\st M}{\st N \rightarrow \st G}}
            }
            \hypo{(A)}

            \hypo{
                    \stJugComput[\Pi_\comput u]
                    { x :: \st M'}
                    {\comput u}
                    {\st F}}
        \typeRuleLetin{
            \stJugComput
                {\emptyset}
                {\letin{x}{\effect{\val v}{y}{\comput s}}{\comput u}}
                {\itEffect
                    {\sigma}{\st M}
                    {\st N \rightarrow \st E}}
        }
    \end{prooftree}
\end{equation*}
\begin{equation*}
    (A) \coloneqq
    \begin{prooftree}
            \hypo{\varphi}
            \ellipsis{}{
                \st G
                \replaceLeaf{
                    \st M' 
                    \rightarrow 
                    \st F
                    }
                \st E}
        \infer1{
            \itEffect
                {\sigma}{\st M}
                {\st N \rightarrow \st G}
            \replaceLeaf{
                    \st M' 
                    \rightarrow 
                    \st F
                }
            \itEffect
                {\sigma}{\st M}
                {\st N \rightarrow \st E}}
    \end{prooftree}
\end{equation*}
Let $\Pi'$ be the following typing derivation:
\begin{equation*}
    \begin{prooftree}
        \hypo{\stJugValue[\Pi_{\val v}]{\emptyset}{\val v}{\st M}}
        \hypo{(B_i)}
    \typeRuleEff{
        \stJugComput
            {\emptyset}
            {\effect{\val v}{y}{\letin{x}{\comput s}{\comput u}}}
            {\itEffect
                    {\sigma}{\st M}
                    {\st N \rightarrow \st E}}}
    \end{prooftree}
\end{equation*}
\begin{equation*}
    (B_i) \coloneqq
    \begin{prooftree}
            \hypo{\stJugComput[\Pi_{\comput s}]{ y :: \st N}{\comput s}{\st G}}
            \hypo{\varphi}
            \ellipsis{}{
                \st G
                \replaceLeaf{
                    \st M' 
                    \rightarrow 
                    \st F
                    }
                \st E}
            \hypo{\stJugComput[\Pi_\comput u^{(y::\st N)}]
                {y::\st N; x :: \st M'}
                {\comput u}
                {\st F}}
        \typeRuleLetin{
            \stJugComput
                { y :: \st N}
                {\letin{x}{\comput s}{\comput u}}
                {\st E}}
    \end{prooftree}
\end{equation*}
One concludes this case by observing that $\Pi'$ is typing
$\comput t'$.

\paragraph*{Case $\reductHdlRet$.}
Let $\comput t, \comput t' \in \setComput$ such that $\comput t$ is
typed and $\comput t \reductArr'_\reductHdlRet \comput t'$. By
definition, $\comput t \coloneqq
\handle{\handler h}{\ret{\val v}}$ and $\comput t'
\coloneqq \comput u\sub{x}{\val v}$ with
$\retClause{\val v}{\comput u} \in \handler h$. Since
$\comput t$ is typable, there exists a derivation
$\stJugComput[\Pi]{\emptyset}{\handle{\handler h}{\ret{\val v}}}{\st E}$
for some type $\st E$. By
case analysis, we deduce that $\Pi$ is necessarily of the following
form:
\begin{equation*}
    \begin{prooftree}
            \hypo{\stJugComput[\Pi_{\comput u}]{ x :: \st M}{\comput u}{\st E}}
           \hypo{ \{\retClause{x}{\comput u}\}\handlerCompat \handler h'}
            \typeRuleHandlerRet[2]{\stJugHandler{\emptyset}{\{\retClause{x}{\comput u}\} \cup \handler h'}{\itEffectReturn{\st M} \Rightarrow \st E}}
            \hypo{\stJugValue[\Pi_\val v]{\emptyset}{\val v}{\st M}}
            \typeRuleRet{\stJugComput{\emptyset}{\ret{\val v}}{\itEffectReturn{\st M}}}
        \typeRuleHandle{\stJugComput{\emptyset}{\handle{\{\retClause{x}{\comput u}\}\cup \handler h'}{\ret{\val v}}}{\st E}}
    \end{prooftree}
\end{equation*}
By substitution lemma (\Cref{lem:SimpleSubstitutionLemma})
with $\Pi_{\comput u}$ and $\Pi_{\val v}$, there exists
$\stJugComput[\Pi']{\emptyset}{\comput u\sub{x}{\val v}}{\st E}$,
which concluding this case since $\comput t' =
\comput u\sub{x}{\val v}$.

\paragraph*{Case $\reductHdlEff$.}
Let $\comput t, \comput t' \in \setComput$ such that $\comput t$ is
typed and $\comput t \reductArr'_\reductHdlEff \comput t'$. By definition,
$\comput t \coloneqq
\handle{\handler h}{\effect{\val v}{x}{\comput u}}$ and
$\comput t' \coloneqq
\comput s\sub{y}{\val v}\sub{r}{\abs{x}{\handle{\handler h}{\comput u}}}$
with $\effectClause[\sigma]{y}{r}{\comput s} \in
\handler h$. Since $\comput t$ is typable, there exists a derivation
$\stJugComput[\Pi]{\emptyset}{\handle{\handler h}{\effect{\val v}{x}{\comput u}}}{\st E}$
for some type $\st E$. By
case analysis, we deduce that $\Pi$ is necessarily of the following
form:
\begin{equation*}
    \begin{prooftree}
            \hypo{(B)}
                \hypo{\stJugValue[\Pi_{\val v}]{\emptyset}{\val v}{\st M}}
                \hypo{
                    \stJugComput[\Pi_{\comput u}]
                        { x :: \st N}
                        {\comput u}
                        {\st F}}
            \typeRuleEff{
                \stJugComput
                    {\emptyset}
                    {\effect{\val v}{x}{\comput u}}
                    {\itEffect
                        {\sigma}{\st M}
                        {\st N \rightarrow \st F}}}
        \typeRuleHandle{
            \stJugComput
                {\emptyset}
                {\handle
                    {\{\effectClause[\sigma]{x}{r}{\comput s}\} \cup \handler h'}
                    {\effect{\val v}{x}{\comput u}}}
                {\st E}}
    \end{prooftree}
\end{equation*}

\begin{equation*}
    (B) \;\coloneqq\;
    \begin{prooftree}
            \hypo{
                \stJugComput[\Pi_{\comput s}]
                    { y :: \st M; r :: \st N
                        \rightarrow \st G}
                    {\comput s}
                    {\st E}}
        %
            \hypo{
                \stJugHandler[\Pi_\handler h']
                    {\emptyset}
                    {\handler h'}
                    {\st F \Rightarrow \st G}}
        %
        %
        \typeRuleHandler{
            \stJugHandler
                {\emptyset}
                {\{\effectClause{y}{r}{\comput s}\} \cup \handler h'}
                {\itEffect
                        {\sigma}{\st M}{\st N
                    \rightarrow \st F}
                    \Rightarrow
                \st E}}
    \end{prooftree}
\end{equation*}
where $\{\effectClause{y}{r}{\comput s}\} \handlerCompat
\handler h'$.

By substitution lemma (\Cref{lem:SimpleSubstitutionLemma})
with $\Pi_{\comput s}$ and $\Pi_{\val v}$, there exists
$\stJugComput[\Pi_{\comput s\sub{y}{\val v}}]{\Delta; r ::\st N\rightarrow \st G}
{\comput s\sub{y}{\val v}}{\st E}$.
Since $\{\effectClause{y}{r}{\comput s}\} \handlerCompat
\handler h'$, by handler weakening (\Cref{lem:HEBHandlerWeakening})
on $\Pi_\handler h'$, one has that
$\stJugHandler[\Pi_\handler h]{\emptyset}{\handler h}
{\st F \Rightarrow \st G}$.
Let $\Pi_{\handle*{}{}}$ be the following derivation:
\begin{equation*}
    \begin{prooftree}
                \hypo{
                    \stJugHandler[\Pi_\handler h^{(x::\st N)}]
                        {x::\st N}
                        {\handler h}
                        {\st F
                            \Rightarrow
                        \st G}}
                \hypo{
                    \stJugComput[\Pi_{\comput u}]
                        { x :: \st N}
                        {\comput u}
                        {\st F}}
            \typeRuleHandle{
                \stJugComput
                    { x :: \st N}
                    {\handle{\handler h}{\comput u}}
                    {\st G}
            }
        \typeRuleAbs{
            \stJugValue
                {\emptyset}
                {\abs{x}{\handle{\handler h}{\comput u}}}
                {\st N \rightarrow \st G}
        }
    \end{prooftree}
\end{equation*}
We conclude using the substitution lemma
(\Cref{lem:SimpleSubstitutionLemma}) with
$\Pi_{\comput s\sub{y}{\val v}}$ and $\Pi_{\handle*{}{}}$
yielding $\stJugComput[\Pi']{\emptyset}{\comput s\sub{y}{\val v}\sub{r}{\abs{x}{\handle{\handler h}{\comput u}}}}{\st E}$.

%% file: Definitions/NewRefinementRelationContd.tex
\begin{tabular}{c}
\framebox{$
    \begin{array}{c}
    
    \begin{prooftree}
                \hypo{\stJugComput[\Pi']
                    {\Delta; x :: \st M}
                    {\comput t}{\st E}}
            \typeRuleAbs{\stJugValue
                {\Delta}
                {\abs{x}{\comput t}}
                {\st M \rightarrow \st E}}
        \end{prooftree}
        \refrelsym
        \begin{prooftree}
                \hypo{(\itJugComput[\Xi_i]
                    {\Gamma_i; x : \it M_i }
                    {\comput t}{\it E_i})_{\rangeI}}
            \typeRuleAbs{\itJugValue
                {+_{\rangeI} \Gamma_i}
                {\abs{x}{\comput t}}
                {\itSet{\it M_i \rightarrow \it E_i}_{\rangeI}}}
        \end{prooftree}\\
        \text{if } \refrel{\Pi'}{\Xi_i} \; \forall \rangeI
    \\[0.5cm]

        \begin{prooftree}
                \hypo{\stJugValue[\Pi]
                    {\Delta}{\val v}
                    {\typeInt n}}
                \hypo{(
                    \stJugComput[\Pi_k]{\Delta}{\comput t_k}{\st E})_{ 0<k\le n}}
            \typeRuleCase[2]{\stJugComput
                {\Delta }
                {\case{\val v}{\comput t_1, \ldots, \comput t_n}}
                {\st E}}
        \end{prooftree}
    \\[0.5cm]
    \refrelsym 
    \begin{prooftree}
                \hypo{\itJugValue[\Xi]
                    {\Gamma_1}{\val v}
                    {\itSet{\typeInt m}}}
                \hypo{\itJugComput[\Xi_m]
                    {\Gamma_m}{\comput t_m}
                    {\it E}}
                \hypo{0 < m \leq n}
            \typeRuleCase[3]{\itJugComput
                {\Gamma_1 + \Gamma_m}
                {\case{\val v}{\comput t_1, \ldots, \comput t_n}}
                {\it E}}
        \end{prooftree}\\
        \text{if } \refrel{\Pi}{\Xi} \text{ and } \refrel{\Pi_m}{\Xi_m}
        \\[0.5cm]
        \begin{prooftree}
                \hypo{\stJugValue[\Pi]{\Delta}{\val v}
                    {\st M}}
            \typeRuleRet{\stJugComput
                {\Delta}{\ret{\val v}}
                {\itEffectReturn{\st M}}}
        \end{prooftree}
        \refrelsym 
        \begin{prooftree}
                \hypo{\itJugValue[\Xi]{\Gamma}{\val v}
                    {\it M}}
            \typeRuleRet{\itJugComput
                {\Gamma}{\ret{\val v}}
                {\itEffectReturn{\it M}}}
        \end{prooftree}\\
        \text{if } \refrel{\Pi}{\Xi} 
        \\[0.5cm] 
    \begin{prooftree}
                \hypo{\stJugValue[\Pi_1]{\Delta}{\val v}                    {\st M \rightarrow \st U}}
                \hypo{\stJugValue[\Pi_2]{\Delta}{\val w}{\st M}}
            \typeRuleApp{\stJugComput
                {\Delta}
                {\app{\val v}{\val w}}
                {\st U}}
        \end{prooftree}
        \refrelsym
        \begin{prooftree}
                \hypo{\itJugValue[\Xi_1]{\Gamma_1}{\val v}
                    {\itSet{\it M \rightarrow \it E}}}
                \hypo{\itJugValue[\Xi_2]{\Gamma_2}{\val w}{\it M}}
            \typeRuleApp{\itJugComput
                {\Gamma_1 + \Gamma_2}
                {\app{\val v}{\val w}}
                {\it E}}
        \end{prooftree}\\      
        \text{if } \refrel{\Pi_1}{\Xi_1} \text{ and } \refrel{\Pi_2}{\Xi_2}\\[0.5cm]

        \begin{prooftree}
            \hypo{\stJugHandler[\Pi_1]{\Delta}{\handler h}
                {\st F \Rightarrow \st E}}
            \hypo{\stJugComput[\Pi_2]{\Delta}{\comput t}{\st F}}
            \typeRuleHandle{\stJugComput
                {\Delta}
                {\handle{\handler h}{\comput t}}{\st E}}
        \end{prooftree}
    \\[0.5cm]
    \refrelsym
    \begin{prooftree}
            \hypo{\itJugHandler[\Xi_1]{\Gamma_1}{\handler h}
                {\it F \Rightarrow \it E}}
            \hypo{\itJugComput[\Xi_2]{\Gamma_2}{\comput t}{\it F}}
            \typeRuleHandle{\itJugComput
                {\Gamma_1 + \Gamma_2}
                {\handle{\handler h}{\comput t}}
                {\it E}}
        \end{prooftree}\\
        \text{if } \refrel{\Pi_1}{\Xi_1} \text{ and } \refrel{\Pi_2}{\Xi_2}
        \\[0.5cm]

    \begin{prooftree}
    \hypo{
        \stJugComput[\Pi]
            {\Delta; y :: \st M}
            {\comput t}
            {\st E}
            }
             \hypo{\{\retClause{y}{\comput t}\} \handlerCompat \handler h}
            \typeRuleHandlerRet[2]{\stJugHandler
                {\Delta'}
                {\{\retClause{y}{\comput t}\} \cup \handler h}
                {\itEffectReturn{\st M} \Rightarrow \st E}
            }
    \end{prooftree}
    \\[0.5cm]
    \refrelsym
    \begin{prooftree}
    \hypo{
        \itJugComput[\Xi]
            {\Gamma; y :: \it M}
            {\comput t}
            {\it E}
            }
             \hypo{\{\retClause{y}{\comput t}\} \handlerCompat \handler h}
            \typeRuleHandlerRet[2]{\itJugHandler
                {\Gamma}
                {\{\retClause{y}{\comput t}\} \cup \handler h}
                {\itEffectReturn{\it M} \Rightarrow \it E}
            }
    \end{prooftree}\\ 
    \text{if } \refrel{\Pi}{\Xi}
    \\[0.5cm]
    \end{array}
$}
\end{tabular}

%% file: Definitions/HEBReductionsContd.tex
\begin{tabular}{c}
    \hspace{-1cm}
\framebox{$
\begin{array}{c}
	\begin{array}{ccl}
 \begin{prooftree}
                        \hypo{\stJugValue
                            {\emptyset}{\int m}{\typeInt n}}
                        \hypo{(\Pi_k\triangleright \stJugComput
                            {\emptyset}{\comput t_k}{\st E}
                            )_{ 0<k\le n}}
                    \typeRuleCase{\stJugComput
                        {\emptyset}
                        {\case{\val v}{\comput t_1, \ldots, \comput t_n}}
                        {\st E}}
                \end{prooftree}
    & \rightsquigarrow & \stJugComput[\Pi_m]{\emptyset}{\comput t_m}{\st E}\\ \\

		\scalebox{0.9}{\begin{prooftree}
				\hypo{\stJugValue[\Pi_\val v]{\emptyset}{\val v}{\st M}}
				\typeRuleRet{\stJugComput{\emptyset}{\ret{\val v}}{\itEffectReturn{\st M}}}
				\infer0{\itEffectReturn{\st M} \replaceLeaf{\st M \rightarrow \st F} \st F}
				\hypo{\stJugComput[\Pi_{\comput u}]{ x :: \st M}{\comput u}{\st F}}
				\typeRuleLetin{\stJugComput{\emptyset}{\letin{x}{\ret{\val v}}{\comput u}}{\st F}}
		\end{prooftree}}
		
		& \rightsquigarrow & \stJugComput[\Pi_{\comput u}\sub{x}{\Pi_\val v}]{ \emptyset}{\comput u\sub{x}{\val v}}{\st F}\\ \\

		\scalebox{0.9}{\begin{prooftree}
				\hypo{\stJugComput[\Pi_{\comput u}]{ x :: \st M}{\comput u}{\st E}}
				\typeRuleHandler[1]{\stJugHandler{\emptyset}{\{\retClause{x}{\comput u}\} }{\itEffectReturn{\st M} \Rightarrow \st E}}
				\hypo{\stJugValue[\Pi_\val v]{\emptyset}{\val v}{\st M}}
				\typeRuleRet{\stJugComput{\emptyset}{\ret{\val v}}{\itEffectReturn{\st M}}}
				\typeRuleHandle{\stJugComput{\emptyset}{\handle{\{\retClause{x}{\comput u}\} }{\ret{\val v}}}{\st E}}
		\end{prooftree}}
		& \rightsquigarrow & \stJugComput[\Pi_{\comput u}\sub{x}{\Pi_\val v}]{\emptyset}{\comput u\sub{x}{\val v}}{\st E}\\ \\
		
	\end{array}\\

	\begin{array}{c}
        \scalebox{0.9}{\begin{prooftree}
			\hypo{
				\stJugValue[\Pi_{\val v}]
				{\Delta}
				{\val v}
				{\st M}
			}
			\hypo{
				\stJugComput[\Pi_{\comput s}]
				{\Delta; y : \st N}
				{\comput s}
				{\st G}}
			\typeRuleEff{
				\stJugComput
				{\Delta}
				{\effect{\val v}{y}{\comput s}}
				{\itEffect{\sigma}{\st M}{\st N \rightarrow \st G}}
			}

			\hypo{\varphi}
			\ellipsis{}{
				\st G
				\replaceLeaf{
					\st M' 
					\rightarrow 
					\st F
				}
				\st E}
			\infer1{
				\itEffect
				{\sigma}{\st M}
				{\st N \rightarrow \st G}
				\replaceLeaf{
					\st M' 
					\rightarrow 
					\st F
				}
				\itEffect
				{\sigma}{\st M}
				{\st N \rightarrow \st E}}

			\hypo{
				\stJugComput[\Pi_\comput u]
				{\Delta; x : \st M'}
				{\comput u}
				{\st F}}
			\typeRuleLetin{
				\stJugComput
				{\Delta}
				{\letin{x}{\effect{\val v}{y}{\comput s}}{\comput u}}
				{\itEffect
					{\sigma}{\st M}
					{\st N \rightarrow \st E}}
			}
		\end{prooftree}}\\
		\rightsquigarrow
		\begin{prooftree}
			\hypo{\stJugValue[\Pi_{\val v}]{\Delta}{\val v}{\st M}}
			\hypo{\stJugComput[\Pi_{\comput s}]{\Delta; y : \st N}{\comput s}{\st G}}
			\hypo{\varphi}
			\ellipsis{}{
				\st G
				\replaceLeaf{
					\st M' 
					\rightarrow 
					\st F
				}
				\st E}
			\hypo{\stJugComput[\Pi_\comput u^{(y::\st N)}]
				{\Delta;y::\st N; x : \st M'}
				{\comput u}
				{\st F}}
			\typeRuleLetin{
				\stJugComput
				{\Delta; y : \st N}
				{\letin{x}{\comput s}{\comput u}}
				{\st E}}
			\typeRuleEff{
				\stJugComput
				{\Delta}
				{\effect{\val v}{y}{\letin{x}{\comput s}{\comput u}}}
				{\itEffect
					{\sigma}{\st M}
					{\st N \rightarrow \st E}}}
		\end{prooftree}\\  \\
	
		\begin{prooftree}
			\hypo{
				\stJugComput[\Pi_{\comput s}]
				{   x :: \st M ;
					r :: \st N \rightarrow \st G}
				{\comput s}
				{\st E}}

			\hypo{
				\begin{matrix}
					\{\effectClause{x}{r}{\comput s}\} \handlerCompat \handler h\\
					(
					\itJugHandler[\Pi_{\handler h}]
					{\emptyset}
					{\handler h}
					{\st F \Rightarrow \st G}
					)
				\end{matrix}  
			}
			\typeRuleHandler{
				\itJugHandler
				{\emptyset }
				{\{\effectClause{x}{r}{\comput s}\} \cup \handler h}
				{\itEffect{\sigma}{\it M}
					{\st N \rightarrow \st F}
					\Rightarrow \st E}}
			\hypo{\stJugValue[\Pi_{\val v}]
				{\emptyset}
				{\val v}{\st M}}
			\hypo{
				\stJugComput[\Pi_{\comput u}]
				{y :: \st N}
				{\comput u}
				{\st F}}
			\typeRuleEff{
				\stJugComput
				{ \emptyset}
				{\effect{\val v}{y}{\comput u}}
				{\itEffect
					{\sigma}{\st M}
					{\st N \rightarrow \st F}}}
			\typeRuleHandle{
				\stJugComput
				{\emptyset}
				{\handle
					{\{\effectClause[\sigma]{x}{r}{\comput s}\} \cup \handler h}
					{\effect{\val v}{y}{\comput u}}}
				{\st E}}
		\end{prooftree}\\ \\

		\rightsquigarrow \;\;\stJugComput[\Pi_\comput{s}\sub{x}{\Pi_\val v^{(r::\st N\rightarrow\st G)}}\sub{r}{\Pi_{\val{\lambda y}}}]{\emptyset }{\comput s \sub{x}{\val v}\sub{r}{\abs{y}{\handle{\handler h}{\comput u}}}}{\st E}\\ \\
	\text{With  }
		\begin{prooftree}
			\hypo{
				\itJugHandler[\Xi_{\handler h}^{(y::\st N)}]
				{y::\st N}
				{\handler h}
				{\st F
					\Rightarrow
					\st G}}
			\hypo{
				\itJugComput[\Xi_{\comput u}]
				{ y :: \st N}
				{\comput u}
				{\st F}}
			\typeRuleHandle{
				\itJugComput
				{ y :: \st N}
				{\handle{\handler h}{\comput u}}
				{\st G}
			}
			\typeRuleAbs{
				\itJugValue[\Pi_{\val{\lambda y}}]
				{\emptyset}
				{\abs{y}{\handle{\handler h}{\comput u}}}
				{\st N \rightarrow \st G}
			}
		\end{prooftree}
    \end{array}
\end{array}	
$}\end{tabular}

%% file: Proofs/NewRefinementElementarySubjectExpansion.tex
By case analysis on the different rewrite rules. 

\paragraph*{Case $\reductBeta$.}
Let $\comput t, \comput t' \in \setComput$ such that $\comput t'$ is
typed and $\comput t \reductArr_\reductBeta \comput t'$. By definition,
$\comput t \coloneqq \app{(\abs{x}{\comput u})}{\val v}$ and $\comput t'
\coloneqq \comput u\sub{x}{\val v}$. By hypothesis,
$\comput t$ is typable with a \HEB\ type derivation $\Pi$ that, by case analysis, has to have the following shape. 

\begin{equation*}
    \begin{prooftree}
        \hypo{\stJugComput[\Pi_{\comput u}]{ x ::\st M}{\comput u}{\st U}}
        \typeRuleAbs{\stJugValue{\emptyset}{\abs{x}{\comput t}}{\st M\rightarrow \st U}}
        \hypo{\stJugValue[\Pi_{\val v}]{\emptyset}{\val v}{\st M}}
        \typeRuleApp{\stJugComput{\emptyset}{\app{(\abs{x}{\comput u})}{\val v}}{\st U}}
    \end{prooftree}
\end{equation*}
Also by hypothesis, there exist two derivations
$\itJugComput[\Xi']{\emptyset}{\comput u\sub{x}{\val v}}{\it E}$ and $\stJugComput[\Pi']{\emptyset}{\comput u\sub{x}{\val v}}{\st U}$ such that $\refrel{\Pi'}{\Xi'}$.
Since by hypothesis $\Pi\rightsquigarrow \Pi'$, by definition of $\rightsquigarrow$ is clear that $\Pi'=\Pi_\comput u\sub{x}{\Pi_\val v}$.
We can then apply anti-substitution lemma
(\Cref{lem:antiSubstLemma-ref}), obtaining that there exist a type $\it M$ and type derivations
$\itJugValue[\Xi_\val v]{\emptyset}{\val v}{\it M}$ and 
$\itJugComput[\Xi_\comput u]{ x : \it M}{\comput u}{\it E}$, where $\refrel{\Pi_\val v}{\Xi_\val v}$ and $\refrel{\Pi_\comput u}{\Xi_\comput u}$.
Therefore, we can construct the following $\Xi$:

\begin{equation*}
    \begin{prooftree}
        \hypo{\itJugComput[\Xi_{\comput u}]{ x : \it M}{\comput u}{\it E}}
        \typeRuleAbs{\itJugValue{\emptyset}{\abs{x}{\comput u}}{\itSet{\it M \rightarrow \it E}}}
        \hypo{\itJugValue[\Xi_{\val v}]{\emptyset}{\val v}{\it M}}
        \typeRuleApp{\itJugComput{\emptyset}{\app{(\abs{x}{\comput u})}{\val v}}{\it E}}
    \end{prooftree}
\end{equation*}
Finally, is easy to see that $\refrel{\Pi}{\Xi}$.

\paragraph*{Case $\reductFix$.}
Let $\comput t, \comput t' \in \setComput$ such that $\comput t'$ is
typed and $\comput t \reductArr_\reductFix \comput t'$. By definition,
$\comput t \coloneqq \app{(\fix{x}{\val v})}{\val w}$ and
$\comput t'\coloneqq
\app{\val v\sub{x}{\fix{x}{\val v}}}{\val w}$. 

Since $\comput t$ is simply typable, $\Pi$ has the following shape.

\begin{equation*}
   \begin{prooftree}
        \hypo{
            \stJugValue[\Pi_\val v]
                                { x :: \st M \rightarrow \st U}
                                {\val v}
                                {\st M \rightarrow \st U}}
                    \typeRuleFix[1]{
                        \stJugValue[\Pi_{\fix*{}{}}]
                            {\emptyset}
                            {\fix{x}{\val v}}
                            {\st M \rightarrow \st U}}
        \hypo{\stJugValue{\emptyset}{\val w}{\st M}}
        \typeRuleApp{\stJugComput{\emptyset}{(\fix{x}{\val v})\val w}{\st U}}

    \end{prooftree} 
\end{equation*}

Since
$\comput t'$ is typable, there exists a derivation
$\itJugComput[\Xi']{\emptyset}
    {\app{\val v\sub{x}{\fix{x}{\val v}}}{\val w}}
    {\it E}$
for some refinement type $\it E$. By
case analysis, we deduce that $\Xi'$ is necessarily  the
following form:

\begin{equation*}
\begin{prooftree}
        \hypo{
            \itJugValue[\Xi'_{\val v}]
                {\emptyset}
                {\val v\sub{x}{\fix{x}{\val v}}}
                {\itSet*{\it M \rightarrow \it E}}}
        \hypo{
            \itJugValue[\Xi'_{\val w}]
                {\emptyset}
                {\val w}
                {\it M}}
    \typeRuleApp{
        \itJugComput
            {\emptyset}
            {\app{\val v\sub{x}{\fix{x}{\val v}}}{\val w}}
                {\it E}}
\end{prooftree}
\end{equation*}
Also, $\Pi'$ has to have the following shape
\begin{equation*}
\begin{prooftree}
        \hypo{
            \stJugValue[\Pi'_{\val v}]
                {\emptyset}
                {\val v\sub{x}{\fix{x}{\val v}}}
                {\st M \rightarrow \st U}}
        \hypo{
            \stJugValue[\Pi'_{\val w}]
                {\emptyset}
                {\val w}
                {\st M}}
    \typeRuleApp{
        \stJugComput
            {\emptyset}
            {\app{\val v\sub{x}{\fix{x}{\val v}}}{\val w}}
                {\st U}}
\end{prooftree}
\end{equation*}
Since by hypothesis $\Pi\rightsquigarrow \Pi'$, by definition of $\rightsquigarrow$ we also have that $\Pi'_\val v=\Pi_\val v\sub{x}{\Pi_{\fix*{}{}}}$ and $\refrel{\Pi'_\val w}{\Xi'_\val w}$.

We can then apply the anti-substitution lemma
(\Cref{lem:antiSubstLemma-ref}) on $\Pi'_{\val v}$. 
One then deduces that
there exist a refinement type $\it N$ and derivations
$\itJugValue[\Xi_{\val v}]
    { x : \it N }
    {\val v}{\itSet*{\it M\rightarrow \it E}}$
and $\itJugValue[\Pi_{\fix*{}{}}]
    {\emptyset}
    {\fix{x}{\val v}}{\it N}$,
where $\refrel{\Pi_\val v}{\Xi_\val v}$ and $\refrel{\Pi_{\fix*{}{}}}{\Xi_{\fix*{}{}}}$ .
Let $\Xi$ be the following typing derivation, if $\it N\neq \itSetEmpty$:

\begin{equation*}
\begin{prooftree}
            \hypo{
                \itJugValue[\Xi_{\val v}]
                    { x : \it N}
                    {\val v}{\itSet*{\it M\rightarrow \it E}}}

            \hypo{
                \itJugValue[\Xi_{\fix*{}{}}]
                    {\emptyset}
                    {\fix{x}{\val v}}{\it N}}
        \typeRuleFixRec{
            \itJugValue[\Xi^1_{\fix*{}{}}]
                {\emptyset}
                {\fix{x}{\val v}}
                {\itSet{\it M \rightarrow \it E}}}
        \hypo{
            \itJugValue[\Xi'_{\val w}]
                {\emptyset}
                {\val w}
                {\it M}}
    \typeRuleApp{
        \itJugComput
            {\emptyset}
            {\app{(\fix{x}{\val v})}{\val w}}
            {\it E}}
\end{prooftree}
\end{equation*}
Notice how by definition of the $\refrelsym$ on the fix rule, we obtain that 
$\refrel{\Pi_{\fix*{}{}}}{\Xi^1_{\fix*{}{}}}$, so finally $\refrel{\Pi}{\Xi}$.
Otherwise, if $\it N=\itSetEmpty$, we can conlude similarly, by using the rule \ruleNameFixBase.

\paragraph*{Case $\reductCase$.}
Let $\comput t, \comput t' \in \setComput$ such that $\comput t'$ is
typed and $\comput t \reductArr_\reductCase \comput t'$. By definition,
$\comput t \coloneqq \case{\int n}{\comput t_1, \ldots,
\comput t_m}$ and $\comput t' \coloneqq \comput t_n$ with $0 < n \leq m$.
Since $\comput t$ is simply typable, $\Pi$ has the following shape.
\begin{equation*}
    \begin{prooftree}
                        \hypo{\stJugValue
                            {\emptyset}{\int n}{\typeInt m}}
                        \hypo{(\Pi_k\triangleright \stJugComput
                            {\emptyset}{\comput t_k}{\st U}
                            )_{ 0<k\le m}}
                    \typeRuleCase{\stJugComput
                        {\emptyset}
                        {\case{\val v}{\comput t_1, \ldots, \comput t_m}}
                        {\st U}}
                \end{prooftree}
\end{equation*}
 
Since $\comput t'$ is typable, we have $\itJugComput[\Xi']{\emptyset}{\comput t_n}{\it E}$
for some type $\it E$. Furthermore, it exists a $\Pi'$ such that $\refrel{\stJugComput[\Pi']{\emptyset}{\comput t_n}{\st U}}{\Xi'}$ and $\Pi\rightsquigarrow \Pi'$.
Let $\Xi$ be the following typing derivation:

\begin{equation*}
    \begin{prooftree}
        \typeRuleInt{\itJugValue{\emptyset}{\int n}{\itSet{\typeInt{n}}}}
        
        \hypo{\itJugComput[\Pi_n]{\emptyset}{\comput t_n}{\it E}}
                \hypo{0 < n \leq m}

        \typeRuleCase[3]{\itJugComput{\emptyset}{\case{\int n}{\comput t_1, \ldots, \comput t_m}}{\it E}}
    \end{prooftree}
\end{equation*}
It is clear that $\refrel{\Pi}{\Xi}$.

\paragraph*{Case $\reductLetRet$.}
Let $\comput t, \comput t' \in \setComput$ such that $\comput t'$ is
typed and $\comput t \reductArr_\reductLetRet \comput t'$. By
definition, $\comput t \coloneqq
\letin{x}{\ret{\val v}}{\comput u}$ and $\comput t'
\coloneqq \comput u\sub{x}{\val v}$. The derivation $\Pi$ is in the following form:
\begin{equation*}
    \begin{prooftree}
        \hypo{\itJugValue[\Pi_\val v]{\emptyset}{\val v}{\st M}}
        \typeRuleRet{\stJugComput{\emptyset}{\ret{\val v}}{\itEffectReturn{\st M}}}
        \infer0{\itEffectReturn{\st M} \replaceLeaf{\st M \rightarrow \st F} \st F}
        \hypo{\itJugComput[\Pi_{\comput u}]{ x : \st M}{\comput u}{\st F}}
    \typeRuleLetin{\itJugComput{\emptyset}{\letin{x}{\ret{\val v}}{\comput u}}{\st F}}
    \end{prooftree}
\end{equation*}

Also,
there exist derivations $\refrel{\stJugComput[\Pi']{\emptyset}{\comput u\sub{x}{\val v}}{\st F}}{\itJugComput[\Xi']{\emptyset}{\comput u\sub{x}{\val v}}{\it F}}$.
Since by hypothesis $\Pi\rightsquigarrow \Pi'$, by definition of $\rightsquigarrow$ we have that $\Pi'=\Pi_\comput u\sub{x}{\Pi_\val v}$.
By anti-substitution lemma (\Cref{lem:antiSubstLemma-ref}), there exist a refinement type $\it M$
and derivations $\itJugValue[\Xi_\val v]{\emptyset}{\val v}{\it M}$
and $\itJugComput[\Xi_{\comput u}]{ x:\it M }{\comput u}{\it F}$,
where $\refrel{\Pi_\val v}{\Xi_\val v}$ and $\refrel{\Pi_\comput u}{\Xi_\comput u}$.
Thus, let $\Xi$ be the following derivation :
\begin{equation*}
\begin{prooftree}
        \hypo{\itJugValue[\Xi_\val v]{\emptyset}{\val v}{\it M}}
        \typeRuleRet{\itJugComput{\emptyset}{\ret{\val v}}{\itEffectReturn{\it M}}}
        \infer0{\itEffectReturn{\it M} \replaceLeaf{\it M \rightarrow \it F} \it F}
        \hypo{\itJugComput[\Xi_{\comput u}]{ x : \it M}{\comput u}{\it F}}
    \typeRuleLetin{\itJugComput{\itJugComput[\Xi']{\emptyset}{\comput u\sub{x}{\val v}}{\it F}}{\letin{x}{\ret{\val v}}{\comput u}}{\it F}}
\end{prooftree}
\end{equation*}
Finally, we can see that $\refrel{\Pi}{\Xi}$.

\paragraph*{Case $\reductLetEff$.}
Let $\comput t, \comput t' \in \setComput$ such that $\comput t'$ is
typed and $\comput t \reductArr_\reductLetEff \comput t'$. By
definition, $\comput t \coloneqq
\letin{x}{\effect{\val v}{y}{\comput s}}{\comput u}$ and
$\comput t' \coloneqq
\effect{\val v}{y}{\letin{x}{\comput s}{\comput t}}$. There exists a derivation
$\itJugComput[\Xi']{\emptyset}{\effect{\val v}{y}{\letin{x}{\comput s}{\comput t}}}{\it E}$
for some type $\it E$. By
case analysis, we deduce that $\Xi'$ is necessarily of the following
form:

\begin{equation*}
\begin{prooftree}
    \hypo{\itJugValue[\Xi_{\val v}]{\emptyset}{\val v}{\it M}}
    \hypo{A_i}
    \delims{\left(}{\right)_{\rangeI}}
\typeRuleEff[2]{
    \itJugComput
        {\emptyset}
        {\effect{\val v}{y}{\letin{x}{\comput s}{\comput u}}}
        {\itEffect
                {\sigma}{\it M}
                {\it N_i \rightarrow \it E_i}_{\rangeI}}}
\end{prooftree}
\end{equation*}
\begin{equation*}
    (A_i) \coloneqq
    \begin{prooftree}
        \hypo{\itJugComput[\Xi^i_{\comput s}]{ y : \it N_i}{\comput s}{\it G_i}}
    \hypo{B_i}
    \delims{\left(}{\right)_{\rangeI}}

        \hypo{(\itJugComput[\Xi_k^i]{x : \it M_k^i}{\comput u}{\it F_k^i})_{\rangeK_i}}
    \typeRuleLetin{
        \itJugComput
            { y : \it N_i}
            {\letin{x}{\comput s}{\comput u}}
            {\it E_i}}
    \end{prooftree}
\end{equation*}
\begin{equation*}
    (B_i) \coloneqq
    \begin{prooftree}
       \hypo{\varphi_i}
        \ellipsis{}{
            \it G_i
            \replaceLeaf{\it M_k^i \rightarrow \it F_k^i}[\rangeK_i]
            \it E_i}
    \end{prooftree}
\end{equation*}
Since $y$ is not a free variable of $\comput u$, by context simplification (\Cref{rem:ContextSimplification}) on each $\Xi_k^i$ one can derive $\itJugComput[\Xi_k^i]{x : \it M_k^i}{\comput u}{\it F_k^i}$.
Furthermore, it exists a $\stJugComput[\Pi']{\emptyset}{\effect{\val v}{y}{\letin{x}{\comput s}{\comput u}}}{\itEffect
                {\sigma}{\st M}
                {\st N \rightarrow \st E}}$.
Let $\Xi$ be the following typing derivation:
\begin{equation*}
    \begin{prooftree}
            \hypo{(A)}
            \hypo{(B)}
            \hypo{(\itJugComput[\Xi_k^i]
                { x : \it M_k^i}
                {\comput u}{\it F_k^i})_{\rangeK_i, \rangeI}}
        \typeRuleLetin{
            \itJugComput
                {\emptyset}
                {\letin{x}{\effect{\val v}{y}{\comput s}}{\comput u}}
                {\itEffect
                    {\sigma}{\it M}
                    {\it N_i \rightarrow \it E_i}_{\rangeI}}
        }
    \end{prooftree}
\end{equation*}
\begin{equation*}
    (A) \coloneqq
    \begin{prooftree}
            \hypo{
                    \itJugValue[\Xi_{\val v}]
                        {\emptyset}
                        {\val v}
                        {\it M}
                }
                \hypo{(
                    \itJugComput[\Xi^i_{\comput s}]
                        { y : \it N_i}
                        {\comput s}
                        {\it G_i}
                )_{\rangeI}}
            \typeRuleEff[2]{
                \itJugComput
                    {\emptyset}
                    {\effect{\val v}{y}{\comput s}}
                    {\itEffect{\sigma}{\it M}{\it N_i \rightarrow \it G_i}_{\rangeI}}
            }
    \end{prooftree}
\end{equation*}
\begin{equation*}
    (B) \coloneqq
    \begin{prooftree}
            \hypo{\varphi_i}
            \ellipsis{}{
                \it G_i
                \replaceLeaf{\it N_i \rightarrow \it F_k^i}[\rangeK_i]
                \it E_i}
            \delims{\left(}{\right)_{\rangeI}}
        \infer1{
            \itEffect
                    {\sigma}{\it M}
                    {\it N_i \rightarrow \it G_i}_{\rangeI}
        \replaceLeaf{
            \it M_k^i \rightarrow 
            \it F_k^i}[\rangeK_i, \rangeI]
            \itEffect
                {\sigma}{\it M}
                {\it N_i \rightarrow \it E_i}_{\rangeI}}
    \end{prooftree}
\end{equation*}
By looking at the definition of this reduction case in \Cref{def:HEBDerivationReductionContd},  from $\refrel{\Pi'}{\Xi'}$ we obtain easily that $\refrel{\Pi}{\Xi}$. In particular, notice that the subderivation $\Pi_\comput u$ of $\Pi'$ is such that $\refrel{\Pi_\comput u^{(y::\st N)}}{\Xi^i_k}$ for each $\rangeI$, $\rangeK$. By \Cref{prop:HEBClosedRefinement}, we obtain $\refrel{\Pi_\comput u}{\Xi^i_k}$.

\paragraph*{Case $\reductHdlRet$.}
Let $\comput t, \comput t' \in \setComput$ such that $\comput t'$ is
typed and $\comput t \reductArr_\reductHdlRet \comput t'$. By
definition, $\comput t \coloneqq
\handle{\handler h}{\ret{\val v}}$ and $\comput t'
\coloneqq \comput u\sub{x}{\val v}$ with
$\handler h = \{\retClause{x}{\comput u}\} \cup \handler h'$. 

Since $\comput t$ is simply typable, by case analysis its derivation $\Pi$ has the following form. 
\begin{equation*}
    \begin{prooftree}
            \hypo{\itJugComput[\Pi_{\comput u}]{ x : \st M}{\comput u}{\st E}}
            \hypo{\{\retClause{x}{\comput u}\} \handlerCompat \handler h}
            \typeRuleHandlerRet[2]{\itJugHandler{\emptyset}{\{\retClause{x}{\comput u}\} \cup \handler h}{\itEffectReturn{\st M} \Rightarrow \st E}}
            \hypo{\itJugValue[\Pi_\val v]{\emptyset}{\val v}{\st M}}
            \typeRuleRet{\itJugComput{\emptyset}{\ret{\val v}}{\itEffectReturn{\st M}}}
        \typeRuleHandle{\itJugComput{\emptyset}{\handle{\{\retClause{x}{\comput u}\}\cup \handler h }{\ret{\val v}}}{\st E}}
    \end{prooftree}
\end{equation*}
There exist two derivations 
$\itJugComput[\Xi']{\emptyset}{\comput u\sub{x}{\val v}}{\it E}$ and $\stJugComput[\Pi']{\emptyset}{\comput u\sub{x}{\val v}}{\st E}$ such that $\refrel{\Pi'}{\Xi'}$. Since by hypothesis $\Pi\rightsquigarrow \Pi'$, by definition of $\rightsquigarrow$ is clear that $\Pi'=\Pi_\comput u\sub{x}{\Pi_\val v}$.
By anti-substitution lemma (\Cref{lem:antiSubstLemma-ref}), there exist
a type $\it M$ and derivations 
$\itJugValue[\Xi_\val v]{\emptyset}{\val v}{\it M}$
and $\itJugComput[\Xi_{\comput u}]{x : \it M}{\comput u}{\it E}$,
where $\refrel{\Pi_\comput u}{\Xi_\comput u}$ and $\refrel{\Pi_\val v}{\Xi_\val v}$.
Therefore, let $\Xi$ be the following derivation :
\begin{equation*}
    \begin{prooftree}
            \hypo{\itJugComput[\Pi_{\comput u}]{ x : \it M}{\comput u}{\it E}}
            \hypo{\{\retClause{x}{\comput u}\} \handlerCompat \handler h}
            \typeRuleHandlerRet[2]{\itJugHandler{\emptyset}{\{\retClause{x}{\comput u}\}\cup \handler h }{\itEffectReturn{\it M} \Rightarrow \it E}}
            \hypo{\itJugValue[\Pi_\val v]{\emptyset}{\val v}{\it M}}
            \typeRuleRet{\itJugComput{\emptyset}{\ret{\val v}}{\itEffectReturn{\it M}}}
        \typeRuleHandle{\itJugComput{\emptyset}{\handle{\{\retClause{x}{\comput u}\} \cup \handler h}{\ret{\val v}}}{\it E}}
    \end{prooftree}
\end{equation*}
Finally, is easy to notice that $\refrel{\Pi}{\Xi}$.

\paragraph*{Case $\reductHdlEff$.}
Let $\comput t, \comput t' \in \setComput$ such that $\comput t'$ is
typed and $\comput t \reductArr_\reductHdlEff \comput t'$.

By definition,
$\comput t \coloneqq
\handle{\handler h}{\effect{\val v}{x}{\comput u}}$ and
$\comput t' \coloneqq
\comput s\sub{x}{\val v}\sub{r}{\abs{y}{\handle{\handler h}{\comput u}}}$
with $\effectClause[\sigma]{x}{r}{\comput s} \in
\handler h$. There exists two derivations
$\itJugComput[\Xi']{\emptyset}{\comput s\sub{x}{\val v}\sub{r}{\abs{y}{\handle{\handler h}{\comput u}}}}{\it E}$ and $\stJugComput[\Pi']{\emptyset}{\comput s\sub{x}{\val v}\sub{r}{\abs{y}{\handle{\handler h}{\comput u}}}}{\st E}$ such that $\refrel{\Pi'}{\Xi'}$. 

The derivation $\Pi$ has the following form. We set $\handler h = \{\effectClause{x}{r}{\comput s}\} \cup \handler h'$.
\begin{equation*}
    \begin{prooftree}
        \hypo{
                \stJugComput[\Pi_{\comput s}]
                    {   x :: \st M ;
                        r :: \st N \rightarrow \st G}
                    {\comput s}
                    {\st E}}

            \hypo{
               (A)
            }
        \typeRuleHandler{
            \stJugHandler
                {\emptyset }
                {\{\effectClause{x}{r}{\comput s}\} \cup \handler h'}
                {\itEffect{\sigma}{\it M}
                        {\st N \rightarrow \st F}
                \Rightarrow \st E}}
                \hypo{\stJugValue[\Pi_{\val v}]
                    {\emptyset}
                    {\val v}{\st M}}
                \hypo{
                    \stJugComput[\Pi_{\comput u}]
                        {y :: \st N}
                        {\comput u}
                        {\st F}}
            \typeRuleEff{
                \stJugComput
                    { \emptyset}
                    {\effect{\val v}{y}{\comput u}}
                    {\itEffect
                        {\sigma}{\st M}
                        {\st N \rightarrow \st F}}}
        \typeRuleHandle{
            \stJugComput
                {\emptyset}
                {\handle
                    {\{\effectClause[\sigma]{x}{r}{\comput s}\} \cup \handler h'}
                    {\effect{\val v}{y}{\comput u}}}
                {\st E}}
    \end{prooftree}
\end{equation*}
\[
(A)::=  \begin{matrix}
                \{\effectClause{x}{r}{\comput s}\} \handlerCompat \handler h'\\
                (
                \itJugHandler[\Pi_{\handler h'}]
                    {\emptyset}
                    {\handler h'}
                    {\st F \Rightarrow \st G}
                    )
            \end{matrix}  
\]

Since by hypothesis $\Pi\rightsquigarrow \Pi'$, by definition of $\rightsquigarrow$ we have that $\Pi'=\Pi_\comput s\sub{x}{\Pi_\val v}\sub{r}{\Pi_\val{\lambda y}}$, with $\stJugValue[\Pi_\val{\lambda y}]{\emptyset}{\lambda y.\handle{\handler h}{\comput u}}{\st N\rightarrow \st G}$ being the following derivation, as defined in \Cref{def:HEBDerivationReductionContd}.

\begin{equation*}
    \begin{prooftree}
                \hypo{
                    \stJugHandler[\Pi_{\handler h'}^{(y::\st N)(\sigma)}]
                        {y::\st N}
                        {\handler h}
                        {\st F
                            \Rightarrow
                        \st G}}
                \hypo{
                    \stJugComput[\Pi_{\comput u}]
                        { y :: \st N}
                        {\comput u}
                        {\st F}}
            \typeRuleHandle{
                \stJugComput
                    { y :: \st N}
                    {\handle{\handler h}{\comput u}}
                    {\st G}
            }
        \typeRuleAbs{
            \stJugValue
                {\emptyset}
                {\abs{y}{\handle{\handler h}{\comput u}}}
                {\st N \rightarrow \st G}
        }
    \end{prooftree}
\end{equation*}

We also have $\stJug[\Pi_\comput s\sub{x}{\Pi_\val{v}}]{\Delta_\comput s;r::\st N\rightarrow \st G}{\comput s \sub{x}{\val v}}{\st E}$.
We can then apply the anti-substitution lemma (\Cref{lem:antiSubstLemma-ref}) on $\Xi'$ with $\Pi_\comput s\sub{x}{\Pi_\val v}$ and 
$\Pi_\val{\lambda y}$, therefore there exist a type $\it L$ and derivations 
$\itJugValue[\Xi_{\val {\lambda y}}]{\emptyset}{\abs{y}{\handle{\handler h}{\comput u}}}{\it L}$ and 
$\itJugComput[\Xi_{\comput s}]{ r : \it L}{\comput s\sub{x}{\val v}}{\it E}$,
where $\refrel{\Pi_\comput s\sub{x}{\Pi_\val v}}{\Xi_{\comput s}}$ and $\refrel{\Pi_\val{\lambda y}}{\Xi_{\val {\lambda y}}}$.
By case analysis, we deduce that $\Xi_{\val {\lambda y}}$ is necessarily of the following
form:

\begin{equation*}
    \begin{prooftree}
                \hypo{
                    \itJugHandler[\Xi_{\handler h}^i]
                        {y:\it N_i^\handler h}
                        {\handler h}
                        {\it F_i
                            \Rightarrow
                        \it G_i}}
                \hypo{
                    \itJugComput[\Xi_{\comput u}^i]
                        { y : \it N_i}
                        {\comput u}
                        {\it F_i}}
            \typeRuleHandle{
                \itJugComput
                    { y : \it N_i}
                    {\handle{\handler h}{\comput u}}
                    {\it G_i}
            }
            \delims{\left(}{\right)_{\rangeI}}
        \typeRuleAbs{
            \itJugValue
                {\emptyset}
                {\abs{y}{\handle{\handler h}{\comput u}}}
                {\itSet*{\it N_i \rightarrow \it G_i}_{\rangeI}}
        }
    \end{prooftree}
\end{equation*}

where $\it L = \itSet{\it N_i \rightarrow \it G_i}_{\rangeI}$. Therefore, using \Cref{prop:HEBClosedRefinement}, we obtain $\refrel{\Pi^{(\sigma)}_{\handler h'}}{\Xi^i_\handler h}$ and $\refrel{\Pi_\comput u}{\Xi^i_\comput u}$ for each $\rangeI$.
Since $y$ is not a free variable of $\handler h$, by context simplification (\Cref{rem:ContextSimplification}) on $\Xi_{\handler h}^i$ one can derive 
$\itJugHandler[\Xi^i_{\handler h}]
    {\emptyset}
    {\handler h}
    {\it F_i \Rightarrow \it G_i}$.
Finally, by \Cref{lem:HEBHandlerElim} on $\Pi^{(\sigma)}_{\handler h'}$ and each $\Xi^i_\handler h$, we obtain derivations $\Psi^i_{\handler h'}$ such that $\refrel{\Pi_{\handler h'}}{\Psi^i_{\handler h'}}$. 
We can now apply again the anti-substitution lemma (\Cref{lem:antiSubstLemma-ref}) on $\Xi_{\comput s}$ with $\Pi_\comput s$ and $\Pi_\val v$, therefore there exist a type $\it M$ and derivations 
$\itJugValue[\Xi_\val v]{ \emptyset}
    {\val v}{\it M}$
and $\itJugComput[\Xi_{\comput s}]
    { r : \itSet*{\it N_i \rightarrow \it G_i}_{\rangeI}; x : \it M}
    {\comput s}{\it E}$,
where $\refrel{\Pi_\comput s}{\Xi_\comput s}$ and $\refrel{\Pi_\val v}{\Xi_\val v}$. 
Let $\Xi$ be the following derivation :

\begin{equation*}
    \begin{prooftree}
            \hypo{(B)}
                \hypo{\itJugValue[\Xi_{\val v}]
                    {\emptyset}
                    {\val v}{\it M}}
                \hypo{\left(
                    \itJugComput[\Xi_{\comput u}^i]
                        {y : \it N_i}
                        {\comput u}
                        {\it F_i}
                \right)_{\rangeI}}
            \typeRuleEff{
                \itJugComput
                    { \emptyset}
                    {\effect{\val v}{y}{\comput u}}
                    {\itEffect
                        {\sigma}{\it M}
                        {\it N_i \rightarrow \it F_i}_{\rangeI}}}
        \typeRuleHandle{
            \itJugComput
                {\emptyset}
                {\handle
                    {\{\effectClause[\sigma]{x}{r}{\comput s}\} \cup \handler h'}
                    {\effect{\val v}{y}{\comput u}}}
                {\it E}}
    \end{prooftree}
\end{equation*}

\begin{equation*}
    (B) \;\coloneqq\;
    \begin{prooftree}
            \hypo{
                \itJugComput[\Xi_{\comput s}]
                    {   x : \it M ;
                        r : \itSet{\it N_i \rightarrow \it G_i}_{\rangeI}}
                    {\comput s}
                    {\it E}}

            \hypo{
                \begin{matrix}
                \{\effectClause{x}{r}{\comput s}\} \handlerCompat \handler h'\\
                (
                \itJugHandler[\Psi_{\handler h'}^i]
                    {\emptyset}
                    {\handler h'}
                    {\it F_i \Rightarrow \it G_i}
                    )_{\rangeI}
            \end{matrix}  
            }
        \typeRuleHandler{
            \itJugHandler
                {\emptyset }
                {\{\effectClause{x}{r}{\comput s}\} \cup \handler h'}
                {\itEffect{\sigma}{\it M}
                        {\it N_i \rightarrow \it F_i}_{\rangeI}
                \Rightarrow \it E}}
    \end{prooftree}
\end{equation*}

One can conclude this case by noticing that 
$\itJugComput[\Xi]{\Gamma}{t}{\it E}$ and $\refrel{\Pi}{\Xi}$.

%% file: Definitions/AlgoLeafLogic.tex
\begin{tabular}{c}
    \hspace{-0.5cm}
\framebox{$
\begin{array}{ccc}    
    \mathcal B\Biggl(
        {\itEffectReturn{\it M}},{\it M_1,\dots,\it M_{|I|}},{\it G_1,\dots,\it G_{|I|}},{\it F}
    \Biggl)
    &=&
     (|I|=1)\wedge(\it M=\it M_1) \wedge(\it F =\it G_1)\\[0.5cm]

    \mathcal B\Biggl(
        {\itEffect{\sigma}{\it M}{\it N_j\rightarrow\it E'_j}_{\rangeJ}},{\it M_1,\dots,\it M_{|I|}},{\it G_1,\dots,\it G_{|I|}},{\it F}
    \Biggl)
    &=&\begin{array}{c}
        \it F=\itEffect{\sigma'}{\it M'}{\it N'_k\rightarrow\it F'_k}_{\rangeK} \wedge\\ (\sigma=\sigma')\wedge\\
        (|J|=|K|) \wedge (\it M=\it M') \wedge\\
         \forall \rangeJ.(\it N_j=\it N'_j) \wedge\\
        \forall\rangeJ.\, \exists K_j\subseteq I. I=\bigcup\limits_{\rangeJ}K_j\\
         \bigwedge_{\rangeJ}\algoleaf{\it E'_j}{\it M_1,\dots,\it M_{|K_j|}}{\it G_1,\dots,\it G_{|K_j|}}{\it F'_j} \\
         
    \end{array}\\[0.5cm]

\end{array}
$}
\end{tabular}

%% file: Definitions/AlgorithmLogic.tex
\begin{tabular}{c}
    \hspace{-0.7cm}
\framebox{$
\begin{array}{ccc}    

    \algosym\Biggl(
        \begin{array}{c}
        \begin{prooftree}
                        \hypo{\stJugComput[\Pi_1]{\Delta}{\comput t}{\st E}}
                        \hypo{
                            \dots}
                        \hypo{
                            \stJugComput[\Pi_2]
                                {\Delta; x :: \st M}
                                {\comput u}
                                {\st G}
                            }
                    \infer3{
                        \stJugComput
                            {\Delta}
                            {\letin{x}{\comput t}{\comput u}}
                            {\st F}
                    }
                \end{prooftree}
    ,\\[0.5cm]
    \qquad\qquad\qquad\qquad\itJugComput
                    {\Gamma}
                    {\letin{y}{\comput t}{\comput u}}
                    {\it F}
        \end{array}
    \Biggl)
    &=&\begin{array}{c}
        \exists\refrel{\st E}{\it E}.\, \exists I.\, |I|\leq |\it F|\\
        \wedge \forall i\in I.\,\exists\refrel{\st M}{\it M_i}.\, \exists\refrel{\st G}{\it G_i}.\\
        \exists \Gamma',\Gamma_i\subseteq \Gamma.\,  (\Gamma=\Gamma'+_{\rangeI}\Gamma_i)\\
        \land\algo{\Pi_1}{\itJugComput{\Gamma'}{\comput t}{\it E}}\\
        \bigwedge_{\rangeI}\algo{\Pi_2}{\itJugComput{\Gamma_i;x:\it M_i}{\comput u}{\it G_i}}\\
        \wedge\, \algoleaf{\it E}{\it M_1,\dots,\it M_{|I|}}{\it G_1,\dots,\it G_{|I|}}{\it F} \\
        
    \end{array}\\[1.7cm]

    \algosym\Biggl(
        \begin{array}{c}
             \begin{prooftree}
                        \hypo{\stJugValue[\Pi_1]{\Delta}{\val v}{\typeInt{n}}}
                        \hypo{(\stJugComput[\Pi_2]{\Delta}{\comput t_i}{\st E}
                            )_{ 0<i\le n}}
                    \infer2{\stJugComput{\Delta }{\case{\val v}{\comput t_1, \ldots, \comput t_n}}{\st E}}
                \end{prooftree}
    ,\\[0.5cm]
   \qquad\qquad\qquad\itJugComput{\Gamma}
        {\case{\val v}{\comput t_1, \ldots, \comput t_n}}{\it E}
        \end{array}    
    \Biggl)
    &=&\begin{array}{c}
        \exists \refrel{\typeInt n}{\itSet{\typeInt{m}}} .\,\exists\Gamma_1,\Gamma_2\subseteq \Gamma.\,\\
         (\Gamma=\Gamma_1+\Gamma_2)\\
        \land\algo{\Pi_1}{\itJugValue{\Gamma_1}{\val v}{\itSet{m}}}\\
        \wedge \algo{\Pi_2}{\itJugComput{\Gamma_2}{\comput t_m}{\it E}} \\
        
    \end{array}\\[1cm]

    \algosym\Biggl(
         \begin{array}{c}
            \begin{prooftree}
                        \hypo{\stJugValue[\Pi_1]{\Delta}{\val v}{\st M}}
                    \infer1{\stJugComput{\Delta}{\ret{\val v}}{\itEffectReturn{\st M}}}
                \end{prooftree}
    ,\\
   \qquad\qquad\itJugComput{\Gamma}
        {\ret{\val v}}{\itEffectReturn{\it M}}
    \end{array}\Biggl)
    &=&\begin{array}{c}
        \algo{\Pi_1}{\itJugValue{\Gamma}{\val v}{\it M}} 
    \end{array}\\[1cm]

    \algosym\Biggl(
        \begin{array}{c}
           \begin{prooftree}
                        \hypo{
                            \stJugValue[\Pi_1]
                                {\Delta}
                                {\val v}
                                {\st M}
                        }
                        \hypo{
                            \stJugComput[\Pi_2]
                                {\Delta; x :: \st N}
                                {\comput t}
                                {\st E}
                        }
                    \infer2{
                        \stJugComput
                            {\Delta}
                            {\effect{\val v}{x}{\comput t}}
                            {\itEffect{\sigma}{\st M}{\st N \rightarrow \st E}}
                        }
                \end{prooftree}
    ,\\[0.5cm]
   \qquad\itJugComput
            {\Gamma}
            {\effect{\val w}{y}{\comput t}}
            {\itEffect{\sigma}{\it M}{\it N_i \rightarrow \it E_i}_{\rangeI}} 
        \end{array}
    \Biggl)
    &=&\begin{array}{c}
        \forall i\in I.\,\exists \Gamma',\Gamma_i\subseteq \Gamma.\\ (\Gamma=\Gamma'+_{\rangeI}\Gamma_i)\\
        \land\algo{\Pi_1}{\itJugValue{\Gamma'}{\val v}{\it M}}\\
        \bigwedge_{i\in I}\algo{\Pi_2}{\itJugComput{\Gamma_i;x:\it N_i}{\comput t}{\it E_i}}\\
        
    \end{array}\\[1cm]

    \algosym\Biggl(
        \begin{array}{c}
           \begin{prooftree}
                \hypo{\stJugComput[\Pi_1]
                    {\Delta;
                        y :: \st M}
                    {\comput t}
                    {\st E}}
                
            \infer1{\stJugHandler
                {\Delta}
                {\{\retClause{y}{\comput t}\}}%
                {\itEffectReturn{\st M} \Rightarrow \st E}
            }%
        \end{prooftree}
    ,\\[0.5cm]
   \quad\itJugComput{\Gamma}
        {\{\retClause{y}{\comput t}\}}{\itEffectReturn{\it M}\Rightarrow \it E} 
        \end{array}
    \Biggl)
    &=&\begin{array}{c}
        \algo{\Pi_1}{\itJugComput{\Gamma;y:\it M}{\comput t}{\it E}} 
    \end{array}\\[1cm]

\end{array}
$}
\end{tabular}

%% file: Proofs/AlgoSoundness.tex
By induction on $\Pi$.

\paragraph*{Case Integers} 
$\Pi$ has the following form:
    \[
         \begin{prooftree}
                    \hypo{0<n\le m}
                    \typeRuleInt[1]{\stJugValue{\Delta}{\int n}{\typeInt{m}}}
                \end{prooftree}
    \]
    We have $J_r\coloneq\itJugValue{\Gamma}{\int n}{\it M}$ with $\it M\in\itSet{\itSetEmpty,\itSet{\typeInt{i}}}$. Since $\models\algo{\Pi}{J_r}$, we have that $\Gamma=\emptyset$ and either $\it M=\itSetEmpty$ or $i=n$. We can construct the derivation $\Xi$ as follows:
    \[
         \begin{prooftree}
            \typeRuleVar{\itJugValue{\emptyset}{\int n}{\it M}}
        \end{prooftree}
    \]
    Finally, $\refrel{\Pi}{\Xi}$.

\paragraph*{Case Variable}
$\Pi$ has the following form:
    \[
         \begin{prooftree}
            \typeRuleVar{\stJugValue{\Delta;x::\st M}{x}{\st M}}
        \end{prooftree}
    \]
    We have $J_r\coloneq\itJugValue{\Gamma;x:\it M'}{x}{\it M}$, and since $\refrel{J_s}{J_r}$, we have that $\refrel{\st M}{\it M}$ and $\refrel{\st M}{\it M'}$. Moreover, since $\models\algo{\Pi}{J_r}$, then $\Gamma=\emptyset$ and $\it M=\it M'$.
    We then can construct the derivation $\Xi$ as follows:
    \[
         \begin{prooftree}
            \typeRuleVar{\itJugValue{x:\it M}{x}{\it M}}
        \end{prooftree}
    \]
    Finally, $\refrel{\Pi}{\Xi}$.       
\paragraph*{Case Abstraction} 
$\Pi$ has the following form : 

\begin{equation*}
     \begin{prooftree}
                        \hypo{\stJugComput[\Pi']{\Delta; x :: \st M}{\comput t}{\st E}}
                    \typeRuleAbs{\stJugValue{\Delta}{\abs{x}{\comput t}}{\st M \rightarrow \st E}}
                \end{prooftree}
\end{equation*}
Being $J_s$ the conclusion of $\Pi$, $J_r$ is such that $\refrel{J_s}{J_r\coloneq\itJugValue{\Gamma}{\abs{x}{\comput t}}{\itSet*{\it M_i \rightarrow \it E_i}_{\rangeI}}}$, therefore $\refrel{\Delta}{\Gamma}$, $\refrel{\st M}{\it M_i}$ and $\refrel{\st E}{\it E_i}$ for each $\rangeI$. Let $J'_s$ be the conclusion of $\Pi'$.
Since $\models\algo{\Pi}{J_r}$, there exist $\Gamma_i\subseteq\Gamma$ such that $\Gamma=+_{\rangeI}\Gamma_i$. We set $J^i_r\coloneq\itJugComput{\Gamma_i;x:\it M_i} {\comput t}{\it E_i}$ and have that $\models\algo{\Pi'}{J^i_r}$ for each $\rangeI$. Notice also that $\refrel{J'_s}{J^i_r}$, therefore, by induction hypothesis, we have a derivations $\refrel{\Pi'}{\Xi_i\triangleright J^i_r}$ for each $\rangeI$. We can then construct the following derivation :

\begin{equation*}
     \begin{prooftree}
                        \hypo{(\itJugComput[\Xi_i]{\Gamma_i;x:\it M_i}{\comput t}{\it E_i})_{\rangeI}}
                    \typeRuleAbs{\itJugValue{+_{\rangeI}\Gamma_i}{\abs{x}{\comput t}}{\itSet*{\it M \rightarrow \it E}_{\rangeI}}}
                \end{prooftree}
\end{equation*}
Finally, we have $\refrel{\Pi}{\Xi}$.

\paragraph*{Case Fixpoint} 
$\Pi$ has the following form :

    \begin{equation*}
     \begin{prooftree}
                    \hypo{
                        \stJugValue[\Pi']
                            {\Delta; x :: \st M \rightarrow \st E}
                            {\val v}
                            {\st M \rightarrow \st E}}
                    
                \typeRuleFix[1]{
                    \stJugValue
                        {\Delta}
                        {\fix{x}{\val v}}
                        {\st M \rightarrow \st E}}
            \end{prooftree}
    \end{equation*}

Being $J_s$ the conclusion of $\Pi$, $J_r$ is such that $\refrel{J_s}{J_r\coloneq \itJugValue{\Gamma}{\fix{x}{\val v}}{\itSet*{\it M_i\rightarrow \it E_i}_{\rangeI}}}$, so $\refrel{\Delta}{\Gamma}$, $\refrel{\st M\rightarrow \st E}{\itSet{\it M_i\rightarrow \it E_i}_{\rangeI}}$. Let $J'_s$ be the conclusion of $\Pi'$. Since $\models\algo{\Pi}{J_r}$, there exists an integer $n$, some contexts $\Gamma_j\subseteq \Gamma$ for $j\leq n$ such that $\Gamma=+_{j\leq n}\Gamma_j$ and some \HEBI\ types $\it M_j$ for $j\leq n+1$ such that $\refrel{\st M \rightarrow \st E}{\it M_i}$, $\it M_0=\itSetEmpty$ and $\it M_{n+1}=\itSet*{\it M_i\rightarrow \it E_i}_{\rangeI}$. We set $J^j_r\coloneq \itJugValue{\Gamma_j;x:\it M_j}{\val v}{\it M_{j+1}}$ and have that $\models\algo{\Pi_1}{J^j_r}$ for each $j\leq n$. Notice also that $\refrel{J'_s}{J^j_r}$, therefore, by induction hypothesis, we obtain derivations $\refrel{\Pi'}{\Xi_j\triangleright J^j_r}$ for each $j\leq n$. 
By \Cref{lem:RefinementFixSoundness} it exists a derivation $\itJugValue[\Omega]{+_{0\leq j\leq n}\Gamma_j}{\fix{x}{\val v}}{\it M_{n+1}}$ and $\refrel{\Pi}{\Omega}$.

\paragraph*{Case Application}
$\Pi$ has the following form:
\[
 \begin{prooftree}
                        \hypo{\stJugValue[\Pi_1]{\Delta}{\val v}{\st M \rightarrow \st E}}
                        \hypo{\stJugValue[\Pi_2]{\Delta}{\val w}{\st M}}
                    \typeRuleApp{\stJugComput{\Delta}{\app{\val v}{\val w}}{\st E}}
                \end{prooftree}
\]
Being $J_s$ the conclusion of $\Pi$, $J_r$ is such that $\refrel{J_s}{J_r\coloneq\itJugValue{\Gamma}{\app{\val v}{\val w}}{ \it E}}$, so $\refrel{\st E}{\it E}$ and $\refrel{\Delta}{\Gamma}$. Let $J^1_s$ and $J^2_s$ be the conclusions respectively of $\Pi_1$ and $\Pi_2$. Since $\models\algo{\Pi}{J_r}$, there exists two context $\Gamma_1,\Gamma_2\subseteq\Gamma$ such that $\Gamma=\Gamma_1+\Gamma_2$ and an $\it M$ such that $\refrel{\st M}{\it M}$. We set $J^1_r\coloneq\itJugValue{\Gamma_1}{\val v}{\it M \rightarrow\it E}$ and $J^1_r\coloneq\itJugValue{\Gamma_2}{\val w}{\it M}$, having $\models\algo{\Pi_1}{J^1_r}$ and $\models\algo{\Pi_2}{J^2_r}$. 
By \Cref{lem:ContextRefinement} we obtain $\refrel{J^1_s}{J^1_r}$ and $\refrel{J^2_s}{J^2_r}$. Therefore, by induction hypothesis, we have derivations $\refrel{\Pi_1}{\Xi_1\triangleright J^1_r}$ and $\refrel{\Pi_2}{\Xi_2\triangleright J^2_r}$.  
We can then construct $\Xi$ as follows:
\[
 \begin{prooftree}
                        \hypo{\itJugValue[\Xi_1]{\Gamma_1}{\val v}{\it M \rightarrow \it E}}
                        \hypo{\itJugValue[\Xi_2]{\Gamma_2}{\val w}{\it M}}
                    \typeRuleApp{\itJugComput{\Gamma_1 + \Gamma_2}{\app{\val v}{\val w}}{\it E}}
                \end{prooftree}
\]
Finally, $\refrel{\Pi}{\Xi}$.

\paragraph*{Case Let-Binding}
$\Pi$ has the following form :
\begin{equation*}
    \hspace{-0.5cm}
\begin{prooftree}
                        \hypo{\stJugComput[\Pi_1]{\Delta}{\comput t}{\st E}}
                        \hypo{
                            \st E
                                \replaceLeaf{\st M \rightarrow \st G}
                            \st F}
                        \hypo{
                            \stJugComput[\Pi_2]
                                {\Delta; x :: \st M}
                                {\comput u}
                                {\st G}
                            }
                    \typeRuleLetin{
                        \stJugComput
                            {\Delta}
                            {\letin{x}{\comput t}{\comput u}}
                            {\st F}
                    }
                \end{prooftree}
\end{equation*}
Being $J_s$ the conclusion of $\Pi$, $J_r$ is such that $\refrel{J_s}{J_r\coloneq\itJugComput
                    {\Gamma}
                    {\letin{y}{\comput t}{\comput u}}
                    {\it F}}$, so $\refrel{\st F}{\it F}$ and $\refrel{\Delta}{\Gamma}$. 
We start by finding two contexts $\Gamma'$ and $\Gamma_i$ for each $\rangeI$ such that $\Gamma=\Gamma' +_{\rangeI} \Gamma_i$.
Since $\models\algo{\Pi}{J_r}$, there exist a set $I$, contexts $\Gamma',\Gamma_i\subseteq \Gamma$ such that $\Gamma=_+{\rangeI}\Gamma_i$ and \HEBI\ types $\refrel{\st E}{\it E}$, $\refrel{\st M}{\it M_i}$ and $\refrel{\st G}{\it G_i}$.
We set $J'_r\coloneq\itJugComput{\Gamma'}{\comput t}{\it E}$ and $J^i_r\coloneq\itJugComput{\Gamma_i;x:\it M_i}{\comput u}{\it G_i}$ and have $\models\algo{\Pi_1}{J'_r}$ and $\models\algo{\Pi_2}{J^i_r}$ for each $\rangeI$.
Notice also that, by \Cref{lem:ContextRefinement}, $\refrel{J^1_s}{J'_r}$ and $\refrel{J^2_s}{J^i_r}$.
Therefore, by induction hypothesis we have derivations $\refrel{\Pi_1}{\Xi'\triangleright J'_r}$ and $\refrel{\Pi_2}{\Xi_i\triangleright J^i_r}$ for each $\rangeI$.
Furthermore, $\models\algoleaf{\it E}{\it M_1,\dots,\it M_{|I|}}{\it G_1,\dots,\it G_{|I|}}{\it F}$, so by \Cref{lem:AlgoLeafSoundness} the relation $ \it E\replaceLeaf{\it M_i \rightarrow \it G_i}[\rangeI] \it F$ holds.
Thus, one can derive the following $\Xi$.
\begin{equation*}
\begin{prooftree}
                \hypo{\itJugComput[\Xi']{\Gamma }{\comput t}{\it E}}
                \hypo{
                    \it E
                        \replaceLeaf{\it M_i \rightarrow \it G_i}[\rangeI]
                    \it F}
                \hypo{
                    (\itJugComput[\Xi_i]
                        {\Gamma_i;x:\it M_i}
                        {\comput u}
                        {\it G_i}
                    )_{\rangeI}}
            \typeRuleLetin{
                \itJugComput
                    {\Gamma +_{\rangeI} \Gamma_i}
                    {\letin{y}{\comput t}{\comput u}}
                    {\it F}
            }
        \end{prooftree}
\end{equation*}
Finally, obtaining that $\refrel{\Pi}{\Xi}$.

\paragraph*{Case Case}
$\Pi$ has the following form :
\[
 \begin{prooftree}
                        \hypo{\stJugValue{\Delta}{\val v}{\typeInt{n}}}
                        \hypo{(\stJugComput{\Delta}{\comput t_i}{\st E}
                            )_{ 0<i\le n}}
                    \typeRuleCase[2]{\stJugComput{\Delta }{\case{\val v}{\comput t_1, \ldots, \comput t_n}}{\st E}}
                \end{prooftree}
\]
Being $J_s$ the conclusion of $\Pi$, $J_r$ is such that $\refrel{J_s}{J_r\coloneq\itJugComput{\Gamma}
        {\case{\val v}{\comput t_1, \ldots, \comput t_n}}{\it E}}$, so $\refrel{\st E}{\it E}$ and $\refrel{\Delta}{\Gamma}$. 
Let $J^1_s$ be the conclusion of $\Pi_1$ and $J^2_s$ be the conclusion of $\Pi_2$. Since $\models\algo{\Pi}{J_r}$, there exists two contexts $\Gamma_1,\Gamma_2\subseteq\Gamma$ such that $\Gamma=\Gamma_1 +\Gamma_2$ and an \HEBI\ type $\itSet{\typeInt{m}}$ such that $\refrel{\typeInt n}{\itSet{\typeInt{m}}}$. 
We set $J^1_r\coloneq\itJugValue{\Gamma_1}{\val v}{\itSet{m}}$ and $J^2_r\coloneq\itJugComput{\Gamma_2}{\comput t_m}{\it E}$, having $\models\algo{\Pi_1}{J^1_r}$ and $\models\algo{\Pi_2}{J^2_r}$.  
By \Cref{lem:ContextRefinement} we obtain $\refrel{J^1_s}{J^1_r}$ and $\refrel{J^2_s}{J^2_r}$. 
Therefore, by induction hypothesis, we have derivations $\refrel{\Pi_1}{\Xi_1\triangleright J^1_r}$ and $\refrel{\Pi_2}{\Xi_2\triangleright J^2_r}$.
We can derive the following $\Xi$.
\begin{equation*}
\begin{prooftree}
        \hypo{\itJugValue[\Xi_1]{\Gamma_1}{\val v}{\itSet{\typeInt{m}}}}
        \hypo{\itJugComput[\Xi_2]{\Gamma_2}{\comput t_m}{\it E}}
    \typeRuleCase{\itJugComput{\Gamma_1 + \Gamma_2}
        {\case{\val v}{\comput t_1, \ldots, \comput t_n}}{\it E}}
\end{prooftree}
\end{equation*}
Finally, $\refrel{\Pi}{\Xi}$.

\paragraph*{Case Return}
$\Pi$ has the following form :
\[
    \begin{prooftree}
                        \hypo{\stJugValue[\Pi_1]{\Delta}{\val v}{\st M}}
                    \typeRuleRet{\stJugComput{\Delta}{\ret{\val v}}{\itEffectReturn{\st M}}}
                \end{prooftree}
\]
Being $J_s$ the conclusion of $\Pi$, $J_r$ is such that $\refrel{J_s}{J_r\coloneq\itJugComput{\Gamma}
        {\ret{\val v}}{\itEffectReturn{\it M}}}$, so $\refrel{\itEffectReturn{\st M}}{\itEffectReturn{\it M}}$, therefore $\refrel{\st M}{\it M}$ and $\refrel{\Delta}{\Gamma}$. 
Let $J^1_s$ be the conclusion of $\Pi_1$, we set $J^1_r\coloneq\itJugValue{\Gamma}{\val v}{\it M}$ and have $\refrel{J^1_s}{J^1_r}$. Since $\models\algo{\Pi}{J_r}$, we have also $\models\algo{\Pi_1}{J^1_r}$.
Therfore, one can derive the following $\Xi$.
\begin{equation*}
\begin{prooftree}
        \hypo{\itJugValue[\Xi_1]{\Gamma}{\val v}{\it M}}
    \typeRuleRet{\itJugComput{\Gamma}
        {\ret{\val v}}{\itEffectReturn{\it M}}}
\end{prooftree}
\end{equation*}

Finally, $\refrel{\Pi}{\Xi}$.

\paragraph*{Case Effect}
$\Pi$ has the following form :

\[
 \begin{prooftree}
                        \hypo{
                            \stJugValue[\Pi_1]
                                {\Delta}
                                {\val v}
                                {\st M}
                        }
                        \hypo{
                            \stJugComput[\Pi_2]
                                {\Delta; x :: \st N}
                                {\comput t}
                                {\st E}
                        }
                    \typeRuleEff{
                        \stJugComput
                            {\Delta}
                            {\effect{\val v}{x}{\comput t}}
                            {\itEffect{\sigma}{\st M}{\st N \rightarrow \st E}}
                        }
                \end{prooftree}
\]
Being $J_s$ the conclusion of $\Pi$, $J_r$ is such that $\refrel{J_s}{J_r\coloneq\itJugComput
            {\Gamma}
            {\effect{\val w}{y}{\comput t}}
            {\itEffect{\sigma}{\it M}{\it N_i \rightarrow \it E_i}_{\rangeI}}}$, so $\refrel{\itEffect{\sigma}{\st M}{\st N \rightarrow \st E}}{\itEffect{\sigma}{\it M}{\it N_i \rightarrow \it E_i}_{\rangeI}}$ and $\refrel{\Delta}{\Gamma}$. 
Let $J^1_s$ and $J^2_s$ be the conclusions of $\Pi_1$ and $\Pi_2$ respectively. Since $\models\algo{\Pi}{J_r}$, there exists contexts $\Gamma',\Gamma_i\subseteq \Gamma$ such that $\Gamma=\Gamma+_{\rangeI}\Gamma_i$.
We set $J'_r\coloneq\itJugComput{\Gamma'}{\comput t}{\it M}$ and $J^i_r\coloneq\itJugComput{\Gamma_i;x:\it N_i}{\comput u}{\it E_i}$ for each $\rangeI$. We have $\models\algo{\Pi_1}{J'_r}$ and $\models\algo{\Pi_2}{J^i_r}$ and by \Cref{lem:ContextRefinement} we obtain $\refrel{J^1_s}{J'_r}$ and $\refrel{J^2_s}{J^i_r}$ for each $\rangeI$. We can then apply the induction hypothesis, obtaining derivations $\refrel{\Pi_1}{\Xi'\triangleright J'_r}$ and $\refrel{\Pi_2}{\Xi_i\triangleright J^2_r}$.
Therefore, we construct the following derivation $\Xi$.

\begin{equation*}
\begin{prooftree}
        \hypo{\itJugValue[\Xi']
            {\Gamma'}
            {\val w}
            {\it M}}
        \hypo{(\itJugComput[\Xi_i]
            {\Gamma_i; y : \it N_i}
            {\comput t}
            {\it E_i}
        )_{\rangeI}}
    \typeRuleEff[2]{
        \itJugComput
            {\Gamma' +_{\rangeI} \Gamma_i}
            {\effect{\val w}{y}{\comput t}}
            {\itEffect{\sigma}{\it M}{\it N_i \rightarrow \it E_i}_{\rangeI}}}
\end{prooftree}
\end{equation*}
Finally, $\refrel{\Pi}{\Xi}$.

\paragraph*{Case Handler}
$\Pi$ has the following form :
\begin{equation*}
\begin{prooftree}
                    \hypo{\stJugHandler{\Delta}{\handler h}{\st F \Rightarrow \st E}}
                    \hypo{\stJugComput{\Delta}{\comput t}{\st F}}
                    \typeRuleHandle{\stJugComput{\Delta}{\handle{\handler h}{\comput t}}{\st E}}
                \end{prooftree}
\end{equation*}
Being $J_s$ the conclusion of $\Pi$, $J_r$ is such that $\refrel{J_s}{J_r\coloneq\itJugHandler{\Gamma}{\handle{\handler h}{\comput t}}{\it E}}$, so $\refrel{\st E}{\it E}$ and $\refrel{\Delta}{\Gamma}$. 
Let $J^1_s$ be the conclusion of $\Pi_1$ and $J^2_s$ be the conclusion of $\Pi_2$. Since $\models\algo{\Pi}{J_r}$, there exist contexts $\Gamma_1,\Gamma_2\subseteq \Gamma$ such that $\Gamma=\Gamma_1+\Gamma_2$ and a type $\it F$ such that $\refrel{\st F}{\it F}$. We set $J^1_r\coloneq\itJugHandler{\Gamma_1}{\handler h}{\it F \Rightarrow \it E}$ and $J^2_r\coloneq\itJugComput{\Gamma_2}{\comput t}{\it F}$. We have $\models\algo{\Pi_1}{J^1_r}$ and $\models\algo{\Pi_2}{J^2_r}$ and, by \Cref{lem:ContextRefinement}, $\refrel{J^1_s}{J^1_r}$ and $\refrel{J^2_s}{J^2_r}$.
We can then apply the induction hypothesis, obtaining derivations $\refrel{\Pi_1}{\Xi_1\triangleright J^1_r}$ and $\refrel{\Pi_2}{\Xi_2\triangleright J^2_r}$.
Therfore, one can derive the following $\Xi$.
\begin{equation*}
\begin{prooftree}
        \hypo{\itJugHandler[\Xi_1]{\Gamma_1}{\handler h}{\it F \Rightarrow \it E}}
        \hypo{\itJugComput[\Xi_2]{\Gamma_2}{\comput t}{\it F}}
    \typeRuleHandle{\itJugHandler{\Gamma_1 + \Gamma_2}{\handle{\handler h}{\comput t}}{\it E}}
\end{prooftree}
\end{equation*}
Finally, $\refrel{\Pi}{\Xi}$.

\paragraph*{Case Handler Branching Return}
$\Xi$ has the following form :

\[
 \begin{prooftree}
                \hypo{\stJugComput[\Pi_1]
                    {\Delta;
                        y :: \st M}
                    {\comput t}
                    {\st E}}
                                            \hypo{\{\retClause{y}{\comput u}\} \handlerCompat \handler h}

            \typeRuleHandlerRet[2]{\stJugHandler
                {\Delta}
                {\{\retClause{y}{\comput t}\}}%
                {\itEffectReturn{\st M} \Rightarrow \st E}
            }%
        \end{prooftree}
\]

Being $J_s$ the conclusion of $\Pi$, $J_r$ is such that $\refrel{J_s}{J_r\coloneq\itJugComput{\Gamma}
        {\ret{\val v}}{\itEffectReturn{\it M}\Rightarrow \it E}}$, therefore $\refrel{\st M}{\it M}$, $\refrel{\st E}{\it E}$ and $\refrel{\Delta}{\Gamma}$. 
We can then construct a judgement $J^1_r\coloneq\itJugComput
                    {\Gamma;
                        y : \it M}
                    {\comput t}
                    {\it E}$. 
Let $J^1_s$ be the conclusion of $\Pi_1$ we obtain $\refrel{J^1_s}{J^1_r}$. Since $\models\algo{\Pi}{J_r}$, then also $\models\algo{\Pi_1}{J^1_r}$. By induction hypothesis we obtain a derivation $\Xi_1$ such that $\refrel{\Pi_1}{\Xi_1}$.
Therefore, one can derive the following $\Xi_1$:
\begin{equation*}
    \begin{prooftree}
                \hypo{\itJugComput[\Xi_1]
                    {\Gamma;
                        y : \it M}
                    {\comput t}
                    {\it E}}
                                            \hypo{\{\retClause{y}{\comput u}\} \handlerCompat \handler h}

            \typeRuleHandlerRet[2]{\itJugHandler
                {\Gamma}
                {\{\retClause{y}{\comput t}\} }%
                {\itEffectReturn{\it M} \Rightarrow \it E}
            }%
        \end{prooftree}
\end{equation*}

Finally, $\refrel{\Pi}{\Xi}$.

\paragraph*{Case Handler Branching}
$\Pi$ has the following form :
\[
 \begin{prooftree}
                        \hypo{%
                            \stJugComput%
                                {\Delta;
                                    y :: \st M;
                                    r :: \st N \rightarrow \st G}%
                                {\comput t}%
                                {\st E}%
                            }%
                        \hypo{%
                            \stJugHandler%
                                {\Delta}%
                                {\handler h}%
                                {\st F \Rightarrow \st G}%
                        }%
                        \hypo{\{\effectClause{y}{r}{\comput t}\} \handlerCompat \handler h}%
                    \typeRuleHandler[3]{%
                        \stJugHandler%
                            {\Delta}%
                            {\{\effectClause{y}{r}{\comput t}\} \cup \handler h}%
                            {
                                \itEffect{\sigma}{\st M}
                                    {\st N
                                        \rightarrow 
                                    \st F
                                }
                                    \Rightarrow 
                                \st E
                            }%
                    }%
                \end{prooftree}
\]

Being $J_s$ the conclusion of $\Pi$, $J_r$ is such that $\refrel{J_s}{J_r\coloneq\itJugHandler
                {\Gamma +_{\rangeI} \Gamma_i}
                {\{\effectClause{y}{r}{\comput t}\} \cup \handler h}%
                {\itEffect{\sigma}{\it M}
                    {\it N_i\rightarrow \it F_i
                        }_{\rangeI}
                    \Rightarrow \it E}}$, so $\refrel{\st M}{\it M}$, $\refrel{\st E}{\it E}$, $\refrel{\st N}{\it N_i}$, $\refrel{\st F}{\it F_i$} for each $\rangeI$ and $\refrel{\Delta}{\Gamma}$. 
Let $J^1_s$ be the conclusion of $\Pi_1$ and $J^2_s$ be the conclusion of $\Pi_2$. Since, $\models\algo{\Pi}{J_r}$, there exist contexts $\Gamma',\Gamma_i\subseteq\Gamma$ and types $\it G_i$ such that $\refrel{\st G}{\it G_i}$. We set $J'_r\coloneq\itJugComput[\Xi_\comput t]
                    {\Gamma_\comput t;
                        y : \it M ;
                        r : \itSet*{
                                \it N_i \rightarrow \it G_i
                                }_{\rangeI}}
                    {\comput t}
                    {\it E}$
and $J^i_r\coloneq\itJugHandler
                        {\Gamma_i}%
                        {\handler h}%
                        {\it F_i \Rightarrow \it G_i}$.
for each $\rangeI$. We have $\models\algo{\Pi_1}{J'_r}$ and $\models\algo{\Pi_2}{J^i_r}$ and, by \Cref{lem:ContextRefinement}, we obtain $\refrel{J^1_s}{J'_r}$ and $\refrel{J^2_s}{J^i_r}$ for each $\rangeI$. We can then apply the induction hypothesis, obtaining derivations $\refrel{\Pi_1}{\Xi'\triangleright J'_r}$ and $\refrel{\Pi_2}{\Xi_i\triangleright J^i_r}$ for each $\rangeI$.
Therefore, we construct the following derivation $\Xi$:
\begin{equation*}
    \hspace{-0.5cm}
\begin{prooftree}
                \hypo{\itJugComput[\Xi_\comput t]
                    {\Gamma_\comput t;
                        y : \it M ;
                        r : \itSet*{
                                \it N_i \rightarrow \it G_i
                                }_{\rangeI}}
                    {\comput t}
                    {\it E}}
                \hypo{
                    (\itJugHandler[\Xi_i]
                        {\Gamma_i}%
                        {\handler h}%
                        {\it F_i \Rightarrow \it G_i}
                    )_{\rangeI}
                }
                \hypo{(A)}
            \typeRuleHandler[3]{\itJugHandler
                {\Gamma +_{\rangeI} \Gamma_i}
                {\{\effectClause{y}{r}{\comput t}\} \cup \handler h}%
                {\itEffect{\sigma}{\it M}
                    {\it N_i\rightarrow \it F_i
                        }_{\rangeI}
                    \Rightarrow \it E}}%
        \end{prooftree}
\end{equation*}
\[
(A)= \begin{matrix}
                    \{\effectClause{y}{r}{\comput t}\} 
                    \handlerCompat \handler h
            \end{matrix}
\]
Finally, $\refrel{\Pi}{\Xi}$.

%% file: Proofs/AlgoCompleteness.tex
By induction on $\Pi$.

\paragraph*{Case Integers} 
$\Pi$ has the following form:
    \[
         \begin{prooftree}
            \hypo{0<n \le m}
            \typeRuleInt[1]{\itJugValue{\Delta}{\int n}{\typeInt{m}}}
        \end{prooftree}
        \refrelsym
        \begin{prooftree}
                \hypo{\it M \in \{\itSetEmpty, \itSet{\typeInt{n}}\}}
            \typeRuleInt[1]{\itJugValue{\emptyset}{\int n}{\it M}}
        \end{prooftree}
    \]
    It is clear that $J_r$ respects the condition given by the defition of the function, therefore $\models\algo{\Pi}{J_r}$. 

\paragraph*{Case Variable}
$\Pi$ has the following form:
    \[
    \begin{array}{ccc}
        \begin{prooftree}
            \typeRuleVar{\stJugValue[\Pi]{\Delta;x::\st M}{x}{\st M}}
        \end{prooftree}& \refrelsym &
        \begin{prooftree}
            \typeRuleVar{\itJugValue[\Xi]{x:\it M}{x}{\it M}}
        \end{prooftree}
    \end{array} 
    \]
    By definition of the refinement relation, $\refrel{\st M}{\it M}$. Therefore, $\models\algo{\Pi}{J_r}$.
\paragraph*{Case Abstraction} 
$\Pi$ has the following form : 

\[
\begin{array}{ccc}
    \begin{prooftree}
                        \hypo{\stJugComput[\Pi']{\Delta; x :: \st M}{\comput t}{\st E}}
                    \typeRuleAbs{\stJugValue{\Delta}{\abs{x}{\comput t}}{\st M \rightarrow \st E}}
                \end{prooftree}
                & \refrelsym & 
     \begin{prooftree}
                        \hypo{(\itJugComput[\Xi_i]{\Gamma_i;x:\it M_i}{\comput t}{\it E_i})_{\rangeI}}
                    \typeRuleAbs{\itJugValue[\Xi]{+_{\rangeI}\Gamma_i}{\abs{x}{\comput t}}{\itSet*{\it M \rightarrow \it E}_{\rangeI}}}
                \end{prooftree}
\end{array}
\]
With $\refrel{\Pi'}{\Xi_i}$ for each $\rangeI$. Let $J'_s$ be the conclusion of $\Pi'$ and $J_i$ the conclusion of $\Xi_i$, by \Cref{rem:NewRefineConclusion}, we have $\refrel{J'_s}{J_i}$ for each $\rangeI$. Then, by induction hypothesis, $\models\algo{\Pi'}{J_i}$ for each $\rangeI$. Moreover, it obviously exists $\Gamma=+_{\rangeI}\Gamma_i$, thus $\models\algo{\Pi}{J_r}$.

\paragraph*{Case Fixpoint} 
$\Pi$ has the following form :

    \[\begin{prooftree}
                \hypo{
                    \stJugValue[\Pi']
                        {\Delta; x::\st M\rightarrow \st E}
                        {\val v}
                        {\st M \rightarrow \st E}}
            \typeRuleFix[1]{
                \stJugValue[\Pi]
                    {\Delta}
                    {\fix{x}{\val v}}
                    {\st M \rightarrow \st E}}
        \end{prooftree}
        \refrelsym 
                \stJugValue[\Xi]
                    {\Gamma}
                    {\fix{x}{\val v}}
                    {\it M}\]
By \Cref{lem:RefinementFixCompleteness} there exists a context $\Gamma_0\subseteq\Gamma$ and a simple path $\phi$ between $(\Gamma_0,\itSetEmpty)$ and $(\Gamma,\it M')$ in $\graph{\Pi'}{\st M \rightarrow \st E}{\Gamma}$. Moreover, \Cref{rem:PathToDerivations} there exist some \HEBI\ types $\it N_i$ for $i\leq |\phi|+1$ such that $\refrel{\st M \rightarrow \st E}{\it N_i}$, $\it N_0=\itSetEmpty$ and $\it N_{|\phi|+1}=\it M$. Furthermore, there exist derivations $\itJugValue[\Xi_i]{\Gamma_i;x:\it N_i}{\val v}{\it N_{i+1}}$ for $i\leq |\phi|$ such that $\refrel{\Pi'}{\Xi_i}$ and $\Gamma=+_{i\leq |\phi|}\Gamma_i$.
Let $J^i_r$ be the conclusion of each $\Xi_i$ and $J'_s$ be the conlcusion of $\Pi'$, by \Cref{rem:NewRefineConclusion}, we have $\refrel{J'_s}{J^i_r}$ for each $i\leq n$. Therefore by induction hypothesis $\models\algo{\Pi'}{J^i_r}$ for each $i\leq n$.
Finally, since $\phi$ is simple, \Cref{rem:GraphFinite} implies $|\phi|\leq |\Refn(\st M\rightarrow\st E)|\cdot |Sub(\Gamma)|$. 
Taking $n= |\phi|$ and $\Gamma=+_{i\leq n}\Gamma_i$, we can conclude $\models\algo{\Pi}{J_r}$.

\paragraph*{Case Application}
$\Pi$ has the following form:
\[
  \begin{prooftree}
                \hypo{\stJugValue[\Pi_1]{\Delta}{\val v}                    {\st M \rightarrow \st U}}
                \hypo{\stJugValue[\Pi_2]{\Delta}{\val w}{\st M}}
            \typeRuleApp{\stJugComput
                {\Delta}
                {\app{\val v}{\val w}}
                {\st U}}
        \end{prooftree}
        \refrelsym
        \begin{prooftree}
                \hypo{\itJugValue[\Xi_1]{\Gamma_1}{\val v}
                    {\itSet{\it N \rightarrow \it E}}}
                \hypo{\itJugValue[\Xi_2]{\Gamma_2}{\val w}{\it N}}
            \typeRuleApp{\itJugComput
                {\Gamma_1 + \Gamma_2}
                {\app{\val v}{\val w}}
                {\it E}}
        \end{prooftree}
\]
With $\refrel{\Pi_1}{\Xi_1}$ and $\refrel{\Pi_2}{\Xi_2}$. Let $J^1_s$ and $J^2_s$ be the conclusions of $\Pi_1$ and $\Pi_2$, moreover let $J^1_r$ and $J^2_r$ be the conclusions of $\Xi_1,\Xi_2$. By \Cref{rem:NewRefineConclusion}, we have $\refrel{J_s}{J_r}$, $\refrel{J^1_s}{J^1_r}$ and $\refrel{J^1_s}{J^1_r}$. By induction hypothesis $\models\algo{\Pi_1}{J^1_r}$ and $\models\algo{\Pi_1}{J^1_r}$, thus $\models\algo{\Pi}{J_r}$.
Notice that $\refrel{\st M}{\it N}$ and $\Gamma=\Gamma_1+\Gamma_2$, thus $\models\algo{\Pi}{J_r}$.

\paragraph*{Case Let-Binding}
$\Pi$ has the following form :
\[\begin{array}{c}
  \begin{prooftree}
                \hypo{\stJugComput[\Pi_1]{\Delta}{\comput t}{\st E}}
            \hypo{
                    \st E
                        \replaceLeaf{\st M \rightarrow \st G}
                    \st F}
                \hypo{\stJugComput[\Pi_2]
                        {\Delta; x :: \st M}
                        {\comput u}
                        {\st G}}
            \typeRuleLetin[3]{
                \stJugComput
                    {\Delta}
                    {\letin{x}{\comput t}{\comput u}}
                    {\st F}
            }
        \end{prooftree}
    \\[0.5cm]
    \refrelsym 
    \begin{prooftree}
                \hypo{\itJugComput[\Xi]{\Gamma'}{\comput t}{\it E}}
                \hypo{
                    \it E
                        \replaceLeaf{\it M_j \rightarrow \it G_j}[\rangeJ]
                    \it F}
                \hypo{
                    (\itJugComput[\Xi_j]
                        {\Gamma_j; x : \it M_j}
                        {\comput u}
                        {\it G_j}
                    )_{\rangeJ}}
            \typeRuleLetin{
                \itJugComput
                    {\Gamma' +_{\rangeJ} \Gamma_j}
                    {\letin{x}{\comput t}{\comput u}}
                    {\it F}
            }
        \end{prooftree}
        \end{array}
\]
With $\refrel{\Pi_1}{\Xi'}$ and $\refrel{\Pi_2}{\Xi_i}$ for each $\rangeI$. Let $J^1_s$ and $J^2_s$ be the conclusions of $\Pi_1$ and $\Pi_2$, moreover let $J'_r$ and $J^j_r$ be the conclusions respectively of $\Xi'$ and $\Xi_i$. By \Cref{rem:NewRefineConclusion}, we have $\refrel{J_s}{J_r}$, $\refrel{J^1_s}{J'_r}$ and $\refrel{J^2_s}{J^j_r}$. By induction hypothesis $\models\algo{\Pi_1}{J'_r}$ and $\models\algo{\Pi_2}{J^j_r}$ for each $\rangeJ$. Moreover, by \Cref{lem:LeafReplaceBound}, $|J|\leq |\it F|$ and by \Cref{lem:AlgoLeafSoundness}, $\models\algoleaf{\it E}{\it M_1,\dots,\it M_{|J|}}{\it G_1,\dots,\it G_{|J|}}{\it F}$. Having $\Gamma=\Gamma'+_{\rangeJ}\Gamma_j$, all the conditions in the definition of $\mathcal A$ are satisfied, thus $\models\algo{\Pi}{J_r}$.

\paragraph*{Case Case}
$\Pi$ has the following form :
\[\begin{array}{c}
    \begin{prooftree}
                \hypo{\stJugValue[\Pi]
                    {\Delta}{\val v}
                    {\typeInt n}}
                \hypo{(
                    \stJugComput[\Pi_k]{\Delta}{\comput t_k}{\st U})_{ 0<k\le n}}
            \typeRuleCase[2]{\stJugComput
                {\Delta }
                {\case{\val v}{\comput t_1, \ldots, \comput t_n}}
                {\st U}}
        \end{prooftree}
    \\[0.5cm]
    \refrelsym 
    \begin{prooftree}
                \hypo{\itJugValue[\Xi]
                    {\Gamma_1}{\val v}
                    {\itSet{\typeInt m}}}
            \hypo{0 < m \leq n}
                \hypo{\itJugComput[\Xi_m]
                    {\Gamma_m}{\comput t_m}
                    {\it E}}
            \typeRuleCase[3]{\itJugComput
                {\Gamma_1 + \Gamma_m}
                {\case{\val v}{\comput t_1, \ldots, \comput t_n}}
                {\it E}}
        \end{prooftree}
\end{array}
\]
With $\refrel{\Pi_1}{\Xi_1}$ and $\refrel{\Pi_m}{\Xi_m}$. Let $J^1_s$ and $J^m_s$ be the conclusions of $\Pi_1$ and $\Pi_m$, moreover let $J^m_r$ and $J^m_r$ be the conclusions of $\Xi_1,\Xi_m$. By \Cref{rem:NewRefineConclusion}, we have $\refrel{J_s}{J_r}$, $\refrel{J^1_s}{J^1_r}$ and $\refrel{J^m_s}{J^m_r}$. By induction hypothesis $\models\algo{\Pi_1}{J^1_r}$ and $\models\algo{\Pi_1}{J^1_r}$. Since $\Gamma=\Gamma_1+\Gamma_m$ and $\refrel{\typeInt n}{\itSet{\typeInt{m}}}$, we conclude $\models\algo{\Pi}{J_r}$.

\paragraph*{Case Return}
$\Pi$ has the following form :
\[
    \begin{prooftree}
                \hypo{\stJugValue[\Pi']{\Delta}{\val v}
                    {\st M}}
            \typeRuleRet{\stJugComput
                {\Delta}{\ret{\val v}}
                {\itEffectReturn{\st M}}}
        \end{prooftree}
        \refrelsym 
        \begin{prooftree}
                \hypo{\itJugValue[\Xi']{\Gamma}{\val v}
                    {\it M}}
            \typeRuleRet{\itJugComput
                {\Gamma}{\ret{\val v}}
                {\itEffectReturn{\it M}}}
        \end{prooftree}
\]
With $\refrel{\Pi'}{\Xi'}$. Let $J'_s$ be the conclusions of $\Pi'$ and $J'_r$ be the conclusions of $\Xi'$. By \Cref{rem:NewRefineConclusion}, we have $\refrel{J_s}{J_r}$ and $\refrel{J'_s}{J'_r}$.
By induction hypothesis $\models\algo{\Pi'}{J'_r}$, thus $\models\algo{\Pi}{J_r}$.

\paragraph*{Case Effect}
$\Pi$ has the following form :

\[\begin{array}{c}
 \begin{prooftree}
                \hypo{\stJugValue[\Pi_1]
                    {\Delta}
                    {\val v}
                    {\st M}}
                \hypo{
                    \stJugComput[\Pi_2]
                        {\Delta; x :: \st N}
                        {\comput t}
                        {\st E}}

            \typeRuleEff[2]{
                \stJugComput
                    {\Delta }
                    {\effect{\val v}{x}{\comput t}}
                    {\itEffect{\sigma}{\st M}{\st N \rightarrow \st E}}
                }
        \end{prooftree}
    \\[0.5cm]
     \begin{prooftree}
                \hypo{\itJugValue[\Xi]
                    {\Gamma'}
                    {\val v}
                    {\it M}}
                \hypo{(
                    \itJugComput[\Xi_i]
                        {\Gamma_i; x : \it N_i}
                        {\comput t}
                        {\it E_i}
                )_{\rangeI}}
            \typeRuleEff{
                \itJugComput
                    {\Gamma' +_{\rangeI} \Gamma_i}
                    {\effect{\val v}{x}{\comput t}}
                    {\itEffect{\sigma}{\it M}{\it N_i \rightarrow \it E_i}_{\rangeI}}
                }
        \end{prooftree}
        \end{array}
\]
With $\refrel{\Pi_1}{\Xi'}$ and $\refrel{\Pi_2}{\Xi_i}$ for each $\rangeI$. Let $J^1_s$ and $J^2_s$ be the conclusions of $\Pi_1$ and $\Pi_2$, moreover let $J'_r$ and $J^i_r$ be the conclusions respectively of $\Xi'$ and $\Xi_i$. By \Cref{rem:NewRefineConclusion}, we have $\refrel{J_s}{J_r}$, $\refrel{J^1_s}{J'_r}$ and $\refrel{J^2_s}{J^i_r}$. By induction hypothesis $\models\algo{\Pi_1}{J'_r}$ and $\models\algo{\Pi_2}{J^i_r}$. Having $\Gamma=\Gamma'+_{\rangeI}\Gamma_i$, we conclude $\models\algo{\Pi}{J_r}$.

\paragraph*{Case Handler}
$\Pi$ has the following form :
\begin{equation*}
  \begin{prooftree}
            \hypo{\stJugHandler[\Pi_1]{\Delta}{\handler h}
                {\st F \Rightarrow \st E}}
            \hypo{\stJugComput[\Pi_2]{\Delta}{\comput t}{\st F}}
            \typeRuleHandle{\stJugComput
                {\Delta}
                {\handle{\handler h}{\comput t}}{\st E}}
        \end{prooftree}
    \\[0.5cm]
    \refrelsym
    \begin{prooftree}
            \hypo{\itJugHandler[\Xi_1]{\Gamma_1}{\handler h}
                {\it F \Rightarrow \it E}}
            \hypo{\itJugComput[\Xi_2]{\Gamma_2}{\comput t}{\it F}}
            \typeRuleHandle{\itJugComput
                {\Gamma_1 + \Gamma_2}
                {\handle{\handler h}{\comput t}}
                {\it E}}
        \end{prooftree}
\end{equation*}
With $\refrel{\Pi_1}{\Xi_1}$ and $\refrel{\Pi_2}{\Xi_2}$. Let $J^1_s$ and $J^2_s$ be the conclusions of $\Pi_1$ and $\Pi_2$, moreover let $J^1_r$ and $J^2_r$ be the conclusions of $\Xi_1,\Xi_2$. By \Cref{rem:NewRefineConclusion}, we have $\refrel{J_s}{J_r}$, $\refrel{J^1_s}{J^1_r}$ and $\refrel{J^2_s}{J^2_r}$. By induction hypothesis $\models\algo{\Pi_1}{J^1_r}$ and $\models\algo{\Pi_1}{J^2_r}$. Having $\Gamma=\Gamma_1+\Gamma_2$ and $\refrel{\st F}{\it F}$, we conclude $\models\algo{\Pi}{J_r}$.

\paragraph*{Case Handler Branching Return}
$\Xi$ has the following form :

\[\begin{array}{c}
    \begin{prooftree}
    \hypo{
        \stJugComput[\Pi']
            {\Delta; y :: \st M}
            {\comput t}
            {\st E}
            }
            \typeRuleHandlerRet[1]{\stJugHandler
                {\Delta'}
                {\{\retClause{y}{\comput t}\} \cup \handler h}
                {\itEffectReturn{\st M} \Rightarrow \st E}
            }
    \end{prooftree}
    \\[0.5cm]
    \refrelsym
    \begin{prooftree}
    \hypo{
        \itJugComput[\Xi']
            {\Gamma; y :: \it M}
            {\comput t}
            {\it E}
            }
            \typeRuleHandlerRet[1]{\itJugHandler
                {\Gamma}
                {\{\retClause{y}{\comput t}\} \cup \handler h}
                {\itEffectReturn{\it M} \Rightarrow \it E}
            }
    \end{prooftree}
\end{array}
\]

With $\refrel{\Pi'}{\Xi'}$ and $\{\retClause{y}{\comput u}\} \handlerCompat \handler h$. Let $J'_s$ be the conclusions of $\Pi'$ and $J'_r$ be the conclusions of $\Xi'$. By \Cref{rem:NewRefineConclusion}, we have $\refrel{J_s}{J_r}$ and $\refrel{J'_s}{J'_r}$.
By induction hypothesis $\models\algo{\Pi'}{J'_r}$, thus $\models\algo{\Pi}{J_r}$.

\paragraph*{Case Handler Branching}
$\Pi$ has the following form :
\[\begin{array}{c}
  \begin{prooftree}
    \hypo{
        \stJugHandler[\Pi_1]
            {\Delta; y :: \st M;r :: \st N \rightarrow \st G}
            {\comput t}
            {\st E}
            }
            \hypo{
                \stJugHandler[\Pi_2]
                    {\Delta}%
                    {\handler h}%
                    {\st F \Rightarrow \st G}
                }
            \hypo{\{\effectClause{y}{r}{\comput t}\} \handlerCompat \handler h}
            \typeRuleHandlerSigma[3]{\stJugHandler
                {\Delta }
                {\{\effectClause{y}{r}{\comput t}\} \cup \handler h}
                {\itEffect{\sigma}{\st M}
                                    {\st N
                                        \rightarrow 
                                    \st F
                                }
                                    \Rightarrow 
                                \st E}
            }
            \end{prooftree}
        \\[0.5cm] \refrelsym
        \begin{prooftree}
                        \hypo{%
                            \itJugComput[\Xi]
                                {\Gamma';
                                    y : \it M;
                                    r : \itSet*{\it N_i \rightarrow \it G_i}_{\rangeI}}%
                                {\comput t}%
                                {\it E}%
                            }%
                        \hypo{(%
                            \itJugHandler[\Xi_i]
                                {\Gamma_i}%
                                {\handler h}%
                                {\it F_i \Rightarrow \it G_i}%
                        )_{\rangeI}}%
                        \hypo{\{\effectClause{y}{r}{\comput t}\} \handlerCompat \handler h}%
                    \typeRuleHandlerSigma[3]{%
                        \itJugHandler%
                            {\Gamma' +_{\rangeI} \Gamma_i}%
                            {\{\effectClause{y}{r}{\comput t}\} \cup \handler h}%
                            {
                                \itEffect{\sigma}{\it M}
                                    {\it N_i
                                        \rightarrow 
                                    \it F_i
                                }_{\rangeI}
                                    \Rightarrow 
                                \it E
                            }%
                    }%
        \end{prooftree}
        \end{array}
\]
With $\refrel{\Pi_1}{\Xi'}$ and $\refrel{\Pi_2}{\Xi_i}$ for each $\rangeI$. Let $J^1_s$ and $J^2_s$ be the conclusions of $\Pi_1$ and $\Pi_2$, moreover let $J'_r$ and $J^i_r$ be the conclusions respectively of $\Xi'$ and $\Xi_i$. By \Cref{rem:NewRefineConclusion}, we have $\refrel{J_s}{J_r}$, $\refrel{J^1_s}{J'_r}$ and $\refrel{J^2_s}{J^i_r}$. By induction hypothesis $\models\algo{\Pi_1}{J'_r}$ and $\models\algo{\Pi_2}{J^i_r}$. Having $\Gamma=\Gamma'+_{\rangeI}\Gamma_i$ and $\refrel{\st G}{\it G_i}$ for all $\rangeI$, we conclude $\models\algo{\Pi}{J_r}$.

%% file: Definitions/RefinementRelation.tex
\begin{tabular}{c}
    \hspace{-0.5cm}
\framebox{$
    \begin{array}{c}
        \begin{prooftree}
                \hypo{\stJugComput[\Pi_1]{\Delta}{\comput t}{\stEffect{\st M}}}
                \hypo{\stJugComput[\Pi_2]
                        {\Delta; x :: \st M}
                        {\comput u}
                        {\stEffect{\st N}}}
            \typeRuleLetin[2]{
                \stJugComput
                    {\Delta}
                    {\letin{x}{\comput t}{\comput u}}
                    {\stEffect{\st N}}
            }
        \end{prooftree}
    \\[0.5cm]
    \refrelsym 
    \begin{prooftree}
                \hypo{\itJugComput[\Xi]{\Gamma}{\comput t}{\it E}}
                \hypo{
                    \it E
                        \replaceLeaf{\it M_i \rightarrow \it G_i}[\rangeI]
                    \it F}
                \hypo{
                    (\itJugComput[\Xi_i]
                        {\Gamma_i; x : \it M_i}
                        {\comput u}
                        {\it G_i}
                    )_{\rangeI}}
            \typeRuleLetin{
                \itJugComput
                    {\Gamma +_{\rangeI} \Gamma_i}
                    {\letin{x}{\comput t}{\comput u}}
                    {\it F}
            }
        \end{prooftree}\\
        \text{if } \refrel{\Pi_1}{\Xi} \text{ and } \refrel{\Pi_2}{\Xi_i} \text{ for each } \rangeI
        \\[0.5cm]
        \begin{prooftree}
                \hypo{\stJugValue[\Pi]{\Delta}{\val v}
                    {\st M}}
            \typeRuleRet{\stJugComput
                {\Delta}{\ret{\val v}}
                {\stEffect{\st M}}}
        \end{prooftree}
        \refrelsym 
        \begin{prooftree}
                \hypo{\itJugValue[\Xi]{\Gamma}{\val v}
                    {\it M}}
            \typeRuleRet{\itJugComput
                {\Gamma}{\ret{\val v}}
                {\itEffectReturn{\it M}}}
        \end{prooftree}\\
        \text{if } \refrel{\Pi}{\Xi} 
        \\[0.5cm] 
    \begin{prooftree}
                \hypo{\stJugValue[\Pi_1]{\Delta}{\val v}                    {\st M \rightarrow \st U}}
                \hypo{\stJugValue[\Pi_2]{\Delta}{\val w}{\st M}}
            \typeRuleApp{\stJugComput
                {\Delta}
                {\app{\val v}{\val w}}
                {\st U}}
        \end{prooftree}
        \refrelsym
        \begin{prooftree}
                \hypo{\itJugValue[\Xi_1]{\Gamma_1}{\val v}
                    {\itSet{\it M \rightarrow \it E}}}
                \hypo{\itJugValue[\Xi_2]{\Gamma_2}{\val w}{\it M}}
            \typeRuleApp{\itJugComput
                {\Gamma_1 + \Gamma_2}
                {\app{\val v}{\val w}}
                {\it E}}
        \end{prooftree}\\      
        \text{if } \refrel{\Pi_1}{\Xi_1} \text{ and } \refrel{\Pi_2}{\Xi_2}\\[0.5cm] 

     \begin{prooftree}
                \hypo{\stJugValue[\Pi_1]
                    {\Delta}
                    {\val v}
                    {\st M}}
                \hypo{
                    \stJugComput[\Pi_2]
                        {\Delta; x :: \st N}
                        {\comput t}
                        {\stEffect{\st M'}}}
                \hypo{
                    \effectSpec{\st M}{\st N} \in \effectCtxt E
                   }
            \typeRuleEff[3]{
                \stJugComput
                    {\Delta}
                    {\effect{\val v}{x}{\comput t}}
                    {\stEffect{\st M'}}
                }
        \end{prooftree}
    \\[0.5cm]
     \begin{prooftree}
                \hypo{\itJugValue[\Xi]
                    {\Gamma}
                    {\val v}
                    {\it M}}
                \hypo{(
                    \itJugComput[\Xi_i]
                        {\Gamma_i; x : \it N_i}
                        {\comput t}
                        {\it E_i}
                )_{\rangeI}}
            \typeRuleEff{
                \itJugComput
                    {\Gamma +_{\rangeI} \Gamma_i}
                    {\effect{\val v}{x}{\comput t}}
                    {\itEffect{\sigma}{\it M}{\it N_i \rightarrow \it E_i}_{\rangeI}}
                }
        \end{prooftree}\\
        \text{if } \refrel{\Pi_1}{\Xi} \text{ and } \refrel{\Pi_2}{\Xi_i} \text{ for each }\rangeI
        \\[0.5cm]
       
        \begin{prooftree}
            \hypo{\stJugHandler[\Pi_1]{\Delta}{\handler h}
                {\st S \Rightarrow \st U}}
            \hypo{\stJugComput[\Pi_2]{\Delta}{\comput t}{\st S}}
            \typeRuleHandle{\stJugComput
                {\Delta}
                {\handle{\handler h}{\comput t}}{\st U}}
        \end{prooftree}
    \\[0.5cm]
    \refrelsym
    \begin{prooftree}
            \hypo{\itJugHandler[\Xi_1]{\Gamma_1}{\handler h}
                {\it F \Rightarrow \it E}}
            \hypo{\itJugComput[\Xi_2]{\Gamma_2}{\comput t}{\it F}}
            \typeRuleHandle{\itJugComput
                {\Gamma_1 + \Gamma_2}
                {\handle{\handler h}{\comput t}}
                {\it E}}
        \end{prooftree}\\
        \text{if } \refrel{\Pi_1}{\Xi_1} \text{ and } \refrel{\Pi_2}{\Xi_2}
        \\[0.5cm]

     \end{array}
$}
\end{tabular}

%% file: Definitions/RefinementRelationHand.tex
\begin{tabular}{c}
\framebox{$
    \begin{array}{c}

     \begin{prooftree}
    \hypo{(
        \stJugComput
            {\Delta; x : \st M_i';r :: \st N_i \rightarrow \st U}
            {\comput t_i}
            {\st U})_{1\leq i\leq n}
            }
            \hypo{
                \begin{matrix}
            \stJugComput[\Pi_\comput t]
            {\Delta; y :: \st M}
            {\comput t}
            {\st U}\\
            \{\sigma_i:{\st M'_i}\rightsquigarrow{\st N_i}\in \st E \vsep 1 \leq i \leq n\}
                \end{matrix}
            }
            %
            %
            \typeRuleHandler[2]{\stJugHandler
                {\Delta}
                {\{\retClause{y}{\comput t}\} \cup\{\effectClause[\sigma_{\spaceNegII i}]{x}{r}{\comput t_i} \vsep 1 \leq i \leq n\}}
                {\stEffect{\st M} \Rightarrow \st U}
            }
            \end{prooftree}
    \\[1cm]
    \refrelsym
    \begin{prooftree}
    \hypo{
        \itJugComput[\Xi]
            {\Gamma; y :: \it M}
            {\comput t}
            {\it E}
            }
            \hypo{\{\retClause{y}{\comput t}\} \handlerCompat \handler h}
            \typeRuleHandlerRet[2]{\itJugHandler
                {\Gamma}
                {\{\retClause{y}{\comput t}\} \cup \handler h}
                {\itEffectReturn{\it M} \Rightarrow \it E}
            }
    \end{prooftree}\\[0.5cm]
    \text{if } \refrel{\Pi_\comput t}{\Xi}
    \\[0.5cm]
         \begin{prooftree}
    \hypo{(
        \stJugComput[\Pi_i]
            {\Delta; x :: \st M_j';r :: \st N_j \rightarrow \st U}
            {\comput t_j}
            {\st U})_{1\leq j\leq n}
            }
            \hypo{
                \begin{matrix}
            \stJugComput
            {\Delta; y :: \st M}
            {\comput t}
            {\st U}\\
            \{\sigma_j:{\st M'_j}\rightsquigarrow{\st N_j}\in \st E \vsep 1 \leq j \leq n\}
                \end{matrix}
            }
            %
            %
            \typeRuleHandler[2]{\stJugHandler
                {\Delta}
                {\{\retClause{y}{\comput t}\} \cup\{\effectClause[\sigma_{\spaceNegII j}]{x}{r}{\comput t_j} \vsep 1 \leq j \leq n\}}
                {\stEffect{\st M} \Rightarrow \st U}
            }
            \end{prooftree}
        \\[1cm] \refrelsym
        \begin{prooftree}
                        \hypo{%
                            \itJugComput[\Xi]
                                {\Gamma;
                                    y : \it M;
                                    r : \itSet*{\it N_i \rightarrow \it G_i}_{\rangeI}}%
                                {\comput t}%
                                {\it E}%
                            }%
                        \hypo{
                        \begin{matrix}
                            \{\effectClause[\sigma_{\spaceNegII k}]{y}{r}{\comput t}\} \handlerCompat \handler h\\
                            (%
                            \itJugHandler[\Xi_i]
                                {\Gamma_i}%
                                {\handler h}%
                                {\it F_i \Rightarrow \it G_i}%
                        )_{\rangeI}
                        \end{matrix}    
                        }%
                    \typeRuleHandlerSigma[2]{%
                        \itJugHandler%
                            {\Gamma +_{\rangeI} \Gamma_i}%
                            {\{\effectClause[\sigma_{\spaceNegII k}]{y}{r}{\comput t}\} \cup \handler h}%
                            {
                                \itEffect{\sigma_k}{\it M}
                                    {\it N_i
                                        \rightarrow 
                                    \it F_i
                                }_{\rangeI}
                                    \Rightarrow 
                                \it E
                            }%
                    }%
        \end{prooftree}\\[0.5cm]
        \text{if } \refrel{\Pi_k}{\Xi} \text{ and } \refrel{\stJugHandler[\Pi]{\Delta}{\handler h}{\prescript{}{\st E}{\st M}}\Rightarrow\st U}{\Xi_i} \text{ for each }\rangeI

    \end{array}
    
$}
\end{tabular}

%% file: Definitions/HEPCFReductions.tex
\begin{tabular}{c}
    \hspace{-0.5cm}
\framebox{$

\begin{array}{c}

    \begin{array}{ccl}
    \begin{prooftree}
        \hypo{\stJugValue[\Pi_\val v]{\emptyset}{\val v}{\st M}}
        \typeRuleRet{\stJugComput{\emptyset}{\ret{\val v}}{\stEffect{\st M}}}
        \hypo{\stJugComput[\Pi_{\comput u}]{ x ::\st M}{\comput u}{\stEffect{\st N}}}
    \typeRuleLetin[2]{\stJugComput{\emptyset}{\letin{x}{\ret{\val v}}{\comput u}}{\stEffect{\st N}}}
\end{prooftree}
    & \rightsquigarrow & \stJugComput[\Pi_{\comput u}\sub{x}{\Pi_\val v}]{ \emptyset}{\comput u\sub{x}{\val v}}{\stEffect{\st N}}\\[1cm]
    \end{array}\\\\
    \begin{prooftree}
            \hypo{\stJugComput[\Pi_{\comput u}]{ x :: \st M}{\comput u}{\st U}}
            \hypo{\begin{matrix}
        \{\sigma_j:{\st M'_j}\rightsquigarrow{\st N_j}\in \st E \vsep 1 \leq j \leq n\}\\
        (
        \stJugComput[\Pi_i]
            {\Delta; x :: \st M_j';r :: \st N_j \rightarrow \st U}
            {\comput t_j}
            {\st U})_{1\leq j\leq n}\\
    \end{matrix}}
            \typeRuleHandler[2]{\stJugHandler[\Pi_\handler{h}]{\emptyset}{\{\retClause{x}{\comput u}\}\cup \handler h'}{\stEffect{\st M} \Rightarrow \st U}}
            \hypo{\stJugValue[\Pi_\val v]{\emptyset}{\val v}{\st M}}
            \typeRuleRet{\stJugComput{\emptyset}{\ret{\val v}}{\stEffect{\st M}}}
        \typeRuleHandle{\stJugComput{\emptyset}{\handle{\{\retClause{x}{\comput u}\} \cup \handler h'}{\ret{\val v}}}{\st U}}
    \end{prooftree}\\[1.2cm]
     \rightsquigarrow  \stJugComput[\Pi_{\comput u}\sub{x}{\Pi_\val v}]{\emptyset}{\comput u\sub{x}{\val v}}{\st U}\\[0.5cm]

    \begin{prooftree}
            \hypo{
                    \stJugValue[\Pi_{\val v}]
                        {\emptyset}
                        {\val v}
                        {\st M}
                }
            \hypo{
        \begin{matrix}
                \effectSpec{\st M}{\st N} \in \effectCtxt E\\
                \stJugComput[\Pi_{\comput s}]
                        {y :: \st N}
                        {\comput s}
                        {\stEffect{\st N'}}
        \end{matrix}}
            \typeRuleEff[2]{
                \stJugComput
                    {\emptyset}
                    {\effect{\val v}{y}{\comput s}}
                    {\stEffect{\st N'}}
            }
            \hypo{\stJugComput[\Pi_\comput u]
                {x ::\st N'}{\comput u}
                {\stEffect{\st M'}}}
        \typeRuleLetin[2]{
            \stJugComput
                {\emptyset}
                {\letin{x}{\effect{\val v}{y}{\comput s}}{\comput u}}
                {\stEffect{\st M'}}
        }
    \end{prooftree}\qquad\\[1.5cm]

    \rightsquigarrow
    \begin{prooftree}
            \hypo{
                    \stJugValue[\Pi_{\val v}]
                        {\emptyset}
                        {\val v}
                        {\st M}
                }

                \hypo{\begin{matrix}
                    \stJugComput[\Pi_\comput u^{(y::\st N)}]
                {y::\st N;x ::\st N'}{\comput u}
                {\stEffect{\st M'}}
                    \\
                    \stJugComput[\Pi_{\comput s}]
                        { y :: \st N}
                        {\comput s}
                        {\stEffect{\st N'}}
                \end{matrix}
                    }
        \typeRuleLetin[1]{
                \stJugComput
                    {y::\st N}
                    {\letin{x}{\comput s}{\comput u}}
                    {\stEffect{\st M'}}
            }
        \hypo{
        \begin{matrix}
                \effectSpec{\st M}{\st N} \in \effectCtxt E
        \end{matrix}}
        \typeRuleEff[3]{
            \stJugComput
                {\emptyset}
                {\effect{\val v}{y}{\letin{x}{\comput s}{\comput u}}}
                {\stEffect{\st M'}}
        }
    \end{prooftree}\\[1cm]

    \begin{prooftree}
    \hypo{(
        \stJugComput[\Pi_i]
            {\Delta; x :: \st M_j';r :: \st N_j \rightarrow \st U}
            {\comput t_j}
            {\st U})_{1\leq j\leq n}
            }
            \hypo{
                (A)
            }
            %
            %
            \typeRuleHandler[2]{\stJugHandler[\Pi_\handler h]
                {\Delta}
                {\{\retClause{y}{\comput t}\} \cup\{\effectClause[\sigma_{\spaceNegII j}]{x}{r}{\comput t_j} \vsep 1 \leq j \leq n\}}
                {\stEffect{\st M} \Rightarrow \st U}
            }

        \hypo{
            \begin{matrix}
                \stJug[\Pi_\val v]{\emptyset}{\val v}{\st M'_k}\\
                \stJug[\Pi_\comput u]{y::\st N_k}{\comput u}{\stEffect{\st M}}\\
                \effectSpec[\sigma_k]{\st M'_k}{\st N_k} \in \effectCtxt E
            \end{matrix}       
        }   

        \typeRuleEff[1]{\stJugComput
            {\emptyset}{\effect[\sigma_k]{\val v}{y}{\comput u}}
            {\stEffect{\st M}}}
        \typeRuleHandle[2]{
            \stJugComput
                {\emptyset}
                {\handle{\{\effectClause[\sigma_k]{x}{r}{\comput t_k}\} \cup \handler h}{\effect[\sigma_k]{\val v}{y}{\comput u}}}
                {\st U}
        }
    \end{prooftree}\\[1.5cm]
    (A):= \; \begin{matrix}
            \stJugComput
            {\Delta; y :: \st M}
            {\comput t}
            {\st U}\\
            \{\sigma_j:{\st M'_j}\rightsquigarrow{\st N_j}\in \st E \vsep 1 \leq j \leq n\}
                \end{matrix}
    \\[0.5cm]
    \rightsquigarrow \stJugComput[\Pi_k\sub{x}{\Pi_\val v}\sub{r}{\Pi_{\val{\lambda y}}}]{\emptyset }{\comput t_k \sub{x}{\val v}\sub{r}{\abs{y}{\handle{\handler h}{\comput u}}}}{\st U}\\[0.5cm]
\text{With}\\
\begin{prooftree}

    \hypo{
        \stJugHandler[\Pi_\handler h^{(y::\st N_k)}]{y::\st N_k}{\handler h}{\stEffect{\st M}\Rightarrow \st U}
    }
    \hypo{
        \stJugComput[\Pi_\comput u]{y::\st N_k}{\comput u}{\stEffect{\st M}}
    }
    \typeRuleHandler{
        \stJugComput{ y::\st N_k}{\handle{\handler h}{\comput u}}{\st U}
    }
    \typeRuleAbs{
        \stJugValue[\Pi_\val{\lambda y}]{\emptyset}{\lambda y.\handle{ \handler h}{\comput u}}{\st N_k\Rightarrow \st U}
    }
\end{prooftree}
\end{array}
$}
\end{tabular}

%% file: Proofs/RefinementElementarySubjectExpansion.tex
By case analysis on the different rewrite rules. We are going to detail only the new rewriting rules, the other cases being identical to \Cref{lem:NewElementarySubjectExpansion-ref}

\paragraph*{Case $\reductLetRet$.}
Let $\comput t, \comput t' \in \setComput$ such that $\comput t'$ is
typed and $\comput t \reductArr_\reductLetRet \comput t'$. By
definition, $\comput t \coloneqq
\letin{x}{\ret{\val v}}{\comput u}$ and $\comput t'
\coloneqq \comput u\sub{x}{\val v}$. Since $t$ is simply typable, we get the existence of two derivations $\stJugValue[\Pi_\val v]{\emptyset}{\val v}{\st M}$ and $\stJugComput[\Pi_\comput u]{x::\st M}{u}{\stEffect{\st N}}$ for some $\st M$. Also,
there exist derivations $\refrel{\stJugComput[\Pi']{\emptyset}{\comput u\sub{x}{\val v}}{\stEffect{\st M}}}{\itJugComput[\Xi']{\emptyset}{\comput u\sub{x}{\val v}}{\it F}}$.
Since by hypothesis $\Pi\rightsquigarrow \Pi'$, by definition of $\rightsquigarrow$ we have that $\Pi'=\Pi_\comput u\sub{x}{\Pi_\val v}$.
By anti-substitution lemma (\Cref{lem:HEPCFantiSubstLemma-ref}), there exist a refinement type $\it M$
and derivations $\itJugValue[\Xi_\val v]{\emptyset}{\val v}{\it M}$
and $\itJugComput[\Xi_{\comput u}]{ x:\it M }{\comput u}{\it F}$,
where $\refrel{\Pi_\val v}{\Xi_\val v}$ and $\refrel{\Pi_\comput u}{\Xi_\comput u}$.
Thus, let $\Xi$ be the following derivation :
\begin{equation*}
\begin{prooftree}
        \hypo{\itJugValue[\Xi_\val v]{\emptyset}{\val v}{\it M}}
        \typeRuleRet{\itJugComput{\emptyset}{\ret{\val v}}{\itEffectReturn{\it M}}}
        \infer0{\itEffectReturn{\it M} \replaceLeaf{\it M \rightarrow \it F} \it F}
        \hypo{\itJugComput[\Xi_{\comput u}]{ x : \it M}{\comput u}{\it F}}
    \typeRuleLetin{\itJugComput{\itJugComput[\Xi']{\emptyset}{\comput u\sub{x}{\val v}}{\it F}}{\letin{x}{\ret{\val v}}{\comput u}}{\it F}}
\end{prooftree}
\end{equation*}
Finally, we can see that $\refrel{\Pi}{\Xi}$.

\paragraph*{Case $\reductLetEff$.}
Let $\comput t, \comput t' \in \setComput$ such that $\comput t'$ is
typed and $\comput t \reductArr_\reductLetEff \comput t'$. By
definition, $\comput t \coloneqq
\letin{x}{\effect{\val v}{y}{\comput s}}{\comput u}$ and
$\comput t' \coloneqq
\effect{\val v}{y}{\letin{x}{\comput s}{\comput t}}$. There exists a derivation
$\itJugComput[\Xi']{\emptyset}{\effect{\val v}{y}{\letin{x}{\comput s}{\comput t}}}{\it E}$
for some type $\it E$. By
case analysis, we deduce that $\Xi'$ is necessarily of the following
form:

\begin{equation*}
\begin{prooftree}
    \hypo{\itJugValue[\Xi_{\val v}]{\emptyset}{\val v}{\it M}}
    \hypo{A_i}
    \delims{\left(}{\right)_{\rangeI}}
\typeRuleEff[2]{
    \itJugComput
        {\emptyset}
        {\effect{\val v}{y}{\letin{x}{\comput s}{\comput u}}}
        {\itEffect
                {\sigma}{\it M}
                {\it N_i \rightarrow \it E_i}_{\rangeI}}}
\end{prooftree}
\end{equation*}
\begin{equation*}
    (A_i) \coloneqq
    \begin{prooftree}
        \hypo{\itJugComput[\Xi^i_{\comput s}]{ y : \it N_i}{\comput s}{\it G_i}}
    \hypo{B_i}
    \delims{\left(}{\right)_{\rangeI}}

        \hypo{(\itJugComput[\Xi_k^i]{x : \it M_k^i}{\comput u}{\it F_k^i})_{\rangeK_i}}
    \typeRuleLetin{
        \itJugComput
            { y : \it N_i}
            {\letin{x}{\comput s}{\comput u}}
            {\it E_i}}
    \end{prooftree}
\end{equation*}
\begin{equation*}
    (B_i) \coloneqq
    \begin{prooftree}
       \hypo{\varphi_i}
        \ellipsis{}{
            \it G_i
            \replaceLeaf{\it M_k^i \rightarrow \it F_k^i}[\rangeK_i]
            \it E_i}
    \end{prooftree}
\end{equation*}
Since $y$ is not a free variable of $\comput u$, by context simplification (\Cref{rem:ContextSimplification}) on each $\Xi_k^i$ one can derive $\itJugComput[\Xi_k^i]{x : \it M_k^i}{\comput u}{\it F_k^i}$.
Furthermore, it exists a $\stJugComput[\Pi']{\emptyset}{\effect{\val v}{y}{\letin{x}{\comput s}{\comput u}}}{\stEffect{\st M'}}$.
Let $\Xi$ be the following typing derivation:
\begin{equation*}
    \begin{prooftree}
            \hypo{(A)}
            \hypo{(B)}
            \hypo{(\itJugComput[\Xi_k^i]
                { x : \it M_k^i}
                {\comput u}{\it F_k^i})_{\rangeK_i, \rangeI}}
        \typeRuleLetin{
            \itJugComput
                {\emptyset}
                {\letin{x}{\effect{\val v}{y}{\comput s}}{\comput u}}
                {\itEffect
                    {\sigma}{\it M}
                    {\it N_i \rightarrow \it E_i}_{\rangeI}}
        }
    \end{prooftree}
\end{equation*}
\begin{equation*}
    (A) \coloneqq
    \begin{prooftree}
            \hypo{
                    \itJugValue[\Xi_{\val v}]
                        {\emptyset}
                        {\val v}
                        {\it M}
                }
                \hypo{(
                    \itJugComput[\Xi^i_{\comput s}]
                        { y : \it N_i}
                        {\comput s}
                        {\it G_i}
                )_{\rangeI}}
            \typeRuleEff[2]{
                \refJugComput
                    {\emptyset}
                    {\effect{\val v}{y}{\comput s}}
                    {\itEffect{\sigma}{\it M}{\it N_i \rightarrow \it G_i}_{\rangeI}}
            }
    \end{prooftree}
\end{equation*}
\begin{equation*}
    (B) \coloneqq
    \begin{prooftree}
            \hypo{\varphi_i}
            \ellipsis{}{
                \it G_i
                \replaceLeaf{\it N_i \rightarrow \it F_k^i}[\rangeK_i]
                \it E_i}
            \delims{\left(}{\right)_{\rangeI}}
        \infer1{
            \itEffect
                    {\sigma}{\it M}
                    {\it N_i \rightarrow \it G_i}_{\rangeI}
        \replaceLeaf{
            \it M_k^i \rightarrow 
            \it F_k^i}[\rangeK_i, \rangeI]
            \itEffect
                {\sigma}{\it M}
                {\it N_i \rightarrow \it E_i}_{\rangeI}}
    \end{prooftree}
\end{equation*}
By looking at the definition of this reduction case in \Cref{def:HEPCFReduction},  from $\refrel{\Pi'}{\Xi'}$ we obtain easily that $\refrel{\Pi}{\Xi}$. In particular, notice that the subderivation $\Pi_\comput u$ of $\Pi'$ is such that $\refrel{\Pi_\comput u^{(y::\st N)}}{\Xi^i_k}$ for each $\rangeI$, $\rangeK$. By \Cref{prop:HEPCFClosedRefinement}, we obtain $\refrel{\Pi_\comput u}{\Xi^i_k}$.

\paragraph*{Case $\reductHdlRet$.}
Let $\comput t, \comput t' \in \setComput$ such that $\comput t'$ is
typed and $\comput t \reductArr_\reductHdlRet \comput t'$. By
definition, $\comput t \coloneqq
\handle{\handler h}{\ret{\val v}}$ and $\comput t'
\coloneqq \comput u\sub{x}{\val v}$ with
$\handler h = \{\retClause{x}{\comput u}\} \cup \handler h'$. 

Since $\comput t$ is simply typable, by case analysis its derivation $\Pi$ has the following form. 
\begin{equation*}
    \hspace{-1cm}
    \begin{prooftree}
            \hypo{\stJugComput[\Pi_{\comput u}]{ x :: \st M}{\comput u}{\st U}}
            \hypo{...}
            \typeRuleHandler[2]{\stJugHandler[\Pi_\handler{h}]{\emptyset}{\{\retClause{x}{\comput u}\}\cup \handler h'}{\stEffect{\st M} \Rightarrow \st U}}
            \hypo{\stJugValue[\Pi_\val v]{\emptyset}{\val v}{\st M}}
            \typeRuleRet{\stJugComput{\emptyset}{\ret{\val v}}{\stEffect{\st M}}}
        \typeRuleHandle{\stJugComput{\emptyset}{\handle{\{\retClause{x}{\comput u}\}\cup \handler h' }{\ret{\val v}}}{\st U}}
    \end{prooftree}
\end{equation*}
There exist two derivations 
$\itJugComput[\Xi']{\emptyset}{\comput u\sub{x}{\val v}}{\it E}$ and $\stJugComput[\Pi']{\emptyset}{\comput u\sub{x}{\val v}}{\st U}$ such that $\refrel{\Pi'}{\Xi'}$. Since by hypothesis $\Pi\rightsquigarrow \Pi'$, by definition of $\rightsquigarrow$ is clear that $\Pi'=\Pi_\comput u\sub{x}{\Pi_\val v}$.
By anti-substitution lemma (\Cref{lem:HEPCFantiSubstLemma-ref}), there exist
a type $\it M$ and derivations 
$\itJugValue[\Xi_\val v]{\emptyset}{\val v}{\it M}$
and $\itJugComput[\Xi_{\comput u}]{x : \it M}{\comput u}{\it E}$,
where $\refrel{\Pi_\comput u}{\Xi_\comput u}$ and $\refrel{\Pi_\val v}{\Xi_\val v}$.
Since $\comput t$ is not a clash, by \Cref{lem:CompatImpliesClashFree}
$\{\retClause{x}{\comput u}\} \handlerCompat \handler h'$.
Therefore, let $\Xi$ be the following derivation :
\begin{equation*}
    \begin{prooftree}
            \hypo{\itJugComput[\Xi_{\comput u}]{ x : \it M}{\comput u}{\it E}}
            \hypo{\{\retClause{x}{\comput u}\} \handlerCompat \handler h'}
            \typeRuleHandler[2]{\itJugHandler{\emptyset}{\{\retClause{x}{\comput u}\}\cup \handler h' }{\itEffectReturn{\it M} \Rightarrow \it E}}
            \hypo{\itJugValue[\Xi_\val v]{\emptyset}{\val v}{\it M}}
            \typeRuleRet{\itJugComput{\emptyset}{\ret{\val v}}{\itEffectReturn{\it M}}}
        \typeRuleHandle{\itJugComput{\emptyset}{\handle{\{\retClause{x}{\comput u}\} \cup \handler h'}{\ret{\val v}}}{\it E}}
    \end{prooftree}
\end{equation*}
Finally, is easy to notice that $\refrel{\Pi}{\Xi}$.

\paragraph*{Case $\reductHdlEff$.}
Let $\comput t, \comput t' \in \setComput$ such that $\comput t'$ is
typed and $\comput t \reductArr_\reductHdlEff \comput t'$. 

By definition,
$\comput t \coloneqq
\handle{\handler h}{\effect[\sigma_k]{\val v}{x}{\comput u}}$ and
$\comput t' \coloneqq
\comput t_k\sub{x}{\val v}\sub{r}{\abs{y}{\handle{\handler h}{\comput u}}}$
with $\effectClause[\sigma_k]{x}{r}{\comput t_k} \in
\handler h$. There exists two derivations
$\itJugComput[\Xi']{\emptyset}{\comput t_k\sub{x}{\val v}\sub{r}{\abs{y}{\handle{\handler h}{\comput u}}}}{\it E}$ and $\stJugComput[\Pi']{\emptyset}{\comput t_k\sub{x}{\val v}\sub{r}{\abs{y}{\handle{\handler h}{\comput u}}}}{\st U}$ such that $\refrel{\Pi'}{\Xi'}$. 
The derivation $\Pi$ has the following form. 
\begin{equation*}
     \begin{prooftree}
    \hypo{(
        \stJugComput[\Pi_i]
            { x :: \st M_j';r :: \st N_j \rightarrow \st U}
            {\comput t_j}
            {\st U})_{1\leq j\leq n}
            }
            \hypo{
                (A)
            }
            %
            %
            \typeRuleHandler[2]{\stJugHandler[\Pi_\handler h]
                {\emptyset}
                {\{\retClause{y}{\comput t}\} \cup\{\effectClause[\sigma_{\spaceNegII j}]{x}{r}{\comput t_j} \vsep 1 \leq j \leq n\}}
                {\stEffect{\st M} \Rightarrow \st U}
            }

        \hypo{
            \begin{matrix}
                \stJug[\Pi_\val v]{\emptyset}{\val v}{\st M'_k}\\
                \stJug[\Pi_\comput u]{y::\st N_k}{\comput u}{\stEffect{\st M}}\\
                \effectSpec[\sigma_k]{\st M'_k}{\st N_k} \in \effectCtxt E
            \end{matrix}       
        }   

        \typeRuleEff[1]{\stJugComput
            {\emptyset}{\effect[\sigma_k]{\val v}{y}{\comput u}}
            {\stEffect{\st M}}}
        \typeRuleHandle[2]{
            \stJugComput
                {\emptyset}
                {\handle{\{\effectClause[\sigma_k]{x}{r}{\comput t_k}\} \cup \handler h}{\effect[\sigma_k]{\val v}{y}{\comput u}}}
                {\st U}
        }
    \end{prooftree}
\end{equation*}
\begin{equation*}
        (A):= \; \begin{matrix}
            \stJugComput
            {\Delta; y :: \st M}
            {\comput t}
            {\st U}\\
            \{\sigma_j:{\st M'_j}\rightsquigarrow{\st N_j}\in \st E \vsep 1 \leq j \leq n\}
                \end{matrix}
\end{equation*}
Since by hypothesis $\Pi\rightsquigarrow \Pi'$, by definition of $\rightsquigarrow$ we have that $\Pi'=\stJugComput[\Pi_k\sub{x}{\Pi_\val v}\sub{r}{\Pi_{\val{\lambda y}}}]{\emptyset }{\comput t_k \sub{x}{\val v}\sub{r}{\abs{y}{\handle{\handler h}{\comput u}}}}{\st U}$, with $\Pi_\val{\lambda y}$ being the following derivation, as defined in \Cref{def:HEPCFReduction}.

\begin{equation*}
    \begin{prooftree}

    \hypo{
        \stJugHandler[\Pi_\handler h^{(y::\st N_k)}]{y::\st N_k}{\handler h}{\stEffect{\st M}\Rightarrow \st U}
    }
    \hypo{
        \stJugComput[\Pi_\comput u]{y::\st N_k}{\comput u}{\stEffect{\st M}}
    }
    \typeRuleHandle{
        \stJugComput{ y::\st N_k}{\handle{\handler h}{\comput u}}{\st U}
    }
    \typeRuleAbs{
        \stJugValue[\Pi_\val{\lambda y}]{\emptyset}{\lambda y.\handle{\handler h}{\comput u}}{\st N_k\Rightarrow \st U}
    }
\end{prooftree}
\end{equation*}

We also have $\stJug[\Pi_k\sub{x}{\Pi_\val{v}}]{\Delta;r::\st N_k\rightarrow \st U}{\comput t_k \sub{x}{\val v}}{\st U}$.
We can then apply the anti-substitution lemma (\Cref{lem:HEPCFantiSubstLemma-ref}) on $\Xi'$ with $\Pi_k\sub{x}{\Pi_\val v}$ and 
$\Pi_\val{\lambda y}$, therefore there exist a type $\it L$ and derivations 
$\itJugValue[\Xi_{\val {\lambda y}}]{\emptyset}{\abs{y}{\handle{\handler h}{\comput u}}}{\it L}$ and 
$\itJugComput[\Xi_{\comput s}]{ r : \it L}{\comput t_k\sub{x}{\val v}}{\it E}$,
where $\refrel{\Pi_k\sub{x}{\Pi_\val v}}{\Xi_{\comput s}}$ and $\refrel{\Pi_\val{\lambda y}}{\Xi_{\val {\lambda y}}}$.
By case analysis, we deduce that $\Xi_{\val {\lambda y}}$ is necessarily of the following
form:

\begin{equation*}
    \begin{prooftree}
                \hypo{
                    \itJugHandler[\Xi_{\handler h}^i]
                        {y:\it N_i^\handler h}
                        {\handler h}
                        {\it F_i
                            \Rightarrow
                        \it G_i}}
                \hypo{
                    \itJugComput[\Xi_{\comput u}^i]
                        { y : \it N_i}
                        {\comput u}
                        {\it F_i}}
            \typeRuleHandle{
                \itJugComput
                    { y : \it N_i}
                    {\handle{\handler h}{\comput u}}
                    {\it G_i}
            }
            \delims{\left(}{\right)_{\rangeI}}
        \typeRuleAbs{
            \itJugValue
                {\emptyset}
                {\abs{y}{\handle{\handler h}{\comput u}}}
                {\itSet*{\it N_i \rightarrow \it G_i}_{\rangeI}}
        }
    \end{prooftree}
\end{equation*}

where $\it L = \itSet{\it N_i \rightarrow \it G_i}_{\rangeI}$. Therefore we obtain $\refrel{\Pi_\handler h^{(y::\st N')}}{\Xi^i_\handler h}$ and $\refrel{\Pi_\comput u}{\Xi^i_\comput u}$ for each $\rangeI$. By \Cref{prop:HEPCFClosedRefinement}, we have that $\refrel{\Pi_\handler h}{\Xi^i_\handler h}$ for each $\rangeI$. Since $y$ is not a free variable of $\handler h$, by context simplification (\Cref{rem:ContextSimplification}) on $\Xi_{\handler h}^i$ one can derive 
$\itJugHandler[\Xi^i_{\handler h}]
    {\emptyset}
    {\handler h}
    {\it F_i \Rightarrow \it G_i}$.
We can now apply again the anti-substitution lemma (\Cref{lem:HEPCFantiSubstLemma-ref}) on $\Xi_{\comput s}$ with $\Pi_k$ and $\Pi_\val v$, therefore there exist a type $\it M$ and derivations 
$\itJugValue[\Xi_\val v]{ \emptyset}
    {\val v}{\it M}$
and $\itJugComput[\Xi_{\comput t}]
    { r : \itSet*{\it N_i \rightarrow \it G_i}_{\rangeI}; x : \it M}
    {\comput t_k}{\it E}$,
where $\refrel{\Pi_k}{\Xi_\comput t}$ and $\refrel{\Pi_\val v}{\Xi_\val v}$. 
Since $\comput t$ is not a clash, by \Cref{lem:CompatImpliesClashFree}
$\{\effectClause[\sigma_k]{x}{r}{\comput t_k}\} \handlerCompat \handler h$.
Let $\Xi$ be the following derivation :

\begin{equation*}
    \begin{prooftree}
            \hypo{(B)}
                \hypo{\itJugValue[\Xi_{\val v}]
                    {\emptyset}
                    {\val v}{\it M}}
                \hypo{\left(
                    \itJugComput[\Xi_{\comput u}^i]
                        {y : \it N_i}
                        {\comput u}
                        {\it F_i}
                \right)_{\rangeI}}
            \typeRuleEff{
                \itJugComput
                    { \emptyset}
                    {\effect{\val v}{y}{\comput u}}
                    {\itEffect
                        {\sigma}{\it M}
                        {\it N_i \rightarrow \it F_i}_{\rangeI}}}
        \typeRuleHandle{
            \itJugComput
                {\emptyset}
                {\handle
                    {\{\effectClause[\sigma]{x}{r}{\comput t_k}\} \cup \handler h}
                    {\effect{\val v}{y}{\comput u}}}
                {\it E}}
    \end{prooftree}
\end{equation*}

\begin{equation*}
    \hspace{-1cm}
    (B) \;\coloneqq\;
    \begin{prooftree}
            \hypo{
                \itJugComput[\Xi_{\comput t}]
                    {   x : \it M ;
                        r : \itSet{\it N_i \rightarrow \it G_i}_{\rangeI}}
                    {\comput t_k}
                    {\it E}}

            \hypo{
                \begin{matrix}
                \{\effectClause[\sigma_k]{x}{r}{\comput t_k}\} \handlerCompat \handler h\\
                (
                \itJugHandler[\Xi_{\handler h}^i]
                    {\emptyset}
                    {\handler h}
                    {\it F_i \Rightarrow \it G_i}
                    )_{\rangeI}
            \end{matrix}  
            }
        \typeRuleHandler{
            \itJugHandler[\Xi_\handler h]
                {\emptyset }
                {\{\effectClause[\sigma_k]{x}{r}{\comput t_k}\} \cup \handler h}
                {\itEffect{\sigma_k}{\it M}
                        {\it N_i \rightarrow \it F_i}_{\rangeI}
                \Rightarrow \it E}}
    \end{prooftree}
\end{equation*}
To conclude, we have to prove that $\Xi_\handler h$ refines
\[
\begin{prooftree}
\hypo{(\stJugComput[\Pi_i]
            { x :: \st M_j';r :: \st N_j \rightarrow \st U}
            {\comput t_j}
            {\st U})_{1\leq j\leq n}
            }
            \hypo{
                \begin{matrix}
            \stJugComput
            { y :: \st M}
            {\comput t}
            {\st U}\\
            \{\sigma_j:{\st M'_j}\rightsquigarrow{\st N_j}\in \st E \vsep 1 \leq j \leq n\}
                \end{matrix}
            }
            %
            %
            \typeRuleHandler[2]{\stJugHandler[\Pi_\handler h]
                {\emptyset}
                {\{\retClause{y}{\comput t}\} \cup\{\effectClause[\sigma_{\spaceNegII j}]{x}{r}{\comput t_j} \vsep 1 \leq j \leq n\}}
                {\stEffect{\st M} \Rightarrow \st U}
            }
            \end{prooftree}
\]
We have $\refrel{\Pi_k}{\Xi_\comput t}$ and $\refrel{\Pi_\handler h}{\Xi^i_{\handler h}}$ for each $\rangeI$, so $\refrel{\Pi_\handler h}{\Xi_\handler h}$.
One can conclude this case by noticing that 
$\itJugComput[\Xi]{\Gamma}{t}{\it E}$ and $\refrel{\Pi}{\Xi}$.